\documentclass[reqno]{article}

\usepackage[utf8]{inputenc}

\usepackage{amsmath, amsthm}
\usepackage{amsfonts}
\usepackage{amssymb}
\usepackage{amsxtra}
\usepackage{mathrsfs}
\usepackage{mathtools}

\usepackage{bbm}
\usepackage{bm}
\usepackage{color}
\usepackage{float}
\usepackage{graphicx}
\usepackage{caption}
\usepackage{subfigure}
\usepackage{cancel}
\usepackage{comment}
\usepackage{latexsym}
\usepackage{soul}
\usepackage{todonotes}
\usepackage{accents}

\usepackage{authblk}
\usepackage{natbib}

\usepackage{tikz}
\usetikzlibrary{calc,arrows.meta, decorations.pathreplacing, calligraphy}

\usepackage{slashed}
\usepackage{leftidx}
\usepackage{tensor}
\usepackage{epstopdf}

\usepackage{enumerate}
\usepackage{enumitem}

\usepackage{hyperref}

\newtheorem{theorem}{Theorem}[section]
\newtheorem{lemma}[theorem]{Lemma}
\newtheorem{corollary}[theorem]{Corollary}
\newtheorem{proposition}[theorem]{Proposition}

\newtheorem{definition}[theorem]{Definition}

\newtheorem{convention}[theorem]{Convention}

\theoremstyle{remark}
\newtheorem{remark}[theorem]{Remark}
\newcommand{\mr}{\mathring}
\numberwithin{equation}{section}

\usepackage[margin = 1in]{geometry}

\makeatletter

\newcommand{\Rmnum}[1]{\expandafter\@slowromancap\romannumeral #1@}
\makeatother

\begin{document}

\title{Nonlinear stability of continuously self-similar naked singularities for the Einstein-scalar field equations \Rmnum{2}: linearized stability}

\author[1]{Jaydeep Singh}
\author[2]{Weihao Zheng\thanks{wz344@math.rutgers.edu}}
\affil[1]{\small Department of Mathematics, Princeton University, Washington Road, Princeton, NJ, USA}
\affil[2]{\small Department of Mathematics, Rutgers University, Hill Center, 110 Frelinghuysen Road, Piscataway, NJ, USA}
\date{\today}

\maketitle
\begin{abstract}
This is the second part of a series of papers proving the nonlinear stability of a one-parameter family of continuous self-similar $C^{1,\alpha}$ naked singularity solutions, with $0<\alpha\ll1$, to the spherically symmetric Einstein-scalar field equations. These solutions were constructed by Christodoulou and are known to be unstable under sufficiently rough perturbations due to the blue-shift instability mechanism. In complete contrast to the previous instability results, we establish the linearized stability for those naked singularity spacetimes under perturbations of the same regularity as the background, revealing the central role of regularity in determining the strength of the blue-shift instability mechanism, and showing that it is not triggered at the regularity level of the background spacetime.

The linear analysis carried out in this paper provides the foundation for the nonlinear stability result established in the companion paper [W. Zheng, Nonlinear stability of the continuous self-similar naked singularities for the Einstein-scalar field equations \Rmnum{1}: main results]. Together with that companion paper, this yields the nonlinear stability of these continuously self-similar naked singularities.

\end{abstract}
\tableofcontents
\newpage
\section{Introduction}
In this paper and the companion paper~\cite{zhenglinear}, we prove the nonlinear stability of the $k$-self-similar naked singularity solution to the Einstein-scalar field equations, which are a geometric system of PDEs describing a Lorentzian manifold $(\mathcal{M},g)$ coupled to a real scalar field $\phi: \mathcal{M}\rightarrow\mathbb{R}$
\begin{align}
    &Ric[g]_{\mu\nu} = \partial_{\mu}\phi\partial_{\nu}\phi,\label{eq: ESf 1}\\&
    \Box_{g}\phi = 0,\label{eq: ESf 2}
\end{align}
where Ric is the standard Ricci curvature of $(\mathcal{M},g)$ and $\Box_{g}$ is the covariant wave operator on $(\mathcal{M},g)$. This paper focuses on the linearized stability of these $k$-self-similar naked singularity solutions.

The family of $k$-self-similar naked singularity solutions $(\mathcal{M}_{k},g_{k},\phi_{k})$ with $0<k^{2}<1/3$, first constructed by Christodoulou~\cite{chris94}, is a family of spherically symmetric spacetimes solving \eqref{eq: ESf 1}--\eqref{eq: ESf 2} that contain a curvature singularity without an event horizon and possess a complete future null infinity. Notably, the regularity class of the initial data of $(\mathcal{M}_{k},g_{k},\phi_{k})$ is not smooth. To be more precise, let $\mathcal{O}$ be the naked singularity in $\mathcal{M}_{k}$, $\underline{C}_{o}^{-}$ be the past light cone emanating from the singularity, and $C_{out}$ the initial null hypersurface emanating from the center $\{r = 0\}$. The causal future of $\underline{C}_{o}^{-}$ is called the exterior region of the naked singularity spacetime $\mathcal{M}_{k}^{(ex)}$, while the causal past of $\underline{C}_{o}^{-}$ is called the interior region $\mathcal{M}_{k}^{(in)}$. Then the initial data of $(\mathcal{M}_{k},g_{k},\phi_{k})$ is smooth on $C_{out}\backslash \underline{C}_{o}^{-}$ and has $C^{1,k^{2}/(1-k^{2})}$-Hölder regularity at the intersecting sphere $C_{out}\cap \underline{C}_{o}^{-}$; see Figure~\ref{fig:penrose}.

These $k$-self-similar naked singularity spacetimes $(\mathcal{M}_{k},g_{k})$ have been proved to be unstable to trapped surface formation under initial perturbations that are rougher than the background~\cite{chris99,liuli,liuli_outsidesymm,li2025interior,an_highcodim}, while they are stable for spherically symmetric perturbations supported in the exterior region of $C_{out}$ at the same regularity level as the background~\cite{singh1,singh2,singhphdtheis}. To obtain a complete understanding of their stability properties, it is therefore sufficient to analyze how the interior region of $(\mathcal{M}_{k},g_{k},\phi_{k})$ responds to initial perturbations that are nontrivial within the interior of $C_{out}$\footnote{In fact, the exterior stability results in~\cite{singh2,singhphdtheis} apply to a wider class of naked singularity spacetimes that are asymptotically $k$-self-similar. As a corollary of~\cite{singh2,singhphdtheis}, for general perturbations, the asymptotic stability of the interior region implies asymptotic stability of the entire spacetime.}.

It has been shown that solutions to the linear wave equation on the interior of $(\mathcal{M}_{k},g_{k},\phi_{k})$ are stable under perturbations of the same regularity as the background~\cite{singh2}.  As a next step toward nonlinear stability for general perturbations at this regularity level, we therefore consider the linearized spherically symmetric Einstein-scalar field equations in the present work.

Before we state the rough version of our main theorem, we define the function space we shall consider in this paper. We parametrize the interior region of $C_{out}$ by $v\in[-1,0]$, where $v = -1$ corresponds to $r = 0$, while $v = 0$ corresponds to $C_{out}\cap \underline{C}_{o}^{-}$. We define the following Banach space capturing the regularity of $\phi_{k}|_{C_{out}\cap\mathcal{M}_{k}^{(in)}}$.
\begin{definition}
    Let $I=[-1,0]$ and $\delta\in(0,1)$. We define \begin{equation*}
        \mathcal{C}_{N}^{1/(1-k^{2}),\delta}: = \left\{f(v): I\rightarrow\mathbb{R}|f\in C^{N}(I\backslash\{0\})\cap C^{1}(I),\ \vert v\vert^{j-1/(1-k^{2})}\frac{d^{j}f}{dv^{j}}\in C^{0,\delta},\ 2\leq j\leq N\right\},
    \end{equation*}
    associated with the norm \begin{equation*}
        \Vert f\Vert_{\mathcal{C}_{N}^{1/(1-k^{2}),\delta}} = \sum_{i = 0}^{1}\left\Vert\frac{d^{i}f}{dv^{i}}\right\Vert_{C^{0,\delta}}+\sum_{i = 2}^{N}\left\Vert(-v)^{i-1/(1-k^{2})}\frac{d^{i}f}{dv^{i}}\right\Vert_{C^{0,\delta}},
    \end{equation*}
    where $C^{0,\delta}$ is the standard Hölder space.
\end{definition}
{Functions in the above Banach space are $C^{N}$ away from $v = 0$, while their second and higher-order derivatives are allowed to exhibit polynomial blow-up at $v = 0$. This localized singular behavior accurately captures the regularity of $\phi_{k}|_{C_{out}}\in\mathcal{C}_{N}^{1/(1-k^{2})}$ for any $N\geq 2$, and the space is therefore referred to as a localized Hölder space.}
\begin{theorem}[Rough version of Theorem~\ref{thm: linearized result at the threshold}]
\label{thm: rough version of the linearized result}
   Fix a non-zero value of $k$ sufficiently small. For any initial data of the scalar field in $\mathcal{C}_{N}^{1/(1-k^{2}),\delta}$ with $\delta\in(0,1-O(k))$ and $N\geq 5$, the solution to the spherically symmetric Einstein-scalar field equations linearized around $(\mathcal{M}_{k}^{(in)},g_{k},\phi_{k})$ contains a naked singularity and asymptotes to $(\mathcal{M}_{k},g_{k},\phi_{k})$ when approaching the singularity.
\end{theorem}
\paragraph{Difference between the linear wave equation on $(\mathcal{M}_{k},g_{k})$ and the linearized Einstein-scalar field equations}
Although stability has been established for the linear wave equation on $(\mathcal{M}_{k},g_{k})$ under initial perturbations in $\mathcal{C}_{N}^{1/(1-k^{2}),\delta}$~\cite{singh2}, this is insufficient to justify expectations of nonlinear stability. The argument in~\cite{singh2} crucially relies on the regularity of the background geometry $g_{k}$, which is one derivative more regular than the background scalar field $\phi_{k}$. However, the Einstein-scalar field equations linearized around $(\mathcal{M}_{k},g_{k},\phi_{k})$ capture both the contribution of the background scalar field $\phi_{k}$ and its coupling with the geometry $g_{k}$. The presence of these effects introduces substantial analytical difficulties; see Section~\ref{sec: previous works on the interior perturbations}, Section~\ref{sec: review on the linear wave}, and Section~\ref{sec: analysis on the linearized operator} for further discussions.

\paragraph{Nonlinear stability result in the companion paper~\cite{zhenglinear}}In the companion paper \cite{zhenglinear}, we use the linearized result established here to further prove the following nonlinear stability of $(\mathcal{M}_{k},g_{k},\phi_{k})$ under small perturbations of the same regularity as $\phi_{k}|_{C_{out}}$, \textbf{thereby falsifying the Weak Cosmic Censorship conjecture for initial data of this regularity} and showing the importance of the regularity class of the initial data in the formulation of the Weak Cosmic Censorship conjecture.
\begin{theorem}[\cite{zhenglinear}]
  Fix a non-zero value of $k$ sufficiently small. There exists a small open neighborhood $\mathcal{U}$ of $\phi_{k}|_{C_{out}\cap\mathcal{M}^{(in)}_{k}}$ in the topology of $\mathcal{C}_{N}^{1/(1-k^{2}),\delta}$ with $\delta\in(0,1-O(k))$ and $N\geq 6$, such that for any initial data of $\phi$ in $\mathcal{U}$, the maximal globally hyperbolic development of this initial data for the Einstein-scalar field equations \eqref{eq: ESf 1}-\eqref{eq: ESf 2} terminates at a naked singularity and has an incomplete future null infinity. Moreover, the spacetime asymptotes to $(\mathcal{M}_{k},g_{k},\phi_{k})$ when approaching the naked singularity.
\end{theorem}
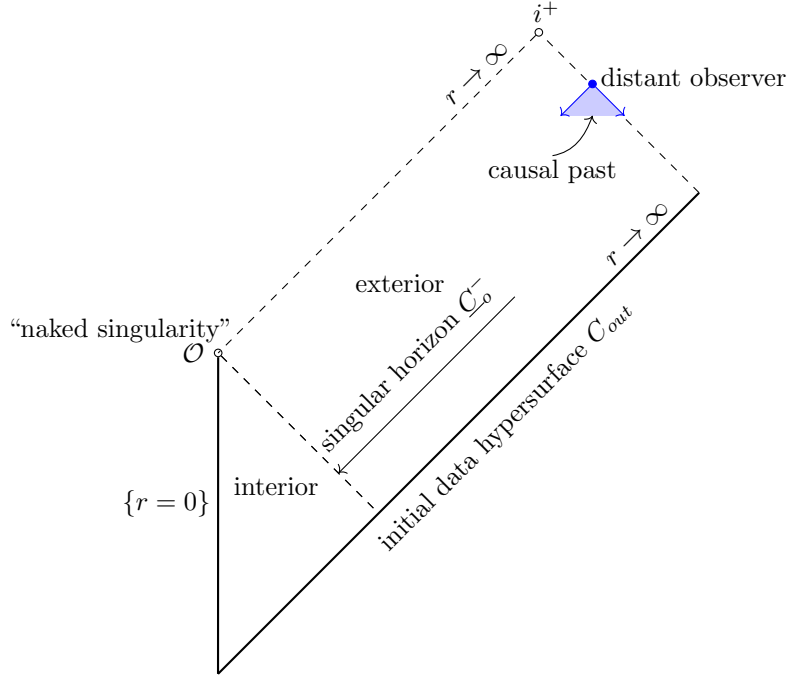
\begin{figure}[htbp]
\centering


\begin{tikzpicture}[scale=1.5]

\coordinate (n1) at (0,0);
\coordinate (n2) at (45:6);
\draw[thick] (n1) -- (n2);

\coordinate (n3) at ($(n2) + (135:2) $);
\coordinate (b1) at ($(n3) + (-135:2) $);
\node[circle, inner sep = 0, minimum size = .1cm, label = {[xshift=1mm, yshift=-.5mm]above:{$i^+$}}, draw] (node1) at (n3) {};
\draw[dashed] (n2) -- (node1);

\node[label = left:{$\{r=0\}$}] at ($(0,1.5) + (.1,0)$) {};

\coordinate (n8) at ($(node1) + (-135:4)$) ;
\draw[dashed] (node1) --+ (-135:4);

\node[circle, draw, inner sep = 0, minimum size = .1cm, label = left:{$\mathcal{O}$}] (Onode) at (n8) {};

\node[label = left:{``naked singularity"}] (Bnode) at ($(n8)+(0,.2) + (.3,0)$) {};

\draw[thick] (n1) -- (Onode);
\draw[dashed] (Onode) -- ($(Onode) + (-45:2)$);
\coordinate (Hend) at ($(Onode)+(-45:2)$);
\draw[dashed] (Onode) -- (Hend);

\coordinate (Hmid) at ($(Onode)!0.75!(Hend)$);

\draw[->]
($(Hmid)+(45:2.2)$)
-- node[midway, above=2pt, sloped] {singular horizon $\underline{C}_{o}^{-}$}
(Hmid);
\coordinate (n7) at (45:4) ;
\node[label = {[rotate = 45]below:{initial data hypersurface $C_{out}$}}] at (45:3.4) {};

\coordinate (n95) at ($(0,1.9) + (45:4.67) $ );
\coordinate (n10) at ($(n95) + (-45:.4) $);
\coordinate (n11) at ($(n95) + (-135:.4) $);
\fill[blue!20] ($(0,1.9) + (45:4.67) $ ) -- (n10) -- (n11);
\node[circle, blue, fill, draw, inner sep = 0, minimum size = .1cm, label = {[yshift = 1mm, xshift = -.7mm]right:{distant observer}}] (obs) at (n95) {};

\draw[->, blue] (obs) -- (n10) ;
\draw[->, blue] (obs) -- (n11) ;

\coordinate (n12) at ($(n95) + (0,-.28) $);
\coordinate (n13) at ($(n12) + (-135:.5) $ );

\draw[->] (n13) to[out=10, in = -110] (n12);

\node[label = {[yshift = 2mm]below:{causal past}}] at (n13) {};

\coordinate (b2) at ($(n2) + (-135:.7)  $ ) ;
\coordinate (b3) at ($(n3) + (-135:.7)  $ ) ;
\coordinate (b4) at ($(Onode) + (-45:1)  $ ) ;
\node[label = {[rotate = 45, yshift=-1mm]above:{ $ r\rightarrow \infty $ }}] at (b2) {};
\node[label = {[rotate = 45,yshift=-1mm]above:{ $ r\rightarrow \infty $ }}] at (b3) {};
\coordinate (p1) at (60:3);

\node at (65:3.8) {exterior};
\coordinate (p2) at (70:1.5);
\node[label = above:{interior}] at (p2){};
\end{tikzpicture}

\caption{Penrose diagram for naked singularity spacetimes}
\label{fig:penrose}
\end{figure}

For the rest of the introduction, we will briefly recall some key properties of the naked singularity spacetime $(\mathcal{M}_{k},g_{k},\phi_{k})$ in Section \ref{sec: review on the k-self-similar spacetime}, review the previous works on the (in)stability of the interior of $(\mathcal{M}_{k},g_{k},\phi_{k})$ in Section~\ref{sec: previous works on the interior perturbations}, and discuss the transition mechanisms for low-regularity instability and high-regularity stability in Section \ref{sec: blue-shift mechanism}. We refer the interested reader to the introduction of our companion paper \cite{zhenglinear} for a more detailed review of the previous results and some related works.

\subsection{Global double-null gauge and the $k$-self-similarity}
\label{sec: review on the k-self-similar spacetime}
Although the original construction of $(\mathcal{M}_{k},g_{k},\phi_{k})$ by Christodoulou \cite{chris94} was carried out in the so-called Bondi coordinates, in this paper, we will work with spherically symmetric global double-null coordinates $(u,v,\theta,\varphi)$ \begin{equation}
    g =-\frac{1}{2}\Omega^{2}(u,v)du\otimes dv-\frac{1}{2}\Omega^{2}(u,v)dv\otimes du+r^{2}(u,v)d\sigma^{2},\label{eq: double null gauge intro}
\end{equation}
where $\Omega^{2}$ is the time-lapse function, $r$ is the area radius of the orbit of the $SO(3)$ action, and $d\sigma^{2}$ is the standard sphere metric on $S^{2}$. Then the study of the spherically symmetric Einstein-scalar field equations reduces to the study of $(r,\Omega^{2},\phi)$ with suitable initial conditions. In \cite{chris94}, Christodoulou further assumed that the solution is $k$-self-similar, i.e., there exists a conformal Killing vector field such that \begin{equation*}
    \mathcal{L}_{K}g = 2g,\quad \mathcal{L}_{K}\phi = -k,
\end{equation*} 
where $\mathcal{L}$ is the standard Lie derivative. Up to the gauge freedom of the $(u,v)$ coordinates, one natural choice of $K$ seems to be $K = u\partial_{u}+v\partial_{v}$. However, such a gauge choice fails to make $(u,v,\theta,\varphi)$ a global double-null coordinate system; see the discussion in Section~\ref{sec: pre on k naked singularity sapcetime} for details. Therefore, we assume that $K = u\partial_{u}+(1-k^{2})v\partial_{v}$. Then the naked singularity $\mathcal{O}$ in $(\mathcal{M}_{k},g_{k},\phi_{k})$ corresponds to the vanishing point $(u,v) = (0,0)$ of $K$. The past light cone $\underline{C}_{o}^{-}$ emanating from $\mathcal{O}$ corresponds to $\{v= 0\}$. We assume the initial hypersurface corresponds to $\{u = -1\}$. Then the spherically symmetric $k$-self-similar solution $(r_{k},\Omega_{k}^{2},\phi_{k})$ to the Einstein-scalar field equations takes the form \begin{equation*}
r_{k}(u,v) = (-u)\mr{r}(z),\quad \Omega_{k}^{2}(u,v)= \mr{\Omega}^{2}(z),\quad \phi_{k}(u,v) = \mr{\phi}(z)-k\log(-u),
\end{equation*}
where $z = \frac{v}{(-u)^{1-k^{2}}}$ is called the self-similar coordinate. We can further fix the gauge choice of the center to be $z = -1$. Then on the initial hypersurface $\{u = -1\}$, the triple $(r_{k},\Omega_{k}^{2},\phi_{k})$ has the following regularity properties
\begin{equation}
\label{eq: regularity of the background intro}
\begin{aligned}
    &r_{k}|_{\{u = -1\}}\in C^{\infty}(C_{out}\backslash\{v = 0\})\cap C^{2,\frac{k^{2}}{1-k^{2}}}(C_{out}),\\&  \Omega_{k}^{2}|_{\{u = -1\}}\in C^{\infty}(C_{out}\backslash\{v = 0\})\cap C^{1,\frac{k^{2}}{1-k^{2}}}(C_{out}),\\&  \phi_{k}|_{\{u = -1\}}\in C^{\infty}(C_{out}\backslash\{v = 0\})\cap C^{1,\frac{k^{2}}{1-k^{2}}}(C_{out}).
\end{aligned}
\end{equation}
One should note that the exact $k$-self-similar solution to the Einstein-scalar field equations is not asymptotically flat. In \cite{chris94}, Christodoulou truncated the exact $k$-self-similar spacetime in the far-away region $r\gg1$ to make it asymptotically flat. Nonetheless, we still call the asymptotically flat naked singularity spacetimes in \cite{chris94} the $k$-self-similar naked singularity spacetimes.

\subsection{Previous works on the interior perturbations}
\label{sec: previous works on the interior perturbations}
The first instability result for the $k$-self-similar naked singularity spacetime $(\mathcal{M}_{k},g_{k},\phi_{k})$ was proved in the pioneering work of Christodoulou \cite{chris99}. Subsequently, the stability and instability of the $k$-self-similar naked singularities under perturbations supported in the exterior region have been extensively studied~\cite{liuli,liuli_outsidesymm,an_highcodim,singh1,singh2}. More precisely, the works \cite{liuli,singh2} proved the instability of $(\mathcal{M}_{k},g_{k},\phi_{k})$ to trapped surface formation under exterior spherical perturbations of regularity $C^{1,\alpha}$ with $\alpha<k^{2}/(1-k^{2})$; the work \cite{singh2} established stability under exterior spherical perturbations of regularity $C^{1,\alpha}$ with $\alpha\geq k^{2}/(1-k^{2})$; and \cite{liuli_outsidesymm,an_highcodim} proved instability under exterior perturbations that are non-spherically symmetric and rougher than the background spacetime. See Section 1.3 of the companion paper \cite{zhenglinear} for a more detailed review of these results. 

Moreover, the mechanism for those low-regularity instability results~\cite{chris99,liuli,liuli_outsidesymm,an_highcodim,singh2} is called the \textit{blue-shift instability mechanism}, while the mechanism for those stability results~\cite{singh1,singh2} is called the \textit{high-regularity stabilizing effect}; see Section \ref{sec: blue-shift mechanism} for a more detailed discussion of these mechanisms. Motivated by the transition from instability to stability as one increases the regularity of the initial perturbations, we refer to the regularity of $\phi_{k}|_{C_{out}}$ as the \textit{threshold regularity}.

As mentioned above, to fully understand the stability aspect of the $k$-self-similar naked singularity, one must consider general perturbations rather than the initial perturbations supported in the exterior region considered in \cite{singh1,singh2}. Since the approach in \cite{singh1,singh2} for the exterior perturbations is robust enough to be generalized to any naked singularity spacetime that is asymptotically $k$-self-similar, the key difficulty in studying (in)stability under general perturbations lies in the study of the interior region under such perturbations.

The study of the (in)stability of the $k$-self-similar naked singularity spacetime was initiated by the first author in \cite{singh2}, where the linear wave equation \begin{equation}
    r_{k}\Box_{g_{k}}\phi=r_{k}\partial_{u}\partial_{v}\phi+\partial_{u}r_{k}\partial_{v}\phi+\partial_{v}r_{k}\partial_{u}\phi = 0\label{eq: linear wave equation lit review}
\end{equation}
on $(\mathcal{M}_{k},g_{k})$ is considered for $0<k^{2}\ll1$. At the level of the linear wave equation, stability is understood in terms of an asymptotic rate relative to the self-similar rate when approaching the singularity $\mathcal{O}$. More precisely, the scalar field $\phi_{k}$ in the $k$-self-similar naked singularity spacetime has the following bounds\footnote{To avoid the confusion from the gauge choice, we state the bound for the gauge invariant quantities $\frac{1}{\partial_{u}r_{k}}\partial_{u}\phi_{k}$ and $\frac{1}{\partial_{v}r_{k}}\partial_{v}\phi_{k}$.} \begin{equation}
    \left\vert\frac{1}{\partial_{u}r_{k}}\partial_{u}\phi_{k}\right\vert\lesssim \left\vert u\right\vert^{-1},\quad \left\vert\frac{1}{\partial_{v}r_{k}}\partial_{v}\phi_{k}\right\vert\lesssim \left\vert u\right\vert^{-1}.\label{eq: self-similar rate intro}
\end{equation}
We refer to Section \ref{sec: pre on k naked singularity sapcetime} for a detailed introduction to the $k$-self-similar spacetime.

The solutions to the linear wave equation $\Box_{g_{k}}\phi = 0$ are said to be stable if the upper bounds of $\partial_{u}\phi/\partial_{u}r_{k}$ and $\partial_{v}\phi/\partial_{v}r_{k}$ are comparable to, or polynomially faster than the corresponding self-similar rates \eqref{eq: self-similar rate intro}, and unstable if they have polynomially slower rates than \eqref{eq: self-similar rate intro}. It is shown in \cite{singh2} that solutions arising from interior perturbations of regularity below the threshold exhibit instability, whereas perturbations with regularity at or above the threshold lead to stability. This instability aspect was subsequently established at the nonlinear level in \cite{li2025interior}, where Li proved that, for interior perturbations below the threshold, solutions to the Einstein-scalar field equations contain a trapped surface.

However, the stability result for the wave equation on the $k$-self-similar naked singularity interior does not suffice to justify the expectation of linearized stability for the following two reasons:
\begin{itemize}
    \item As seen from \eqref{eq: regularity of the background intro} and \eqref{eq: linear wave equation lit review}, the spherically symmetric linear wave equation depends only on the background area radius $r_{k}$, which is $C^{2,k^{2}/(1-k^{2})}$ and hence one derivative more regular than the background scalar field $\phi_{k}$. Since the analysis in~\cite{singh2} crucially relies on the $C^{1,k^{2}/(1-k^{2})}$ regularity of $\partial r_{k}$, the linear wave equation avoids the lowest-regularity terms coming from the background scalar field. In contrast, the linearized Einstein--scalar field equations contain the contribution of $\phi_k$, making the stability problem more delicate. More precisely, if one writes the spherically symmetric Einstein--scalar field equations under the double-null gauge~\eqref{eq: double null gauge intro} as three wave equations for $r$, $\Omega^{2}$, and $\phi$, then the wave equation for $\Omega^{2}$ involves the quadratic term $\partial\phi\cdot\partial\phi$. Upon linearization around the $k$-self-similar spacetime, this produces a term of regularity $C^{0,k^{2}/(1-k^{2})}$ in the wave equation for $\Omega^{2}$, which falls outside the framework developed for the linear wave equation in~\cite{singh2}.

    \item As we will review in detail in Section~\ref{sec: review on the linear wave}, the approach in~\cite{singh2} to proving stability for the linear wave equation is to take the Fourier--Laplace transform of the corresponding wave equation and analyze the distribution of the scattering resonances. A key step in establishing stability is to rule out exponentially growing modes, except for those generated by the symmetries of the original equations. However, for the linearized Einstein--scalar field equations, the coupling between the background geometry and the background scalar field may produce additional nontrivial exponentially growing modes. The presence of such modes would lead to a further instability mechanism; see Section~\ref{sec: difficulty for the scattering} for a detailed discussion.
\end{itemize}

\subsection{Blue-shift instability mechanism, high-regularity stabilizing effect, and model problems for (in)stability}
\label{sec: blue-shift mechanism}
As we have mentioned above, the $k$-self-similar naked singularity spacetime is nonlinearly unstable under perturbations below the threshold \cite{chris99,liuli,liuli_outsidesymm,an_highcodim,li2025interior} and is nonlinearly stable under exterior perturbations above the threshold \cite{singh1,singh2}. The instability is driven by the so-called blue-shift instability mechanism, whereas a stabilizing effect emerges for initial perturbations of high regularity. In this section, we illustrate these mechanisms and describe a transition between the low-regularity instability and high-regularity stability regimes by considering the behavior of linear waves on the exterior region and the interior region of $k$-self-similar naked singularity spacetimes, respectively. These two model problems provide insight into the complex relationship among the support and regularity of the initial perturbation, the regularity of the background spacetime, and the leading order asymptotics of the solution.

Let $\varphi(u,v)$ denote a spherically symmetric solution to the linear wave equation 
\begin{equation}
    \label{intro:eq:5}
   r_{k} \Box_{g_k}\varphi =r_{k}\partial_{u}\partial_{v}\varphi+\partial_{u}r_{k}\partial_{v}\varphi+\partial_{v}r_{k}\partial_{u}\varphi =  0.
\end{equation}

By analogy with the bounds satisfied by the $k$-self-similar scalar field \eqref{eq: self-similar rate intro}, we make the following definition:
\begin{definition}
    A sufficiently regular, spherically symmetric, solution $\varphi(u,v)$ to (\ref{intro:eq:5}) is said to satisfy \textbf{$C^1$-self-similar bounds} if the following pointwise bounds hold:
    \begin{equation*}
        \left|\frac{1}{\partial_ur_k}\partial_u\varphi(u,v)\right| \lesssim |u|^{-1},\quad 
          \left|\frac{1}{(\partial_vr_k)}\partial_v\varphi(u,v)\right| \lesssim |u|^{-1}.
    \end{equation*}
\end{definition}
\begin{remark}
    More generally, one can consider self-similar bounds for higher-order derivatives and for non-spherically symmetric solutions. Given the finite Hölder regularity of the background solution, appropriate modifications must be made in the definition for higher-order derivatives.
\end{remark}

In the case of a single linear wave equation, the effect of the perturbation is reflected in the solution to the wave equation under the perturbed initial data. Solutions $\varphi$ which fail to satisfy $C^1$-self-similar bounds are indicative of a linear instability for the background spacetime, whereas those satisfying \textit{better-than-self-similar} bounds are indicative of linear asymptotic stability. Hence, we will be interested in the conditions under which such bounds occur.

\subsubsection{Problem 1: linear wave equation in the exterior region}
The study of the exterior perturbation, in the setting of a single linear wave equation, is reduced to the following characteristic problem:
\begin{equation}
\label{intro:eq:6}
    \begin{cases}
        \Box_{g_k}\varphi = 0, \\
        \partial_u\varphi(u,0) = 0, \\
        \partial_v\varphi(-1,v) = g_0(v),
    \end{cases}
\end{equation}
where $g_0(v): [0,1] \rightarrow \mathbb{R}$ admits a Hölder expansion
\begin{equation}
\label{intro:eq:6.5}
    g_0(v) = c_0 + c_1 v^{\alpha} + O(v), \quad v\in[0,1]
\end{equation}
for $\alpha \in (0,1)$, and constants $c_0, c_1$. Define $\theta \doteq \frac{r_k}{\partial_v r_k}\partial_v\varphi$, for which the wave equation \eqref{intro:eq:5} reduces to the following transport equation in the ingoing direction:
\begin{equation}
\label{intro:eq:7}
    \partial_u\theta - \frac{\mu_k (-\partial_u r_k)}{(1-\mu_k)r_k} \theta = -\partial_u \varphi.
\end{equation}
\paragraph{Blue-shift instability mechanism for BV data}
In this section, we consider the case $c_{0}\neq 0$, i.e., $\partial_{v}\varphi$ has a jump discontinuity on $\{v = 0\}$. On the past light cone of the singularity $\{v = 0\}$, the above equation \eqref{intro:eq:7} will be reduced to \begin{equation}
    \partial_{u}\theta(u,0) = \frac{k^{2}}{(-u)}\theta(u,0).\label{eq: intro: blue-shift on the ingoing cone}
\end{equation}
Hence, directly integrating the equation \eqref{eq: intro: blue-shift on the ingoing cone}, we can derive \begin{equation}
    \theta(u,0) = \theta(-1,0)(-u)^{-k^{2}} = c_{0}(-u)^{-k^{2}},
\end{equation}
which means that for $\varphi$ satisfying the self-similar bound on the past light cone of the singularity, we must have $c_{0} = 0$.
If $c_{0} \neq 0$, then the lower bound $|\theta| \gtrsim |u|^{-k^2}$ translates to
\begin{equation}
\label{intro:eq:9}
    \big|\frac{1}{\partial_v r_k} \partial_v\varphi\big| \gtrsim |u|^{-1-k^2}.
\end{equation}
We interpret this result as indicating that solutions are unstable for generic choices of outgoing data $g_0(v)$ in function spaces where $c_0$ can vary. In particular, since $g_{0}(v)$ vanishes identically in the interior region on the initial hypersurface, $\partial_{v}\varphi$ has a jump discontinuity at $(u,v) = (-1,0)$. Consequently, the initial data for $\partial_{v}\varphi$ belongs to the function class of bounded variation. This instability is due to the geometry intrinsic to the ingoing null-cone, and is termed the \textbf{blue-shift instability}. We refer to the particular instability rate (\ref{intro:eq:9}) as the \textbf{blue-shift rate}.
\paragraph{Higher-regularity stabilizing effect}
The natural follow-up question is whether $c_{0} = 0$ is \textit{sufficient} for self-similar bounds to hold. Here, the Hölder regularity of $g_0(v)$, particularly the value of $\alpha$, is decisive. We only present heuristics here; a rigorous analysis is given in \cite{singh2}.

It is convenient to expand the metric quantities and the solution in a formal series in powers of the similarity coordinate $z = v (-u)^{-1+k^2}$ in a neighborhood of the singular horizon where $0 \leq z \ll 1$. For the metric quantities, this follows from the \textit{regularity} of the $k$-self-similar spacetimes. In particular, 
\begin{equation*}
    \frac{\mu_k (-\partial_u r_k)}{(1-\mu_k)r_k} = \frac{k^2(1+O(z))}{(-u)} 
\end{equation*}
holds in a fixed spacetime region $\{z \ll 1\}$. For $\partial_u\varphi$, we use the wave equation to write 
\begin{equation*}
    \partial_u\varphi = O(z|\partial_u\varphi| + z (-u)^{-1}|\theta|).
\end{equation*}
Inserting these expansions into (\ref{intro:eq:7}) yields
\begin{equation}
    \partial_u \theta -\frac{k^2}{(-u)}\theta =  O(z|\partial_u\varphi| + z (-u)^{-1}|\theta|).\label{eq: intro: smoothing effect equation}
\end{equation}
Integrating from initial data and using the assumption $c_{0}=0$ gives
\begin{align*}
    \theta(u,v) &= (-u)^{-k^2}g_0(v)+ (-u)^{-k^2}\int_{-1}^{u}z(-u')^{-1+k^2}O(|(-u')\partial_u\varphi| + |\theta|)du' \\
    &= (-u)^{-k^2}\big(c_1v^\alpha + O(v)\big)+ (-u)^{-k^2}\int_{-1}^{u}v (-u')^{-2+2k^2}O(|(-u')\partial_u\varphi| + |\theta|)du'\\
    &= c_1 (-u)^{-k^2}v^\alpha + O((-u)^{-k^2}v+z \|(-u)\partial_u\varphi \|_{L^\infty} + z \|\theta\|_{L^\infty}).
\end{align*}
For simplicity, we have treated the error terms under the assumption that $\partial_{u}\varphi$ satisfies the self-similar bound. The conclusion, then, is that when $0<z \ll 1$, one expects the leading order term in the asymptotic for $\theta$ to be 
\begin{equation*}
    \theta(u,v) \sim c_1 (-u)^{-k^2}v^\alpha + O(1) = c_1 z^\alpha (-u)^{\alpha(1-k^2)-k^2}+O(1).
\end{equation*}
We find the following:
\begin{itemize}
    \item If $\alpha \in (0, \frac{k^2}{1-k^2})$, the observed asymptotic behavior in $u$ is inconsistent with the self-similar bound. The instability rate is strictly between the blue-shift and self-similar rates, approaching the former as $\alpha \rightarrow 0$.
    \item If $\alpha \geq \frac{k^2}{1-k^2}$, the asymptotics are consistent with self-similar bounds. In fact, the argument suggests that $\alpha > \frac{k^2}{1-k^2}$ implies \textit{better than self-similar} bounds are obtained. 
\end{itemize}
This reveals that the blue-shift rate, determined by the geometry along the singular horizon, more generally gives an upper bound for the instability rates of Hölder continuous data in small, self-similar, exterior neighborhoods of the singular horizon. One can understand the above formal computation as a competition between a high-regularity stabilizing effect and the low-regularity blue-shift instability mechanism. As one increases the Hölder regularity of the initial exterior perturbation, the stabilizing effect becomes stronger. The Hölder regularity $C^{1,\alpha_{k}}$
 for $\alpha_{k} = \frac{k^{2}}{1-k^{2}}$ is moreover seen to hold a special role as a \textbf{threshold regularity} separating the \textbf{below-threshold} instability range $\alpha \in (0, \alpha_k)$ from the \textbf{above-threshold} stability range $\alpha \in (\alpha_k,1).$

 Finally, we close the discussion in this section of the exterior perturbations with the following three remarks.
 \begin{remark}
 Although we only consider the characteristic problem \eqref{intro:eq:6} with the trivial initial data on the ingoing cone, as shown in \cite{singh1}, one can also consider the non-trivial ingoing data $\partial_{u}\varphi = \frac{f_{0}(u)}{(-u)}$ with some additional decay assumption on $f_{0}(u)$. The heuristics and their conclusions provided in this section will follow line by line.
 \end{remark}
 \begin{remark}
The underlying mechanism for the exterior perturbations can be understood as a local ODE-type mechanism. The nature of these perturbations provides a priori control of certain quantities in both null directions, allowing both the blue-shift instability mechanism and high-regularity stabilizing effects to be analyzed locally near the past light cone emanating from the singularity.

 \end{remark}
\begin{remark}
The analysis on the wave equation \eqref{intro:eq:6} suggests that the solution to the linear wave equation is only orbitally stable for perturbations at the threshold. This aspect is further exploited in the nonlinear analysis in \cite{singh2} to show the orbital stability at the nonlinear level for exterior perturbations at the threshold. However, using the fact that Einstein-scalar field equations are scaling invariant, one can in fact upgrade the orbital exterior stability at the threshold to asymptotic stability; see Section 3 in the companion paper \cite{zhenglinear}.
\end{remark}

\subsubsection{Problem 2: linear wave equation in the interior region}
\label{eq: linear wave equation in the interior region}
The heuristics for exterior perturbations rely essentially on a priori control in the ingoing direction. On the one hand, the above approach is robust in the sense that, once such a priori control is available, it can be applied to understand (in)stability of naked singularity solutions to Einstein equations coupled with other matter models in the exterior region. On the other hand, for interior perturbations—where no a priori estimate in the ingoing direction is available—the above ODE-type approach fails to provide any insight into the behavior of the interior perturbations. In this section, we provide a way to understand the blue-shift mechanism and the high-regularity stabilizing effect for the interior perturbation.

We consider solutions in the interior region of \textit{mode-type}
\begin{equation}
    \label{intro:eq:mode-type equation}
    (r_k\varphi)(u,v) = (-u)^{\rho} f(z),
\end{equation}
where $z = v/(-u)^{1-k^{2}}\in[-1,0]$ is the self-similar coordinate. If $\Re(\rho)>1$, the mode solution satisfies a bound better than the self-similar bound; if $1-k^{2}<\Re(\rho)<1$, the mode solution satisfies a bound between the self-similar rate and the blue-shift rate; if $\Re(\rho)\leq1-k^{2}$, the mode solution satisfies a bound worse than the self-similar rate. As discussed in the previous section, any rate worse than the self-similar rate is expected to correspond to the instability.

Then, for mode-type solutions, the wave equation \eqref{intro:eq:5} reduces to an ODE for the profile $f(z)$
\begin{equation}
\label{intro:eq:ode profile}
    (1-k^2)z f''(z)+(1-k^2-\rho)f'(z) + V_k(z) f(z) = 0.
\end{equation}

Here, $V_k(z)$ is a non-negative, Hölder continuous potential determined by the background metric. Since $r_{k} = 0$ at the center, the solution $f(z)$ should satisfy the \textit{Dirichlet} boundary condition, i.e., $f(-1) = 0$.

A computation gives that the Frobenius indices for the ODE \eqref{intro:eq:ode profile} are $0, \frac{\rho}{1-k^2}$, and hence in the region $\{\text{Re}(\rho) < \frac{3}{2}\} \setminus \{\frac{\rho}{1-k^2}\in \mathbb{Z}\},$ there exists a local basis of solutions $\{f_{reg}(z), f_{irreg}(z)\}$ for which $f_{reg}(z) \in C^2$ and $f_{irreg}(z) \sim c_0 + c_1 (-z)^{\frac{\rho}{1-k^2}} + \text{l.o.t}  \notin C^2.$ 

Whether or not $f_{irreg}(z)$ lies in a given regularity class depends on $\rho$, leading to the following trichotomy:
\begin{itemize}
    \item For $\text{Re}(\rho) \in (1-k^2,1)$, we have that $f_{irreg}(z) \in C^{1,\frac{\Re(\rho)-(1-k^{2})}{1-k^{2}}} \setminus C^{1,\frac{k^2}{1-k^2}}$, and hence the regularity of $f_{irreg}$ is \textit{strictly worse} than the threshold. It follows that if the admissible class of data is restricted to a Hölder space $C^{1,\alpha}$ with $\alpha < \frac{k^2}{1-k^2}$, any solution $f$ to \eqref{intro:eq:ode profile} satisfying the Dirichlet boundary condition and with $\rho = (1-k^{2})(\alpha+1)\in(1-k^{2},1)$ will lie within the given Hölder regularity. In this case, the mode solution with $\rho = (1-k^{2})(\alpha+1)$ will satisfy the regularity constraint of the initial data and has a lower bound worse than the self-similar rate. Hence, the instability is expected. This instability can be interpreted as the failure of mode stability for perturbations below the threshold. Recently, this instability has been rigorously achieved nonlinearly in the work \cite{li2025interior} for any perturbation strictly below the threshold.
    
    \item For $\Re(\rho)=1,$ $f_{irreg}(z) \notin C^{1,\alpha}$ for any $\alpha>\frac{k^{2}}{1-k^{2}}$. In particular, when $\rho =1$, then $f_{irreg}\in C^{1,\frac{k^{2}}{1-k^{2}}}$; when $\Im(\rho)\neq0$, then $f_{irreg}\notin C^{1,\frac{k^{2}}{1-k^{2}}}$ yet it still possesses Hölder regularity better than any $C^{1,\alpha}$ with $\alpha\in(0,\frac{k^{2}}{1-k^{2}})$. Consequently, this scenario is expected to exhibit a complicated phenomenon.
    \item For $\text{Re}(\rho) > 1$, $f_{irreg}(z)$ has regularity \textit{strictly better} than that of the background solution $\phi$. If the admissible class of initial data is contained in a Hölder space $C^{1,\alpha}$ with $\alpha > \frac{k^2}{1-k^2}$, then mode instability is the question of whether there exists $\text{Re}(\rho) < 1$ satisfying $W[f_{dir}, f_{reg}] = 0$, where $f_{dir}(z)$ is a solution to \eqref{intro:eq:ode profile} satisfying Dirichlet conditions $f(-1) = 0$, and $W[f,g]$ denotes the Wronskian.
\end{itemize}

Although this mode perspective immediately gives instability at the level of the linear wave equation for regularities below threshold, it leaves open the case of regularities above threshold, which is treated in~\cite{singh2}. However, it clarifies the essential difference between the two cases: whereas dynamics in the former case are determined by a continuum of modes defined by the \textit{local} behavior of solutions to (\ref{intro:eq:7}) near $\{z=0\}$, in the latter the dynamics are determined by a discrete set of modes reflecting the \textit{global} behavior of solutions to (\ref{intro:eq:7}). Mode stability in the high-regularity setting is therefore more subtle, necessitating the use of some global properties of the potential $V_k(z)$.
\subsubsection{Difference between the linear wave equation and the linearized Einstein--scalar field equations}

Lastly, as pointed out in Section~\ref{sec: previous works on the interior perturbations}, the stability mechanism for the linear wave equation in the interior of the $k$-self-similar naked singularity, discussed in Section~\ref{eq: linear wave equation in the interior region}, is not sufficient to fully capture the stability mechanism for the linearized Einstein--scalar field equations, for the following two reasons.

{The first issue concerns the regularity of the background scalar field. To illustrate this point, consider mode-type solutions to \eqref{intro:eq:ode profile}. The rigorous construction of the corresponding outgoing solutions relies on the regularity of $V_{k}$, which depends only on the background function $r_{k}$ and belongs to $C^{1,k^{2}/(1-k^{2})}$. However, if one formulates the linearized Einstein--scalar field equations as wave equations for $r$, $\phi$, and $\Omega^{2}$, then the mode-type solutions to the corresponding ODE system involve potential terms containing derivatives of $\phi_{k}$, which have only $C^{0,k^{2}/(1-k^{2})}$ regularity. This loss of regularity obstructs the rigorous construction of outgoing solutions; see Section~\ref{sec: difficulty for the scattering} for details.}

The second issue concerns the coupling between the background geometry and the background scalar field. To rigorously implement the mechanism described in Section~\ref{eq: linear wave equation in the interior region}, the approach in~\cite{singh2} is to take the Fourier--Laplace transform of the linear wave equation for $\psi = r_{k}\varphi$ and to analyze the scattering resolvent and resonances of the resulting ODE. For the linear wave equation, there is only one scattering resonance, corresponding to the translation invariance of the Einstein--scalar field equations. By contrast, due to the coupling between the background geometry and the background scalar field, the analogous ODE system arising from the linearized Einstein--scalar field equations may admit additional scattering resonances. Such resonances would generate an additional instability mechanism and hence fall outside the framework of~\cite{singh2}; see Section~\ref{sec: difficulty for the scattering} for details.

\subsection{Overview of the paper}
\label{intro:subs:overviewofpaper}
In Section \ref{prelims:sec}, we recall some important properties of the $k$-self-similar naked singularities. In Section \ref{sec: main result and proof outline}, we introduce our main results and the proof outline. The analysis of the linearized Einstein-scalar field equations will take the majority of this paper. Section \ref{sec:multiplier_estimates} establishes the energy estimate for the linearized system and obtains a non-sharp decay result. Section \ref{sec: scattering theory for the linearized system} establishes the scattering theory for the linearized Einstein-scalar field operator. Sections \ref{leading order expansion}--\ref{sec: decay in the interior region} conclude the decay estimate for the linearized system.

\subsection{Acknowledgments}
\label{intro:subsec:acknowledgments}
Zheng would like to thank his advisor, Maxime Van de Moortel, for his kind support, patience, encouragement, and many enlightening discussions. Zheng also gratefully acknowledges the support
from the NSF Grant DMS-2247376. The authors would like to thank Igor Rodnianski for proposing this project to us and for many useful discussions. The authors would also like to thank Yakov Shlapentokh-Rothman, Mihalis Dafermos, and Jonathan Luk for their interest in this work and useful conversations.

\section{Preliminary}
\label{prelims:sec}

\subsection{Einstein-scalar field equations under spherical symmetry}
The Einstein-scalar field equations are a coupled system for metric $g$ and the scalar field $\phi$, taking the form of
\begin{align}
Ric(g)_{\mu\nu}&= 2\partial_{\mu}\phi\partial_{\nu}\phi,\label{ric equation}\\
\Box_{g}\phi&=0.\label{scalar field equation}
\end{align}
The solution to the Einstein-scalar field equations is a Lorentzian manifold $\mathcal{M}$ equipped with a Lorentzian metric $g$ and the associated scalar field $\phi:\mathcal{M}\rightarrow \mathbb{R}$. If the solution $(\mathcal{M},g,\phi)$ is invariant under the $SO(3)$ action, then we call the solution spherically symmetric. Let $\mathcal{Q}$ be the quotient manifold $\mathcal{M}/SO(3)$ and $p:\mathcal{M}\rightarrow\mathcal{Q}$ be the projection map. Then $p$ can naturally induce a quotient Lorentzian metric $g_{\mathcal{Q}}$ on $\mathcal{Q}$. The set of fixed points under the $SO(3)$ action is called the center or the axis of the spacetime and is denoted as $\Gamma$.

Assume that there exists a global double-null system $(u,v)$ on $\mathcal{Q}$, under which the quotient metric $g_{Q}$ takes the form of \begin{equation*}
    g_{\mathcal{Q}} = -\frac{1}{2}\Omega^{2}(u,v)du\otimes dv-\frac{1}{2}\Omega^{2}(u,v)dv\otimes du,
\end{equation*}
where $\Omega^{2}$ is the lapse function. For any $q\in\mathcal{Q}$, we can define the associated area radius function by \begin{equation*}
    r(q): = \sqrt{\frac{Area(p^{-1}(q))}{4\pi}}.
\end{equation*}
In particular, by the definition of the center $\Gamma$, we have $r|_{\Gamma} = 0$.

Then, we can express the metric $g$ on $\mathcal{M}$ by
\begin{equation*}
g = -\frac{1}{2}\Omega^{2}du\otimes dv-\frac{1}{2}\Omega^{2}dv\otimes du+r^{2}(u,v)d\sigma^{2},
\end{equation*}
where $d\sigma^{2} = d\theta^{2}+\sin^{2}(\theta)d\varphi^{2}$ is the standard metric on $S^{2}$.

Under the spherically symmetric assumption, the Einstein-scalar field equations \eqref{ric equation}-\eqref{scalar field equation} can be reduced to\begin{align}
\partial_{u}\left(\frac{r_{u}}{\Omega^{2}}\right) &= -\frac{r}{\Omega^{2}}(\partial_{u}\phi)^{2},\label{eq:u-Ray equation}\\
\partial_{v}\left(\frac{r_{v}}{\Omega^{2}}\right)& = -\frac{r}{\Omega^{2}}(\partial_{v}\phi)^{2},\label{eq:v-Ray equation}\\
r\partial_{u}\partial_{v}r+\partial_{u}r\partial_{v}r &= -\frac{1}{4}\Omega^{2},\label{eq:wave equation for r}\\
\partial_{u}\partial_{v}\log(\Omega^{2})& = \frac{\Omega^{2}}{2r^{2}}+\frac{2r_{u}r_{v}}{r^{2}}-2\partial_{u}\phi\partial_{v}\phi,\label{eq:wave equation for omega}\\
r\partial_{u}\partial_{v}\phi +r_{u}\partial_{v}\phi+r_{v}\partial_{u}\phi &= 0\label{eq:wave equation for phi}.  
\end{align} 
The transport equations \eqref{eq:u-Ray equation}-\eqref{eq:v-Ray equation} are constraint equations and are often called the Raychaudhuri equations. The equations \eqref{eq:wave equation for r}-\eqref{eq:wave equation for phi} are the wave equations for $r$, $\Omega^{2}$, and $\phi$, respectively. One should note that the above equations are overdetermined. In fact, one can deduce the wave equation for $\Omega$ \eqref{eq:wave equation for omega} from the wave equations \eqref{eq:wave equation for r} and \eqref{eq:wave equation for phi}, and constraint equations \eqref{eq:u-Ray equation}--\eqref{eq:v-Ray equation}. We will not use the wave equation for $\Omega^{2}$ in the later study.

Under the global double-null coordinates, instead of solving the Einstein-scalar field equations for $(\Omega^{2},r,\phi)$, it is often more convenient to solve for $(r,\phi,m)$, where $m$ is the so-called Hawking mass defined as \begin{align}
m: &= \frac{r}{2}\left(1-g(\nabla r,\nabla r)\right).
\end{align}
The mass ratio $\mu$ is defined as \begin{equation*}
    \mu := \frac{2m}{r}.
\end{equation*}
For the Hawking mass $m$ and the mass ratio $\mu$, the Raychaudhuri equations \eqref{eq:u-Ray equation}-\eqref{eq:v-Ray equation} can be written as the following transport equations for $m$ and $\mu$:
\begin{align}
\partial_{u}m +\frac{r}{r_{u}}\left(\partial_{u}\phi\right)^{2}m &= \frac{1}{2}\frac{r^{2}}{r_{u}}\left(\partial_{u}\phi\right)^{2},\label{eq:u-equ for m}\\
\partial_{v}m+\frac{r}{r_{v}}\left(\partial_{v}\phi\right)^{2}m &= \frac{1}{2}\frac{r^{2}}{r_{v}}\left(\partial_{v}\phi\right)^{2},\label{eq:v-equ for m}\\
\partial_{u}\mu+\left(\frac{r_{u}}{r}+\frac{r}{r_{u}}\left(\partial_{u}\phi\right)^{2}\right)\mu&=\frac{r}{r_{u}}\left(\partial_{u}\phi\right)^{2},\label{eq:u-equ for mu}\\
\partial_{v}\mu+\left(\frac{r_{v}}{r}+\frac{r}{r_{v}}\left(\partial_{v}\phi\right)^{2}\right)\mu&=\frac{r}{r_{v}}\left(\partial_{v}\phi\right)^{2}.\label{eq:v-equ for mu}
\end{align}
\subsection{Christodoulou's naked singularity spacetime}
\label{sec: pre on k naked singularity sapcetime}
The coordinate system in the original construction of the $k$-self-similar spacetime in \cite{chris94} was under the so-called Bondi coordinates. In this section, we give a quick review of these spacetimes under the global double-null coordinates. For a more detailed review, we refer to Section 2.2 in the companion paper~\cite{zhenglinear} or Appendix B of~\cite{singh1}.

First, we recall the scaling and translation symmetries of the Einstein-scalar field equations. Under the spherical double-null coordinates $(\hat{u},\hat{v},\theta,\varphi)$\begin{equation*}
    g = -\frac{1}{2}\Omega^{2}d\hat{u}\otimes d\hat{v}-\frac{1}{2}\Omega^{2}d\hat{v}\otimes d\hat{u}+r^{2}d\sigma^{2},
\end{equation*}
if $(r,\Omega^{2},\phi)$ is a solution, then $(r_{a},\Omega_{a}^{2},\phi_{a,b})$
\begin{equation*}
    r_{a}(\hat{u},\hat{v}): = ar\left(\frac{\hat{u}}{a},\frac{\hat{v}}{a}\right),\quad \Omega_{a}^{2}(u,v): = \Omega^{2}\left(\frac{\hat{u}}{a},\frac{\hat{v}}{a}\right),\quad \phi_{a,b}(\hat{u},\hat{v}): = \phi\left(\frac{\hat{u}}{a},\frac{\hat{v}}{a}\right)+b
\end{equation*}
is also a solution. Taking $b = -k\log a$ and assuming the solution spacetime $(\mathcal{M},g,\phi)$ is invariant under scaling and translation \begin{equation*}
    r(\hat{u},\hat{v}) = ar\left(\frac{\hat{u}}{a},\frac{\hat{v}}{a}\right),\quad\Omega^{2}(\hat{u},\hat{v}) = \Omega^{2}\left(\frac{\hat{u}}{a},\frac{\hat{v}}{a}\right),\quad \phi(\hat{u},\hat{v}) = \phi\left(\frac{\hat{u}}{a},\frac{\hat{v}}{a}\right)+k\log a,
\end{equation*}
then the resulting solution $\left(r,\Omega^{2},\phi\right)$ to \eqref{ric equation}-\eqref{scalar field equation}, known as the $k$-self-similar solution, takes the form of \begin{equation}
    r\left(\hat{u},\hat{v}\right) = \left(-\hat{u}\right)\mr{r}\left(\hat{z}\right),\quad \Omega^{2}\left(\hat{u},\hat{v}\right) = \mr{\Omega}^{2}\left(\hat{z}\right),\quad \phi(\hat{u},\hat{v}) = \mr{\phi}\left(\hat{z}\right)-k\log(-\hat{u}),\label{eq: k self similar form of the equations}
\end{equation}
where $\hat{z}$ is called the self-similar coordinate, defined as \begin{equation*}
    \hat{z}:= -\frac{\hat{v}}{\hat{u}}.
\end{equation*}
Christodoulou's original construction of the $k$-self-similar naked singularity took advantage of the $k$-self-similarity \eqref{eq: k self similar form of the equations}, under which the Einstein-scalar field equations can be reduced to a system of ODEs. Immediately from the form \eqref{eq: k self similar form of the equations}, one can show that the curvature blows up at $(\hat{u},\hat{v})$. Therefore, the point $(0,0)$ corresponds to the naked singularity $\mathcal{O}$. We can further fix the gauge choice of the center $\Gamma$ to be $\Gamma = \{\hat{z} = -1\}$. Then the interior region $\mathcal{Q}^{(in)}$ of the naked singularity spacetime will be $\{-1\leq \hat{z}\leq 0\}$, while the exterior region $\mathcal{Q}^{(ex)}$ becomes $\{\hat{z}\geq 0\}$.

However, the ODE analysis in the original paper by Christodoulou~\cite{chris94} was carried out in Bondi coordinates. For double-null coordinates $(\hat{u},\hat{v})$, the analogous ODE analysis will show that under the $k$-self-similar assumption~\eqref{eq: k self similar form of the equations}, $(\hat{u},\hat{v})$ will no longer be a global double-null coordinate system. More precisely, one can show that $\mr\Omega^{2}(\hat{z})$, $\frac{d\mr{\phi}(\hat{z})}{d\hat{z}}$, and $\frac{d\mr{r}}{d\hat{z}}(\hat{z})$ will blow up at the rate $\vert \hat{z}\vert^{-k^{2}}$ as $\hat{z}\rightarrow 0$. Therefore, to work under a global double-null gauge and extend spacetime beyond the cone $\hat{z} = 0$, we introduce the following renormalized double-null coordinates \begin{equation*}
    u = \hat{u},\quad (-v) = (-\hat{v})^{q_{k}},\quad (-z) = (-\hat{z})^{q_{k}},
\end{equation*}
where $q_{k} := 1-k^{2}$, which will be frequently used later. Let $p_{k}: = \frac{1}{q_{k}}$. The constant-$u$ hypersurface is called the outgoing cone and is denoted as $\Sigma_{const}$, while the constant-$v$ hypersurface is called the ingoing cone and is denoted as $\underline{\Sigma}_{const}$. The ingoing cone emanating from the naked singularity will play an important role in the later analysis. We refer to this cone as \textit{singular horizon}. Although Christodoulou's original construction of those naked singularity spacetimes does not proceed via solving an initial value problem for the Einstein-scalar field equations,  one can nevertheless view the $k$-self-similar naked singularity spacetime as arising from some initial data prescribed on the outgoing null cone $\Sigma_{-1}$.

We collect the quantitative properties of the interior of the $k$-self-similar naked singularity spacetime in \cite{chris94} under the double-null coordinates $(u,v)$ in the following proposition:
\begin{proposition}[\cite{chris94,singh1,zhenglinear}]
\label{prop: estimate on the exact k self similar spacetime}
    Assume that the spacetime solution $(\mathcal{M}_{k},g_{k},\phi_{k})$ is a $k$-self-similar solution to the spherically symmetric Einstein-scalar field equation \eqref{ric equation}--\eqref{scalar field equation}. Then $(r_{k},\Omega^{2}_{k},\phi_{k},m_{k},\mu_{k})$ takes the form of \begin{equation*}
        r_{k} = (-u)\mr{r}(z),\quad \Omega^{2} = (-u)^{k^{2}}\mr{\Omega}^{2}(z),\quad \phi_{k} = \mr{\phi}(z)-k\log(-u),\quad m_{k} = (-u)\mr{m}(z),\quad \mu_{k} = \mr{\mu}(z).
    \end{equation*}
    Let $\nu_{k} = \partial_{u}r_{k}$ and $\lambda_{k} = \partial_{v}r_{k}$. Then we have \begin{equation*}
        \nu_{k} = \mr{\nu}(z) = -\mr{r}(z)+q_{k}z\frac{d\mr{r}}{dz},\quad \lambda_{k} = (-u)^{k^{2}}\mr{\lambda}(z) = (-u)^{k^{2}}\frac{d\mr{r}}{dz}.
    \end{equation*}
    Moreover, we can fix the gauge freedom of $(r_{k},\phi_{k},m_{k})$ to be $r(u,0) = (-u)$ and $\phi_{k}(u,0) = -k\log(-u)$. In the interior region $\mathcal{Q}^{(in)}$ of the $k$-self-similar naked singularity spacetime, we have the following properties on the regularity and quantitative estimates of $(r_{k},\Omega_{k}^{2},\phi_{k},m_{k},\mu_{k})$: \begin{itemize}
        \item[(1)] The spacetime $\left(\mathcal{M}_{k},g_{k},\phi_{k}\right)$ is smooth in the region $$\{-1\leq u<0,v<0\}.$$ In other words, the $k$-self-similar spacetime is smooth away from the singular horizon.
    \item[(2)] The low regularity nature of the spacetime only arises at the singular horizon $\{v = 0\}$. In the interior region $\{-1\leq u<0,\ -1\leq z\leq 0\}$, we have \begin{equation*}
        \mr{m},\ \mr{\Omega},\ \mr{\phi}\in C^{1,p_{k}k^{2}}\left([-1,0]\right),\quad \mr{r}\in C^{2,p_{k}k^{2}}([-1,0]),
    \end{equation*}
    with the estimates \begin{equation*}
        \left\vert\vert z\vert^{j-1-p_{k}k^{2}}\partial_{z}^{j+1}\mr{r}\right\vert+\left\vert\vert z\vert^{j-1-p_{k}k^{2}}\partial_{z}^{j}\mr{\Omega}\right\vert+\left\vert\vert z\vert^{j-1-p_{k}k^{2}}\partial_{z}^{j}\mr{m}\right\vert+\left\vert\vert z\vert^{j-1-p_{k}k^{2}}\partial_{z}^{j}\mr{\phi}\right\vert\lesssim_{k}1,\quad j\geq 2.
    \end{equation*}
    Moreover, the finite regularity for the scalar field $\mr{\phi}$ is sharp in the sense that \begin{align*}
        \frac{d\mr{\phi}}{dz}\sim\frac{1}{k},\quad 
        \frac{d^{2}\mr{\phi}}{dz^{2}}\sim \vert z\vert^{-1+p_{k}k^{2}},\quad z\rightarrow 0.
    \end{align*}
    \item[(3)] In the region $\mathcal{Q}^{(in)}$, we have the following estimates on the geometric quantities: \begin{align*}
        &\mr{r}(-1) = 0,\quad  \mr{r}\approx 1-\left\vert z\right\vert,\quad (-\nu_{k})\approx 1.\quad \lambda_{k}\approx (-u)^{k^{2}},\\&
        \Omega_{k}^{2}\approx (-u)^{k^{2}},\quad 0<\frac{\mr{m}}{\mr{r}^{3}}\lesssim k^{2},\quad 0<\frac{\mr{\mu}}{\mr{r}^{2}}\lesssim k^{2}.
    \end{align*}
    \item[(4)] In the region $\mathcal{Q}^{(in)}$, we have the following estimates on the quantities related to the scalar field: \begin{align*}
    \mr{\phi}^{\prime}(z = 0) = \frac{C}{k}, \quad \left\Vert\mr{\phi}^{\prime}\right\Vert_{L^{\infty}([-1,-\frac{1}{2}])}\lesssim k^{2},\quad 
    \left\Vert\vert z\vert^{\epsilon}\mr{\phi}^{\prime}\right\Vert_{L^{\infty}([-1,0])}\lesssim \frac{k}{\epsilon},\quad k\ll\epsilon<1.
    \end{align*}
    \item[(4)] As a consequence of Property $(4)$ and our gauge choice of $\phi_{k}$, we have \begin{equation*}
        \left\vert\mr{\phi}\right\vert\lesssim \frac{k}{\epsilon}\left\vert z\right\vert^{1-\epsilon}.
    \end{equation*}
    \item[(5)] For the higher order derivatives, we have \begin{align*}
    &\vert z\vert^{1+\epsilon}\left\vert \mr{\phi}^{\prime\prime}\right\vert\lesssim \frac{k^{2}}{\epsilon},\quad \vert z\vert^{\epsilon}\left\vert\frac{d^{2}\mr{r}}{dz^{2}}\right\vert+\vert z\vert^{1+\epsilon}\left\vert\frac{d^{3}\mr{r}}{dz^{3}} \right\vert\lesssim \frac{k^{2}}{\epsilon^{2}},\\&
    \left\vert z\right\vert^{\epsilon}\left\vert\frac{\partial_{z}\mr{\mu}}{(1+z)}\right\vert+\vert z\vert^{1+\epsilon}\left\vert\partial_{z}^{2}\mr{\mu}\right\vert\lesssim \frac{k^{2}}{\epsilon^{2}},\quad \vert z\vert^{\epsilon}\left\vert\frac{\partial_{z}\mr{m}}{(z+1)^{2}}\right\vert+\vert z\vert^{1+\epsilon}\left\vert\frac{\partial_{z}^{2}\mr{m}}{z+1}\right\vert\lesssim \frac{k^{2}}{\epsilon^{2}}.
    \end{align*}
    \end{itemize}
\end{proposition}

\subsection{Self-similar coordinates}
We use the same self-similar coordinates $(s,z)$ as in~\cite{zhenglinear}\begin{equation}
    s: = -\log(-u),\quad z = \frac{v}{(-u)^{q_{k}}}.
\end{equation}
In these coordinates, the center $\Gamma$ in the quotient spacetime $\mathcal{Q}$ corresponds to the curve $\{z = -1\}$, the singular horizon $\mathcal{H}$ corresponds to $\{z = 0\}$, and the interior region $\mathcal{Q}^{(in)}$ corresponds to $\{-1\leq z\leq0\}$. Moreover, we have the following relations \begin{align}
    &\partial_{u} = e^{s}\partial_{s}+q_{k}ze^{s}\partial_{z},\quad \partial_{v} = e^{q_{k}s}\partial_{z},\\&
    \partial_{s} =(-u)\partial_{u}+q_{k}(-v)\partial_{v},\quad \partial_{z} = (-u)^{q_{k}}\partial_{v}.
\end{align}

\subsection{Deriving the linearized system for the interior region}
\label{sec: linearized system}
In this section, we derive the linearized system on the interior of the $k$-self-similar background $(r_{k},\phi_{k},m_{k})$. For a technical reason arising in the later analysis of the scattering theory---namely, the need for the coefficients in the linearized system under self-similar coordinates $(s,z)$ to be independent of $s$, we introduce the following renormalized quantity. Let \begin{equation}
\Psi = r\phi+k\log(-u)r.
\end{equation}
Then we can write down the equations for $(\Psi,r,m)$:\begin{align}
\partial_{u}\partial_{v}r-\frac{\mu}{1-\mu}\frac{\partial_{u}r\partial_{v}r}{r}& = 0,\label{eq:new formulation first}\\
\partial_{u}\partial_{v}\Psi+\frac{k}{(-u)}\partial_{v}r-\frac{\partial_{u}\partial_{v}r}{r}\Psi& = 0,\\
\partial_{u}m+\frac{r}{\partial_{u}r}(\partial_{u}\phi)^{2}m & = \frac{1}{2}\frac{r^{2}}{\partial_{u}r}(\partial_{u}\phi)^{2},\\
\partial_{v}m+\frac{r}{\partial_{v}r}(\partial_{v}\phi)^{2}m& = \frac{1}{2}\frac{r^{2}}{\partial_{v}r}(\partial_{v}\phi)^{2}.\label{eq: new formulatin last}
\end{align} 
We keep $\phi$ in the transport equations for $m$ under the new formulation $(r,\Psi,m)$ to reduce the algebraic complexity in the later study. We can schematically write the above Einstein-scalar field equations as $\mathcal{F}(r,\Psi,m) = 0$. Then the $k$-self-similar spacetime $(r_{k},\Psi_{k},m_{k})$ is the kernel of the nonlinear operator $\mathcal{F}$. Let \begin{equation*}
    r = r_{k}+r_{p},\quad \Psi = \Psi_{k}+\Psi_{p},\quad m = m_{k}+m_{p}.
\end{equation*}
Then we have \begin{equation*}
    0 = \mathcal{F}(r_{k}+r_{p},\Psi_{k}+\Psi_{p},m_{k}+m_{p})-\mathcal{F}(r_{k},\Psi_{k},m_{k}).
\end{equation*}
Taylor expanding the above expression and taking out the first-order terms of $(r_{p},\Psi_{p},m_{p})$, we can derive the following linearized operator.

\begin{align}
&
\partial_{u}\partial_{v}r_{p}-\frac{\mu_{k}}{1-\mu_{k}}\frac{\partial_{u}r_{k}}{r_{k}}\partial_{v}r_{p}-\frac{\mu_{k}}{1-\mu_{k}}\frac{\partial_{v}r_{k}}{r_{k}}\partial_{u}r_{p}+\frac{\mu_{k}(2-\mu_{k})}{(1-\mu_{k})^{2}}\frac{\partial_{u}r_{k}\partial_{v}r_{k}}{r_{k}^{2}}r_{p} = \frac{2\partial_{u}r_{k}\partial_{v}r_{k}}{(1-\mu_{k})^{2}r_{k}^{2}}m_{p},\label{eq:linearized r equation}\\[1em]
&\partial_{u}\partial_{v}\Psi_{p}-\frac{\mu_{k}}{1-\mu_{k}}\frac{\partial_{u}r_{k}\partial_{v}r_{k}}{r_{k}^{2}}\Psi_{p}-\left(\frac{1}{r_{k}^{2}}\frac{\mu_{k}}{1-\mu_{k}}\Psi_{k}\partial_{u}r_{k}+\frac{k}{(-u)}\right)\partial_{v}r_{p}-\frac{1}{r_{k}^{2}}\frac{\mu_{k}}{1-\mu_{k}}\Psi_{k}\partial_{v}r_{k}\partial_{u}r_{p} \nonumber\\&= -\frac{\mu_{k}(3-2\mu_{k})}{r_{k}^{3}(1-\mu_{k})^{2}}\Psi_{k} \partial_{u}r_{k}\partial_{v}r_{k}r_{p}+\frac{2}{r_{k}^{3}}\frac{\Psi_{k}\partial_{u}r_{k}\partial_{v}r_{k}}{(1-\mu_{k})^{2}}m_{p},\label{eq:linearized psi equation}\\[1em]
&\partial_{u}m_{p}+\frac{r_{k}}{\partial_{u}r_{k}}(\partial_{u}\phi_{k})^{2}m_{p}-(1-\mu_{k})\frac{r_{k}}{\partial_{u}r_{k}}\partial_{u}\phi_{k}\partial_{u}\Psi_{p}+(1-\mu_{k})\partial_{u}\phi_{k}\Psi_{p}\nonumber\\&=-
\left(\frac{1}{2}\left(\frac{r_{k}}{\partial_{u}r_{k}}\right)^{2}(\partial_{u}\phi_{k})^{2}+\frac{r_{k}}{\partial_{u}r_{k}}\partial_{u}\phi_{k}\left(\phi_{k}+k\log(-u)\right)\right)(1-\mu_{k})\partial_{u}r_{p}\nonumber\\&\hspace{.4cm}+\left(\frac{1}{2}\frac{r_{k}}{\partial_{u}r_{k}}(\partial_{u}\phi_{k})^{2}\mu_{k}+(1-\mu_{k})\partial_{u}\phi_{k}\left(\phi_{k}+k\log(-u)\right)+\frac{k}{(-u)}(1-\mu_{k})\frac{r_{k}}{\partial_{u}r_{k}}\partial_{u}\phi_{k}\right)r_{p},
\label{eq:linerized dum equation}\allowdisplaybreaks\\[1em]
&
\partial_{v}m_{p}+\frac{r_{k}}{\partial_{v}r_{k}}(\partial_{v}\phi_{k})^{2}m_{p}-(1-\mu_{k})\frac{r_{k}}{\partial_{v}r_{k}}\partial_{v}\phi_{k}\partial_{v}\Psi_{p}+(1-\mu_{k})\partial_{v}\phi_{k}\Psi_{p}\nonumber\\&=-
\left(\frac{1}{2}\left(\frac{r_{k}}{\partial_{v}r_{k}}\right)^{2}(\partial_{v}\phi_{k})^{2}+\frac{r_{k}}{\partial_{v}r_{k}}\partial_{v}\phi_{k}\left(\phi_{k}+k\log(-u)\right)\right)(1-\mu_{k})\partial_{v}r_{p}\nonumber\\&\hspace{.4cm}+\left(\frac{1}{2}\frac{r_{k}}{\partial_{v}r_{k}}(\partial_{v}\phi_{k})^{2}\mu_{k}+(1-\mu_{k})\partial_{v}\phi_{k}\left(\phi_{k}+k\log(-u)\right)\right)r_{p}.\label{eq:linearized dvm equation}
\end{align}
Since we will work with the self-similar coordinates later, we derive the corresponding equations.
\begin{align}
&
\partial_{s}\partial_{z}r_{p}+q_{k}z\partial_{z}^{2}r_{p}+\left(q_{k}-\frac{\mu_{k}}{1-\mu_{k}}\frac{\partial_{s}r_{k}+2q_{k}z\partial_{z}r_{k}}{r_{k}}\right)\partial_{z}r_{p}-\frac{\mu_{k}}{1-\mu_{k}}\frac{\partial_{z}r_{k}}{r_{k}}\partial_{s}r_{p}\nonumber\\&=-\frac{\mu_{k}(2-\mu_{k})}{(1-\mu_{k})^{2}}\frac{(\partial_{s}r_{k}+q_{k}z\partial_{z}r_{k})\partial_{z}r_{k}}{r_{k}^{2}}r_{p}+\frac{2}{(1-\mu_{k})^{2}}\frac{(\partial_{s}r_{k}+q_{k}z\partial_{z}r_{k})\partial_{z}r_{k}}{r_{k}^{2}}m_{p},
\label{eq:wave eq for rp under self-similar coordinates}\allowdisplaybreaks\\[1em]
&
\partial_{s}\partial_{z}\Psi_{p}+q_{k}z\partial_{z}^{2}\Psi_{p}+q_{k}\partial_{z}\Psi_{p}-\frac{\mu_{k}}{1-\mu_{k}}\frac{(\partial_{s}r_{k}+q_{k}z\partial_{z}r_{k})\partial_{z}r_{k}}{r_{k}^{2}}\Psi_{p}\nonumber\\&=\frac{1}{r_{k}^{2}}\frac{\mu_{k}}{1-\mu_{k}}\Psi_{k}\partial_{z}r_{k}\partial_{s}r_{p}+\left(\frac{1}{r_{k}^{2}}\frac{\mu_{k}}{1-\mu_{k}}\Psi_{k}(\partial_{s}r_{k}+q_{k}z\partial_{z}r_{k})+k\right)\partial_{z}r_{p}\nonumber\\&\hspace{.4cm}-\left(\frac{\mu_{k}(3-2\mu_{k})}{r_{k}^{3}(1-\mu_{k})^{2}}\Psi_{k}(\partial_{s}r_{k}+q_{k}z\partial_{z}r_{k})\partial_{z}r_{k}\right)r_{p}+\frac{2\Psi_{k}(\partial_{s}r_{k}+q_{k}z\partial_{z}r_{k})\partial_{z}r_{k}}{r_{k}^{3}(1-\mu_{k})^{2}}m_{p},
\label{eq:wave eq for psip under self-similar coordinate}\allowdisplaybreaks\\[1em]
&(\partial_{s}+q_{k}z\partial_{z})m_{p}+\frac{r_{k}\left((\partial_{s}+q_{k}z\partial_{z})\phi_{k}\right)^{2}}{(\partial_{s}r_{k}+q_{k}z\partial_{z}r_{k})}m_{p} \nonumber\\&=(1-\mu_{k})\frac{r_{k}(\partial_{s}+q_{k}z\partial_{z})\phi_{k}}{(\partial_{s}+q_{k}z\partial_{z})r_{k}}(\partial_{s}+q_{k}z\partial_{z})\Psi_{p}-(1-\mu_{k})(\partial_{s}+q_{k}z\partial_{z})\phi_{k}\Psi_{p}\nonumber\\&\quad-(1-\mu_{k})\left(\frac{1}{2}\left(\frac{r_{k}(\partial_{s}+q_{k}z\partial_{z})\phi_{k}}{(\partial_{s}+q_{k}z\partial_{z})r_{k}}\right)^{2}+\frac{r_{k}(\partial_{s}+q_{k}z\partial_{z})\phi_{k}}{(\partial_{s}+q_{k}z\partial_{z})r_{k}}(\phi_{k}-ks)\right)(\partial_{s}+q_{k}z\partial_{z})r_{p}\nonumber\\&\quad +\left(\frac{r_{k}((\partial_{s}+q_{k}z\partial_{z})\phi_{k})^{2}}{2(\partial_{s}+q_{k}z\partial_{z})r_{k}}\mu_{k}+(1-\mu_{k})(\phi_{k}-ks)+(1-\mu_{k})\frac{kr_{k}(\partial_{s}+q_{k}z\partial_{z})\phi_{k}}{(\partial_{s}+q_{k}z\partial_{z})r_{k}}\right)r_{p},\label{eq:u transport eq for m under self-similar coordinates}\allowdisplaybreaks\\[1em]&
\partial_{z}m_{p}+\frac{r_{k}}{\partial_{z}r_{k}}(\partial_{z}\phi_{k})^{2}m_{p}-(1-\mu_{k})\frac{r_{k}}{\partial_{z}r_{k}}\partial_{z}\phi_{k}\partial_{z}\Psi_{p}+(1-\mu_{k})\partial_{z}\phi_{k}\Psi_{p}\nonumber\\=&-\left(\frac{1}{2}\left(\frac{r_{k}}{\partial_{z}r_{k}}\right)^{2}(\partial_{z}\phi_{k})^{2}+\frac{r_{k}}{\partial_{z}r_{k}}\partial_{z}\phi_{k}(\phi_{k}-ks)\right)(1-\mu_{k})\partial_{z}r_{p}\nonumber\\&+
\left(\frac{1}{2}\frac{r_{k}}{\partial_{z}r_{k}}(\partial_{z}\phi_{k})^{2}\mu_{k}+(1-\mu_{k})\partial_{z}\phi_{k}(\phi_{k}-ks)\right)r_{p}.\label{eq:v transport eq for m under self-similar coordinates}
\end{align}

To simplify the notation, we abbreviate equations \eqref{eq:wave eq for rp under self-similar coordinates}-\eqref{eq:v transport eq for m under self-similar coordinates} as a linear operator: \begin{equation*}
    \mathcal{P}(r_{p},\Psi_{p},m_{p}) = 0.
\end{equation*}
Using the local well-posedness result established in \cite{chris93}, one can deduce that the linearized Einstein-scalar field equations are also well-posed. In view of the transport equations of $m_{p}$, we can express $m_{p}$ as a non-local operator of $r_{p}$ and $\Psi_{p}$, integrating from the center in the $z$-direction. Therefore, the free initial data for the linearized equations is \begin{equation*}
    r_{p}|_{\Sigma_{-1}} = r_{p}^{0},\quad \Psi_{p}|_{\Sigma_{-1}} = \Psi_{p}^{0}.
\end{equation*}

\subsection{Function spaces}
We now recall the localized Hölder-type spaces used to measure the interior perturbations.
\begin{definition}
\label{def: definition of the localized holder space}
    Let $\alpha\in(1,2)$, $\gamma\in(-\frac{1}{2},\frac{1}{2})$, and $I:=[-1,0]$. We can define \begin{align}
        \mathcal{C}^{\alpha,\delta}_{N}(I): &=\left\{f(z):I\rightarrow\mathbb{R}| f\in C_{z}^{N}(I\backslash\{0\})\cap C_{z}^{1}(I),\ \vert z\vert^{j-\alpha}\frac{d^{j}f}{dz^{j}}\in C_{z}^{0,\delta},\ 2\leq j\leq N\right \},
        \\
        H^{\gamma}_{N}(I):&=\left\{f(z): I\rightarrow \mathbb{R}|f\in C_{z}^{N}(I\backslash\{0\})\cap W_{z}^{1,2}(I),\ (-z)^{j-2+\gamma}\frac{d^{j}f}{dz^{j}}\in L^{2}(I),\ 2\leq j\leq N\right\},
    \end{align}
    associated with the norms \begin{align}
        \Vert f\Vert_{\mathcal{C}_{N}^{\alpha,\delta}}: &= \sum_{j = 0}^{1}\left\Vert\frac{d^{j}f}{dz^{j}}\right\Vert_{L^{\infty}(I)}+\sum_{j = 2}^{N}\left\Vert \vert z\vert^{j-\alpha}\frac{d^{j}f}{dz^{j}}\right\Vert_{C^{0,\delta}_{z}(I)},\\
        \left\Vert f\right\Vert_{H^{\gamma}_{N}}:&=\sum_{j = 0}^{1}\left\Vert\frac{d^{j}f}{dz^{j}}\right\Vert_{L^{2}(I)}+\sum_{j = 2}^{N}\left\Vert \vert z\vert^{j-2+\gamma}\frac{d^{j}f}{dz^{j}}\right\Vert_{L^{2}(I)},
    \end{align}
    where $\Vert\cdot\Vert_{C_{z}^{0,\delta}}$ means the standard Hölder norm. When $\delta = 0$, $C_{z}^{0,0}$ will be reduced to the standard continuous function space. We simplify the notation $\mathcal{C}_{N}^{\alpha,0}$ to be $\mathcal{C}_{N}^{\alpha}$.

    Finally, we define the function space 
    \begin{equation}
        H_{N,a}^{\gamma} = \left\{f:\mathbb{R}_{+}\times I\rightarrow\mathbb{R}|f(s,\cdot)\in H_{N}^{\gamma}(I),\ \sup_{s\geq0}e^{(1+a)s}\left\Vert f(s,\cdot)\right\Vert_{H_{N}^{\gamma}}<\infty\right\}.
    \end{equation}

    associated with the norm \begin{equation}
        \left\Vert f\right\Vert_{H_{N,a}^{\gamma}} = \sup_{s\geq0}e^{(1+a)s}\left\Vert f(s,\cdot)\right\Vert_{H_{N}^{\gamma}}.
    \end{equation}
\end{definition}
\begin{remark}
The function spaces introduced above will be used to manifest the regularity aspect of the initial data. For instance, by Proposition~\ref{prop: estimate on the exact k self similar spacetime}, the interior initial data $(r_{k},\Psi_{k},m_{k})|_{\Sigma_{-1}^{(in)}}$ of the $k$-self-similar naked singularity spacetime constructed by Christodoulou \cite{chris94} belongs to the function space $\mathcal{C}_{\infty}^{2-}\times\mathcal{C}_{\infty}^{p_{k},\delta}(I)\times\mathcal{C}_{\infty}^{p_{k},\delta}(I)$ for all $\delta\in(0,1-O(k))$. The function spaces with $\alpha>p_{k}$ are considered to be of high regularity in the paper.
\end{remark}
An easy example of a function $f$ in the function space $\mathcal{C}_{N}^{\alpha,\delta}$ is $f = (-z)^{\alpha}$, where $f$ is smooth away from the point $z = 0$ and is of finite Hölder regularity at $z = 0$. The following decomposition lemma illustrates that the motivation of the function space $\mathcal{C}_{N}^{\alpha,\delta}$ is to generalize the function $(-z)^{\alpha}$.

\begin{lemma}{(\cite{singh2})}
\label{lemma: decomposition of the function spaces}
    For any $f\in\mathcal{C}_{N}^{\alpha,\delta}$ with $\delta>0$, there exist a unique constant $c_{0}$ and a unique function $f_{1}\in\mathcal{C}_{N}^{\alpha+\delta^{\prime}}$ with $\delta^{\prime}\in(0,\delta)$ such that \begin{equation}
        f = c_{0}\vert z\vert^{\alpha}+f_{1}.\label{eq: decompose holder function}
    \end{equation}
    Moreover, we have the estimate \begin{equation}
        \left\vert c_{0}\right\vert+\left\Vert f_{1}\right\Vert_{\mathcal{C}^{\alpha+\delta^{\prime}}}\lesssim \left\Vert f\right\Vert_{\mathcal{C}_{N}^{\alpha}}.
        \label{eq: estimate on the decomposition}
    \end{equation}
    For $\alpha = \frac{1}{1-k^{2}}$, the function $\vert z\vert^{\alpha}$ in \eqref{eq: decompose holder function} can be replaced by $\mr{\phi}$ and the estimate \eqref{eq: estimate on the decomposition} continues to hold.
\end{lemma}
\begin{proof}
The proof can be found in Section 2.9 of \cite{singh2}.
\end{proof}

As we have discussed before, in view of the linearized constraint equations \eqref{eq:linerized dum equation}-\eqref{eq:linearized dvm equation}, the gauge choice of the center $\Gamma$, and the boundary condition $m_{p}|_{\Gamma} = 0$, the quantity $m_{p}$ can be expressed in terms of an integral in the $v$-direction of $r_{p}$ and $\Psi_{p}$ from the center. Hence, the free interior initial data of the linearized equations $\mathcal{P}(r_{p},\Psi_{p},m_{p}) = 0$ on the initial characteristic hypersurface $\Sigma_{-1}^{(in)}$ consist of $(r_{p},\Psi_{p})\bigl|_{\Sigma_{-1}^{(in)}}$. Let \begin{equation*}
    (r_{p},\Psi_{p})|_{\Sigma_{-1}^{(in)}} = ( r_{p}^{0},\Psi_{p}^{0}).
\end{equation*}
For perturbations above the threshold, the choice of the interior initial data will be $(r_{p}^{0},\Psi_{p}^{0})\in\mathcal{C}_{N}^{\alpha^{\prime}}\times\mathcal{C}_{N}^{\alpha}$ with $p_{k}<\alpha<\frac{3}{2}$, $\alpha+\frac{1}{4}\leq\alpha^{\prime}<2$, and $N\geq 5$. For perturbations at the threshold, the choice of the initial data will be $(r_{p}^{0},\Psi_{p}^{0})\in\mathcal{C}^{\alpha^{\prime}}_{N}\times\mathcal{C}_{N}^{p_{k},\delta}$ with $p_{k}+\frac{1}{2}\leq\alpha^{\prime}<2$, $\delta\in(0,1-O(k))$, and $N\geq 5$.
\begin{remark}
We remark that the threshold in this context refers to the regularity of the scalar field in the background $k$-self-similar spacetime. Recall that the geometric quantity $r_{k}$ is in the function space $\mathcal{C}_{N}^{2-}$ and is $C^{2,k^{2}p_{k}}$. Consequently, in our choice of the initial interior perturbation for the linearized problem, the regularity requirement of $r_{p}$ is already lower than that of the background geometry. Nevertheless, the results of this paper show that the linearized stability of the $k$-self-similar naked singularity spacetime depends crucially on the regularity of the scalar field perturbation, while being largely insensitive to the regularity of $r_{p}$.
\end{remark}
\begin{remark}
Due to the constraint equations of the Einstein-scalar field equation under the spherically symmetric assumption, the regularity of $r$ is closely related to the regularity of the scalar field $\phi$. In particular, in the nonlinear problem, once the initial regularity of $\phi$ is given, the regularity of both $\phi$ and $r$ is consequently fixed. Therefore, one should distinguish the regularity of the nonlinear problem and the freedom of imposing $r_{p}^{0}$ in the linearized problem. The aforementioned regularity freedom of $r_{p}^{0}$ in the linearized problem should be interpreted as follows: in estimating the decay rate of solutions to the linearized Einstein-scalar field equations, only the $\mathcal{C}_{N}^{\alpha^{\prime}}$-norm of $r_{p}$ enters the analysis. This does not imply that, in the nonlinear problem, the regularity of $r_{p}$ can be below $C^{2,k^{2}p_{k}}$.
\end{remark}

\subsection{Fourier--Laplace transform}
In this section, we recall the theory of the Fourier--Laplace transform and its inverse. 
\begin{definition}
    Let $f: \mathbb{R}\rightarrow \mathbb{R}$ be a Hölder continuous function satisfying \begin{enumerate}
        \item[(1)] $f(x) = 0$ for $x\leq0$;
        \item[(2)] There exist two constants $M>0$ and $C_{0}$ such that \begin{equation*}
            \vert f\vert\leq M e^{c_{0}x}.
        \end{equation*}
    \end{enumerate}
    Then we can define \begin{equation*}
        \hat{f}(\sigma) := \int_{0}^{\infty}f(x)e^{-\sigma x}dx,\quad \sigma\in\mathbb{C}
    \end{equation*}
    to be the Fourier--Laplace transform of $f$. Moreover, $\hat{f}$ is well-defined and holomorphic in the region $\Re{\sigma}>C_{0}$.
\end{definition}
Moreover, for any $a>C_{0}$, we have the following inverse Fourier--Laplace transform:
\begin{equation*}
        f(x) = \frac{1}{2\pi}\int_{a-i\infty}^{a+i\infty}e^{\sigma x}\hat{f}(\sigma)d\sigma.
    \end{equation*}
\section{Main results and the outline of the proof}
\label{sec: main result and proof outline}

\subsection{Linearized result for the interior region}
For the linearized Einstein-scalar field equations introduced in Section \ref{sec: linearized system} with initial data $(r_{p},\phi_{p})|_{\Sigma_{-1}^{(in)}} = (r_{p}^{0},\phi_{p}^{0})$, we can prove the following theorem.

\begin{theorem}[Linear stability for perturbations at the threshold]
\label{thm: linearized result at the threshold}
Fix $k^{2}$ to be a sufficiently small non-zero number. Let $\delta_{0}\in(0,1)$ be any value such that $\phi_{k}|_{\Sigma_{-1}^{(in)}}\in\mathcal{C}_{N}^{p_{k},\delta_{0}}$. Then for any $\delta\in(0,\delta_{0}]$, $\beta\in(p_{k}+\frac{1}{4},2)$, and $N\geq 5$, let $(r_{p},\phi_{p},m_{p})$ be the solution to the spherically symmetric Einstein-scalar field equations linearized around the $k$-self-similar naked singularity interior with the initial data $(r_{p},\phi_{p})|_{\Sigma_{-1}^{(in)}} = (r_{p}^{0},\phi_{p}^{0})\in\mathcal{C}_{N}^{\beta}\times \mathcal{C}_{N}^{p_{k},\delta}$. According to Lemma \ref{lemma: decomposition of the function spaces}, we further assume that \begin{equation*}
    \phi_{p}^{0} =c_{0}\phi_{k}|_{\Sigma_{-1}^{(in)}}+\phi_{1}= c_{0}\mr{\phi}+\phi_{1},
\end{equation*}
with $\phi_{1}\in\bar{\mathcal{C}}_{N}^{p_{k}+\delta^{\prime}}$ for any $\delta^{\prime}\in(0,\min\{\beta-\frac{1}{4}-p_{k},\delta\})$. Then there exist constants $c_{\infty}$ depending on the initial data and $C$ depending on $k$, such that for any $\alpha^{\prime}\in(p_{k},\min\{\frac{3}{2},p_{k}+\delta^{\prime}\})$, the solution $(r_{p},\phi_{p},m_{p})$ satisfies the following bounds: \begin{align}
    \sum_{0\leq i+j\leq 1}\left\vert\partial_{u}^{i}\partial_{v}^{j}(\phi_{p}-c_{\infty})\right\vert&\leq C\left(\left\Vert r_{p}^{0}\right\Vert_{\mathcal{C}_{N}^{\beta}}+\left\Vert\phi_{p}^{0}\right\Vert_{\mathcal{C}_{N}^{p_{k},\delta}}\right)(-u)^{\alpha^{\prime}q_{k}-1-i-q_{k}j},\\\sum_{0\leq i+j\leq 1}\left\vert\partial_{u}^{i}\partial_{v}^{j}(r_{p}-c_{0}r_{k})\right\vert&\leq C\left(\left\Vert r_{p}^{0}\right\Vert_{\mathcal{C}_{N}^{\beta}}+\left\Vert\phi_{p}^{0}\right\Vert_{\mathcal{C}_{N}^{p_{k},\delta}}\right)(-u)^{\alpha^{\prime}q_{k}-i-q_{k}j},\\\sum_{0\leq i+j\leq 1}\left\vert\partial_{u}^{i}\partial_{v}^{j}(m_{p}-c_{0}m_{k})\right\vert&\leq C\left(\left\Vert r_{p}^{0}\right\Vert_{\mathcal{C}_{N}^{\beta}}+\left\Vert\phi_{p}^{0}\right\Vert_{\mathcal{C}_{N}^{p_{k},\delta}}\right)(-u)^{\alpha^{\prime}q_{k}-i-q_{k}j}.
\end{align}
Moreover, for any fixed $k$, we have the following estimate on the constant $c_{\infty}$ in terms of the initial data \begin{equation}
    \vert c_{\infty}\vert\lesssim_{k}\left\Vert r_{p}^{0}\right\Vert_{\mathcal{C}_{N}^{\beta}}+\left\Vert\phi_{p}^{0}\right\Vert_{\mathcal{C}_{N}^{p_{k},\delta}}.
\end{equation}
\end{theorem}

The proof of the above theorem heavily relies on the following theorem, proving the linearized stability result for perturbations above the threshold.

\begin{theorem}[Linear stability for perturbations above the threshold]
\label{thm: linearized result}
    For $k\neq0$ sufficiently small, any $\alpha\in(p_{k},\frac{3}{2})$, and $\beta\in(\alpha+\frac{1}{4},2)$, let $(r_{p},\phi_{p},m_{p})$ be the solution to the spherically symmetric Einstein-scalar field equations linearized around the $k$-self-similar naked singularity interior with the initial data $(r_{p},\phi_{p})|_{\Sigma_{-1}^{(in)}} = (r_{p}^{0},\phi_{p}^{0})\in\mathcal{C}_{N}^{\beta}\times\mathcal{C}_{N}^{\alpha}$ for $N\geq 5$. Then there exist constants $c_{\infty}$ depending on the initial data and $C$ depending on $k$, such that for any $\alpha^{\prime}\in(p_{k},\alpha)$, the solution $(r_{p},\phi_{p},m_{p})$ satisfies the following bounds: \begin{align}
        \sum_{0\leq i+j\leq 1}\left\vert\partial_{u}^{i}\partial_{v}^{j}(\phi_{p}-c_{\infty})\right\vert\leq& C\left(\left\Vert r_{p}^{0}\right\Vert_{\mathcal{C}_{N}^{\beta}}+\left\Vert \phi_{p}^{0}\right\Vert_{\mathcal{C}_{N}^{\alpha}}\right)(-u)^{\alpha^{\prime}q_{k}-1-i-q_{k}j},\\\sum_{0\leq i+j\leq 1}\left\vert\partial_{u}^{i}\partial_{v}^{j}r_{p}\right\vert\leq& C\left(\left\Vert r_{p}^{0}\right\Vert_{\mathcal{C}_{N}^{\beta}}+\left\Vert \phi_{p}^{0}\right\Vert_{\mathcal{C}_{N}^{\alpha}}\right)(-u)^{\alpha^{\prime}q_{k}-i-q_{k}j},\\\sum_{0\leq i+j\leq 1}\left\vert\partial_{u}^{i}\partial_{v}^{j}m_{p}\right\vert\leq& C\left(\left\Vert r_{p}^{0}\right\Vert_{\mathcal{C}_{N}^{\beta}}+\left\Vert \phi_{p}^{0}\right\Vert_{\mathcal{C}_{N}^{\alpha}}\right)(-u)^{\alpha^{\prime}q_{k}-i-q_{k}j}.
    \end{align}
    Moreover, for any fixed $k$, we have the following estimate on the constant $c_{\infty}$ in terms of the initial data: \begin{equation}
        \vert c_{\infty}\vert\lesssim_{k}\left\Vert r_{p}^{0}\right\Vert_{\mathcal{C}_{N}^{\beta}}+\left\Vert \phi_{p}^{0}\right\Vert_{\mathcal{C}_{N}^{\alpha}}.
    \end{equation}
\end{theorem}

\noindent We will prove these two theorems in Section \ref{leading order expansion}.

\subsection{Trivial mode from the translation invariance}
\label{sec: trivial mode from the translation invariance}
To provide sufficient motivation for our proof outline, in this section, we address one of the main difficulties in the proof, namely, the trivial mode generated by the translation invariance of the Einstein-scalar field equations.

Recall that under the global double-null coordinates $(u,v)$ introduced in Section \ref{prelims:sec}, the interior of the $k$-self-similar naked singularity solution $(r_{k},m_{k},\phi_{k})$ to the spherically symmetric Einstein-scalar field equations has the following behaviors when approaching the singularity $(0,0)$ for $0\leq i+j\leq 1$:\begin{equation*}
    \partial_{u}^{i}\partial_{v}^{j} r_{k}\approx (-u)^{1-i-q_{k}j},\quad\partial_{u}^{i}\partial_{v}^{j}m_{k}\approx (-u)^{1-i-q_{k}j},\quad \partial_{u}^{i}\partial_{v}^{j}\phi_{k}\approx (-u)^{-i-q_{k}j}.
\end{equation*}
Let $(r,m,\phi)$ be the solution of the perturbed initial data and $(r_{p},m_{p},\phi_{p}):=(r-r_{k},m-m_{k},\phi-\phi_{k})$ be the collection of the perturbed quantities. A naive approach, which has been widely employed in the study of nonlinear wave equations, is to construct a suitable vector field and then use the energy method together with Sobolev embedding to show that the perturbed quantities $(r_{p},m_{p},\phi_{p})$ behave better than those of the $k$-self-similar solution. In particular, one has the following expectation \begin{align}
    &\left\vert\partial_{u}^{i}\partial_{v}^{j}r_{p}\right\vert\lesssim (-u)^{1-i-q_{k}j+\epsilon},\quad \left\vert\partial_{u}^{i}\partial_{v}^{j}m_{p}\right\vert\lesssim (-u)^{1-i-q_{k}j+\epsilon},\label{eq: expectation for rp and mp}\\& \left\vert\partial_{u}\phi_{p}\right\vert\lesssim (-u)^{-1+\epsilon},\quad \left\vert\partial_{v}\phi_{p}\right\vert\lesssim (-u)^{-q_{k}+\epsilon},
\end{align}
for $\epsilon$ small enough.

However, since the wave equation for $\phi$ \eqref{eq:wave equation for phi} degenerates at the center $r = 0$, to avoid this degeneracy,  it is convenient to work with $\psi = r\phi$ instead of $\phi$. Then the expectation for $\psi_{p}$ becomes \begin{equation*}
    \left\vert\partial_{u}^{i}\partial_{v}^{j}\psi_{p}\right\vert\lesssim (-u)^{1-i-q_{k}j+\epsilon}.
\end{equation*}

As we have discussed extensively in Section \ref{prelims:sec}, the spherically symmetric Einstein-scalar field equations exhibit two important symmetries: the scaling symmetry and the translation symmetry. These symmetries, on the one hand, play a crucial role in the construction of the $k$-self-similar naked singularities in \cite{chris94}. On the other hand, they create a severe obstacle to demonstrating the nonlinear stability under perturbations above the threshold. To address this, under the $(r,m,\psi)$-formulation of solutions to the Einstein-scalar-field equations, if $(r,m,\psi)$ is a solution, then so is $(r,m,\psi+cr)$ for any $c\in\mathbb{R}$. Although these two solutions are essentially the same up to a translation, they correspond to different initial data for $\psi$, differing by a multiple of $r$. By our gauge choice $r|_{\Sigma_{-1}} = r_{k}\in C^{2}$ on the initial hypersurface, the regularity of this difference lies above the threshold. Consequently, the regularity constraint on our initial perturbation cannot detect this trivial mode generated by the translation symmetry. Since $r$ only behaves like \begin{equation*}
    \partial_{u}^{i}\partial_{v}^{j}r = \partial_{u}^{i}\partial_{v}^{j}r_{k}+\partial_{u}^{i}\partial_{v}^{j}r_{p}\approx (-u)^{1-i-q_{k}j},
\end{equation*}
one can generically establish only the following bound for $\psi_{p}$:
\begin{equation*}
\left\vert\partial_{u}^{i}\partial_{v}^{j}\psi_{p}\right\vert\lesssim (-u)^{1-i-q_{k}j}.
\end{equation*}
Put it differently, the trivial mode $r = r_{k}+r_{p}$ generated by the translation invariance prevents us from obtaining a sharper bound for $\psi_{p}$. Consequently, even for arbitrarily small perturbations of regularity above the threshold, establishing improved behavior for $\psi_{p}$ requires a method to exclude this trivial mode:
\begin{equation*}
    \left\vert\partial_{u}^{i}\partial_{v}^{j}\left(\psi_{p}-cr_{k}\right)\right\vert\lesssim (-u)^{1-i-q_{k}j+\epsilon}.
\end{equation*} 
\subsection{A recap of the study on a linear wave equation}
\label{sec: review on the linear wave}
In fact, even at the level of a linear wave equation on the $k$-self-similar naked singularity spacetime, the same difficulty of excluding the trivial mode presents. If one considers the wave equation for the quantity $\psi = r_{k}\phi$, the linear wave equation will be reduced to \begin{equation}
    \partial_{s}\partial_{z}\psi+q_{k}z\partial_{z}^{2}\psi+q_{k}\partial_{z}\psi+V_{k}(z)\psi = 0\label{eq: linear wave equation for psi}
\end{equation}
under the self-similar coordinates, where the potential term $V_{k}(z)$ is a function only depending on $z$ due to the $k$-self-similarity and has the following expansion:\begin{equation}
    V_{k}(z) = \frac{\mu_{k}}{1-\mu_{k}}\frac{\partial_{u}r_{k}\partial_{v}r_{k}}{r_{k}^{2}} = \omega_{k}k^{2}+O_{k}(z)+O_{k}\left(\left\vert z\right\vert^{\frac{1}{1-k^{2}}}\right).\label{eq: the expansion of Vk}
\end{equation}
Since $\phi = const$ is always a solution to the linear wave equation, the equation \eqref{eq: linear wave equation for psi} admits a trivial solution $\psi = cr_{k}$, which prevents us from seeing a better-than-self-similar rate for $\psi$.

In the work \cite{singh2}, Singh carefully analyzed the behavior of the solution to the linear wave equation and provided a scattering approach to exclude the trivial mode. Since the approach in the present paper is closely related to the work of the first author \cite{singh2}, we give a detailed review here.
\paragraph{Step 1: Backward scattering}
Let $t = -\frac{1}{2}\log(-u)-\frac{1}{2q_{k}}\log(-v)$ and $x = \frac{1}{2}\log(-u)-\frac{1}{2q_{k}}\log(-v)$, which we refer to as hyperbolic coordinates. Then the linear wave equation \eqref{eq: linear wave equation for psi} takes the following standard Cauchy form under $(t,x)$ coordinates \begin{equation*}
    \partial_{t}^{2}\psi-\partial_{x}^{2}\psi+V_{k}(x)e^{-2q_{k}x}\psi = 0.
\end{equation*}
Note that the center $\Gamma$ corresponds to the curve $x = 0$ and the singularity horizon $\{v = 0\}$ corresponds to $x = \infty$ (which can be made rigorous by a suitable compactification process) under the hyperbolic coordinates.

To translate the above characteristic initial value problem to a standard Cauchy data formulation of a wave equation in the one-dimensional scattering theory, we first give free data on the characteristic hypersurface $\{v = 0,\ u\leq-1\}$ and propagate the solution backwards. By choosing the free data properly, one can draw a connection between the regularity of the initial perturbation and the asymptotic tail on the Cauchy hypersurface $\{t = 0\}$: \begin{equation}
    \partial_{t}^{i}\partial_{x}^{j}\psi|_{t = 0}\approx e^{-\alpha q_{k}x},
\end{equation}
where $\alpha$ is the regularity of the initial data. A function space collecting functions with such an asymptotic tail is denoted by $\mathcal{D}_{\infty}^{\alpha,l}$.  The figure \ref{fig:myfigure} depicts the process of the backward scattering.
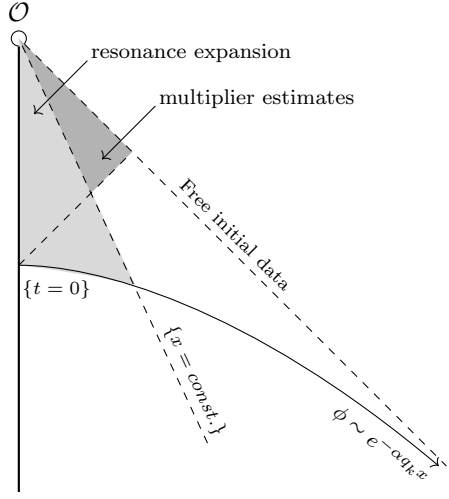
\begin{figure}
 \centering


\begin{tikzpicture}[scale=1]

\node[circle, label= above:{$\mathcal{O}$}, draw, inner sep =0, minimum size = .2cm] (n2) at (90:6) {};

\fill[gray!30] (90:3) -- (90:6) -- ($(90:3) +  (1.54,-.26)$ ) -- ($(90:3) +  (1,-.15)$ ) -- ($(90:3) +  (.4,-.05)$ );

\fill[gray!60] ($(90:3) + (45:1.35) $ ) --($(90:3) + (45:2.12) $ ) -- (90:6);

\draw[dashed] (90:3) -- ($(90:3) + (45:2.1) $ );

\coordinate (n1) at (0,0) {};

\draw[thick] (n1) -- (n2);
\draw[dashed] (n2) -- ($(n2) + (-45:8) $);

\coordinate (n3) at (90:3);

\draw[ ->] (n3) .. controls ($(n3) + (1,0)  $) and ($(n3) + (3,-.5)  $) .. ($(n2) + (-45:8) + (-.1,0) $);

\node[label = {[yshift=.5mm, xshift=-5mm]below:\scriptsize $\{t=0\}$   }] at ($(n3) + (1,0)  $) {};

\node[label = {[yshift=6mm, xshift=-2mm, rotate=-40]below:\footnotesize $\phi \sim e^{-\alpha q_{k} x}$   }] at ($(n2) + (-45:8) + (-.5,0) $) {};

\draw[dashed] (n2) -- ($(n2) + (-65:6) $ ) ;
\node[label = {[rotate=-65, xshift=-8mm, yshift=-1mm]above:\scriptsize{$\{x = const.\} $} } ] at ($(n2) + (-65:6) $ ) {};

\node[label = {[rotate=-45, xshift=0mm, yshift=-1mm]above:\scriptsize{Free initial data } } ] at ($(n2) + (-45:4) $) {};

\coordinate (b1) at ($(90:5)+(.2,0) $);
\coordinate (b2) at ($(90:5)+(.2,0) + (45:1) $);
\draw[->] (b2) -- (b1);
\node[label={[xshift=-2mm, yshift=1mm]right:\footnotesize{resonance expansion}    }] at (b2) {};

\coordinate (b3) at ($(90:4.4)+(1.1,0) $);
\coordinate (b4) at ($(90:4.4)+(1.1,0) + (45:1) $);
\draw[->]  (b4) -- (b3);
\node[label={[xshift=-2mm, yshift=1mm]right:\footnotesize{multiplier estimates}    }] at (b4) {};

\end{tikzpicture}

 \caption{Backward scattering}
  \label{fig:myfigure}
\end{figure}
The approach to proving the backward scattering is to establish the energy boundedness first and then view the linear wave equation as two transport equations in two characteristic directions to get pointwise asymptotic bounds. \textbf{However, this step is not robust when considering the linearized Einstein-scalar field equations around the $k$-self-similar naked singularity spacetime. In the following, we give a quick guide on redoing the argument in \cite{singh2} purely under the self-similar coordinates.}

\paragraph{Step 2: Scattering theory}
We consider the equation \eqref{eq: linear wave equation for psi} with prescribed initial data on $\{s = 0,\ -1\leq z\leq0\}$ in the function space $\mathcal{C}^{\alpha}$ with $\alpha>p_{k}$. Since the solution is only defined on the upper half plane $s\geq 0$, to extend the solution to the lower half plane and use the standard theory of the Fourier--Laplace transform, we commute the equation \eqref{eq: linear wave equation for psi} with a cut-off function supported on $s\geq 1$. Then the problem is reduced to the following inhomogeneous problem \begin{align*}
    &\partial_{s}\partial_{z}\psi+q_{k}z\partial_{z}^{2}\psi+q_{k}\partial_{z}\psi+V_{k}(z)\psi = \partial_{z}F, \\& \psi(s,z)\equiv0, \quad s\leq 0.
\end{align*}
The forcing term $F$ is determined by the truncation function and the initial data. In particular, we have that $F\in\mathcal{C}^{\alpha}$ is compactly supported in $s$.

Taking the Fourier--Laplace transform, the problem is further reduced to studying the operator: \begin{equation}
  \mathcal{L}[\hat\psi]  = q_{k}z\frac{d^{2}\hat{\psi}}{dz^{2}}+(q_{k}+\sigma)\frac{d\hat{\psi}}{dz}+V_{k}(z)\hat{\psi} = \partial_{z}\hat{F}.\label{eq: equation for scattering theory}
\end{equation}
We follow the standard one-dimensional scattering theory illustrated in \cite{dyatlovzworski} to study the solution to the equation \eqref{eq: equation for scattering theory}.

The asymptotic analysis on the equation $\mathcal{L}[\hat\psi] = 0$ will give a solution basis near $z = 0$: \begin{equation*}
    \hat{\psi}_{out,\sigma}(z)\approx 1,\quad \hat{\psi}_{in,\sigma}(z)\approx (-z)^{-\frac{\sigma}{q_{k}}}.
\end{equation*}
Moreover, we require that the outgoing solution $\hat{\psi}_{out,\sigma}$ and the ingoing solution $\hat{\psi}_{in,\sigma}$ be holomorphic functions of $\sigma$ in the domain of interest $\Re\sigma>-\frac{3}{2}$ for fixed $z$. On the one hand, the solution $\psi$ to \eqref{eq: linear wave equation for psi} should be regular at the center, i.e., $\psi$ satisfies the Dirichlet boundary condition $\psi|_{z = -1} = \hat{\psi}|_{z = -1} = 0$. On the other hand, for $\Re\sigma>-1$, the regularity of the initial data of $\psi\in\mathcal{C}^{\alpha}$ with $\alpha>p_{k}$ requires $\hat\psi$ to satisfy the outgoing boundary condition. Hence, if the Dirichlet solution $\hat{\psi}_{dir,\sigma}$ and the outgoing solution $\hat{\psi}_{out,\sigma}$ to $\mathcal{L}[\hat{\psi}] = 0$ can be constructed and are not parallel, we can use Green's formula to write down the \textit{scattering resolvent}, i.e., a solution to \eqref{eq: equation for scattering theory} with the Dirichlet boundary condition and the outgoing boundary condition: \begin{align*}
    \hat{\psi}(\sigma,z) = R(\sigma)[\hat{F}] =& \hat{\psi}_{out,\sigma}(z)\int_{-1}^{z}\frac{1}{q_{k}(-z^{\prime})}\frac{\partial_{z}\hat{F}\hat{\psi}_{dir,\sigma}}{\mathcal{W}[\hat{\psi}_{dir,\sigma},\hat{\psi}_{out,\sigma}]}(z^{\prime})dz^{\prime}+\hat{\psi}_{dir,\sigma}\int_{z}^{0}\frac{1}{q_{k}(-z^{\prime})}\frac{\partial_{z}\hat{F}\hat{\psi}_{out,\sigma}}{\mathcal{W}[\hat{\psi}_{dir,\sigma},\hat{\psi}_{out,\sigma}]}dz^{\prime}.
\end{align*}
where $\mathcal{W}[f,g] = f^{\prime}g-fg^{\prime}$ is the Wronskian of $f$ and $g$. However, in the whole domain of interest $\Re\sigma>-\frac{3}{2}$, the above expression might not be integrable near $z = 0$. This can be resolved by the integration by parts. We have \begin{equation*}
    \hat{\psi}(\sigma,z) = R(\sigma)[\hat{F}] =\hat{\psi}_{out,\sigma}\int_{-1}^{z}\partial_{z}\left(\frac{\hat{\psi}_{dir,\sigma}}{\hat{\psi}_{out,\sigma}}\right)G(z^{\prime})\partial_{z}\hat{F}(z^{\prime})+\frac{\hat{\psi}_{dir,\sigma}}{\hat{\psi}_{out,\sigma}}G(z^{\prime})\partial_{z}^{2}\hat{F}dz^{\prime}+\hat{\psi}_{dir,\sigma}\int_{z}^{0}G(z^{\prime})\partial_{z}^{2}\hat{F}(z^{\prime})dz^{\prime},
\end{equation*}
where $G(z)^{\prime} = \frac{1}{q_{k}z}\frac{\hat{\psi}_{out,\sigma}}{\mathcal{W}[\hat{\psi}_{dir,\sigma},\hat{\psi}_{out,\sigma}]}$.

Moreover, the scattering resolvent $R(\sigma)[\hat{F}]$ can be viewed as a meromorphic function of $\sigma$ on $\Re\sigma>-\frac{3}{2}$, with poles at the values of $\sigma$ where $\hat{\psi}_{dir,\sigma}$ is linearly dependent on $\hat{\psi}_{out,\sigma}$. Such values of $\sigma$ are called the \textit{scattering resonances}.

The original information of $\psi$ can be recovered by taking the inverse Fourier--Laplace transform in a suitable contour where $R(\sigma)[\hat{F}]$ is well-defined
\begin{equation*}
    \psi(s,z) = \int_{\Re\sigma = \frac{1}{10}}R(\sigma)[\hat{F}]d\sigma.
\end{equation*}
\paragraph{Step 3: Construction of the outgoing and ingoing solutions via the Volterra iteration}
However, although the asymptotic analysis will immediately give the asymptotic behaviors of the outgoing solution and the ingoing solution, a rigorous mathematical construction of the outgoing and ingoing solutions is highly nontrivial. Since the construction itself will give insight into the location of the scattering resonance, we give a sketch here.

 We notice that the equation $\mathcal{L}[\hat{\psi}] = 0$ in the scattering analysis is a singular equation with a regular singularity $z = 0$. In the standard theory of the regular singularity, the zeroth-order term in $\mathcal{L}[\hat\psi] = 0$ is viewed as a lower-order term. This motivates us to consider the following equation as an approximation equation
\begin{equation}
    q_{k}z\frac{d^{2}\hat{\psi}}{dz^{2}}+\left(q_{k}+\sigma\right)\frac{d\hat{\psi}}{dz} = 0,\label{eq: toy mode equation for a single wave equation}
\end{equation}
whose solution basis can be easily found: $\{1,\ (-z)^{-\frac{\sigma}{q_{k}}}\}$.

Treating the zeroth-order term with potential $V_{k}(z)$ as an error term and using the exact solutions to \eqref{eq: toy mode equation for a single wave equation}, we can write the outgoing solution to $\mathcal{L}[\hat\psi] = 0$ in the following integral form by using Green's formula: 
\begin{equation}
\begin{aligned}
    \hat{\psi}(\sigma,z) =& 1+\int_{z}^{0}\underbrace{\left(1-\left(\frac{(-z^{\prime})}{(-z)}\right)^{\frac{\sigma}{q_{k}}}\right)\frac{V_{k}(z^{\prime})}{\sigma}}_{P(z,z^{\prime})}\hat{\psi}(\sigma,z^{\prime})dz^{\prime}\\=&
    1+\sum_{n = 0}^{\infty}\int_{z}^{\infty}\int_{z_{1}}^{\infty}\cdots\int_{z_{n-1}}^{\infty}P(z,z_{1})P(z_{1},z_{2})\cdots P(z_{n-1},z_{n})dz_{1}\cdots dz_{n},
\end{aligned}
\label{eq: volterra iteration in the toy model equation}
\end{equation}
where the second identity is from iterating the first identity, and it is called the \textit{Volterra iteration}. The ingoing solution can be constructed similarly. However, the above iteration will only be convergent in the domain $\Re\sigma>-q_{k}$, which does not cover the whole domain of interest $\Re\sigma>-\frac{3}{2}$ for our later analysis. In fact, the lower bound for $\Re\sigma$ in the domain of the convergence is closely related to the asymptotic tail of the error term. To extend the domain of existence of outgoing and ingoing solutions in our Volterra iteration, we take out the leading order expansion of $V_{k}$ and put it into the approximation equations \begin{equation}
    q_{k}z\frac{d^{2}\hat{\psi}}{dz^{2}}+(q_{k}+\sigma)\frac{d\hat{\psi}}{dz}+\omega_{k}k^{2}\hat{\psi} = 0. \label{eq: second toy model equation}
\end{equation}
Since the resulting error term has a faster decay rate when $z\rightarrow 0$, the domain of convergence in the corresponding Volterra iteration will be larger and sufficient for our later analysis.

To construct outgoing solutions to \eqref{eq: second toy model equation}, we expand the solution into power series:\begin{equation*}
    \hat\psi = \sum_{n = 0}^{\infty}a_{n}(-z)^{n},
\end{equation*} 
and then plug the expansion into the equation to obtain the following recurrence relation: \begin{equation}
    \left(n+\frac{\sigma}{q_{k}}\right)na_{n} = \frac{\omega_{k}k^{2}}{q_{k}}a_{n-1}.\label{eq: recurence relation in outline of the proof}
\end{equation}
Hence, given the value of $a_{0}$, we can determine $a_{n}$ for $n\geq 1$ inductively. The outgoing and ingoing solutions to $\mathcal{L}[\hat{\psi}] = 0$ will be given by the standard Volterra iteration as in \eqref{eq: volterra iteration in the toy model equation}.
\paragraph{Step 4: Reading the location of the scattering resonance from the recurrence relation}
Since only the coefficient for $a_{1}$ in the above recurrence relation will degenerate in our domain of interest $\Re\sigma>-\frac{3}{2}$, viewing each $a_{n}$ as a function of $\sigma$, this degeneration means $a_{1}$ is a meromorphic function of $\sigma$ with a single pole at $\sigma = -q_{k}$, and so for each $a_{n}$ with $n\geq 1$. However, this single pole can be eliminated by renormalizing the value of $a_{0}$. Choosing $a_{0} = 1+\frac{\sigma}{q_{k}}$, one can immediately deduce from the recurrence relation that each $a_{n}$ is a holomorphic function of $\sigma$ in the domain $\Re\sigma>-\frac{3}{2}$. Heuristically, since $k$ is taken to be sufficiently small, one can think that the outgoing solution to $\mathcal{L}[\hat\psi] = 0$ is well-approximated by $\hat{\psi}_{out,\sigma}^{(approx)}(z):= a_{0}= 1+\frac{\sigma}{q_{k}}$. Computing the value of $\hat{\psi}_{out,\sigma}^{(approx)}(-1) = 1+\frac{\sigma}{q_{k}}$ at the center $\Gamma$, which has a single zero point at $\sigma = -q_{k}$, suggests that the location of the scattering resonance will be in a neighborhood of $\sigma = -q_{k}$ and has a single multiplicity. In the work \cite{singh2}, an argument using complex analysis validates this. To determine the exact location of the scattering resonance, plugging the trivial mode solution $\psi = r_{k} = e^{-s}\mr{r}(z)$ into the equation \eqref{eq: linear wave equation for psi}, one can show that $\mr{r}(z)$ is the Dirichlet solution to $\mathcal{L}[\hat{\psi}]$ satisfying the outgoing boundary condition at $\sigma= - 1$.
\begin{remark}
    The $k$-smallness of the potential term $V_{k}$ is used here to determine the location and the multiplicity of the scattering resonance. Since $V_{k}$ is globally $k$-small, we can use $\hat{\psi}_{out,\sigma}^{(approx)}$ to approximate the outgoing solution $\hat{\psi}_{out,\sigma}$. A rigorous argument can show that, the difference between $\hat{\psi}_{out,\sigma}$ and $\hat{\psi}_{out,\sigma}^{(approx)}$ is at most of size $k$.\label{rmk: esstential k smallness}
\end{remark}
\begin{remark}
\label{rmk: the role of the regularity}
    In the rigorous construction of the outgoing solution, the regularity aspect of $V_{k}$ also plays a crucial role. In view of \eqref{eq: the expansion of Vk}, the sub-leading term in the expansion of $V_{k}$ is of order $O(z)$. Therefore, after taking out the leading order term of $V_{k}$ and constructing a solution basis of \eqref{eq: second toy model equation}, the sub-leading behavior of $V_{k}$ ensures convergence of the Volterra iteration in the region $\Re\sigma>-\frac{3}{2}$. However, if the regularity of $V_{k}$ was only $C^{0,\epsilon}$, the sub-leading behavior would only be $O((-z)^{\epsilon})$. In this case, the above argument breaks down, as removing the leading-order term of $V_{k}$ would no longer suffice to guarantee convergence of the Volterra iteration in the domain $\Re\sigma>-\frac{3}{2}$.
\end{remark}
\paragraph{Step 5: Contour deformation and proof concluded}
After determining the location of the scattering resonance and its multiplicity, a standard approach in the one-dimensional scattering theory is to take the inverse Fourier--Laplace transform in the contour $\Re\sigma = \frac{1}{10}$ (the parts $\Gamma_{1,R}^{-}$, $\Gamma_{1,R}^{+}$, and the dashed line in Figure \ref{fig: contour deformation figure}). Since we know the only pole is at $\sigma = -1$, we deform the contour as in Figure \ref{fig: contour deformation figure}. By the residue theorem, the circular contour around $-1$ will be used to exclude the trivial mode $cr_{k}$, and the rest of the contour will give a better-than-self-similar decay rate. 
\begin{figure}[h]

\centering
\begin{tikzpicture}[scale=1.2]

\def\R{3}
\def\eta{-0.8}
\def\beta{1.8}
\def\eps{0.3}

\draw[->] (-4,0) -- (4,0) node[right] {$\Im \sigma$};
\draw[->] (0,-3) -- (0,1) node[above] {$\Re \sigma$};

\draw[thick]
(0,-1) circle (\eps);
\node at (0.4,-0.7) {$\Gamma_\epsilon$};

\draw[thick]
(\R,-\eta) -- (4,-\eta);
\node[right] at (4,-\eta) {$\Gamma_{1,R}^+$};

\draw[thick]
(-\R,-\eta) -- (-4,-\eta);
\node[left] at (-4,-\eta) {$\Gamma_{1,R}^-$};

\draw[thick]
(\R,-\eta) -- (\R,-\beta+\eps);
\node[right] at (\R,-1.3) {$\Gamma_{2,R}^+$};

\draw[thick]
(-\R,-\eta) -- (-\R,-\beta+\eps);
\node[left] at (-\R,-1.3) {$\Gamma_{2,R}^-$};

\draw[thick]
(-\R,-\beta+\eps) -- (\R,-\beta+\eps);
\node[below] at (0,-\beta+\eps) {$\Gamma_{3,R}$};

\fill (0,-1) circle (1.5pt);
\node[left] at (0,-1) {$-1$};
\draw[dashed] (-\R,-\eta)--(\R,-\eta);
\end{tikzpicture}
\caption{Oriented contours in the complex plane}
\label{fig: contour deformation figure}
\end{figure}
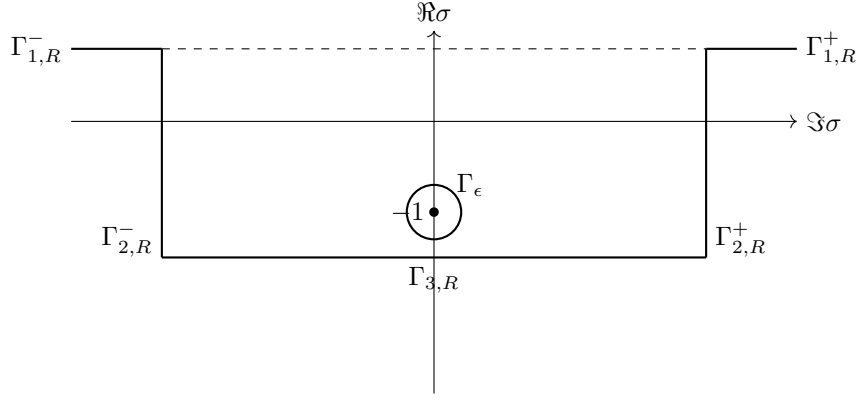

\subsection{Analysis on the linearized operator}
\label{sec: analysis on the linearized operator}
As mentioned before, a key step in proving the nonlinear stability is a careful analysis of the linearized Einstein-scalar field equations, denoted by $\mathcal{P}(r_{p},\Psi_{p},m_{p}) = 0$. We follow the same framework as in \cite{singh2}. In this section, we point out the main difficulties and their resolutions.

\subsubsection{Linear analysis for perturbations at the threshold}
One of the advantages of using $(r,\Psi,m)$ formulation is that all three quantities respect the scaling invariance of the Einstein-scalar field equations, i.e., if $(r,\Psi,m)$ is a solution, then so is $(ar,a\Psi,am)$ for any $a\in\mathbb{R}$. This invariance naturally induces a non-trivial kernel for the linearized operator $\mathcal{P}$, namely $\mathcal{P}(r_{k},\Psi_{k},m_{k}) = 0$. Therefore, for the linearized Einstein-scalar field equations with initial data at the threshold, using Lemma \ref{lemma: decomposition of the function spaces} and subtracting the kernel arising from the scaling invariance, we can reduce the problem to perturbations above the threshold. Therefore, in the following sections, we focus on perturbations above the threshold.

\subsubsection{The role of the energy estimate}
Although we have mentioned that it is impossible to prove a good upper bound for $\Psi$ by only using the energy estimate without excluding the trivial mode solution in Section \ref{sec: trivial mode from the translation invariance}, the vector field method still plays a role here. 

First, even at the level of the linear wave equation, the scattering approach emphasized in the previous section can only prove a better bound in the near-axis region. To extend the result to the entire interior region, one has to appeal to the energy estimate. The same situation also appears for the linearized Einstein-scalar field equations.

Second, applying the inverse Fourier--Laplace transform requires an a priori $L^{\infty}$-estimate for the solution. For the linear wave equation, this a priori estimate can be obtained without invoking any energy estimates, since the contour lies in the right half-plane. In the linearized Einstein-scalar field equations, however, some fake modes appear in the right half-plane; see Section \ref{sec: difficulty for the scattering} for an overview of the presence of these fake modes. Using the vector field method and Sobolev embedding, we obtain some non-trivial decay (Theorem \ref{thm:a priori energy estimate}), allowing the contour to be placed in the left half-plane and avoiding the fake modes.

In the energy estimate, although the vector $\partial_{s}$ in the self-similar coordinates corresponds to the natural conformal Killing vector field $K$ in the $k$-self-similar naked singularity spacetime, it turns out that $\partial_{s}$ is not useful in obtaining energy estimates. In fact, we use the null vector field $\partial_{z}$ extensively.

Our basic strategy in the energy estimate is to construct suitable multipliers for the wave equations for $r_{p}$ and $\Psi_{p}$. For $m_{p}$, we always use the $z$-transport equation \eqref{eq:v transport eq for m under self-similar coordinates} of $m_{p}$ to control the $L^{2}$ norm of $m_{p}(s,\cdot)$ on $z\in[-1,0]$ in terms of the $L^{2}$-norm of $r_{p}$ and $\Psi_{p}$.

To close the first-order energy estimate for $r_{p}$ and $\Psi_{p}$, one has to control the $L^{2}$-norm of $\frac{m_{p}}{(z+1)^{2}}$ by the first-order derivatives of $r_{p}$ and $\Psi_{p}$. However, near the axis, the quantity $\frac{m_{p}}{(z+1)^{2}}\approx \partial_{z}^{2}m_{p}(s,-1)$ in the sense of $L^{\infty}$-norm and $\partial_{z}^{2}m_{p}$ can only be controlled by the second-order derivatives of $r_{p}$ and $\Psi_{p}$. Consequently, it will be potentially dangerous to estimate $\frac{m_{p}}{(z+1)^{2}}$ without losing derivatives. However, by a careful use of the Hardy inequality, Fubini theorem, and the explicit form of $m_{p}$, one can bound the $L^{2}$-norm of $\frac{m_{p}}{(z+1)^{2}}$ by the $L^{2}$-norm of $\partial_{z}r_{p}$ and $\partial_{z}\Psi_{p}$; see Section \ref{sec: first-order energy estimate for mp} for details.

\subsubsection{Main difficulty in the scattering theory for the linearized Einstein-scalar field equations}
\label{sec: difficulty for the scattering}
Besides the technical difficulties arising in the energy estimates, additional substantial obstacles emerge in the scattering analysis of the linearized Einstein-scalar field equations.
\paragraph{Problematic $\partial_{v}\phi_{k}$ term}
The linear wave equation \eqref{eq: linear wave equation for psi} on the $k$-self-similar metric is only a toy model for the nonlinear problem or the linearized problem, since it does not capture the most dangerous behavior generated by $\partial_{v}\phi_{k}$ in the $k$-self-similar naked singularity spacetime. The presence of $\partial_{v}\phi_{k}$ is problematic for two reasons:\begin{itemize}
    \item (Loss of global smallness) In the linear wave equation \eqref{eq: linear wave equation for psi}, due to the absence of $\partial_{v}\phi_{k}$ and the presence of $\mu_{k}$, the potential term $V_{k}$ is a globally small quantity. This global smallness is used in an essential way in \cite{singh2} to determine the location of the scattering resonance (see Remark \ref{rmk: esstential k smallness}). However,  as we have mentioned in Section \ref{sec: pre on k naked singularity sapcetime}, $\partial_{v}\phi_{k}$ behaves like $\frac{1}{k}$ near the singular horizon, which is a large quantity if we take $k$ to be sufficiently small.
    \item (Bad regularity) As opposed to Remark \ref{rmk: the role of the regularity}, the quantity $\partial_{v}\phi_{k}$ has only $C^{0,p_{k}k^{2}}$ regularity, which generates a bad sub-leading term in the expansion. Consequently, this lack of enough regularity creates serious difficulties in constructing outgoing and ingoing solutions.
\end{itemize}

\paragraph{Non-trivial coupling between the geometry and the scalar field}
In the a priori energy estimates, the difficulties do not arise from the number of equations. That is, although we deal with two wave equations, once we understand how to do the energy estimate for the wave equation for $r_{p}$, the corresponding energy estimate for $\Psi_{p}$ will follow similarly. Surprisingly, in the analysis of the scattering resonance, the presence of coupled multiple equations introduces additional challenges. Similar to the analysis in Section \ref{sec: review on the linear wave}, the trivial solution arising from the translation symmetry to the linearized Einstein-scalar field equations $(r_{p},\Psi_{p},m_{p}) = (0,r_{k},0)$ will give scattering resonance at $\sigma = -1$. To exclude other scattering resonances, heuristically, in the construction of the outgoing solution to \eqref{eq: second toy model equation}, the renormalization factor $1+\frac{\sigma}{q_{k}}$ in the recurrence relation~\eqref{eq: recurence relation in outline of the proof} suggests that $\sigma = -1$ should be the only scattering resonance. 

To illustrate the difference between one equation and a system of equations, we consider the following toy model equations: \begin{equation}
    q_{k}z\frac{d^{2}f}{dz^{2}}+(1+\sigma)\frac{df}{dz}+\omega_{1}k^{2}f = 0,\quad q_{k}z\frac{d^{2}g}{dz^{2}}+\left(q_{k}+\sigma\right)\frac{dg}{dz}+\omega_{2}k^{2}g = 0.\label{eq: decoupled toy model equations}
\end{equation}
{The coefficients in these equations are obtained from the linearized Einstein-scalar field equations after neglecting the coupling terms.}

The corresponding Dirichlet solutions, should be a two-dimensional vector space with a basis $\{(f_{1},g_{1}),\ (f_{2},g_{2})\}$, where $(f_{1},g_{1})$ and $(f_{2},g_{2})$ satisfy the boundary conditions $$(f_{1},f_{1}^{\prime},g_{1},g_{1}^{\prime})|_{z =  -1} = (0,1,0,0),\quad(f_{2},f_{2}^{\prime},g_{2},g_{2}^{\prime})|_{z = -1} = (0,0,0,1),$$ respectively. The corresponding outgoing solutions should also be a two-dimensional vector space with a basis $\{(f_{out,\sigma},0),(0,g_{out,\sigma})\}$. The generalized notion of the scattering resonance should be the point of $\sigma$ where the determinant of the matrix \begin{equation*}
    \begin{bmatrix}
        f_{1}&f_{2}&f_{out,\sigma}&0\\g_{1}&g_{2}&0&g_{out,\sigma}\\f_{1}^{\prime}&f_{2}^{\prime}&f_{out,\sigma}^{\prime}&0\\g_{1}^{\prime}&g_{2}^{\prime}&0&g_{out,\sigma}^{\prime}
    \end{bmatrix}
\end{equation*}
is zero at $z = -1$. Using the Dirichlet boundary conditions, this can be translated to \begin{equation*}
    f_{out,\sigma}g_{out,\sigma} (-1)= 0.
\end{equation*}
Since in the toy model equations, $f$ and $g$ decouple, using the recurrence relation to construct the corresponding outgoing solutions, renormalizing the leading order coefficients to avoid the presence of poles in the domain $\Re\sigma>-\frac{3}{2}$, and taking $k$ to be sufficiently small, we can derive \begin{equation*}
    f_{out,\sigma}(0)\approx 1+\frac{\sigma+k^{2}}{q_{k}},\quad g_{out,\sigma}(0)\approx 1+\frac{\sigma }{q_{k}}.
\end{equation*}
This will suggest that the equation for determining the scattering resonance $f_{out,\sigma}g_{out,\sigma}(0) = 0$ is very likely to have two zero points in the domain $\Re\sigma>-\frac{3}{2}$, resulting in the failure of the original argument in \cite{singh2}.

{Since the equations in the above toy model problem decouple, the preceding argument suggests that, when the equations are considered individually, each equation has only one scattering resonance in the region $\Re\sigma>-\frac{3}{2}$. However, when the equations are considered as a system, the same argument suggests that there are essentially two scattering resonances, which may or may not be located at the same point; in the former case, this corresponds to a scattering resonance of multiplicity two. In the context of the linearized Einstein-scalar field equations, which are a system of coupled wave equations, the corresponding difficulty is expected to arise from the nontrivial coupling between the geometric quantity $r$ and the scalar field $\phi$.} On the one hand, if stability holds, one expects this system of wave equations to admit only one single scattering resonance from the translation symmetry of the Einstein-scalar field equations. On the other hand, the toy model suggests \textbf{a possible additional instability mechanism arising from the interaction between the geometry and the scalar field}, resulting in the presence of multiple scattering resonances.
\subsubsection{Formulation of the scattering analysis}
\label{sec: formulation of the scattering analysis}
In this section, we will state the formulation of the scattering analysis in terms of the self-similar coordinates $(s,z)$; see \eqref{eq:scattering equation for rp}-\eqref{scattering equation for mp} for the precise equations. Under the self-similar coordinates $(s,z)$, the trivial mode from the translation invariance will be $\sigma = -1$ and the domain of interest of $\sigma$ will be $\Re\sigma>-\frac{3}{2}$.

Under the $(r_{p},\Psi_{p},m_{p})$-formulation of the linearized Einstein-scalar field equations, we need to deal with the wave equation for $r_{p}$, the wave equation for $\Psi_{p}$, and two transport equations for $m_{p}$. {To apply the framework in~\cite{singh2}, it is natural to derive a wave equation for $m_{p}$ from those two transport equations for $m_{p}$~\eqref{eq:u-equ for m}--\eqref{eq:v-equ for m}. The resulting wave equation for $m_{p}$ will take the form of }\begin{equation*}
    \partial_{u}\partial_{v}m_{p}+\cdots = \frac{r_{k}}{\partial_{v}r_{k}}\partial_{u}m_{k}\partial_{v}\phi_{k}\partial_{v}\phi_{p}+\cdots.
\end{equation*}
However, this equation will lead to severe complications because $\partial_{v}\phi_{k}$ is present in the scattering analysis, {in view of the discussion in Section~\ref{sec: difficulty for the scattering}, Remark~\ref{rmk: esstential k smallness}, and Remark~\ref{rmk: the role of the regularity}.}

An explanation for why this complication happens is that by only looking at three wave equations, one cannot retrieve the original important information carried by the constraint equations. Hence, we keep the original two wave equations and two transport equations in the formulation and directly apply the Fourier--Laplace transform. It turns out that the $\partial_s$-transport equation for $m_{p}$, after the Fourier--Laplace transform, will be an algebraic equation in the following form:
\begin{equation}
\label{eq: schematic equation for mphat}
\begin{aligned}
    (\sigma +\text{Terms on the background})\hat{m}_{p} =&(\cdots)z\partial_{z}\phi_{k}\partial_{z}\hat{\Psi}_{p}+(\cdots)z\partial_{z}\phi_{k}\partial_{z}\hat{r}_{p}+(\cdots)z\partial_{z}\phi_{k}\hat{\Psi}_{p}+(\cdots)z\partial_{z}\phi_{k}\hat{r}_{p}\\&+ \text{Other terms only involving ${\hat{r}}_{p}$ and $\hat{\Psi}_{p}$ without $\partial_{z}\phi_{k}$},
\end{aligned}
\end{equation}
where $\cdots$ represents terms depending on the background that do not involve $\phi_{k}$. We can directly express $\hat{m}_{p}$ by $\hat{r}_{p}$ and $\hat{\Psi}_{p}$ from this algebraic equation and then substitute the $\hat{m}_{p}$ terms in differential equations for $\hat{r}_{p}$ and $\hat{\Psi}_{p}$ by this expression. The coefficient before $\hat{m}_{p}$ in the above algebraic equation will be zero for some points on the right half-plane of $\sigma$. We call these points \textit{fake modes}, and they will be excluded by restricting our domain of interest in the left half-plane of $\sigma$. This aspect is different from the work \cite{singh2}, where no fake mode appears and the inverse Fourier--Laplace transform is performed in the upper half-plane. The consequence of this formulation is that we only need to analyze the differential equations for $\hat{r}_{p}$ and $\hat{\Psi}_{p}$ in the scattering analysis.

\subsubsection{Resolution to the presence of $\partial_{v}\phi_{k}$}

{In view of the discussion in Section~\ref{sec: review on the linear wave} on the linear wave equation, Remark~\ref{rmk: esstential k smallness}, and Remark~\ref{rmk: the role of the regularity}, the presence of the problematic $\partial_{v}\phi_{k}$ terms essentially obstructs the construction of the outgoing and ingoing solutions to the corresponding ODE system after taking the Fourier--Laplace transform. For convenience, we focus only on the outgoing solutions here. The way to resolve the difficulty caused by the $\partial_{v}\phi_{k}$ terms in the construction of outgoing solutions relies on the following three observations.}

\paragraph{Observation 1: good regularity in the $\partial_{s}$-direction}
By the nature of the $k$-self-similar naked singularity spacetime (Proposition~\ref{prop: estimate on the exact k self similar spacetime}), the $\partial_{v}$ derivative, or equivalently the $\partial_{z}$ derivative, should be regarded as a bad derivative, due to the $\frac{1}{k}$-largeness of $\partial_{v}\phi_{k}$ and its poor regularity on the singular horizon $v=0$. However, since $u\partial_{u}$ agrees with the conformal Killing vector field
\[
    K = u\partial_{u}+v\partial_{v}=\partial_{s}
\]
of $(\mathcal{M}_{k},g_{k},\phi_{k})$ on the singular horizon, the quantity $u\partial_{u}\phi_{k}$ is in fact smooth there and enjoys $k$-smallness, by the $k$-self-similarity with $0<k^{2}\ll1$. Therefore, the vector field $\partial_{u}$ in double-null coordinates and the conformal Killing vector field $\partial_{s}$ in self-similar coordinates should be viewed as good derivatives. This is analogous to the standard theory of linear wave equations, where vector fields transverse to the light cone are regarded as bad derivatives.

\paragraph{Observation 2: advantage of the formulation in Section~\ref{sec: formulation of the scattering analysis}}
In the formulation of the scattering analysis introduced in Section~\ref{sec: formulation of the scattering analysis}, the differential equations for $\hat{r}_{p}$ and $\hat{\Psi}_{p}$; see \eqref{eq:scattering equation for rp}--\eqref{eq:scattering equation for psip}, do not contain terms involving $\partial_{v}\phi_{k}$ directly. The $\partial_{v}\phi_{k}$ terms enter the system only through the equation for $\hat{m}_{p}$ \eqref{scattering equation for mp}. However, this equation is obtained by taking the Fourier--Laplace transform of $\partial_{s}m_{p}$, which is a good derivative of $m_{p}$ in the sense of Observation~1. Thus, this formulation turns out to overcome the difficulty caused by the problematic $\partial_{v}\phi_{k}$ terms. Indeed, by Proposition~\ref{prop: estimate on the exact k self similar spacetime}, in self-similar coordinates we have
\begin{equation*}
    \left|\partial_{z}\phi_{k}\right|\lesssim \frac{1}{k},\qquad
    \left|\partial_{z}\phi_{k}\right|\lesssim \frac{k}{\epsilon}(-z)^{-\epsilon},\qquad
    \left|z\partial_{z}\phi_{k}\right|\lesssim \frac{k}{\epsilon}(-z)^{1-\epsilon},
\end{equation*}
where one recovers $k$-smallness for $\partial_{z}\phi_{k}$ at the cost of a loss in the $z$-weight. Moreover, on the right-hand side of~\eqref{eq: schematic equation for mphat}, the terms involving $\partial_{z}\phi_{k}$ always come with a compensating $z$-weight. Therefore, using the formulation introduced in Section~\ref{sec: formulation of the scattering analysis}, the ODEs for $\hat{r}_{p}$ and $\hat{\Psi}_{p}$ can be systematically written in the form
\begin{equation*}
    z\frac{d^{2}f}{dz^{2}}
    +\bigl(c_{1}+E_{1}(\sigma,z)\bigr)\frac{df}{dz}
    +\bigl(c_{2}+E_{2}(\sigma,z)\bigr)f=0,
\end{equation*}
with $\vert E_{1}\vert\lesssim k(-z)^{1-\epsilon}$ and $\vert E_{2}\vert\lesssim k(-z)^{1-\epsilon}$.

\paragraph{Observation 3: correct order of the lower-order term}
Consider a model ODE with a regular singularity of the form
\begin{equation}
    z\frac{d^{2}f}{dz^{2}}+(c_{1}+E_{1}(\sigma,z))\frac{df}{dz}+(c_{2}+E_{2}(\sigma,z))f  = 0.
    \label{eq: toy model single equation}
\end{equation}
The construction of outgoing and ingoing solutions for the linear wave equation in Section~\ref{sec: review on the linear wave} shows that, if $E_{1}\equiv 0$ and $\vert E_{2}\vert\lesssim k(-z)^{1-\epsilon}$, then the solutions to the approximate equation obtained by neglecting $E_{2}$ provide good approximations to the outgoing and ingoing solutions of the original equation via the Volterra iteration, throughout the region $\Re\sigma>-\frac{3}{2}$. If $E_{1}\neq 0$ and one assumes that a solution to \eqref{eq: toy model single equation} behaves asymptotically like $(-z)^{\alpha}$ as $z\rightarrow0$ for some constant $\alpha$, then one needs $\vert E_{1}\vert\lesssim k(-z)^{2-\epsilon}$ in order for the term $E_{1}\frac{df}{dz}$ to be of the same order as $E_{2}f$, and hence to be treatable as an error term in the Volterra iteration. However, the asymptotic behavior of the outgoing solution $f_{\mathrm{out}}$ to \eqref{eq: toy model single equation} is actually
\begin{equation*}
    f_{\mathrm{out}}\approx 1,\qquad \frac{df_{\mathrm{out}}}{dz}\approx 1.
\end{equation*}
Therefore, the same computation shows that $E_{1}$ is allowed to satisfy the rougher bound $\vert E_{1}\vert\lesssim k(-z)^{1-\epsilon}$ when one considers outgoing solutions.

Combining the ideas from the above three observations and carefully running the Volterra iteration for constructing outgoing solutions, we can construct the desired outgoing and ingoing solutions to the linearized Einstein--scalar field operator; see Section~\ref{sec: construction of outgoing solutions}. Finally, we conclude that the resolution of the difficulty caused by the problematic $\partial_{v}\phi_{k}$ terms ultimately relies on the favorable $\partial_{s}$-regularity of the background $k$-self-similar spacetime. Beyond the technical observations in this section, the analysis of the scattering resolvent and resonances is naturally associated with the $\partial_{s}$ vector field in self-similar coordinates, with respect to which the background quantities exhibit good regularity.
\subsubsection{Resolution to the coupling between the geometry and the scalar field: cancellation from the expansion at the singular horizon}
The second difficulty from the interaction between the geometric quantity $r$ and the scalar field $\phi$ in the scattering analysis can be resolved by a cancellation from the expansion of the $k$-self-similar naked singularity spacetime at the singular horizon. 

{For a system of decoupled equations, such as the toy model problem~\eqref{eq: decoupled toy model equations}, one cannot avoid the presence of multiple scattering resonances, since each equation directly contributes a scattering resonance. However, for a system of coupled equations with suitable coefficients, cancellations may occur, so that only a single scattering resonance remains. To illustrate this, we consider the following toy model problem with two coupled equations}: \begin{equation}
    q_{k}z\frac{d^{2}f}{dz^{2}}+(1+\sigma)\frac{df}{dz}+(1+\sigma)k^{2}g+(1+\sigma)k^{2}f = 0,\quad q_{k}z\frac{d^{2}g}{dz^{2}}+(q_{k}+\sigma)\frac{dg}{dz}+\alpha k^{2}f = 0.\label{eq: toy model for coupled equations}
\end{equation}
Assume that the corresponding outgoing solutions take the form of \begin{equation*}
    f = \sum a_{n}(-z)^{n},\quad g = \sum A_{n}(-z)^{n}.
\end{equation*}
Plugging the above expressions into~\eqref{eq: toy model for coupled equations}, we can get the following analogous recurrence relations as in~\eqref{eq: recurence relation in outline of the proof} \begin{equation*}
    n(q_{k}n+k^{2}+\sigma)a_{n} = (1+\sigma)k^{2}A_{n-1}+(1+\sigma)k^{2}a_{n-1},\quad n(q_{k}n+\sigma)A_{n} = \alpha k^{2}a_{n-1}.
\end{equation*}
{Therefore, each outgoing solution can be determined if one specifies the values of $a_{0}$ and $A_{0}$. Just as in~\eqref{eq: recurence relation in outline of the proof}, to make each $a_{n}$ and $A_{n}$ holomorphic functions of $\sigma$ for $\Re\sigma>-\frac{3}{2}$, one has to renormalize the values of $a_{0}$ and $A_{0}$. In this toy model problem, a solution basis of the outgoing solutions can be obtained by choosing $(a_{0},A_{0}) = (1+\sigma,0)$ and $(a_{0},A_{0}) = (0,1)$, respectively. One can easily check that for these choices of $(a_{0},A_{0})$, the resulting sequence $(a_{n},A_{n})$ depends holomorphically on $\sigma$ when $\Re\sigma>-\frac{3}{2}$. Moreover, as we have discussed in Section~\ref{sec: review on the linear wave}, this suggests that there is only one scattering resonance for this system. We remark that the choices of $(a_{0},A_{0})$ in this toy model problem crucially rely on the coefficients in~\eqref{eq: toy model for coupled equations}.}

Now, let us go back to the linearized Einstein-scalar field equations. Using the formulation of the scattering analysis mentioned above, we only need to deal with two coupled differential equations for $\hat{r}_{p}$ and $\hat{\Psi}_{p}$. In the study of the linear wave \cite{singh2}, the explicit expansion of $V_{k}$ does not matter, i.e., we do not need the explicit value of $\omega_{k}$ in the expansion \eqref{eq: the expansion of Vk}. However, to deal with the danger of producing multiple scattering resonances arising from the coupling of $r$ and $\phi$, we expand various gauge invariant quantities that appear in the equations \eqref{eq:scattering equation for rp}-\eqref{scattering equation for mp} at the singular horizon $\{z = 0\}$ and compute the leading order coefficients explicitly. Then a cancellation analogous to that in~\eqref{eq: toy model for coupled equations} occurs. In fact, this cancellation removes the need to introduce a renormalization factor in the construction of one of the outgoing solutions, thereby eliminating the aforementioned potential instability mechanism; see Proposition~\ref{prop: model argument for square coefficient}. Therefore, we conclude the presence of exactly one scattering resonance $\sigma = -1$ for the linearized Einstein-scalar field equations in Section~\ref{sec: location of the scattering resonance}

\label{mainresults:sec}
\section{Multiplier estimates for the linearized equations}
\label{sec:multiplier_estimates}
In this section, we will establish the non-sharp decay result for the solution $(\Psi_{p},r_{p},m_{p})$ to the linearized Einstein-scalar field equations $\mathcal{P}(r_{p},\Psi_{p},m_{p}) = 0$ using energy estimates. 

Precisely, we will prove the following theorem.
\begin{theorem}[Non-sharp decay estimate]
Let $(r_{p},\Psi_{p},m_{p})$ be the solution to the linearized equations $\mathcal{P}(r_{p},\Psi_{p},m_{p}) = 0$ with the initial data $(r_{p},\Psi_{p})|_{\Sigma_{-1}^{(in)}} = (r_{p}^{0},\Psi_{p}^{0})\in \mathcal{C}_{N}^{\beta}\times\mathcal{C}_{N}^{\alpha}$ imposed on $\Sigma_{-1}^{(in)}$ for $\alpha\in(p_{k},\frac{3}{2})$ and $\beta\in(\alpha+\frac{1}{4},2)$. Then for any $k\neq0$ sufficiently small, there exists a constant $C$ depending on $k$ and the initial data, such that for any $0\leq i+j\leq 1$, we have\begin{align}
\Vert \partial_{u}^{i}\partial_{v}^{j}\Psi_{p}\Vert_{L^{\infty}(\Sigma_{u}^{(in)})}&\leq C(-u)^{q_{k}-k^{\frac{1}{8}}-i-q_{k}j},\\
\Vert \partial_{u}^{i}\partial_{v}^{j}r_{p}\Vert_{L^{\infty}(\Sigma_{u}^{(in)})}&\leq C(-u)^{q_{k}-k^{\frac{1}{8}}-i-q_{k}j},\\
\Vert \partial_{u}^{i}\partial_{v}^{j}m_{p}\Vert_{L^{\infty}(\Sigma_{u}^{(in)})}&\leq C(-u)^{q_{k}-k^{\frac{1}{8}}-i-q_{k}j}.
\end{align}
Equivalently, under the self-similar coordinates $(s,z)$, for $0\leq i+j\leq 1$, we have the following decay estimates: \begin{align}
    \left\Vert\partial_{s}^{i}\partial_{z}^{j}\Psi_{p}(s,\cdot)\right\Vert_{L_{z}^{\infty}([-1,0])}&\leq Ce^{(-q_{k}+k^{\frac{1}{8}})s},\\\left\Vert\partial_{s}^{i}\partial_{z}^{j}r_{p}(s,\cdot)\right\Vert_{L^{\infty}_{z}([-1,0])}&\leq Ce^{(-q_{k}+k^{\frac{1}{8}})s},\\
    \left\Vert\partial_{s}^{i}\partial_{z}^{j}m_{p}(s,\cdot)\right\Vert_{L^{\infty}_{z}([-1,0])}&\leq Ce^{(-q_{k}+k^{\frac{1}{8}})s}.
\end{align}
\label{thm:a priori energy estimate}
\end{theorem}

To prove Theorem \ref{thm:a priori energy estimate}, we apply the energy estimates to the equations \eqref{eq:wave eq for rp under self-similar coordinates}-\eqref{eq:v transport eq for m under self-similar coordinates} under the self-similar coordinates $(s,z)$. Let \begin{equation}
\Psi_{\rho} = e^{(q_{k}-\rho)s}\Psi_{p},\quad r_{\rho} = e^{(q_{k}-\rho)s}r_{p},\quad m_{\rho} = e^{(q_{k}-\rho)s}m_{p}.\label{eq: def of psirho}
\end{equation}
The desired decay rate in Theorem \ref{thm:a priori energy estimate} will follow from the boundedness of $(r_{\rho},\Psi_{\rho},m_{\rho})$ for some suitable choice of $\rho$, while the boundedness of $(r_{\rho},\Psi_{\rho},m_{\rho})$ can be established by the standard multiplier estimates and the Sobolev embedding. One should note that the smaller values of $\rho$ will correspond to a sharper decay.

\begin{remark}
Under the change of variables \eqref{eq: def of psirho}, assuming the boundedness of $(r_{\rho},\Psi_{\rho},m_{\rho})$ and their derivatives, the case $\rho = 0$ corresponds to the blue-shift decay rate of $\Psi_{p}$; the case $\rho>0$ corresponds to a decay rate slower than the blue-shift rate for $\Psi_{p}$; the case $-k^{2}<\rho<0$ corresponds to a decay rate between the self-similar rate and the blue-shift rate for $\Psi_{p}$; and $\rho<-k^{2}$ corresponds to a rate faster than the self-similar rate. 
\end{remark}

\begin{remark}
\label{rmk: observation why we could not get sharp decay from energy}
If the stability of the $k$-self-similar spacetime is expected under the high-regularity perturbations, one may ask whether there exists an energy approach that pushes the value of $\rho$ below the threshold $\rho = -k^{2}$, thereby yielding a (linearized) stability result directly. We explain why the decay rate in Theorem \ref{thm:a priori energy estimate} is not sharp and cannot be made sharp solely by energy estimates. The Einstein-scalar field equations are invariant under the translation $(r,\phi,m)\rightarrow(r,\phi+c,m)$. Accordingly, the linearized equations $\mathcal{P}(r_{p},\Psi_{p},m_{p}) = 0$ are also invariant under the translation $(r_{p},\Psi_{p},m_{p})\rightarrow(r_{p},\Psi_{p}+cr_{k},m_{p})$. From the viewpoint of the evolution problem, this translation changes the initial data $(r_{p},\Psi_{p})\bigl|_{\Sigma_{-1}}$ by a multiple of $r_{k}$, which will not affect the regularity class of the initial data since $r_{k}$ is at least $C^{2}$. In other words, the regularity restriction of the initial data alone can not distinguish two essentially equivalent sets of initial data modulo this translation. Moreover, in self-similar coordinates, $r_{k}$ decays only at the rate $e^{-s}$. For general high-regularity perturbations, observing improved decay with $\rho<-k^{2}$ therefore requires considering the renormalized quantity $(r_{p},\Psi_{p}-cr_{k},m_{p})$ for a suitable choice of $c$. Hence, without a priori knowledge of the correct value of $c$, one cannot expect to close the energy estimate for $(r_{\rho},\Psi_{\rho},m_{\rho})$ when $\rho<-k^{2}$.
\end{remark}

\subsection{Schematic notation}
In the following energy estimate, we shall use the self-similar coordinates $(s,z)$. It turns out that precise structures of many tedious coefficients in \eqref{eq:wave eq for rp under self-similar coordinates}-\eqref{eq:v transport eq for m under self-similar coordinates} are irrelevant when running the energy estimates.  \textbf{In this section (Section \ref{sec:multiplier_estimates}) only}, we adopt the following schematic notation:
\begin{definition}
Let $\mathcal{F}_{a}$, $\mathcal{G}_{a}$, and $\mathcal{K}$ be three classes of functions independent of $s$ satisfying the following properties:
\begin{enumerate}
  \item[(a)] For any $F_{a}\in \mathcal{F}_{a}$, we have\begin{equation}
  \left\Vert F_{a}\right\Vert_{L^{\infty}([-1,-\frac{1}{2}])}\lesssim a,\quad \left\Vert z\vert^{\epsilon} F_{a}\right\Vert_{L^{\infty}\left([-1,0]\right)}\lesssim\frac{a}{\epsilon},\quad\epsilon>0.
  \end{equation}
  \item[(b)] For any $G_{a}\in\mathcal{G}_{a}$, we have $\frac{dG_{a}}{dz}\in\mathcal{F}_{a}$ and \begin{equation}
  \Vert G_{a}\Vert_{L^{\infty}([-1,0])}\lesssim a.
  \end{equation} 
  \item[(c)] For any $K\in\mathcal{K}$, we have $K(z)$ is bounded.
\end{enumerate}
In the following, without any confusion, we will simply use $F_{a}(z)$, $G_{a}(z)$, and $K(z)$ to denote any function in the function classes $\mathcal{F}_{a}$, $\mathcal{G}_{a}$, and $\mathcal{K}$, respectively. Moreover, $F_{a}^{n}$, $G_{a}^{n}$, and $K^{n}$ denote the product of any $n$ functions in $\mathcal{F}_{a}$, $\mathcal{G}_{a}$, and $\mathcal{K}$, respectively.
\label{def:schematic notation}
\end{definition}
\begin{proposition}
Adopting the notation in Definition \ref{def:schematic notation}, we can rewrite the equations \eqref{eq:wave eq for rp under self-similar coordinates}-\eqref{eq:v transport eq for m under self-similar coordinates} in the following forms:\begin{align}
&\partial_{s}\partial_{z}r_{\rho}+q_{k}z\partial_{z}^{2}r_{\rho}+(\rho+G_{k^{2}}(z))\partial_{z}r_{\rho}+G_{k^{2}}(z)\partial_{s}r_{\rho}+G_{k^{2}}(z)r_{\rho} = K(z)\frac{m_{\rho}}{\mr{r}^{2}},\label{eq:wave equation for r rho}\\&
\partial_{s}\partial_{z}\Psi_{\rho}+q_{k}z\partial_{z}^{2}\Psi_{\rho}+\rho\partial_{z}\Psi_{\rho}+G_{k^{2}}(z)\Psi_{\rho} = G_{k^{2}}(z)\partial_{s}r_{\rho}+(G_{k^{2}}(z)+k)\partial_{z}r_{\rho}+G_{k^{2}}(z)r_{\rho}\nonumber\\&\hspace{6.5cm}+\vert z\vert^{\frac{1}{2}}G_{k}(z)\frac{m_{\rho}}{\mr{r}^{2}},\label{eq:wave equation for Psi rho}\\&
\partial_{z}\left(m_{\rho}\right)+\left(\frac{\mr{r}}{\partial_{z}\mr{r}}(\partial_{z}\mr{\phi})^{2}\right)m_{\rho} = \mr{r}^{2}\left(F_{k^{2}}(z)\frac{r_{\rho}}{\mr{r}}+F_{k^{2}}(z)\partial_{z}r_{\rho}+F_{k}(z)\left(\partial_{z}\left(\frac{\Psi_{\rho}}{\mr{r}}\right)-\partial_{z}\left(\mr{\phi}\frac{r_{\rho}}{\mr{r}}\right)\right)\right).\label{eq:z-transport eq for mrho}
\end{align}
\end{proposition}
\begin{proof}
Using Proposition \ref{prop: estimate on the exact k self similar spacetime}, a straightforward computation gives the proof of this proposition.
\end{proof}
Since we will also use the second-order energy estimate, we state the following schematic equations after taking one $\partial_{z}$-derivative of \eqref{eq:wave equation for r rho}-\eqref{eq:z-transport eq for mrho}:
\begin{proposition}
Adopting the notation in Definition \ref{def:schematic notation}, commuting the equations \eqref{eq:wave equation for r rho}-\eqref{eq:wave equation for Psi rho} with $\partial_{z}$, we have the following schematic forms:\begin{align}
&\begin{aligned}
&\partial_{s}\partial_{z}^{2}r_{\rho}+q_{k}z\partial_{z}^{3}r_{\rho}+\left(\rho+q_{k}+G_{k^{2}}(z)\right)\partial_{z}^{2}r_{\rho}\\ =& F_{k^{2}}(z)\partial_{z}r_{\rho}+F_{k^{2}}(z)\partial_{s}r_{\rho}+F_{k^{2}}(z)r_{\rho}+K(z)\partial_{z}\left(\frac{m_{\rho}}{\mr{r}^{2}}\right)+(K(z)+F_{k^{2}}(z))\frac{m_{\rho}}{\mr{r}^{2}},
\end{aligned}\label{eq:second order wave equation for r rho}\\[1em]
&\begin{aligned}
&\partial_{s}\partial_{z}^{2}\Psi_{\rho}+q_{k}z\partial_{z}^{3}\Psi_{\rho}+(\rho+q_{k})\partial_{z}^{2}\Psi_{\rho}+G_{k^{2}}(z)\partial_{z}\Psi_{\rho}+F_{k^{2}}(z)\Psi_{\rho}\\=& G_{k}(z)\partial_{z}^{2}r_{\rho}+F_{k^{2}}(z)\partial_{s}r_{\rho}+F_{k^{2}}(z)\partial_{z}r_{\rho}+F_{k^{2}}(z)r_{\rho}+\vert z\vert^{\frac{1}{2}}G_{k}(z)\partial_{z}\left(\frac{m_{\rho}}{\mr{r}^{2}}\right)+F_{k}(z)\frac{m_{\rho}}{\mr{r}^{2}}.
\end{aligned}\label{second order wave equation for psi rho}
\end{align}
\end{proposition}
For the convenience of the integration, we define \begin{equation*}
    \mathcal{R}(s_{0}):=\left\{-1\leq z\leq 0,\ 0\leq s\leq s_{0}\right\},\quad \mathcal{H}(s_{0}): = \left\{v = 0,\ -s_{0}\leq u\leq0\right\}.
\end{equation*}

\subsection{Preliminary inequalities}

We prove some useful inequalities in this section. The following lemma can be directly proved by using integration by parts.
\begin{lemma}
Let $f\in C^{1}([-1,0])$ be given and $\omega>0$. Then we have:
\begin{equation}
\begin{aligned}
\int_{-1}^{0}(-z)^{2\omega-1}f^{2}(z)dz&\leq \frac{1}{\omega}f^{2}(-1)+\frac{1}{\omega^{2}}\int_{-1}^{0}(-z)^{2\omega+1}\left(\frac{df}{dz}\right)^{2}dz.
\end{aligned}
    \label{eq: fundamental theorem of calculus}
    \end{equation}
    If $f(-1) = 0$, then we have \begin{equation}
        \int_{-1}^{0}(-z)^{2\omega-1}f^{2}(z)dz\leq\frac{1}{\omega^{2}}\int_{-1}^{0}(-z)^{2\omega+1}\left(\frac{df}{dz}\right)^{2}dz.   \label{eq: standard hardy with zero boundary condition}
        \end{equation}
        Moreover, for $\omega\geq \omega_{0}>0$, we have \begin{equation}
        \label{eq: fundamental theorem of calculus for lesssim}
            \int_{-1}^{0}(-z)^{2\omega-1}f^{2}(z)dz\lesssim \frac{1}{\omega_{0}}f^{2}(-1)+\frac{1}{\omega_{0}^{2}}\int_{-1}^{0}(-z)^{2\omega+1}\left(\frac{df}{dz}\right)^{2}.
        \end{equation}
    \end{lemma}
\begin{proof}
Applying integration by parts and the Cauchy--Schwarz inequality, we have \begin{align*}
    \int_{-1}^{0}(-z)^{2\omega-1}f^{2}(z)dz &= -\frac{1}{2\omega}\int_{-1}^{0}f^{2}(z)d(-z)^{2\omega}\\ &= \frac{1}{2\omega}f^{2}(-1)+\frac{1}{2\omega}\int_{-1}^{0}2(-z)^{2\omega}f(z)f^{\prime}(z)dz\\&\leq \frac{1}{2\omega}f^{2}(-1)+\frac{1}{2}\int_{-1}^{0}(-z)^{2\omega-1}f^{2}(z)dz+\frac{1}{2\omega^{2}}\int_{-1}^{0}(-z)^{2\omega+1}\left(f^{\prime}(z)\right)^{2}dz.
\end{align*}
Hence, we have \begin{equation*}
    \int_{-1}^{0}(-z)^{2\omega-1}f^{2}(z)dz\leq \frac{1}{\omega}f^{2}(-1)+\frac{1}{\omega^{2}}\int_{-1}^{0}(-z)^{2\omega+1}\left(f^{\prime}(z)\right)^{2}dz.
\end{equation*}
The inequality \eqref{eq: standard hardy with zero boundary condition} follows trivially from taking $f(-1) = 0$.
\end{proof}


Next, we recall the standard weighted Hardy inequality:
\begin{lemma}[Weighted Hardy inequality]
Let $f\in C^{1}([-1,0])$ with $f(-1) = 0$ and $\omega>0$. Then the following weighted Hardy inequality holds: 
\begin{equation}
\int_{-1}^{0}\frac{(-z)^{2\omega-1}}{(z+1)^{2}}f^{2}(z)dz\leq C(\omega)\int_{-1}^{0}(-z)^{2\omega+1}(\frac{df}{dz})^{2}dz,\label{standard hardy}
\end{equation}
where $C(\omega)$ is \begin{equation*}
    C(\omega) = \frac{4}{\omega^{2}}+2^{2\omega+4}. 
\end{equation*}
Moreover, if $0<\omega_{0}\leq\omega\leq 5$, then we have \begin{equation}
    \int_{-1}^{0}\frac{(-z)^{2\omega-1}}{(z+1)^{2}}f^{2}(z)dz\lesssim \frac{1}{\omega_{0}^{2}}\int_{-1}^{0}(-z)^{2\omega+1}\left(\frac{df}{dz}\right)^{2}dz.\label{eq: standard hardy with lesssim}
\end{equation}
\end{lemma}
\begin{proof}
We split the integral into two parts:\begin{equation}
\int_{-1}^{0}(-z)^{2\omega-1}\frac{f^{2}(z)}{(z+1)^{2}}dz = \int_{-1}^{-\frac{1}{2}}(-z)^{2\omega-1}\frac{f^{2}(z)}{(z+1)^{2}}dz+\int_{-\frac{1}{2}}^{0}(-z)^{2\omega-1}\frac{f^{2}(z)}{(z+1)^{2}}dz.
\end{equation}
For the integration on $[-1,-\frac{1}{2}]$, using integration by parts, we have\begin{equation}
\begin{aligned}
\int_{-1}^{-\frac{1}{2}}\frac{(-z)^{2\omega-1}}{(z+1)^{2}}f^{2}(z)dz &\leq 2\int_{-1}^{-\frac{1}{2}}\frac{f^{2}(z)}{(z+1)^{2}}dz\\&=-2\int_{-1}^{-\frac{1}{2}}f^{2}(z)d\frac{1}{z+1}
\\&= -8f^{2}(-\frac{1}{2})+2\int_{-1}^{-\frac{1}{2}}\frac{2f\frac{df}{dz}}{(z+1)}dz\\&\leq
4\left(\int_{-1}^{-\frac{1}{2}}\frac{f^{2}(z)}{(z+1)^{2}}dz\right)^{\frac{1}{2}}\left(\int_{-1}^{-\frac{1}{2}}\left(\frac{df}{dz}\right)^{2}dz\right)^{\frac{1}{2}}.
\end{aligned}
\end{equation}
Hence, we have \begin{equation*}
    \int_{-1}^{-\frac{1}{2}}\frac{f^{2}(z)}{(z+1)^{2}}dz\leq 4\int_{-1}^{-\frac{1}{2}}\left(f^{\prime}\right)^{2}dz,
\end{equation*}
and \begin{equation}
    \int_{-1}^{-\frac{1}{2}}\frac{(-z)^{2\omega-1}}{(z+1)^{2}}f^{2}(z)dz\leq 8\int_{-1}^{-\frac{1}{2}}\left(f^{\prime}\right)^{2}dz\leq 2^{2\omega+4}\int_{-1}^{-\frac{1}{2}}(-z)^{2\omega+1}\left(f^{\prime}\right)^{2}dz.\label{first part of standard hardy}
\end{equation}
For the integration on $[-\frac{1}{2},0]$, we have\begin{equation}
\begin{aligned}
\int_{-\frac{1}{2}}^{0}(-z)^{2\omega-1}\frac{f^{2}(z)}{(z+1)^{2}}dz&\leq 4\int_{-\frac{1}{2}}^{0}(-z)^{2\omega-1}f^{2}(z)dz\\&\leq 4\int_{-1}^{0}(-z)^{2\omega-1}f^{2}(z)dz\\&\leq \frac{4}{\omega^{2}}\int_{-1}^{0}(-z)^{2\omega+1}\left(f^{\prime}\right)^{2}dz,
\end{aligned}
\label{eq: second part of standard hardy}
\end{equation}
where the last inequality follows from \eqref{eq: fundamental theorem of calculus}.

Combining \eqref{first part of standard hardy} and \eqref{eq: second part of standard hardy} gives \eqref{standard hardy}. The inequality \eqref{eq: standard hardy with lesssim} follows easily from \eqref{standard hardy}. This concludes the proof.
\end{proof}
Applying the standard weighted Hardy inequality \eqref{standard hardy}, we can prove the following version of the Hardy inequality:
\begin{lemma}
Let $F_{b}\in\mathcal{F}_{b}$ and $f\in C^{2}([-1,0])$ with $f(-1)= 0$. Then for any $5\geq \omega\geq\omega_{0}>0$ and $0<\epsilon\ll \omega_{0}$, we have
\begin{equation}
\int_{-1}^{0}(-z)^{2\omega-1}F_{b}^{2}(z)\left(\frac{d}{dz}\left(\frac{f(z)}{z+1}\right)\right)^{2}\lesssim \frac{1}{\omega_{0}^{2}}\frac{b^{2}}{\epsilon}\int_{-1}^{0}(-z)^{2(\omega-\epsilon)+1}\left(\frac{d^{2}f}{dz^{2}}\right)^{2}dz.
\label{eq: hardy with derivative}
\end{equation}
\end{lemma}
\begin{proof}
First, we deal with $\left(\frac{1}{z+1}f\right)^{\prime}$. Using the fact that $f(-1) = 0$, we have \begin{align*}
    \left\vert\frac{d}{dz}\left(\frac{1}{z+1}f\right)\right\vert =& \left\vert\frac{d}{dz}\left(\frac{1}{z+1}\int_{-1}^{z}f^{\prime}(z^{\prime})d(z^{\prime}+1)\right)\right\vert \\=& \left\vert\frac{d}{dz}\left(f^{\prime}(z)-\frac{1}{z+1}\int_{-1}^{z}(z^{\prime}+1)\frac{d^{2}f}{dz^{2}}dz^{\prime}\right)\right\vert\\=&\left\vert\frac{d^{2}f}{dz^{2}}-\frac{d^{2}f}{dz^{2}}+\frac{1}{(z+1)^{2}}\int_{-1}^{z}(z^{\prime}+1)\frac{d^{2}f}{dz^{2}}dz^{\prime}\right\vert\\\leq&\frac{1}{(z+1)}\int_{-1}^{z}\left\vert\frac{d^{2}f}{dz^{2}}\right\vert dz^{\prime}.
\end{align*}
Hence, we have \begin{align*}
    \int_{-1}^{0}(-z)^{2\omega-1}F_{b}^{2}(z)\left(\frac{d}{dz}\left(\frac{1}{z+1}f(z)\right)\right)^{2}dz\lesssim&\frac{b^{2}}{\epsilon}\int_{-1}^{0}\frac{(-z)^{2(\omega-\epsilon)-1}}{(z+1)^{2}}\left(\int_{-1}^{z}\left\vert\frac{d^{2}f}{dz^{2}}\right\vert dz^{\prime}\right)^{2}dz.
\end{align*}
Let \begin{equation*}
    \widetilde{F} = \int_{-1}^{z}\left\vert\frac{d^{2}f}{dz^{2}}\right\vert dz^{\prime}.
\end{equation*}
Then clearly, we have $\widetilde{F}(-1) = 0$. Hence, applying \eqref{eq: standard hardy with lesssim}, we have \begin{equation*}
    \int_{-1}^{0}(-z)^{2\omega-1}F_{b}^{2}\left(\frac{d}{dz}\left(\frac{1}{z+1}f(z)\right)\right)^{2}dz\lesssim \frac{1}{\omega_{0}^{2}}\frac{b^{2}}{\epsilon}\int_{-1}^{0}(-z)^{2\left(\omega-\epsilon\right)+1}\left(\frac{d^{2}f}{dz^{2}}\right)dz.
\end{equation*}
This concludes the proof.
\end{proof}
\subsection{First order estimates for $m_{\rho}$}
\label{sec: first-order energy estimate for mp}
In this section, we establish the $L^{2}$-estimates for $\frac{m_{\rho}}{\mr{r}^{2}}$ and $\partial_{z}\left(\frac{m_{\rho}}{\mr{r}^{2}}\right)$. The main difficulty in estimating the $L^{2}$-norm of $\frac{m_{\rho}}{\mr{r}^{2}}$ lies in the potential loss of losing derivatives. Since the equation \eqref{eq:z-transport eq for mrho} contains terms $\partial_{z}\left(\frac{\Psi_{\rho}}{\mr{r}}\right)$ and $\partial_{z}\left(\frac{r_{\rho}}{\mr{r}}\right)$, if one views $\mr{r}$ as $(z+1)$ in the near-axis region, then these two terms can be heuristically thought as $\partial_{z}^{2}\Psi_{\rho}$ and $\partial_{z}^{2}r_{\rho}$, which suggests that $\left\Vert\frac{m_{\rho}}{\mr{r}^{2}}(s,\cdot)\right\Vert_{L_{z}^{2}([-1,0])}$ would naturally be bounded by $\left\Vert\partial_{z}^{2}\Psi_{\rho}\right\Vert_{L_{z}^{2}}$ and $\left\Vert\partial_{z}^{2}r_{\rho}\right\Vert_{L_{z}^{2}}$. However, by exploiting the $\mr{r}^{2}$-weight in front of these two potentially dangerous terms, we can avoid losing derivatives and instead estimate $\left\Vert\frac{m_{\rho}}{\mr{r}^{2}}(s,\cdot)\right\Vert_{L_{z}^{2}([-1,0])}$ in terms of $\left\Vert\partial_{z}r_{\rho}\right\Vert_{L_{z}^{2}}$ and $\left\Vert\partial_{z}\Psi_{\rho}\right\Vert_{L_{z}^{2}}$. More specifically, we can prove the following proposition.
\begin{proposition}
\label{prop: second order estimate of mrho}
    For any $5\geq\omega\geq\omega_{0}>0$ and a fixed small real number $\epsilon$, we have the following the estimate
    \begin{equation}
\begin{aligned}
    \int_{-1}^{0}(-z)^{2\omega-1}\left(\frac{m_{\rho}}{\mr{r}^{2}}\right)^{2}dz\lesssim& \frac{1}{\omega_{0}^{3}}\frac{k^{4}}{\epsilon}\int_{-1}^{0}(-z)^{2\omega-2\epsilon}(\partial_{z}r_{\rho})^{2}dz+\frac{1}{\omega_{0}^{3}}\frac{k^{2}}{\epsilon}\int_{-1}^{0}(-z)^{2\omega-2\epsilon}(\partial_{z}\Psi_{\rho})^{2}dz.
    \end{aligned}
    \label{eq: first order energy estimate for mrho}
    \end{equation}
    Moreover, using the second-order derivatives, we can further derive the following estimate: 
    \begin{equation}
        \begin{aligned}
            \int_{-1}^{0}(-z)^{2\omega-1}\left(\frac{m_{\rho}}{\mr{r}^{2}}\right)^{2}dz\lesssim&\frac{1}{\omega_{0}^{4}}\frac{k^{4}}{\epsilon}(\partial_{z}r_{\rho})^{2}(s,-1)+\frac{1}{\omega_{0}^{4}}\frac{k^{2}}{\epsilon}(\partial_{z}\Psi_{\rho})^{2}(s,-1)\\&+\frac{1}{\omega_{0}^{5}}\frac{k^{4}}{\epsilon}\int_{-1}^{0}(-z)^{2(\omega-\epsilon)+2}(\partial_{z}^{2}r_{\rho})^{2}+\frac{1}{\omega_{0}^{5}}\frac{k^{2}}{\epsilon}\int_{-1}^{0}(-z)^{2(\omega-\epsilon)+2}\left(\partial_{z}^{2}\Psi_{\rho}\right)^{2}.
        \end{aligned}
        \label{eq: first-order energy estimate using second order terms}
    \end{equation}
\end{proposition}
\begin{proof}
    Using the integration factor for the $z$-transport equation \eqref{eq:z-transport eq for mrho} of $m_{\rho}$, we have\begin{equation}
\partial_{z}\left(e^{\int_{-1}^{z}\frac{\mr{r}}{\partial_{z}\mr{r}}(\partial_{z}\mr{\phi})^{2}d\widetilde{z}}m_{\rho}\right) = e^{\int_{-1}^{z}\frac{\mr{r}}{\partial_{z}\mr{r}}(\partial_{z}\mr{\phi})^{2}d\widetilde{z}}\left[\mr{r}F_{k^{2}}(z)r_{\rho}+\mr{r}^{2}F_{k^{2}}(z)\partial_{z}r_{\rho}+\mr{r}^{2}F_{k}(z)\partial_{z}\left(\frac{\Psi_{\rho}}{\mr{r}}-\frac{\mr{\phi}}{\mr{r}}r_{\rho}\right)\right].
\end{equation}
Directly integrating the above equation and using the fact that $$e^{\int_{-1}^{z}\frac{\mr{r}}{\partial_{z}\mr{r}}(\partial_{z}\mr{\phi})^{2}}d\widetilde{z}\leq e^{Ck^{2}},$$we have\begin{equation}
\begin{aligned}
\vert m_{\rho}\vert\lesssim& \int_{-1}^{z}\mr{r}^{2}F_{k^{2}}\left\vert\frac{r_{\rho}}{\mr{r}}\right\vert+\mr{r}^{2}F_{k^{2}}\vert\partial_{z}r_{\rho}\vert+\mr{r}^{2}F_{k}(z)\left\vert\partial_{z}\left(\frac{\Psi_{\rho}}{\mr{r}}-\frac{\mr{\phi}}{\mr{r}}r_{\rho}\right)\right\vert d\widetilde{z}\\\lesssim&
\int_{-1}^{z}\mr{r}^{2}F_{k^{2}}\left\vert\frac{r_{\rho}}{\mr{r}}\right\vert+\mr{r}^{2}F_{k^{2}}\vert\partial_{z}r_{\rho}\vert+\mr{r}^{2}F_{k}(z)\left\vert\partial_{z}\left(\frac{\Psi_{\rho}}{\mr{r}}\right)\right\vert+\mr{r}^{2}F_{k}^{2}\left\vert\partial_{z}\left(\frac{r_{\rho}}{\mr{r}}\right)\right\vert.
\end{aligned}
\label{eq: integral equation for mrho}
\end{equation}
Taking the square of the above equation, applying the Cauchy-Schwarz inequality, and using the fact that \begin{equation*}
\int_{-1}^{z}\mr{r}^{p}d\widetilde{z}\approx \mr{r}^{p+1},
\end{equation*}
we have\begin{equation}
\left(\frac{m_{\rho}}{\mr{r}^{2}}\right)^{2}(z)\lesssim \int_{-1}^{z}F_{k^{2}}^{2}\left\vert\frac{r_{\rho}}{\mr{r}}\right\vert^{2}+F_{k^{2}}^{2}(z)(\partial_{z}r_{\rho})^{2}d\widetilde{z}+\frac{1}{\mr{r}^{2}}\int_{-1}^{z}\mr{r}^{3}F_{k}^{2}\left\vert\partial_{z}\left(\frac{\Psi_{\rho}}{\mr{r}}\right)\right\vert^{2}+\mr{r}^{3}F_{k^{2}}^{2}\left\vert\partial_{z}\left(\frac{r_{\rho}}{\mr{r}}\right)\right\vert^{2}d\widetilde{z}.\label{eq: the pointwise estimate for mpr}
\end{equation}
Multiplying \eqref{eq: the pointwise estimate for mpr} by $(-z)^{2\omega-1}$, integrating over the region $z\in[-1,0]$, and applying Fubini Theorem, we have\begin{equation}
\begin{aligned}
\int_{-1}^{0}(-z)^{2\omega-1}\left(\frac{m_{\rho}}{\mr{r}^{2}}\right)^{2}dz\leq&\int_{-1}^{0} F_{k^{2}}^{2}\left\vert\frac{r_{\rho}}{\mr{r}}\right\vert^{2}+F_{k^{2}}^{2}(\partial_{z}r_{\rho})^{2}\int_{z}^{0}(-\widetilde{z})^{2\omega-1}d\widetilde{z}dz\\&+\int_{-1}^{0}\frac{(-z)^{2\omega-1}}{\mr{r}^{2}(z)}\int_{-1}^{z}\mr{r}^{3}(\widetilde{z})F_{k}^{2}(\widetilde{z})\left\vert\partial_{z}\left(\frac{\Psi_{\rho}}{\mr{r}}-\frac{\mr{\phi}}{\mr{r}}r_{\rho}\right)\right\vert^{2}d\widetilde{z}dz\\\leq&\frac{1}{2\omega}\int_{-1}^{0}(-z)^{2\omega}F_{k^{2}}^{2}\left\vert\frac{r_{\rho}}{\mr{r}}\right\vert^{2}+(-z)^{2\omega}F_{k^{2}}^{2}\left(\partial_{z}r_{\rho}\right)^{2}dz\\&+\int_{-1}^{0}\mr{r}^{3}(z)F_{k}^{2}(z)\left\vert\partial_{z}\left(\frac{\Psi_{\rho}}{\mr{r}}\right)\right\vert^{2}+\mr{r}^{3}F_{k^{2}}^{2}\left\vert\partial_{z}\left(\frac{r_{\rho}}{\mr{r}}\right)\right\vert^{2}\int_{z}^{0}\frac{(-\widetilde{z})^{2\omega-1}}{\mr{r}^{2}(\widetilde{z})}d\widetilde{z}dz.
\end{aligned}
\label{eq: backward energy estimate for mp; first version}
\end{equation}
For the first term on the right-hand side of \eqref{eq: backward energy estimate for mp; first version}, recall that \begin{equation*}
    \vert z\vert ^{2\epsilon}F_{k^{2}}^{2}\lesssim \frac{k^{4}}{\epsilon}.
\end{equation*}
Then by \eqref{standard hardy}, we have\begin{equation}
\label{eq: intermediate step to estimate the firt term in the backward of mp}
\begin{aligned}
\int_{-1}^{0}(-z)^{2\omega}F_{k^{2}}^{2}\left\vert\frac{r_{p}}{\mr{r}}\right\vert^{2}dz\lesssim&\frac{k^{4}}{\epsilon}\int_{-1}^{0}(-z)^{2\omega-2\epsilon}\left\vert\frac{r_{\rho}}{\mr{r}}\right\vert^{2}dz\\\lesssim&\frac{k^{4}}{\epsilon}\int_{-1}^{0}(-z)^{2\omega-2\epsilon}\frac{r_{\rho}^{2}}{(z+1)^{2}}dz\\\lesssim&\frac{k^{4}}{\epsilon}\frac{1}{\omega_{0}^{2}}\int_{-1}^{0}(-z)^{2\omega-2\epsilon+2}\left(\partial_{z}r_{\rho}\right)^{2}dz.
\end{aligned}
\end{equation}
The second term on the right-hand side of \eqref{eq: backward energy estimate for mp; first version} can be easily bounded by\begin{equation}
\int_{-1}^{0}(-z)^{2\omega}F_{k^{2}}^{2}(z)(\partial_{z}r_{p})^{2}dz\lesssim\frac{k^{4}}{\epsilon} \int_{-1}^{0}(-z)^{2\omega-2\epsilon}(\partial_{z}r_{p})^{2}dz.
\end{equation}
As we have mentioned above, the terms with $\partial_{z}\left(\frac{\Psi_{\rho}}{\mr{r}}\right)$ and $\partial_{z}\left(\frac{r_{\rho}}{\mr{r}}\right)$ on the right-hand side of \eqref{eq: backward energy estimate for mp; first version} are the most dangerous terms. For the term with $\partial_{z}\left(\frac{\Psi_{\rho}}{\mr{r}}\right)$, we have
\begin{equation}
    \begin{aligned}
        &\int_{-1}^{0}\mr{r}^{3}F_{k}^{2}(z)\left\vert\partial_{z}\left(\frac{\Psi_{\rho}}{\mr{r}}\right)\right\vert^{2}\int_{z}^{0}\frac{(-\widetilde{z})^{2\omega-1}}{\mr{r}^{2}(\widetilde{z})}d\widetilde{z}\\=&\int_{-1}^{-\frac{1}{2}}\mr{r}^{3}F_{k}^{2}(z)\left\vert\partial_{z}\left(\frac{\Psi_{\rho}}{\mr{r}}\right)\right\vert^{2}\int_{z}^{0}\frac{(-\widetilde{z})^{2\omega-1}}{\mr{r}^{2}(\widetilde{z})}d\widetilde{z}+\int_{-\frac{1}{2}}^{0}\mr{r}^{3}F_{k}^{2}(z)\left\vert\partial_{z}\left(\frac{\Psi_{\rho}}{\mr{r}}\right)\right\vert^{2}\int_{z}^{0}\frac{(-\widetilde{z})^{2\omega-1}}{\mr{r}^{2}(\widetilde{z})}d\widetilde{z}.
    \end{aligned}
    \label{eq: split the right-side of mp equation}
\end{equation}
To estimate the first term on the right-hand side of \eqref{eq: split the right-side of mp equation}, we first estimate \begin{align*}
    \int_{z}^{0}\frac{(-\widetilde{z})^{2\omega-1}}{\mr{r}^{2}(\widetilde{z})}d\widetilde{z} &= \int_{-\frac{1}{2}}^{0}\frac{(-z)^{2\omega-1}}{\mr{r}^{2}}dz+\int_{z}^{-\frac{1}{2}}\frac{(-\widetilde{z})^{2\omega-1}}{\mr{r}^{2}(\widetilde{z})}d\widetilde{z}\\&\lesssim\frac{1}{2\omega}+\frac{1}{\mr{r}(z)}\lesssim\frac{1}{\omega_{0}}\frac{1}{\mr{r}(z)}.
\end{align*}
Then we have \begin{align*}
    \int_{-1}^{-\frac{1}{2}}\mr{r}^{3}F_{k}^{2}(z)\left\vert\partial_{z}\left(\frac{\Psi_{\rho}}{\mr{r}}\right)\right\vert^{2}\int_{z}^{0}\frac{(-\widetilde{z})^{2\omega-1}}{\mr{r}^{2}(\widetilde{z})}d\widetilde{z}&\lesssim \frac{1}{\omega_{0}}\int_{-1}^{-\frac{1}{2}}\mr{r}^{2}F_{k}^{2}\left\vert\partial_{z}\left(\frac{\Psi_{\rho}}{\mr{r}}\right)\right\vert^{2}dz\\&\lesssim \frac{1}{\omega_{0}}k^{2}\int_{-1}^{-\frac{1}{2}}\mr{r}^{2}\left(\frac{(\partial_{z}\Psi_{\rho})^{2}}{\mr{r}^{2}}+\frac{\Psi_{\rho}^{2}}{\mr{r}^{4}}\right)dz\\&\lesssim\frac{1}{\omega_{0}}k^{2}\int_{-1}^{-\frac{1}{2}}(\partial_{z}\Psi_{\rho})^{2}dz\\&\lesssim \frac{1}{\omega_{0}}k^{2}\int_{-1}^{-\frac{1}{2}}(-z)^{2\omega-2\epsilon}(\partial_{z}\Psi_{\rho})^{2}dz\\&\lesssim \frac{1}{\omega_{0}}k^{2}\int_{-1}^{0}(-z)^{2\omega-2\epsilon}(\partial_{z}\Psi_{\rho})^{2}dz,
\end{align*}
where we have used the Hardy inequality \eqref{standard hardy}.

For the second term on the right-hand side of \eqref{eq: split the right-side of mp equation}, we have \begin{align*}
    \int_{-\frac{1}{2}}^{0}\mr{r}^{3}F_{k}^{2}\left\vert\partial_{z}\left(\frac{\Psi_{\rho}}{\mr{r}}\right)\right\vert^{2}\int_{z}^{0}\frac{(-\widetilde{z})^{2\omega-1}}{\mr{r}^{2}(\widetilde{z})}d\widetilde{z}&\lesssim\frac{1}{\omega_{0}}\int_{-\frac{1}{2}}^{0}(-z)^{2\omega}\mr{r}^{3}F_{k}^{2}\left\vert\partial_{z}\left(\frac{\Psi_{\rho}}{\mr{r}}\right)\right\vert^{2}dz\\&\lesssim\frac{1}{\omega_{0}}\frac{k^{2}}{\epsilon}\int_{-\frac{1}{2}}^{0}(-z)^{2\omega-2\epsilon}\left\vert\partial_{z}\left(\frac{\Psi_{\rho}}{\mr{r}}\right)\right\vert^{2}dz\\&\lesssim\frac{1}{\omega_{0}}\frac{k^{2}}{\epsilon}\int_{-\frac{1}{2}}^{0}(-z)^{2\omega-2\epsilon}\left((\partial_{z}\Psi_{\rho})^{2}+\Psi_{\rho}^{2}\right)dz\\&\lesssim\frac{1}{\omega_{0}}\frac{k^{2}}{\epsilon}\int_{-1}^{0}(-z)^{2\omega-2\epsilon}(\partial_{z}\Psi_{\rho})^{2}dz\\&\quad+\frac{1}{\omega_{0}}\frac{k^{2}}{\epsilon}\frac{1}{\omega_{0}^{2}}\int_{-1}^{0}(-z)^{2\omega-2\epsilon+2}(\partial_{z}\Psi_{\rho})^{2}dz\\&\lesssim\frac{1}{\omega_{0}^{3}}\frac{k^{2}}{\epsilon}\int_{-1}^{0}(-z)^{2\omega-2\epsilon}\left(\partial_{z}\Psi_{\rho}\right)^{2}dz.
\end{align*}

Hence, we have \begin{equation}
\begin{aligned}
    \int_{-1}^{0}\mr{r}^{3}F_{k}^{2}(z)\left\vert\partial_{z}\left(\frac{\Psi_{\rho}}{\mr{r}}\right)\right\vert^{2}\int_{z}^{0}\frac{(-\widetilde{z})^{2\omega-1}}{\mr{r}^{2}(\widetilde{z})}d\widetilde{z}\lesssim\frac{1}{\omega_{0}^{3}}\frac{k^{2}}{\epsilon}\int_{-1}^{0}(-z)^{2\omega-2\epsilon}\left(\partial_{z}\Psi_{\rho}\right)^{2}dz.
    \end{aligned}
\end{equation}

Similarly, we have \begin{equation}
    \begin{aligned}
\int_{-1}^{0}\mr{r}^{3}F_{k}^{2}\left\vert\partial_{z}\left(\mr{\phi}\frac{r_{\rho}}{\mr{r}}\right)\right\vert^{2}\int_{z}^{0}\frac{(-\widetilde{z})^{2\omega-1}}{\mr{r}^{2}(\widetilde{z})}d\widetilde{z}\lesssim\frac{1}{\omega_{0}^{3}}\frac{k^{4}}{\epsilon}\int_{-1}^{0}(-z)^{2\omega-2\epsilon}\left(\partial_{z}r_{\rho}\right)^{2}dz.
\end{aligned}
\end{equation}
Putting everything together, we have \begin{equation}
\begin{aligned}
    \int_{-1}^{0}(-z)^{2\omega-1}\left(\frac{m_{\rho}}{\mr{r}^{2}}\right)^{2}dz\lesssim& \frac{1}{\omega_{0}^{3}}\frac{k^{4}}{\epsilon}\int_{-1}^{0}(-z)^{2\omega-2\epsilon}(\partial_{z}r_{\rho})^{2}dz+\frac{1}{\omega_{0}^{3}}\frac{k^{2}}{\epsilon}\int_{-1}^{0}(-z)^{2\omega-2\epsilon}(\partial_{z}\Psi_{\rho})^{2}dz.
    \end{aligned}
\end{equation}
The equation \eqref{eq: first-order energy estimate using second order terms} follows from \eqref{eq: fundamental theorem of calculus}. This concludes the proof.
\end{proof}
\subsection{Second order estimate for $m_{\rho}$}
In this section, we establish the $L^{2}$-estimate for $\frac{m_{\rho}}{\mr{r}^{2}}$. We rewrite \eqref{eq:z-transport eq for mrho} as: \begin{equation}
    \partial_{z}\left(\frac{m_{\rho}}{\mr{r}^{2}}\right)+\left(\frac{2}{\mr{r}}\frac{d\mr{r}}{dz}+\frac{\mr{r}}{\partial_{z}\mr{r}}(\partial_{z}\mr{\phi})^{2}\right)\frac{m_{\rho}}{\mr{r}^{2}} = F_{k^{2}}\frac{r_{\rho}}{\mr{r}}+F_{k^{2}}\partial_{z}r_{\rho}+F_{k}\partial_{z}\left(\frac{\Psi_{\rho}}{\mr{r}}\right)+F_{k^{2}}\partial_{z}\left(\frac{r_{\rho}}{\mr{r}}\right).\label{eq: z transport equation for mrho for second order estimate}
\end{equation}
We can prove the following proposition.

\begin{proposition}
\label{prop: true second-order energy for mrho}
    For any $5\geq\omega\geq\omega_{0}>0$ and a fixed small real number $\epsilon$, we have the following estimate
  \begin{equation}
        \begin{aligned}
            &\int_{-1}^{0}(-z)^{2\omega-1}\left(\partial_{z}\left(\frac{m_{\rho}}{\mr{r}^{2}}\right)\right)^{2}(s,z)dz\\\lesssim&\frac{k^{4}}{\epsilon}\frac{1}{\omega_{0}^{5}}(\partial_{z}r_{\rho})^{2}(s,-1)+\frac{k^{2}}{\epsilon}\frac{1}{\omega_{0}^{5}}(\partial_{z}\Psi_{\rho})^{2}(s,-1)+\frac{k^{4}}{\epsilon}\frac{1}{\omega_{0}^{6}}\int_{-1}^{0}(-z)^{2(\omega-\epsilon)+1}\left(\partial_{z}^{2}r_{\rho}\right)^{2}\\&+\frac{k^{2}}{\epsilon}\frac{1}{\omega_{0}^{6}}\int_{-1}^{0}(-z)^{2(\omega-\epsilon)+1}\left(\partial_{z}^{2}\Psi_{\rho}\right)^{2}.
        \end{aligned}
    \end{equation}
\end{proposition}
\begin{proof}
    Using \eqref{eq: z transport equation for mrho for second order estimate}, we have \begin{align*}
        (-z)^{2\omega-1}\left(\partial_{z}\left(\frac{m_{\rho}}{\mr{r}^{2}}\right)\right)^{2}\lesssim& (-z)^{2\omega-1}F_{k^{2}}^{2}\left(\frac{r_{\rho}}{\mr{r}}\right)^{2}+(-z)^{2\omega-1}F_{k^{2}}^{2}(\partial_{z}r_{\rho})^{2}+(-z)^{2\omega-1}F_{k}^{2}\left(\partial_{z}\left(\frac{\Psi_{\rho}}{\mr{r}}\right)\right)^{2}\\&+(-z)^{2\omega-1}F_{k^{2}}^{2}\left(\partial_{z}\left(\frac{r_{\rho}}{\mr{r}}\right)\right)^{2}+(-z)^{2\omega-1}\left(\frac{2}{\mr{r}}\frac{d\mr{r}}{dz}+\frac{\mr{r}}{\partial_{z}\mr{r}}(\partial_{z}\mr{\phi})^{2}\right)^{2}\left(\frac{m_{\rho}}{\mr{r}^{2}}\right)^{2}.
    \end{align*}
    For the term with $\frac{r_{\rho}}{\mr{r}}$, we have \begin{equation}
        \begin{aligned}
            &\int_{-1}^{0}(-z)^{2\omega-1}F_{k^{2}}^{2}(z)\left(\frac{r_{\rho}}{\mr{r}}\right)^{2}dz\\\lesssim& \frac{k^{4}}{\epsilon}\int_{-1}^{0}(-z)^{2(\omega-\epsilon)-1}\left(\frac{r_{\rho}}{\mr{r}}\right)^{2}dz\\\lesssim&\frac{k^{4}}{\epsilon}\frac{1}{\omega_{0}^{2}}\int_{-1}^{0}(-z)^{2(\omega-\epsilon)+1}(\partial_{z}r_{\rho})^{2}dz\\\lesssim&\frac{k^{4}}{\epsilon}\frac{1}{\omega_{0}^{2}}\left(\frac{1}{\omega_{0}}(\partial_{z}r_{\rho})^{2}(s,-1)+\frac{1}{\omega_{0}^{2}}\int_{-1}^{0}(-z)^{2(\omega-\epsilon)+3}\left(\partial_{z}^{2}r_{\rho}\right)^{2}dz\right).
        \end{aligned}
    \end{equation}
    For the term with $\partial_{z}r_{\rho}$, we have \begin{equation}
        \begin{aligned}
            \int_{-1}^{0}(-z)^{2\omega-1}F_{k^{2}}^{2}(\partial_{z}r_{\rho})^{2}dz\lesssim&\frac{k^{4}}{\epsilon}\int_{-1}^{0}(-z)^{2(\omega-\epsilon)-1}(\partial_{z}r_{\rho})^{2}dz\\\lesssim&
            \frac{k^{4}}{\epsilon}\frac{1}{\omega_{0}}(\partial_{z}r_{\rho})^{2}(s,-1)+\frac{k^{4}}{\epsilon}\frac{1}{\omega_{0}^{2}}\int_{-1}^{0}(-z)^{2(\omega-\epsilon)+1}(\partial_{z}^{2}r_{\rho})^{2}dz.
        \end{aligned}
    \end{equation}
    For the term with $\partial_{z}\left(\frac{\Psi_{\rho}}{\mr{r}}\right)$, by \eqref{eq: hardy with derivative}, we have \begin{equation}
        \begin{aligned}
            \int_{-1}^{0}(-z)^{2\omega-1}F_{k}^{2}\left(\partial_{z}\left(\frac{\Psi_{\rho}}{\mr{r}}\right)\right)^{2}(s,z)dz\lesssim&\frac{k^{2}}{\epsilon}\frac{1}{\omega_{0}^{2}}(\partial_{z}\Psi_{\rho})^{2}(s,-1)\\&+\frac{k^{2}}{\epsilon}\frac{1}{\omega_{0}^{2}}\int_{-1}^{0}(-z)^{2\omega-2\epsilon+1}(\partial_{z}^{2}\Psi_{\rho})^{2}dz.
        \end{aligned}
    \end{equation}
    
    Similarly, for the term with $\partial_{z}\left(\frac{r_{\rho}}{\mr{r}}\right)$, we have \begin{equation}
        \begin{aligned}
            \int_{-1}^{0}(-z)^{2\omega-1}F_{k}^{2}\left(\partial_{z}\left(\frac{r_{\rho}}{\mr{r}}\right)\right)^{2}(s,z)dz\lesssim&\frac{k^{4}}{\epsilon}\frac{1}{\omega_{0}^{2}}(\partial_{z}r_{\rho})^{2}(s,-1)\\&+\frac{k^{4}}{\epsilon}\frac{1}{\omega_{0}^{2}}\int_{-1}^{0}(-z)^{2(\omega-\epsilon)+1}(\partial_{z}^{2}r_{\rho})^{2}dz.
        \end{aligned}
    \end{equation}
    
    For the term with $\frac{\mr{r}}{\partial_{z}\mr{r}}(\partial_{z}\mr{\phi})^{2}\frac{m_{\rho}}{\mr{r}^{2}}$, by \eqref{eq: first order energy estimate for mrho}, we have \begin{equation}
        \begin{aligned}
            &\int_{-1}^{0}(-z)^{2\omega-1}\left(\frac{\mr{r}}{\partial_{z}\mr{r}}(\partial_{z}\mr{\phi})^{2}\right)^{2}\left(\frac{m_{\rho}}{\mr{r}^{2}}\right)^{2}dz\\\lesssim&\frac{k^{4}}{\epsilon}\int_{-1}^{0}(-z)^{2\omega-\epsilon-1}\left(\frac{m_{\rho}}{\mr{r}^{2}}\right)^{2}dz\\\lesssim&\frac{k^{4}}{\epsilon}\frac{k^{4}}{\epsilon}\frac{1}{\omega_{0}^{3}}\int_{-1}^{0}(-z)^{2(\omega-\epsilon)}\left(\partial_{z}r_{\rho}\right)^{2}+\frac{k^{4}}{\epsilon}\frac{k^{2}}{\epsilon}\frac{1}{\omega_{0}^{3}}\int_{-1}^{0}(-z)^{2(\omega-\epsilon)}(\partial_{z}\Psi_{\rho})^{2}
            \\\lesssim&\frac{k^{8}}{\epsilon^{2}}\frac{1}{\omega_{0}^{4}}(\partial_{z}r_{\rho})^{2}(s,-1)+\frac{k^{6}}{\epsilon^{2}}\frac{1}{\omega_{0}^{4}}(\partial_{z}\Psi_{\rho})^{2}(s,-1)+\frac{k^{8}}{\epsilon^{2}}\frac{1}{\omega_{0}^{5}}\int_{-1}^{0}(-z)^{2(\omega-\epsilon)+2}\left(\partial_{z}^{2}r_{\rho}\right)^{2}\\&+\frac{k^{6}}{\epsilon^{2}}\frac{1}{\omega_{0}^{5}}\int_{-1}^{0}(-z)^{2(\omega-\epsilon)+2}\left(\partial_{z}^{2}\Psi_{\rho}\right)^{2}.
        \end{aligned}
    \end{equation}

    Lastly, for the term with $\frac{2}{\mr{r}}\frac{d\mr{r}}{dz}\frac{m_{\rho}}{\mr{r}^{2}}$, we have \begin{equation}
        \begin{aligned}
            \int_{-1}^{0}(-z)^{2\omega-1}\left(\frac{2}{\mr{r}}\frac{d\mr{r}}{dz}\right)^{2}\left(\frac{m_{\rho}}{\mr{r}^{2}}\right)^{2}dz\lesssim\int_{-1}^{0}(-z)^{2\omega-1}\frac{1}{(1+z)^{2}}\left(\frac{m_{\rho}}{\mr{r}^{2}}\right)^{2}dz.
        \end{aligned}
        \label{eq: first step in estimating mrho term in the second order estimate}
    \end{equation}
    Using \eqref{eq: integral equation for mrho}, we have \begin{equation}
        \left(\frac{m_{\rho}}{\mr{r}^{2}}\right)^{2}(s,z)\lesssim \int_{-1}^{z}\mr{r}F_{k^{2}}^{2}\left\vert\frac{r_{\rho}}{\mr{r}}\right\vert^{2}+\mr{r}F_{k^{2}}^{2}(\partial_{z}r_{\rho})^{2}+\mr{r}F_{k}^{2}\left(\partial_{z}\left(\frac{\Psi_{\rho}}{\mr{r}}\right)\right)^{2}+\mr{r}F_{k^{2}}^{2}\left(\partial_{z}\left(\frac{r_{\rho}}{\mr{r}}\right)\right)^{2}d\widetilde{z}.
        \label{eq: second step in estimaing mrho term in the second order estimates}
    \end{equation}
    Hence, combining \eqref{eq: first step in estimating mrho term in the second order estimate} and \eqref{eq: second step in estimaing mrho term in the second order estimates}, by \eqref{standard hardy} and \eqref{eq: hardy with derivative}, we have \begin{equation}
        \begin{aligned}
            &\int_{-1}^{0}(-z)^{2\omega-1}\left(\frac{2}{\mr{r}}\frac{d\mr{r}}{dz}\right)^{2}\left(\frac{m_{\rho}}{\mr{r}^{2}}\right)^{2}dz\\\lesssim &\frac{1}{\omega_{0}^{2}}\int_{-1}^{0}(-z)^{2\omega-1}F_{k^{2}}^{2}\left(\frac{r_{\rho}}{\mr{r}}\right)^{2}+(-z)^{2\omega-1}F_{k^{2}}^{2}(\partial_{z}r_{\rho})^{2}+(-z)^{2\omega-1}F_{k}^{2}\left(\partial_{z}\left(\frac{\Psi_{\rho}}{\mr{r}}\right)\right)^{2}\\&\quad\quad\quad\quad+(-z)^{2\omega-1}F_{k^{2}}^{2}\left(\partial_{z}\left(\frac{r_{\rho}}{\mr{r}}\right)\right)^{2}dz\\\lesssim&\frac{1}{\omega_{0}^{4}}\frac{k^{4}}{\epsilon}\int_{-1}^{0}(-z)^{2(\omega-\epsilon)-1}(\partial_{z}r_{\rho})^{2}+\frac{1}{\omega_{0}^{4}}\frac{k^{4}}{\epsilon}\int_{-1}^{0}(-z)^{2(\omega-\epsilon)+1}(\partial_{z}^{2}r_{\rho})^{2}\\&+\frac{1}{\omega_{0}^{4}}\frac{k^{2}}{\epsilon}\int_{-1}^{0}(-z)^{2(\omega-\epsilon)+1}(\partial_{z}^{2}\Psi_{\rho})^{2}\\\lesssim& \frac{1}{\omega_{0}^{5}}\frac{k^{4}}{\epsilon}(\partial_{z}r_{\rho})^{2}(s,-1)+\frac{1}{\omega_{0}^{6}}\frac{k^{4}}{\epsilon}\int_{-1}^{0}(-z)^{2(\omega-\epsilon)+1}(\partial_{z}^{2}r_{\rho})^{2}+\frac{1}{\omega_{0}^{4}}\frac{k^{2}}{\epsilon}\int_{-1}^{0}(-z)^{2(\omega-\epsilon)+1}(\partial_{z}^{2}\Psi_{\rho})^{2}.
        \end{aligned}
    \end{equation}
    Hence, putting everything together, we have\begin{equation}
        \begin{aligned}
            &\int_{-1}^{0}(-z)^{2\omega-1}\left(\partial_{z}\left(\frac{m_{\rho}}{\mr{r}^{2}}\right)\right)^{2}(s,z)dz\\\lesssim&\frac{k^{4}}{\epsilon}\frac{1}{\omega_{0}^{5}}(\partial_{z}r_{\rho})^{2}(s,-1)+\frac{k^{2}}{\epsilon}\frac{1}{\omega_{0}^{5}}(\partial_{z}\Psi_{\rho})^{2}(s,-1)+\frac{k^{4}}{\epsilon}\frac{1}{\omega_{0}^{6}}\int_{-1}^{0}(-z)^{2(\omega-\epsilon)+1}\left(\partial_{z}^{2}r_{\rho}\right)^{2}\\&+\frac{k^{2}}{\epsilon}\frac{1}{\omega_{0}^{6}}\int_{-1}^{0}(-z)^{2(\omega-\epsilon)+1}\left(\partial_{z}^{2}\Psi_{\rho}\right)^{2}.
        \end{aligned}
    \end{equation}
\end{proof}

\subsection{First order energy estimates for $r_{\rho}$ and $\Psi_{\rho}$}
In this section, we shall prove the following proposition:
\begin{proposition}
There exist $k$ sufficiently small and constant $C_{0}>0$ depending on the background spacetime $(\mr{r},\mr{\phi},\mr{m})$, such that the following estimate holds in $\mathcal{R}(s_{0})$:\begin{equation}
\begin{aligned}
&\left(\frac{1}{2}-C_{0}k^{2}\right)\int_{s = s_{0}}(\partial_{z}r_{\rho})^{2}+\left(\partial_{z}\Psi_{\rho}\right)^{2}+\frac{1}{2}q_{k}\int_{\Gamma}(\partial_{z}r_{\rho})^{2}+(\partial_{z}\Psi_{\rho})^{2}\\\leq&\left(\frac{1}{2}+C_{0}k^{2}\right)\int_{s = 0}(\partial_{z}r_{\rho})^{2}+(\partial_{z}\Psi_{\rho})^{2}+\left(\frac{1}{2}q_{k}-\rho+C_{0}k\right)\iint_{\mathcal{R}(s_{0})}(\partial_{z}r_{\rho})^{2}+(\partial_{z}\Psi_{\rho})^{2}\\&+C_{0}k^{2}\iint_{\mathcal{R}(s_{0})}(-z)\left(\partial_{z}^{2}r_{\rho}\right)^{2}+(-z)^{2}\left(\partial_{z}^{2}\Psi_{\rho}\right)^{2}.
\end{aligned}
\label{first order energy estimate for the system}.
\end{equation}
\end{proposition}
\begin{proof}
Multiplying \eqref{eq:wave equation for r rho} by $\partial_{z}r_{\rho}$, integrating over $\mathcal{R}(s_{0})$, and using \eqref{standard hardy}, we have\begin{equation}
\begin{aligned}
&\frac{1}{2}\int_{s = s_{0}}(\partial_{z}r_{\rho})^{2}dz+\frac{1}{2}q_{k}\int_{\Gamma}(\partial_{z}r_{\rho})^{2}ds+\left(\rho-\frac{1}{2}q_{k}-C_{0}k^{2}\right)\iint_{\mathcal{R}(s_{0})}(\partial_{z}r_{\rho})^{2}dsdz\\\leq&
\frac{1}{2}\int_{s = 0}(\partial_{z}r_{\rho})^{2}dz+\iint_{\mathcal{R}(s_{0})}G_{k^{2}}(z)\partial_{s}r_{\rho}\partial_{z}r_{\rho}+G_{k^{2}}(z)r_{\rho}\partial_{z}r_{\rho}+K(z)\frac{m_{\rho}}{\mr{r}^{2}}\partial_{z}r_{\rho}dsdz\\\leq&
\frac{1}{2}\int_{s = 0}(\partial_{z}r_{\rho})^{2}dz+C_{0}k\iint_{\mathcal{R}(s_{0})}(\partial_{z}r_{\rho})^{2}dsdz+\frac{1}{k}\iint_{\mathcal{R}(s_{0})}\left(\frac{m_{\rho}}{\mr{r}^{2}}\right)^{2}dsdz\\&+\left\vert\iint_{\mathcal{R}(s_{0})}G_{k^{2}}(z)\partial_{s}r_{\rho}\partial_{z}r_{\rho}dsdz\right\vert.
\end{aligned}
\label{eq: crude version of the first-order energy estimate for rrho}
\end{equation}
For the term with $G_{k^{2}}\partial_{s}r_{\rho}\partial_{z}r_{\rho}$ on the right-hand side of \eqref{eq: crude version of the first-order energy estimate for rrho}, we have\begin{equation}
\begin{aligned}
&\left\vert\iint_{\mathcal{R}(s_{0})}G_{k^{2}}(z)\partial_{s}r_{\rho}\partial_{z}r_{\rho}dsdz\right\vert\\\lesssim& \left\vert\int_{s = s_{0}}G_{k^{2}}r_{\rho}\partial_{z}r_{\rho}\right\vert+\left\vert\int_{s = 0}G_{k^{2}}r_{\rho}\partial_{z}r_{\rho}\right\vert+\left\vert\iint_{\mathcal{R}(s_{0})}G_{k^{2}}r_{\rho}\partial_{s}\partial_{z}r_{\rho}dsdz\right\vert\\\lesssim& k^{2}\int_{s = s_{0}}(\partial_{z}r_{\rho})^{2}dz+k^{2}\int_{s = 0}(\partial_{z}r_{\rho})^{2}dz+k^{2}\iint_{\mathcal{R}(s_{0})}(\partial_{z}r_{\rho})^{2}dsdz\\&+k^{2}\iint_{\mathcal{R}(s_{0})}(-z)\left(\partial_{z}^{2}r_{\rho}\right)^{2}dsdz+k^{2}\iint_{\mathcal{R}(s_{0})}\left(\frac{m_{\rho}}{\mr{r}^{2}}\right)^{2}dsdz+\left\vert\iint_{\mathcal{R}(s_{0})}G_{k^{2}}r_{\rho}\partial_{s}r_{\rho}dsdz\right\vert\\\lesssim&
k^{2}\int_{s = s_{0}}\left(\partial_{z}r_{\rho}\right)^{2}dz+k^{2}\int_{s = 0}(\partial_{z}r_{\rho})^{2}dz+k^{2}\iint_{\mathcal{R}(s_{0})}(\partial_{z}r_{\rho})^{2}dsdz\\&
+k^{2}\iint_{\mathcal{R}(s_{0})}(-z)(\partial_{z}^{2}r_{\rho})^{2}dsdz+k^{2}\iint_{\mathcal{R}(s_{0})}\left(\frac{m_{\rho}}{\mr{r}^{2}}\right)^{2}dsdz.
\end{aligned}
\label{eq: integration by parts in the s direction}
\end{equation}
We can apply \eqref{standard hardy} to bound the bulk term of $\partial_{z}r_{\rho}$ on the right-hand side of \eqref{eq: integration by parts in the s direction}. Thus, we have \begin{equation}
    \begin{aligned}
        \left\vert\iint_{\mathcal{R}(s_{0})}G_{k^{2}}(z)\partial_{s}r_{\rho}\partial_{z}r_{\rho}\right\vert\lesssim& k^{2}\int_{s = s_{0}}\left(\partial_{z}r_{\rho}\right)^{2}+k^{2}\int_{s = 0}(\partial_{z}r_{\rho})^{2}+k^{2}\int_{\Gamma}\left(\partial_{z}r_{\rho}\right)^{2}\\&+k^{2}\iint_{\mathcal{R}(s_{0})}(-z)\left(\partial_{z}^{2}r_{\rho}\right)^{2}+k^{2}\iint_{\mathcal{R}(s_{0})}\left(\frac{m_{\rho}}{\mr{r}^{2}}\right)^{2}.
    \end{aligned}
\end{equation}
Using \eqref{eq: first order energy estimate for mrho} to estimate the terms with $\frac{m_{\rho}}{\mr{r}^{2}}$ and putting everything together, we have \begin{equation}
\begin{aligned}
    &\left(\frac{1}{2}-C_{0}k^{2}\right)\int_{s = s_{0}}(\partial_{z}r_{\rho})^{2}+\frac{1}{2}q_{k}\int_{\Gamma}(\partial_{z}r_{\rho})^{2}\\\leq& \left(\frac{1}{2}+C_{0}k^{2}\right)\int_{s = 0}(\partial_{z}r_{\rho})^{2}+C_{0}k\iint_{\mathcal{R}(s_{0})}(\partial_{z}\Psi_{\rho})^{2}+C_{0}k^{2}\iint_{\mathcal{R}(s_{0})}(-z)\left(\partial_{z}^{2}r_{\rho}\right)^{2}\\&+\left(\frac{1}{2}q_{k}-\rho+C_{0}k^{2}\right)\iint_{\mathcal{R}(s_{0})}(\partial_{z}r_{\rho})^{2}.
    \end{aligned}
    \label{first order energy estimate for r rho}
\end{equation}

Next, we do the first-order energy estimate for $\Psi_{\rho}$. Multiplying \eqref{eq:wave equation for Psi rho} by $\partial_{z}\Psi_{\rho}$ and integrating over the region $\mathcal{R}(s_{0})$, we have\begin{equation}
\begin{aligned}
&\frac{1}{2}\int_{s = s_{0}}(\partial_{z}\Psi_{\rho})^{2}dz+\frac{1}{2}q_{k}\int_{\Gamma}(\partial_{z}\Psi_{\rho})^{2}ds+\left(\rho-\frac{1}{2}q_{k}\right)\iint_{\mathcal{R}(s_{0})}(\partial_{z}\Psi_{\rho})^{2}dsdz\\\leq& \frac{1}{2}\int_{s = 0}(\partial_{z}\Psi_{\rho})^{2}dz+\iint_{\mathcal{R}(s_{0})}G_{k^{2}}(z)\Psi_{\rho}\partial_{z}\Psi_{\rho}dsdz+\iint_{\mathcal{R}(s_{0})}G_{k^{2}}\partial_{s}r_{\rho}\partial_{z}\Psi_{\rho}dsdz\\&+\iint_{\mathcal{R}(s_{0})}\left(G_{k^{2}}(z)+k\right)\partial_{z}r_{\rho}\partial_{z}\Psi_{\rho}+G_{k^{2}}(z)r_{\rho}\partial_{z}\Psi_{\rho}+\vert z\vert^{\frac{1}{2}}G_{k}(z)\frac{m_{\rho}}{\mr{r}^{2}}\partial_{z}\Psi_{\rho}dsdz\\\leq&
\frac{1}{2}\int_{s = 0}(\partial_{z}\Psi_{\rho})^{2}dz+C_{0}k\iint_{\mathcal{R}(s_{0})}(\partial_{z}\Psi_{\rho})^{2}dsdz+C_{0}k\iint_{\mathcal{R}(s_{0})}(\partial_{z}r_{\rho})^{2}dsdz\\&+\iint_{\mathcal{R}(s_{0})}\left(\frac{m_{\rho}}{\mr{r}^{2}}\right)^{2}dsdz+\left\vert\iint_{\mathcal{R}(s_{0})}G_{k^{2}}(z)\partial_{s}r_{\rho}\partial_{z}\Psi_{\rho}dsdz\right\vert.
\end{aligned}
\label{expression of first order energy estimate for Psi rho}
\end{equation}
Applying integration by parts, we can control the term with $\partial_{s}r_{\rho}\partial_{z}\Psi_{\rho}$ by\begin{equation}
\begin{aligned}
&\left\vert\iint_{\mathcal{R}(s_{0})}G_{k^{2}}(z)\partial_{s}r_{\rho}\partial_{z}\Psi_{\rho}dsdz\right\vert\\\lesssim&k^{2} \int_{s = s_{0}}(\partial_{z}r_{\rho})^{2}dz+k^{2}\int_{s = s_{0}}(\partial_{z}\Psi_{\rho})^{2}dz+k^{2}\int_{s = 0}(\partial_{z}r_{\rho})^{2}dz+k^{2}\int_{s = 0}(\partial_{z}\Psi_{\rho})^{2}dz\\&+k^{2}\left\vert\iint_{\mathcal{R}(s_{0})}r_{\rho}\partial_{s}\partial_{z}\Psi_{\rho}dsdz\right\vert
\\\lesssim& k^{2}\int_{s = s_{0}}(\partial_{z}r_{\rho})^{2}dz+k^{2}\int_{s = s_{0}}(\partial_{z}\Psi_{\rho})^{2}dz+k^{2}\int_{s = 0}(\partial_{z}r_{\rho})^{2}dz+k^{2}\int_{s = 0}(\partial_{z}\Psi_{\rho})^{2}dz\\&+k^{2}\iint_{\mathcal{R}(s_{0})}(\partial_{z}r_{\rho})^{2}dsdz+k^{2}\iint_{\mathcal{R}(s_{0})}(-z)^{2}\left(\partial_{z}^{2}\Psi_{\rho}\right)^{2}dsdz+k^{2}\iint_{\mathcal{R}(s_{0})}(\partial_{z}\Psi_{\rho})^{2}dsdz\\&+k^{2}\iint_{\mathcal{R}(s_{0})}\left(\frac{m_{\rho}}{\mr{r}^{2}}\right)^{2}dsdz+k^{4}\left\vert\iint_{\mathcal{R}(s_{0})}r_{\rho}\partial_{s}r_{\rho}dsdz\right\vert\\\leq& k^{2}\int_{s = s_{0}}(\partial_{z}r_{\rho})^{2}dz+k^{2}\int_{s = s_{0}}(\partial_{z}\Psi_{\rho})^{2}dz+k^{2}\int_{s = 0}(\partial_{z}r_{\rho})^{2}dz+k^{2}\int_{s = 0}(\partial_{z}\Psi_{\rho})^{2}dz\\&+k^{2}\iint_{\mathcal{R}(s_{0})}(\partial_{z}r_{\rho})^{2}dsdz+k^{2}\iint_{\mathcal{R}(s_{0})}\left(\partial_{z}\Psi_{\rho}\right)^{2}+k^{2}\iint_{\mathcal{R}(s_{0})}(-z)^{2}\left(\partial_{z}^{2}\Psi_{\rho}\right)^{2}dsdz\\&+k^{2}\iint_{\mathcal{R}(s_{0})}\left(\frac{m_{\rho}}{\mr{r}^{2}}\right)^{2}dsdz,
\end{aligned}
\end{equation}
where in the second inequality we have used the expression of $\partial_{s}\partial_{z}\Psi_{\rho}$ \eqref{eq:wave equation for Psi rho} and used the Cauchy--Schwarz to control each term.

Using \eqref{eq: first order energy estimate for mrho} to bound the term with $\frac{m_{\rho}}{\mr{r}^{2}}$ in \eqref{expression of first order energy estimate for Psi rho}, we have\begin{equation}
\begin{aligned}
&\left(\frac{1}{2}-C_{0}k^{2}\right)\int_{s = s_{0}}(\partial_{z}\Psi_{\rho})^{2}dz+\frac{1}{2}q_{k}\int_{\Gamma}(\partial_{z}\Psi_{\rho})^{2}ds\\\leq&
\left(\frac{1}{2}+C_{0}k^{2}\right)\int_{s =0}(\partial_{z}\Psi_{\rho})^{2}dz+C_{0}k^{2}\int_{s = 0}(\partial_{z}r_{\rho})^{2}dz+C_{0}k^{2}\int_{s = s_{0}}(\partial_{z}r_{\rho})^{2}dz\\&
+\left(\frac{1}{2}q_{k}-\rho+C_{0}k\right)\iint_{\mathcal{R}(s_{0})}(\partial_{z}\Psi_{\rho})^{2}dsdz+C_{0}k\iint_{\mathcal{R}(s_{0})}\left(\partial_{z}r_{\rho}\right)^{2}dsdz\\&+C_{0}k^{2}\iint_{\mathcal{R}(s_{0})}(-z)^{2}\left(\partial_{z}^{2}\Psi_{\rho}\right)^{2}dsdz.
\end{aligned}
\label{first order energy estimate for psi rho}
\end{equation}

Combining \eqref{first order energy estimate for r rho} and \eqref{first order energy estimate for psi rho}, we have\begin{equation}
\begin{aligned}
&\left(\frac{1}{2}-C_{0}k^{2}\right)\int_{s = s_{0}}(\partial_{z}r_{\rho})^{2}+\left(\partial_{z}\Psi_{\rho}\right)^{2}+\frac{1}{2}q_{k}\int_{\Gamma}(\partial_{z}r_{\rho})^{2}+(\partial_{z}\Psi_{\rho})^{2}\\\leq&\left(\frac{1}{2}+C_{0}k^{2}\right)\int_{s = 0}(\partial_{z}r_{\rho})^{2}+(\partial_{z}\Psi_{\rho})^{2}+\left(\frac{1}{2}q_{k}-\rho+C_{0}k\right)\iint_{\mathcal{R}(s_{0})}(\partial_{z}r_{\rho})^{2}+(\partial_{z}\Psi_{\rho})^{2}\\&+C_{0}k^{2}\iint_{\mathcal{R}(s_{0})}(-z)\left(\partial_{z}^{2}r_{\rho}\right)^{2}+(-z)^{2}\left(\partial_{z}^{2}\Psi_{\rho}\right)^{2}.
\end{aligned}
\end{equation}
This concludes the proof.
\end{proof}

\subsection{A quick overview of the second-order energy estimate}
As we can see from \eqref{first order energy estimate for the system}, the first order energy estimate itself is not enough to even get the $L^{2}$-boundedness of the solutions. The main obstacle here is the presence of the bad sign of the bulk term when $\rho$ is sufficiently close to the blue-shift value. To trace the origin of the bad sign, expressing the principal part of the wave operator $\partial_{u}\partial_{v}$ under the self-similar coordinates $(s,z)$ and using $(r_{\rho},\Psi_{\rho})$ as variables, we can schematically write the equations \eqref{eq:wave equation for r rho}-\eqref{eq:wave equation for Psi rho} as \begin{equation}
    \partial_{s}\partial_{z}\Xi+q_{k}z\partial_{z}^{2}\Xi+\rho\partial_{z}\Xi = \text{Error terms},\quad \Xi\in\{r_{\rho},\Psi_{\rho}\}.\label{eq: quick overview of the motivation of the second order estimate}
\end{equation}
Multiplying by $\partial_{z}\Xi$ and applying integration by parts, we can derive: \begin{equation*}
    \frac{1}{2}\int_{s = s_{0}}(\partial_{z}\Xi)^{2}dz+\frac{1}{2}q_{k}\int_{\Gamma}(\partial_{z}\Xi)^{2}ds+\left(\rho-\frac{1}{2}q_{k}\right)\iint_{\mathcal{R}(s_{0})}(\partial_{z}\Xi)^{2}dsdz = \text{Initial data}+\text{Error terms}.
\end{equation*}
To get as sharp a decay rate as possible, the value of $\rho$ shall be chosen to be at least close to the blue-shift value, thus creating the bad sign in the coefficients of the bulk terms in the first-order energy estimates.

Nonetheless, this bad sign can be compensated for by the second-order energy estimate. Commuting the equation \eqref{eq: quick overview of the motivation of the second order estimate} with $\partial_{z}$, we can derive \begin{equation*}
    \partial_{s}\partial_{z}^{2}\Xi+q_{k}z\partial_{z}^{3}\Xi+\left(\rho+q_{k}\right)\partial_{z}^{2}\Xi = \text{Error terms},\quad \Xi\in\{r_{\rho},\Psi_{\rho}\},
\end{equation*}
where we have gained $q_{k}$ in the coefficient of the terms with second-order $z$-derivative. However, due to the low-regularity nature of $\Psi_{\rho}$, we have to add a suitable $z$-weight in the choice of the multiplier. A natural choice of the multiplier might be $(-z)^{2\gamma}\partial_{z}^{2}\Psi_{\rho}$, where $\gamma\in[0,\frac{1}{2})$ depending on the regularity of $\Psi_{\rho}$. If $\Psi_{\rho}\in\mathcal{C}^{\alpha}$ for $\alpha\in(1,2)$, one can roughly view $\Psi_{\rho}$ as $(-z)^{\alpha}$. Then the correct weight $\gamma$ should be $\frac{3}{2}-\alpha+\delta$ for $\delta$ arbitrarily close to $0$. Multiplying $(-z)^{2\gamma}\partial_{z}^{2}\Psi_{\rho}$, we have 
\begin{equation}
    \begin{aligned}
    &\frac{1}{2}\int_{s = s_{0}}(-z)^{2\gamma}\left(\partial_{z}^{2}\Psi_{\rho}\right)^{2}dz+\frac{1}{2}q_{k}\int_{\Gamma}\left(\partial_{z}^{2}\Psi_{\rho}\right)^{2}ds+\left(\rho+q_{k}\left(\frac{1}{2}-\gamma\right)\right)\iint_{\mathcal{R}(s_{0})}(-z)^{2\gamma}\left(\partial_{z}^{2}\Psi_{\rho}\right)^{2}\\=&\text{Initial data}+\text{Error terms}.
\end{aligned}
\label{eq: heuristics of the second-order energy estimate}
\end{equation}
If the regularity of $\Psi_{\rho}$ is sufficiently close to the threshold, i.e., $\alpha\rightarrow p_{k}$, then the coefficient for the bulk term \begin{equation*}
    \rho+q_{k}\left(\frac{1}{2}-\gamma\right)\rightarrow \rho+k^{2}-q_{k}\delta
\end{equation*}
will only be positive for $\rho>-k^{2}$. Recall that $\rho = -k^{2}$ corresponds to the self-similar value $\rho$. Hence, only in view of this bulk term, second-order energy may be closed for $\rho$ arbitrarily close to the self-similar value. However, the estimate for the error terms will generate several bulk terms with coefficients of size $k$ on the left-hand side of \eqref{eq: heuristics of the second-order energy estimate}, which can only be absorbed by the bulk term on the right-hand side of \eqref{eq: heuristics of the second-order energy estimate} if $\rho> Bk$ for some constant $B$.

One key novelty of our second-order energy estimate is that, instead of using $(-z)^{2\gamma}\partial_{z}^{2}\Psi_{\rho}$ as a multiplier, we use $(-z)^{2\gamma}(\mr{r}+1)\partial_{z}^{2}\Psi_{\rho}$. As a result, the weight in the bulk term on the right-hand side of \eqref{eq: heuristics of the second-order energy estimate} will become \begin{equation}
    \left(\rho+q_{k}\left(\frac{1}{2}-\gamma\right)\right)(-z)^{2\gamma}+\frac{1}{2}q_{k}\frac{d\mr{r}}{dz}(-z)^{2\gamma+1}.\label{eq: heuristic bulk term coefficient}
\end{equation}
The advantage of \eqref{eq: heuristic bulk term coefficient} is that we obtain a ``large" term \begin{equation*}
    \frac{1}{2}q_{k}\frac{d\mr{r}}{dz}(-z)^{2\gamma+1}
\end{equation*}
with a faster degenerate rate near the singular horizon $\{z = 0\}$. By exploiting some Hardy-type inequalities near the singular horizon carefully, we can bound all the error terms by \begin{equation*}
    \iint_{\mathcal{R}(s_{0})}\frac{1}{4}q_{k}\frac{d\mr{r}}{dz}(-z)^{2\gamma+1}\left(\partial_{z}^{2}\Psi_{\rho}\right)^{2}.
\end{equation*}
Hence, this new multiplier will allow the value of $\rho$ to be arbitrarily close to the self-similar value.

However, we remark that even if eventually we can get the positivity of bulk terms for some $\rho<-k^{2}$ when $\alpha>p_{k}$, we still cannot close the energy estimates for $\rho<-k^{2}$, which is consistent with the observation in Remark \ref{rmk: observation why we could not get sharp decay from energy}. The reason lies in the first-order energy estimate \eqref{first order energy estimate for the system}. The coefficient of the first-order bulk term on the right-hand side of \eqref{first order energy estimate for the system} can be viewed as $\frac{1}{2}q_{k}-\rho\approx \frac{1}{2}q_{k}$, while the fundamental theorem of calculus can only give \begin{equation*}
    \left(\frac{1}{2}q_{k}-\rho\right)\iint_{\mathcal{R}(s_{0})}(\partial_{z}r_{\rho})^{2}+\left(\partial_{z}\Psi_{\rho}\right)^{2}\leq (1+\epsilon)\left(\frac{1}{2}q_{k}-\rho\right)\int_{\Gamma}(\partial_{z}r_{\rho})^{2}+\text{Second-order terms}.
\end{equation*}
Hence, if $\rho<0$, then we lose control on the first-order terms on the axis, since \begin{equation*}
    (1+\epsilon)\left(\frac{1}{2}q_{k}-\rho\right)>\frac{1}{2}q_{k}.
\end{equation*}

\subsection{Second-order energy estimate for $\Psi_{\rho}$}
\label{sec: second-order energy estimate for psirho}
In this section, we deal with the second-order energy estimate of $\Psi_{\rho}$. We can prove the following proposition.
\begin{proposition}
\label{prop: second-order energy estimate for psirho}
    There exist $k$ sufficiently small, a constant $C_{1}>0$ depending on the background spacetime $(\mr{r},\mr{\phi},\mr{m})$, and a real number $\frac{3}{2}-\alpha<\gamma<\frac{1}{2}$, such that the following estimate holds in $\mathcal{R}(s_{0})$: \begin{equation}
        \begin{aligned}
            &\int_{s = s_{0}}(-z)^{2\gamma}\left(\partial_{z}^{2}\Psi_{\rho}\right)^{2}+\int_{\Gamma}\left(\partial_{z}^{2}\Psi_{\rho}\right)^{2}\\&+\iint_{\mathcal{R}(s_{0})}\left[\left(\rho+q_{k}\left(\frac{1}{2}-\gamma\right)\right)(-z)^{2\gamma}+\frac{1}{4}q_{k}(-z)^{2\gamma+\frac{1}{2}}\right]\left(\partial_{z}^{2}\Psi_{\rho}\right)^{2}\\\leq&C_{1}\Biggl(\int_{s = 0}(-z)^{2\gamma}\left(\partial_{z}^{2}\Psi_{\rho}\right)^{2}+k^{2}\int_{\Gamma}(\partial_{z}r_{\rho})^{2}+k^{2}\int_{\Gamma}(\partial_{z}\Psi_{\rho})^{2}+k\iint_{\mathcal{R}(s_{0})}(-z)^{2\gamma-\frac{1}{2}}\left(\partial_{z}^{2}r_{\rho}\right)^{2}\Biggr).
        \end{aligned}
        \label{eq: second-order energy estimate for psirho}
    \end{equation}
\end{proposition}
\begin{proof}
Let $w(z) = -2(-z)^{\frac{1}{2}}+3$. Multiplying \eqref{second order wave equation for psi rho} by $(-z)^{2\gamma+1}w(z)\partial_{z}^{2}\Psi_{\rho}$ and integrating over the region $\mathcal{R}(s_{0})$ we have \begin{equation}
        \begin{aligned}
            &\frac{1}{2}\int_{s = s_{0}}(-z)^{2\gamma}w(z)\left(\partial_{z}^{2}\Psi_{\rho}\right)^{2}+\frac{1}{2}q_{k}\int_{\Gamma}\left(\partial_{z}^{2}\Psi_{\rho}\right)^{2}\\&+\iint_{\mathcal{R}(s_{0})}\left[\left(\rho+q_{k}\left(\frac{1}{2}-\gamma\right)\right)w(z)(-z)^{2\gamma}+\frac{1}{2}q_{k}(-z)^{2\gamma+\frac{1}{2}}\right]\left(\partial_{z}^{2}\Psi_{\rho}\right)^{2}dzds\\&\lesssim\int_{s = 0}(-z)^{2\gamma}w(z)\left(\partial_{z}^{2}\Psi_{\rho}\right)^{2}+\iint_{\mathcal{R}(s_{0})}(-z)^{2\gamma}G_{k^{2}}\partial_{z}\Psi_{\rho}\partial_{z}^{2}\Psi_{\rho}+(-z)^{2\gamma}F_{k^{2}}\Psi_{\rho}\partial_{z}^{2}\Psi_{\rho}\\&+\iint_{\mathcal{R}(s_{0})}(-z)^{2\gamma}G_{k}\partial_{z}^{2}r_{\rho}\partial_{z}^{2}\Psi_{\rho}+(-z)^{2\gamma}F_{k^{2}}\partial_{s}r_{\rho}\partial_{z}^{2}\Psi_{\rho}+(-z)^{2\gamma}F_{k^{2}}\partial_{z}r_{\rho}\partial_{z}^{2}\Psi_{\rho}+(-z)^{2\gamma}F_{k^{2}}r_{\rho}\partial_{z}^{2}\Psi_{\rho}\\&+\iint_{\mathcal{R}(s_{0})}(-z)^{2\gamma}(-z)^{\frac{1}{2}}G_{k}(z)\partial_{z}\left(\frac{m_{\rho}}{\mr{r}^{2}}\right)\partial_{z}^{2}\Psi_{\rho}+(-z)^{2\gamma}F_{k}(z)\frac{m_{\rho}}{\mr{r}^{2}}\partial_{z}^{2}\Psi_{\rho}.
        \end{aligned}
        \label{eq: crude version of psirho second order energy estimate}
    \end{equation}
    To estimate the term with $\partial_{z}\Psi_{\rho}$ on the right-hand side of \eqref{eq: crude version of psirho second order energy estimate}, we have \begin{equation}
        \begin{aligned}
            \left\vert\iint_{\mathcal{R}(s_{0})}(-z)^{2\gamma}G_{k^{2}}\partial_{z}\Psi_{\rho}\partial_{z}^{2}\Psi_{\rho}\right\vert\lesssim& k^{2}\iint_{\mathcal{R}(s_{0})}(-z)^{2\gamma-\frac{1}{2}}\left(\partial_{z}\Psi_{\rho}\right)^{2}+k^{2}\iint_{\mathcal{R}(s_{0})}(-z)^{2\gamma+\frac{1}{2}}\left(\partial_{z}^{2}\Psi_{\rho}\right)^{2}\\\lesssim&k^{2}\int_{\Gamma}\left(\partial_{z}\Psi_{\rho}\right)^{2}+k^{2}\iint_{\mathcal{R}(s_{0})}(-z)^{2\gamma+\frac{1}{2}}\left(\partial_{z}^{2}\Psi_{\rho}\right)^{2},
        \end{aligned}
    \end{equation}
    where the last inequality follows from \eqref{eq: fundamental theorem of calculus} by choosing $\omega = \gamma+\frac{1}{4}$ and the fact that $\gamma\geq0$.

    To estimate the term with $\Psi_{\rho}$ on the right-hand side of \eqref{eq: crude version of psirho second order energy estimate}, we have \begin{equation}
        \begin{aligned}
            \left\vert\iint_{\mathcal{R}(s_{0})}(-z)^{2\gamma}F_{k^{2}}\Psi_{\rho}\partial_{z}^{2}\Psi_{\rho}\right\vert\lesssim& k^{2}\iint_{\mathcal{R}(s_{0})}(-z)^{2\gamma+\frac{1}{2}}\left(\partial_{z}^{2}\Psi_{\rho}\right)^{2}+\frac{1}{k^{2}}\iint_{\mathcal{R}(s_{0})}(-z)^{2\gamma-\frac{1}{2}}F_{k^{2}}^{2}\Psi_{\rho}^{2}\\\lesssim&k^{2}\iint_{\mathcal{R}(s_{0})}(-z)^{2\gamma+\frac{1}{2}}\left(\partial_{z}^{2}\Psi_{\rho}\right)^{2}+k^{2}\iint_{\mathcal{R}(s_{0})}(-z)^{2\gamma-1}\Psi_{\rho}^{2}\\\lesssim&k^{2}\iint_{\mathcal{R}(s_{0})}(-z)^{2\gamma+\frac{1}{2}}\left(\partial_{z}^{2}\Psi_{\rho}\right)^{2}+k^{2}\iint_{\mathcal{R}(s_{0})}(-z)^{2\gamma+1}(\partial_{z}\Psi_{\rho})^{2}\\\lesssim&
            k^{2}\int_{\Gamma}(\partial_{z}\Psi_{\rho})^{2}+
            k^{2}\iint_{\mathcal{R}(s_{0})}(-z)^{2\gamma+\frac{1}{2}}\left(\partial_{z}^{2}\Psi_{\rho}\right)^{2},
        \end{aligned}
        \label{eq: second order energy estimate with psirho}
    \end{equation}
    where the last two inequalities follow from \eqref{eq: fundamental theorem of calculus}.

    To estimate the term with $\partial_{z}^{2}r_{\rho}$ on the right-hand side of \eqref{eq: crude version of psirho second order energy estimate}, we have \begin{equation}
    \begin{aligned}
        \left\vert\iint_{\mathcal{R}(s_{0})}(-z)^{2\gamma}G_{k}\partial_{z}^{2}r_{\rho}\partial_{z}^{2}\Psi_{\rho}\right\vert\lesssim& k\iint_{\mathcal{R}(s_{0})}(-z)^{2\gamma-\frac{1}{2}}\left(\partial_{z}^{2}r_{\rho}\right)^{2}+k\iint_{\mathcal{R}(s_{0})}(-z)^{2\gamma+\frac{1}{2}}\left(\partial_{z}^{2}\Psi_{\rho}\right)^{2}\\\lesssim& k\iint_{\mathcal{R}(s_{0})}(-z)^{2\gamma-\frac{1}{2}}\left(\partial_{z}^{2}r_{\rho}\right)^{2}+k\iint_{\mathcal{R}(s_{0})}(-z)^{2\gamma+\frac{1}{2}}\left(\partial_{z}^{2}\Psi_{\rho}\right)^{2}.
        \end{aligned}
    \end{equation}

    To estimate the term with $\partial_{z}r_{\rho}$ on the right-hand side of \eqref{eq: crude version of psirho second order energy estimate}, we have \begin{equation}
        \begin{aligned}
            \left\vert\iint_{\mathcal{R}(s_{0})}(-z)^{2\gamma}F_{k^{2}}\partial_{z}r_{\rho}\partial_{z}^{2}\Psi_{\rho}\right\vert\lesssim&k^{2}\iint_{\mathcal{R}(s_{0})}(-z)^{2\gamma+\frac{1}{2}}\left(\partial_{z}^{2}\Psi_{\rho}\right)^{2}+\frac{1}{k^{2}}\iint_{\mathcal{R}(s_{0})}(-z)^{2\gamma-\frac{1}{2}}F_{k^{2}}^{2}\left(\partial_{z}r_{\rho}\right)^{2}\\\lesssim&k^{2}\iint_{\mathcal{R}(s_{0})}(-z)^{2\gamma+\frac{1}{2}}\left(\partial_{z}^{2}\Psi_{\rho}\right)^{2}+k^{2}\iint_{\mathcal{R}(s_{0})}(-z)^{2\gamma-\frac{3}{4}}\left(\partial_{z}r_{\rho}\right)^{2}\\\lesssim&k^{2}\int_{\Gamma}\left(\partial_{z}r_{\rho}\right)^{2}+k^{2}\iint_{\mathcal{R}(s_{0})}(-z)^{2\gamma+\frac{1}{2}}\left(\partial_{z}^{2}\Psi_{\rho}\right)^{2}\\&+k^{2}\iint_{\mathcal{R}(s_{0})}(-z)^{2\gamma-\frac{1}{2}}(\partial_{z}^{2}r_{\rho})^{2},
        \end{aligned}
    \end{equation}
    where the last inequality follows from \eqref{eq: fundamental theorem of calculus} by choosing $\omega = \gamma+\frac{1}{8}$.

To estimate the term with $r_{\rho}$ on the right-hand side of \eqref{eq: crude version of psirho second order energy estimate}, similar to \eqref{eq: second order energy estimate with psirho}, we have \begin{equation}
    \begin{aligned}
        \left\vert\iint_{\mathcal{R}(s_{0})}(-z)^{2\gamma}F_{k^{2}}r_{\rho}\partial_{z}^{2}\Psi_{\rho}\right\vert\lesssim&k^{2}\int_{\Gamma}\left(\partial_{z}r_{\rho}\right)^{2}+k^{2}\iint_{\mathcal{R}(s_{0})}(-z)^{2\gamma+\frac{1}{2}}\left(\partial_{z}^{2}\Psi_{\rho}\right)^{2}\\&+k^{2}\iint_{\mathcal{R}(s_{0})}(-z)^{2\gamma-\frac{1}{2}}\left(\partial_{z}^{2}r_{\rho}\right)^{2}.
    \end{aligned}
\end{equation}
To estimate the term with $\partial_{s}r_{\rho}$ on the right-hand side of \eqref{eq: crude version of psirho second order energy estimate}, we have \begin{equation}
\begin{aligned}
\left\vert\iint_{\mathcal{R}(s_{0})}(-z)^{2\gamma}F_{k^{2}}\partial_{s}r_{\rho}\partial_{z}^{2}\Psi_{\rho}\right\vert\lesssim k^{2}\iint_{\mathcal{R}(s_{0})}(-z)^{2\gamma+\frac{1}{2}}\left(\partial_{z}^{2}\Psi_{\rho}\right)^{2}+k^{2}\iint_{\mathcal{R}(s_{0})}(-z)^{2\gamma-\frac{1}{2}}(\partial_{s}r_{\rho})^{2}.
\end{aligned}
\end{equation}
To estimate $\partial_{s}r_{\rho}$, we view \eqref{eq:wave equation for r rho} as a transport equation for $\partial_{s}r_{\rho}$: \begin{equation*}
    \partial_{z}\left(e^{\int_{-1}^{z}G_{k^{2}}(\widetilde{z})d\widetilde{z}}\partial_{s}r_{\rho}\right) = e^{\int_{-1}^{z}G_{k^{2}}(\widetilde{z})d\widetilde{z}}\left(q_{k}(-z)\partial_{z}^{2}r_{\rho}+\left(\rho+G_{k^{2}}(z)\right)\partial_{z}r_{\rho}+G_{k^{2}}r_{\rho}+K\frac{m_{\rho}}{\mr{r}^{2}}\right).
\end{equation*}
Then we have \begin{equation*}
    \vert\partial_{s}r_{\rho}(s,z)\vert^{2}\lesssim \int_{-1}^{z}\vert \widetilde{z}\vert^{2}\left\vert\partial_{z}^{2}r_{\rho}\right\vert^{2}+\vert\partial_{z}r_{\rho}\vert^{2}+\vert r_{\rho}\vert^{2}+\vert K\vert^{2}\left\vert\frac{m_{\rho}}{\mr{r}^{2}}\right\vert^{2} d\widetilde{z}.
\end{equation*}
Applying the Fubini Theorem, we have \begin{align*}
    \int_{-1}^{0}(-z)^{2\gamma-\frac{1}{2}}\left(\partial_{s}r_{\rho}\right)^{2}dz\lesssim &\int_{-1}^{0}(-z)^{2\gamma-\frac{1}{2}}\left(\partial_{z}^{2}r_{\rho}\right)^{2}+(-z)^{2\gamma-\frac{3}{4}}\left((\partial_{z}r_{\rho})^{2}+r_{\rho}^{2}\right)\\&+\int_{-1}^{0}(-z)^{2\gamma-1}\left\vert K(z)\frac{m_{\rho}}{\mr{r}}\right\vert^{2}.
\end{align*}
Hence, similar to the estimates we have shown above, using the first-order estimate for $\frac{m_{\rho}}{\mr{r}^{2}}$ \eqref{eq: first-order energy estimate using second order terms} we have \begin{equation}
    \begin{aligned}
        &\left\vert\iint_{\mathcal{R}(s_{0})}(-z)^{2\gamma}F_{k^{2}}\partial_{s}r_{\rho}\partial_{z}^{2}\Psi_{\rho}\right\vert\\\lesssim& k^{2}\int_{\Gamma}(\partial_{z}r_{\rho})^{2}+k^{2}\iint_{\mathcal{R}(s_{0})}(-z)^{2\gamma+\frac{1}{2}}\left(\partial_{z}^{2}\Psi_{\rho}\right)^{2}\\&+k^{2}\iint_{\mathcal{R}(s_{0})}(-z)^{2\gamma-\frac{1}{2}}\left(\partial_{z}^{2}r_{\rho}\right)^{2}+\iint_{\mathcal{R}(s_{0})}(-z)^{2\gamma-\frac{1}{2}}\left\vert \frac{m_{\rho}}{\mr{r}^{2}}\right\vert^{2}\\\lesssim&
        k^{2}\int_{\Gamma}(\partial_{z}r_{\rho})^{2}+k^{2}\int_{\Gamma}(\partial_{z}\Psi_{\rho})^{2}+k^{2}\iint_{\mathcal{R}(s_{0})}(-z)^{2\gamma-\frac{1}{2}}\left(\partial_{z}^{2}r_{\rho}\right)^{2}+k^{2}\iint_{\mathcal{R}(s_{0})}(-z)^{2\gamma+\frac{1}{2}}\left(\partial_{z}^{2}\Psi_{\rho}\right)^{2}.
    \end{aligned}
\end{equation}
    To estimate the term with $\partial_{z}\left(\frac{m_{\rho}}{\mr{r}^{2}}\right)$ on the right-hand side of \eqref{eq: crude version of psirho second order energy estimate}, using Proposition \ref{prop: second order estimate of mrho} and taking $\omega = \gamma+\frac{1}{2}$ , we have \begin{equation}
        \begin{aligned}
            &\left\vert\iint_{\mathcal{R}(s_{0})}(-z)^{2\gamma+1}G_{k}(z)\partial_{z}\left(\frac{m_{\rho}}{\mr{r}^{2}}\right)\partial_{z}^{2}\Psi_{\rho}\right\vert\\\lesssim& k\iint_{\mathcal{R}(s_{0})}(-z)^{2\gamma+\frac{1}{2}}\left(\partial_{z}^{2}\Psi_{\rho}\right)^{2}+k\iint_{\mathcal{R}(s_{0})}(-z)^{2\gamma-1}\left\vert\partial_{z}\left(\frac{m_{\rho}}{\mr{r}^{2}}\right)\right\vert^{2}
            \\\lesssim&
            k^{5}\int_{\Gamma}(\partial_{z}r_{\rho})^{2}+k^{3}\int_{\Gamma}(\partial_{z}\Psi_{\rho})^{2}+k^{5}\iint_{\mathcal{R}(s_{0})}(-z)^{2\gamma-\frac{1}{2}}\left(\partial_{z}^{2}r_{\rho}\right)^{2}+k\iint_{\mathcal{R}(s_{0})}(-z)^{2\gamma+\frac{1}{2}}\left(\partial_{z}^{2}\Psi_{\rho}\right)^{2}.
        \end{aligned}
    \end{equation}
    To estimate the term with $\frac{m_{\rho}}{\mr{r}^{2}}$, by \eqref{eq: first-order energy estimate using second order terms}, we have \begin{equation}
    \begin{aligned}
      \left\vert\iint_{\mathcal{R}(s_{0})}(-z)^{2\gamma}F_{k}\frac{m_{\rho}}{\mr{r}^{2}}\partial_{z}^{2}\Psi_{\rho}\right\vert\lesssim&k\iint_{\mathcal{R}(s_{0})}(-z)^{2\gamma+\frac{1}{2}}\left(\partial_{z}^{2}\Psi_{\rho}\right)^{2}+k\iint_{\mathcal{R}(s_{0})}(-z)^{2\gamma-\frac{1}{2}}\left(\frac{m_{\rho}}{\mr{r}^{2}}\right)^{2}\\\lesssim&k^{5}\int_{\Gamma}(\partial_{z}r_{\rho})^{2}+k^{3}\int_{\Gamma}(\partial_{z}\Psi_{\rho})^{2}+k\iint_{\mathcal{R}(s_{0})}(-z)^{2\gamma+\frac{1}{2}}\left(\partial_{z}^{2}\Psi_{\rho}\right)^{2}\\&+k^{5}\iint_{\mathcal{R}(s_{0})}(-z)^{2\gamma-\frac{1}{2}}\left(\partial_{z}^{2}r_{\rho}\right)^{2}.
      \end{aligned}
    \end{equation}
    Putting everything together, for $k$ sufficiently small, we have 
    \begin{equation}
        \begin{aligned}
            &\int_{s = s_{0}}(-z)^{2\gamma}\left(\partial_{z}^{2}\Psi_{\rho}\right)^{2}+\int_{\Gamma}\left(\partial_{z}^{2}\Psi_{\rho}\right)^{2}\\&+\iint_{\mathcal{R}(s_{0})}\left[\left(\rho+q_{k}\left(\frac{1}{2}-\gamma\right)\right)(-z)^{2\gamma}+\frac{1}{4}q_{k}(-z)^{2\gamma+\frac{1}{2}}\right]\left(\partial_{z}^{2}\Psi_{\rho}\right)^{2}\\\lesssim&\int_{s = 0}(-z)^{2\gamma}\left(\partial_{z}^{2}\Psi_{\rho}\right)^{2}+k^{2}\int_{\Gamma}(\partial_{z}r_{\rho})^{2}+k^{2}\int_{\Gamma}(\partial_{z}\Psi_{\rho})^{2}+k\iint_{\mathcal{R}(s_{0})}(-z)^{2\gamma-\frac{1}{2}}\left(\partial_{z}^{2}r_{\rho}\right)^{2}.
        \end{aligned}
    \end{equation}
    This concludes the proof.
\end{proof}

\subsection{Second order energy estimates for $r_{\rho}$}
\label{sec: second-order energy estimate for rrho}
 In this section, we establish the second-order energy estimate for $r_{\rho}$. Due to the presence of $(-z)^{2\gamma-\frac{1}{2}}(\partial_{z}^{2}r_{\rho})^{2}$ in the second-order energy estimate \eqref{eq: second-order energy estimate for psirho} for $\Psi_{\rho}$, we shall prove the following weighted energy estimate for $r_{\rho}$:
\begin{proposition}
\label{prop: second-order energy estimate for rrho}
There exist $k$ sufficiently small, a constant $C_{1}>0$ depending on the background spacetime $(\mr{r},\mr{\phi},\mr{m})$, and a real number $\frac{3}{2}-\alpha<\gamma<\frac{1}{2}$, such that the following estimate holds in $\mathcal{R}(s_{0})$:\begin{equation}
    \begin{aligned}
        &\int_{s = s_{0}}(-z)^{2\gamma-\frac{1}{2}}\left(\partial_{z}^{2}r_{\rho}\right)^{2}+\int_{\Gamma}(\partial_{z}^{2}r_{\rho})+\left(\rho+q_{k}\left(\frac{3}{4}-\gamma\right)-C_{1}k\right)\iint_{\mathcal{R}(s_{0})}(-z)^{2\gamma-\frac{1}{2}}\left(\partial_{z}^{2}r_{\rho}\right)^{2}\\\leq&C_{1}\left(\int_{s = 0}(-z)^{2\gamma-\frac{1}{2}}(\partial_{z}r_{\rho})^{2}+k\int_{\Gamma}(\partial_{z}r_{\rho})^{2}+k\int_{\Gamma}(\partial_{z}\Psi_{\rho})^{2}+k\iint_{\mathcal{R}(s_{0})}(-z)^{2\gamma+\frac{1}{2}}(\partial_{z}^{2}\Psi_{\rho})^{2}\right).
    \end{aligned}
\label{second order energy estimate for the system}
\end{equation}
\end{proposition}
\begin{proof}
Multiplying \eqref{eq:second order wave equation for r rho} by $(-z)^{2\gamma-1}\partial_{z}^{2}r_{\rho}$ and integrating over $\mathcal{R}(s_{0})$, there exists a constant $C_{1}$, such that we have\begin{equation}
\begin{aligned}
&\frac{1}{2}\int_{s = s_{0}}(-z)^{2\gamma-\frac{1}{2}}\left(\partial_{z}^{2}r_{\rho}\right)^{2}+\frac{1}{2}q_{k}\int_{\Gamma}\left(\partial_{z}^{2}r_{\rho}\right)^{2}+\left(\rho+q_{k}\left(\frac{3}{4}-\gamma\right)-C_{1}k^{2}\right)\iint_{\mathcal{R}(s_{0})}(-z)^{2\gamma-\frac{1}{2}}\left(\partial_{z}^{2}r_{\rho}\right)^{2}\\\lesssim&
\frac{1}{2}\int_{s = 0}(-z)^{2\gamma-\frac{1}{2}}\left(\partial_{z}^{2}r_{\rho}\right)^{2}dz+\iint_{\mathcal{R}(s_{0})}(-z)^{2\gamma-\frac{1}{2}}F_{k^{2}}(z)\partial_{z}r_{\rho}\partial_{z}^{2}r_{\rho}+(-z)^{2\gamma-\frac{1}{2}}F_{k^{2}}(z)\partial_{s}r_{\rho}\partial_{z}^{2}r_{\rho}\\&+\iint_{\mathcal{R}(s_{0})}(-z)^{2\gamma-\frac{1}{2}}F_{k^{2}}(z)r_{\rho}\partial_{z}^{2}r_{\rho} dsdz+\iint_{\mathcal{R}(s_{0})}(-z)^{2\gamma-\frac{1}{2}}K(z)\partial_{z}\left(\frac{m_{\rho}}{\mr{r}^{2}}\right)\partial_{z}^{2}r_{\rho}\\&+\iint_{\mathcal{R}(s_{0})}(-z)^{2\gamma-\frac{1}{2}}\left(K(z)+F_{k^{2}}(z)\right)\frac{m_{\rho}}{\mr{r}^{2}}\partial_{z}^{2}r_{\rho}.
\end{aligned}
\label{rough version 1 of second order r rho estimate}
\end{equation}
To control the terms with $r_{\rho}$ and its derivatives on the right-hand side of \eqref{rough version 1 of second order r rho estimate}, arguing similarly to the proof of Proposition \ref{prop: second-order energy estimate for psirho} and using \eqref{standard hardy} repeatedly, we have \begin{align}
    \left\vert\iint_{\mathcal{R}(s_{0})}(-z)^{2\gamma-\frac{1}{2}}F_{k^{2}}\partial_{z}r_{\rho}\partial_{z}^{2}r_{\rho}\right\vert\lesssim& k^{2}\int_{\Gamma}(\partial_{z}r_{\rho})^{2}+k^{2}\iint_{\mathcal{R}(s_{0})}(-z)^{2\gamma-\frac{1}{2}}\left(\partial_{z}^{2}r_{\rho}\right)^{2},\\
    \left\vert\iint_{\mathcal{R}(s_{0})}(-z)^{2\gamma-\frac{1}{2}}F_{k^{2}}r_{\rho}\partial_{z}^{2}r_{\rho}\right\vert\lesssim& k^{2}\int_{\Gamma}(\partial_{z}r_{\rho})^{2}+k^{2}\iint_{\mathcal{R}(s_{0})}(-z)^{2\gamma-\frac{1}{2}}\left(\partial_{z}^{2}r_{\rho}\right)^{2},\\
    \left\vert\iint_{\mathcal{R}(s_{0})}(-z)^{2\gamma-\frac{1}{2}}F_{k^{2}}\partial_{s}r_{\rho}\partial_{z}^{2}r_{\rho}\right\vert\lesssim& k^{2}\int_{\Gamma}(\partial_{z}r_{\rho})^{2}+k^{2}\int_{\Gamma}(\partial_{z}\Psi_{\rho})^{2}+k^{2}\iint_{\mathcal{R}(s_{0})}(-z)^{2\gamma-\frac{1}{2}}\left(\partial_{z}^{2}r_{\rho}\right)^{2}\nonumber\\&+k^{2}\iint_{\mathcal{R}(s_{0})}(-z)^{2\gamma+\frac{1}{2}}\left(\partial_{z}^{2}\Psi_{\rho}\right)^{2}.
\end{align} 
To control the terms with $\frac{m_{\rho}}{\mr{r}^{2}}$ and its derivatives on the right-hand side of \eqref{rough version 1 of second order r rho estimate}, using Propositions \eqref{prop: second order estimate of mrho} and \eqref{prop: true second-order energy for mrho}, we have \begin{align}
    \left\vert\iint_{\mathcal{R}(s_{0})}(-z)^{2\gamma-\frac{1}{2}}K\partial_{z}\left(\frac{m_{\rho}}{\mr{r}^{2}}\right)\partial_{z}^{2}r_{\rho}\right\vert\lesssim& k^{3}\int_{\Gamma}(\partial_{z}r_{\rho})^{2}+k\int_{\Gamma}(\partial_{z}\Psi_{\rho})^{2}+k\iint_{\mathcal{R}(s_{0})}(-z)^{2\gamma-\frac{1}{2}}\left(\partial_{z}^{2}r_{\rho}\right)^{2}\nonumber\\&+k\iint_{\mathcal{R}(s_{0})}(-z)^{2\gamma+\frac{1}{2}}\left(\partial_{z}^{2}\Psi_{\rho}\right)^{2},\\\left\vert\iint_{\mathcal{R}(s_{0})}(-z)^{2\gamma-\frac{1}{2}}\left(K+F_{k^{2}}\right)\frac{m_{\rho}}{\mr{r}^{2}}\partial_{z}^{2}r_{\rho}\right\vert\lesssim&k^{3}\int_{\Gamma}(\partial_{z}r_{\rho})^{2}+k\int_{\Gamma}(\partial_{z}\Psi_{\rho})^{2}+k\iint_{\mathcal{R}(s_{0})}(-z)^{2\gamma-\frac{1}{2}}(\partial_{z}^{2}r_{\rho})^{2}\\&+k\iint_{\mathcal{R}(s_{0})}(-z)^{2\gamma+\frac{1}{2}}\left(\partial_{z}^{2}\Psi_{\rho}\right)^{2}.
\end{align}
Putting everything together, there exists a constant $C_{1}$, such that we have \begin{equation}
    \begin{aligned}
        &\int_{s = s_{0}}(-z)^{2\gamma-\frac{1}{2}}\left(\partial_{z}^{2}r_{\rho}\right)^{2}+\int_{\Gamma}(\partial_{z}^{2}r_{\rho})+\left(\rho+q_{k}\left(\frac{3}{4}-\gamma\right)-C_{1}k\right)\iint_{\mathcal{R}(s_{0})}(-z)^{2\gamma-\frac{1}{2}}\left(\partial_{z}^{2}r_{\rho}\right)^{2}\\\leq&C_{1}\left(\int_{s = 0}(-z)^{2\gamma-\frac{1}{2}}(\partial_{z}r_{\rho})^{2}+k\int_{\Gamma}(\partial_{z}r_{\rho})^{2}+k\int_{\Gamma}(\partial_{z}\Psi_{\rho})^{2}+k\iint_{\mathcal{R}(s_{0})}(-z)^{2\gamma+\frac{1}{2}}(\partial_{z}^{2}\Psi_{\rho})^{2}\right).
    \end{aligned}
\end{equation}
This concludes the proof.
\end{proof}
Combining Propositions \ref{prop: second-order energy estimate for psirho} and \ref{prop: second-order energy estimate for rrho}, we can immediately get the following energy estimate. \begin{proposition}
\label{prop: total second order energy estimate}
    There exist $k$ sufficiently small, a constant $C_{1}>0$ depending on the background spacetime $(\mr{r},\mr{\phi},\mr{m})$, and a real number $\frac{3}{2}-\alpha<\gamma<\frac{1}{2}$, such that the following second-order energy estimate holds in $\mathcal{R}(s_{0})$: \begin{equation}
        \begin{aligned}
            &\int_{s = s_{0}}(-z)^{2\gamma-\frac{1}{2}}\left(\partial_{z}^{2}r_{\rho}\right)^{2}+(-z)^{2\gamma}\left(\partial_{z}^{2}\Psi_{\rho}\right)^{2}+\int_{\Gamma}(\partial_{z}^{2}r_{\rho})^{2}+\int_{\Gamma}(\partial_{z}^{2}\Psi_{\rho})^{2}\\&+\iint_{\mathcal{R}(s_{0})}\left[\rho+q_{k}\left(\frac{5}{8}-\gamma\right)\right](-z)^{2\gamma-\frac{1}{2}}(\partial_{z}^{2}r_{\rho})^{2}+\left[\rho+q_{k}\left(\frac{1}{2}-\gamma\right)\right](-z)^{2\gamma}(\partial_{z}^{2}\Psi_{\rho})^{2}\\\lesssim&C_{1}\left(\int_{s = 0}(-z)^{2\gamma-\frac{1}{2}}(\partial_{z}^{2}r_{\rho})^{2}+(-z)^{2\gamma}(\partial_{z}^{2}\Psi_{\rho})^{2}+k\int_{\Gamma}(\partial_{z}r_{\rho})^{2}+(\partial_{z}\Psi_{\rho})^{2}\right).
        \end{aligned}
    \end{equation}
\end{proposition}
\subsection{Proof of Theorem \ref{thm:a priori energy estimate}}
In this section, we conclude the proof of Theorem \ref{thm:a priori energy estimate}.
\begin{proof}
To show the boundedness of $\partial_{z}r_{\rho}$ and $\partial_{z}\Psi_{\rho}$, by Sobolev embedding, it suffices to prove the $L^{2}$-boundedness of $\partial_{z}r_{\rho}$, $\partial_{z}\Psi_{\rho}$, $\partial_{z}^{2}r_{\rho}$, and $\partial_{z}^{2}\Psi_{\rho}$. First, we refine the inequality \eqref{eq: fundamental theorem of calculus}. For a function $f\in C^{1}([-1,0])$, using integration by parts, for any $a\in(0,1)$, we have 
    \begin{align*}
        \int_{-1}^{0}f^{2}(z)dz =& f^{2}(-1)+2\int_{-1}^{0}(-z)f(z)f^{\prime}(z)dz\\\leq&f^{2}(-1)+a\int_{-1}^{0}f^{2}(z)dz+\frac{1}{a}\int_{-1}^{0}(-z)^{2}\left(\frac{df}{dz}\right)^{2}.
    \end{align*}
Hence, we have \begin{equation}
    \int_{-1}^{0}f^{2}(z)dz\leq \frac{1}{1-a}f^{2}(-1)+\frac{1}{a(1-a)}\int_{-1}^{0}(-z)^{2}\left(\frac{df}{dz}\right)^{2}dz.
\end{equation}
Taking $\gamma = \frac{3}{2}-\alpha+\delta$ for $0<\delta\ll 1$ and using the second order energy estimate \eqref{second order energy estimate for the system} to control the bulk term on the right-hand side of the first-order energy estimate \eqref{first order energy estimate for the system}, we have\begin{equation}
\begin{aligned}
&\left(\frac{1}{2}-C_{0}k^{2}\right)\int_{s = s_{0}}(\partial_{z}r_{\rho})^{2}+(\partial_{z}\Psi_{\rho})^{2}+\frac{1}{2}q_{k}\int_{\Gamma}(\partial_{z}r_{\rho})^{2}+(\partial_{z}\Psi_{\rho})^{2}\\\leq&\left(\frac{1}{2}+C_{0}k^{2}\right)\int_{s =0}(\partial_{z}r_{\rho})^{2}+(\partial_{z}\Psi_{\rho})^{2}+\frac{\left(\frac{1}{2}q_{k}-\rho+C_{0}k\right)}{1-a}\int_{\Gamma}(\partial_{z}r_{\rho})^{2}+(\partial_{z}\Psi_{\rho})^{2}\\&+\left(\frac{\frac{1}{2}q_{k}-\rho+C_{0}k}{a(1-a)}+C_{0}k\right)\iint_{\mathcal{R}(s_{0})}(-z)^{2\gamma-\frac{1}{2}}(\partial_{z}^{2}r_{\rho})^{2}+(-z)^{2\gamma}\left(\partial_{z}^{2}\Psi_{\rho}\right)^{2}\\\leq&\left(\frac{1}{2}+C_{0}k^{2}\right)\int_{s = 0}(\partial_{z}r_{\rho})^{2}+(\partial_{z}\Psi_{\rho})^{2}+\frac{\frac{1}{2}q_{k}-\rho+C_{0}k}{1-a}\int_{\Gamma}(\partial_{z}r_{\rho})^{2}+(\partial_{z}\Psi_{\rho})^{2}\\&+\left(\frac{\frac{1}{2}q_{k}-\rho+C_{0}k}{a(1-a)}+C_{0}k\right)\frac{C_{1}k}{\rho+q_{k}\left(\frac{1}{2}-\gamma\right)}\int_{\Gamma}(\partial_{z}r_{\rho})^{2}+(\partial_{z}\Psi_{\rho})^{2}\\&+\left(\frac{\frac{1}{2}q_{k}-\rho+C_{0}k}{a(1-a)}+C_{0}k\right)\frac{C_{1}}{\rho+q_{k}\left(\frac{1}{2}-\gamma\right)}\int_{s = 0}(-z)^{2\gamma-\frac{1}{2}}(\partial_{z}^{2}r_{\rho})^{2}+(-z)^{2\gamma}(\partial_{z}^{2}\Psi_{\rho})^{2}.
\end{aligned}
\end{equation}
We can choose $\delta = k^{\frac{1}{100}}$, $\rho = k^{\frac{1}{8}}$, and $a = k^{\frac{1}{4}}$. Then, for $k$ sufficiently small independent of the choice of $\alpha$, we have 
\begin{equation}
    \frac{\frac{1}{2}q_{k}-\rho+C_{0}k}{1-a}+\left(\frac{\frac{1}{2}q_{k}-\rho+C_{0}k}{a(1-a)}+C_{0}k\right)\frac{C_{1}k}{\rho+q_{k}\left(\frac{1}{2}-\gamma\right)}<\frac{1}{2}q_{k}.\label{eq: the reason why we could not get sharp decay}
\end{equation}
Hence, we have \begin{equation*}
    \left\Vert\partial_{z}r_{\rho}(s,\cdot)\right\Vert_{L_{z}^{2}},\ \left\Vert\partial_{z}\Psi_{\rho}(s,\cdot)\right\Vert_{L^{2}_{z}},\ \left\Vert(-z)^{\gamma-\frac{1}{4}}\partial_{z}^{2}r_{\rho}(s,\cdot)\right\Vert_{L^{2}_{z}},\ \left\Vert(-z)^{\gamma}\partial_{z}^{2}\Psi_{\rho}(s,\cdot)\right\Vert_{L^{2}_{z}}
\end{equation*}
are uniformly bounded for any $s\geq 0$.

Hence, in the near-axis region $-1\leq z\leq-\frac{1}{2}$, by Sobolev embedding, we have\begin{equation}
\begin{aligned}
&\sup_{z\in[-1,-\frac{1}{2}]}(\partial_{z}r_{\rho})^{2}(s_{0},z)+(\partial_{z}\Psi_{\rho})^{2}(s_{0},z)\\\leq&\int_{-1}^{-\frac{1}{2}}(\partial_{z}r_{\rho})^{2}(s_{0},z)+(\partial_{z}\Psi_{\rho})^{2}(s_{0},z)+(\partial_{z}^{2}r_{\rho})^{2}(s_{0},z)+(-z)^{2\gamma}(\partial_{z}^{2}\Psi_{\rho})^{2}(s_{0},z)dz\\\lesssim&_{k} C,
\end{aligned}
\end{equation}
where $C$ depends on the initial data. For the near-horizon region $-\frac{1}{2}\leq z\leq0$, we have\begin{equation}
\begin{aligned}
&\sup_{-\frac{1}{2}\leq z\leq0} \vert\partial_{z}r_{\rho}\vert(s,z)+\vert\partial_{z}\Psi_{\rho}\vert(s,z)\\\leq&
\vert\partial_{z}r_{\rho}\vert\left(s,-\frac{1}{2}\right)+\vert\partial_{z}\Psi_{\rho}\vert\left(s,-\frac{1}{2}\right)+\int_{-\frac{1}{2}}^{0}\vert\partial_{z}r_{\rho}\vert(s,z)+\vert\partial_{z}\Psi_{\rho}\vert(s,z)dz\\\lesssim&
\vert\partial_{z}r_{\rho}\vert\left(s,-\frac{1}{2}\right)+\vert\partial_{z}\Psi_{\rho}\vert\left(s,-\frac{1}{2}\right)+\left(\int_{-\frac{1}{2}}^{0}(\partial_{z}^{2}r_{\rho})^{2}dz\right)^{\frac{1}{2}}\\&+\left(\int_{-\frac{1}{2}}^{0}(-z)^{-2\gamma}dz\right)^{\frac{1}{2}}\left(\int_{-\frac{1}{2}}^{0}(-z)^{2\gamma}(\partial_{z}^{2}\Psi_{\rho})^{2}dz\right)^{\frac{1}{2}}\\\lesssim&_{k} C,
\end{aligned}
\end{equation}
where $C$ depends on the initial data. Hence we have\begin{equation}
\Vert\partial_{z}r_{\rho}(s,\cdot)\Vert_{L^{\infty}([-1,0])}+\Vert\partial_{z}\Psi_{\rho}(s,\cdot)\Vert_{L^{\infty}([-1,0])}\lesssim_{k} C.
\end{equation}
Using the boundary condition $r_{\rho}(s,-1) = 0$ and $\Psi_{\rho}(s,-1) = 0$, and integrating from the axis $\Gamma$, we have\begin{equation}
\Vert r_{\rho}(s,\cdot)\Vert_{L^{\infty}([-1,0])}+\Vert\Psi_{\rho}(s,\cdot)\Vert_{L^{\infty}([-1,0])}\lesssim_{k} C.
\end{equation}

Next, we bound $\frac{m_{\rho}}{\mr{r}^{2}}$. By Sobolev embedding, Proposition \ref{prop: second order estimate of mrho}, and Proposition \ref{prop: true second-order energy for mrho}, we have\begin{equation}
\begin{aligned}
\sup_{z\in[-1,0]}\left(\frac{m_{\rho}}{\mr{r}^{2}}\right)^{2}(s_{0},z)\leq&\int_{-1}^{0}\left(\frac{m_{\rho}}{\mr{r}^{2}}\right)^{2}(s_{0},z)dz+\int_{-1}^{0}\left(\partial_{z}\left(\frac{m_{\rho}}{\mr{r}^{2}}\right)\right)^{2}(s_{0},z)dz\\\lesssim&k^{4}\left(\int_{\Gamma}(\partial_{z}r_{\rho})^{2}+(\partial_{z}^{2}r_{\rho})^{2}ds+\int_{s = s_{0}}(\partial_{z}^{2}r_{\rho})^{2}dz\right)\\&+k^{2}\left(\int_{\Gamma}(\partial_{z}\Psi_{\rho})^{2}+(\partial_{z}^{2}\Psi_{\rho})^{2}ds+\int_{s = s_{0}}(-z)^{2\gamma}(\partial_{z}^{2}\Psi_{\rho})^{2}dz\right)\\\lesssim&_{k} C.
\end{aligned}
\end{equation}
To estimate $\partial_{s}r_{\rho}$, we use the equation \eqref{eq:wave equation for r rho}:\begin{equation} \partial_{z}\partial_{s}r_{\rho}+G_{k^{2}}(z)\partial_{s}r_{\rho} = q_{k}(-z)\partial_{z}^{2}r_{\rho}+(-\rho+G_{k^{2}}(z))\partial_{z}r_{\rho}+G_{k^{2}}(z)r_{\rho}+K(z)\frac{m_{\rho}}{\mr{r}^2}.
\end{equation}
Using the boundary condition $\partial_{s}r_{\rho}(s,-1) = 0$ and integrating the above equation from the boundary, we have\begin{equation}
\Vert\partial_{s}r_{\rho}(s,\cdot)\Vert_{L^{\infty}([-1,0])}\lesssim_{k} C.
\end{equation}
Similarly, we have\begin{equation}
\Vert \partial_{s}\Psi_{\rho}(s,\cdot)\Vert_{L^{\infty}([-1,0])}\lesssim_{k} C.
\end{equation}
It remains to estimate $\partial_{z}m_{\rho}$ and $\partial_{s}m_{\rho}$. Using the equation \eqref{eq:z-transport eq for mrho}, we have\begin{equation}
\partial_{z}m_{\rho} = -\left(\frac{\mr{r}}{\partial_{z}\mr{r}}(\partial_{z}\mr{\phi})^{2}\right)m_{\rho}+F_{k}(z)r_{\rho}+F_{k}(z)\partial_{z}r_{\rho}+F_{k}(z)\partial_{z}\Psi_{\rho}+F_{k}(z)\Psi_{\rho}.
\end{equation}
Then we have\begin{equation}
\Vert \partial_{z}m_{\rho}(s,\cdot)\Vert_{L^{\infty}([-1,0])}\lesssim_{k}C
\end{equation}
Using the equation \eqref{eq:u transport eq for m under self-similar coordinates}, we have\begin{equation}
\Vert \partial_{s}m_{\rho}(s,\cdot)\Vert_{L^{\infty}([-1,0])}\lesssim_{k} C.
\end{equation}
Lastly, we use the relation \begin{equation}
\partial_{u} = e^{s}\left(\partial_{s}+q_{k}z\partial_{z}\right),\quad \partial_{v} = e^{q_{k}s}\partial_{z}
\end{equation}
to convert the bounds under self-similar coordinates to the bounds under double-null coordinates. Then we can conclude the proof of Theorem \ref{thm:a priori energy estimate}.
\end{proof}
\subsection{Some further remarks}
As one can see from the previous sections, by our carefully designed second-order energy estimates, we can have positive bulk terms for any $\rho>-q_{k}(\alpha-1)$. In particular, when $\alpha>p_{k}$, this allows the value of $\rho$ to be smaller than the self-similar value $\rho = -k^{2}$. However, to prove the decay estimate, due to the presence of the boundary terms $\int_{\Gamma}(\partial_{z}r_{\rho})^{2}$ and $\int_{\Gamma}(\partial_{z}\Psi_{\rho})^{2}$, the inequality \eqref{eq: the reason why we could not get sharp decay} essentially requires that the value of $\rho\geq0$. Nonetheless, by the observation we made in Remark \ref{rmk: observation why we could not get sharp decay from energy} and the second-order energy estimates, we have the following proposition:
\begin{proposition}
\label{prop: key to close the linearized result}
    Fix $k\neq0$ to be sufficiently small and assume $\alpha\in(p_{k},\frac{3}{2})$. If there exist constants $c_{\infty}$ and $C$, and a real number $\alpha^{\prime}\in(p_{k},\alpha)$, such that the following estimates hold on the axis:
    \begin{equation}
    \begin{aligned}
        \sum_{0\leq i+j\leq 2}\left\vert\partial_{u}^{i}\partial_{v}^{j}r_{p}\right\vert\Bigl|_{\Gamma}\leq& C(-u)^{\alpha^{\prime}q_{k}-i-q_{k}j},\\ \sum_{0\leq i+j\leq2}\left\vert\partial_{u}^{i}\partial_{v}^{j}\left(\Psi_{p}-c_{\infty}r_{k}\right)\right\vert\Bigl|_{\Gamma}\leq& C(-u)^{\alpha^{\prime}q_{k}-i-q_{k}j},\\\sum_{0\leq i+j\leq 1}\left\vert\partial_{u}^{i}\partial_{v}^{j}m_{p}\right\vert\Bigl|_{\Gamma}\leq& C(-u)^{\alpha^{\prime}q_{k}-i-q_{k}j},
    \end{aligned}
    \label{eq: good assumption on the axis}
    \end{equation}
    then \eqref{eq: good assumption on the axis} also holds in the whole region of the interior $\mathcal{Q}^{(in)}$ with $\alpha^{\prime}$ replaced by any $\widetilde{\alpha}<\alpha^{\prime}$.
\end{proposition}
\begin{proof}
By the translation invariance of the linearized Einstein-scalar field equations, we have that $$\left(r_{p},\widetilde{\Psi}_{p},m_{p}\right): = (r_{p},\Psi_{p}-c_{\infty}r_{k},m_{p})$$ also solves the linearized equation \eqref{eq:linearized r equation}-\eqref{eq:linearized dvm equation}. 

Under the self-similar coordinates $(s,z)$, the estimates \eqref{eq: good assumption on the axis} can be written as \begin{align*}
    \sum_{0\leq i+j\leq 1}\left\vert\partial_{s}^{i}\partial_{z}^{j}r_{p}\right\vert\Bigl|_{\Gamma}\leq& Ce^{-\alpha^{\prime}q_{k}s},\\\sum_{0\leq i+j\leq 1}\left\vert\partial_{s}^{i}\partial_{z}^{j}\left(\Psi_{p}-c_{\infty}r_{k}\right)\right\vert\Bigl|_{\Gamma}\leq &Ce^{-\alpha^{\prime}q_{k}s},\\\sum_{0\leq i+j\leq 1}\left\vert\partial_{s}^{i}\partial_{z}^{j}m_{p}\right\vert\Bigl|_{\Gamma}\leq& Ce^{-\alpha^{\prime}q_{k}s}.
\end{align*}
Hence, we have the boundedness of $
    \int_{\Gamma}(\partial_{z}r_{\rho})^{2}$ and $ \int_{\Gamma}(\partial_{z}\widetilde{\Psi}_{\rho})^{2}$ for $\rho>(1-\alpha^{\prime})q_{k}$.

Applying Proposition \ref{prop: total second order energy estimate} to $(r_{\rho},\widetilde{\Psi}_{\rho},m_{\rho})$ and using the assumption \eqref{eq: good assumption on the axis} to control the boundary terms on the right-hand side for $\rho>-q_{k}(\alpha^{\prime}-1)$, we have the weighted $L^{2}$-boundedness of $\partial_{z}^{2}r_{\rho}$ and $\partial_{z}^{2}\widetilde{\Psi}_{\rho}$ for any $\rho>-q_{k}(\alpha-1)$. The proof can be concluded by the first-order energy estimate and the fundamental theorem of calculus.
\end{proof}
The remaining part of the linear analysis in this paper will be devoted to the proof of \eqref{eq: good assumption on the axis}.

\section{Scattering theory for the linearized system on the naked singularity background}
\label{sec: scattering theory for the linearized system}
In Section \ref{leading order expansion}, we shall take the Fourier--Laplace transform of $(r_{p},\Psi_{p},m_{p})$ under the self-similar:\begin{equation}
\hat{r}_{p}(\sigma,z):=\int_{0}^{\infty}e^{-\sigma s}r_{p}(s,z)ds,\quad \hat{\Psi}_{p}(\sigma,z): = \int_{0}^{\infty}e^{-\sigma s}\Psi_{p}(s,z)ds,\quad \hat{m}_{p}(\sigma,z): = \int_{0}^{\infty}e^{-\sigma s}m_{p}(s,z)ds.
\label{eq: taking the fourier-laplace transform}
\end{equation}
This motivates the study of the following equations:
\begin{align}
&q_{k}z\partial_{z}^{2}\hat r_{p}+\left(q_{k}+\sigma-\frac{\mu_{k}}{1-\mu_{k}}\frac{\partial_{s}r_{k}+2q_{k}z\partial_{z}r_{k}}{r_{k}}\right)\partial_{z}\hat{r}_{p}+\left(-\frac{\mu_{k}}{1-\mu_{k}}\frac{\partial_{z}r_{k}}{r_{k}}\sigma+\frac{\mu_{k}(2-\mu_{k})}{(1-\mu_{k})^{2}}\frac{(\partial_{s}r_{k}+q_{k}z\partial_{z}r_{k})\partial_{z}r_{k}}{r_{k}^{2}}\right)\hat{r}_{p}\nonumber\\=&\frac{2}{(1-\mu_{k})^{2}}\frac{(\partial_{s}r_{k}+q_{k}z\partial_{z}r_{k})\partial_{z}r_{k}}{r_{k}^{2}}\hat{m}_{p},
\label{eq:scattering equation for rp}\allowdisplaybreaks\\[2em]
&q_{k}z\partial_{z}^{2}\hat{\Psi}_{p}+(q_{k}+\sigma)\partial_{z}\hat{\Psi}_{p}-\frac{\mu_{k}}{1-\mu_{k}}\frac{(\partial_{s}r_{k}+q_{k}z\partial_{z}r_{k})\partial_{z}r_{k}}{r_{k}^{2}}\hat{\Psi}_{p}-\left(\frac{1}{r_{k}^{2}}\frac{\mu_{k}}{1-\mu_{k}}\Psi_{k}(\partial_{s}r_{k}+q_{k}z\partial_{z}r_{k})+k\right)\partial_{z}\hat{r}_{p}\nonumber\\=&-\left(\frac{\mu_{k}(3-2\mu_{k})}{r_{k}^{3}(1-\mu_{k})^{2}}(\partial_{s}r_{k}+q_{k}z\partial_{z}r_{k})\partial_{z}r_{k}-\frac{\partial_{z}r_{k}}{r_{k}^{2}}\frac{\mu_{k}}{1-\mu_{k}}\sigma\right)\Psi_{k}\hat{r}_{p}+\frac{2\Psi_{k}(\partial_{s}r_{k}+q_{k}z\partial_{z}r_{k})\partial_{z}r_{k}}{r_{k}^{3}(1-\mu_{k})^{2}}\hat{m}_{p}
\label{eq:scattering equation for psip}\allowdisplaybreaks\\[2em]
&\sigma \hat{m}_{p}+\left(\frac{r_{k}}{\partial_{s}r_{k}+q_{k}z\partial_{z}r_{k}}\left((\partial_{s}+q_{k}z\partial_{z})\phi_{k}\right)^{2}-q_{k}z\frac{r_{k}}{\partial_{z}r_{k}}(\partial_{z}\phi_{k})^{2}\right)\hat{m}_{p}\nonumber\\=&(1-\mu_{k})\frac{r_{k}}{(\partial_{s}+q_{k}z\partial_{z})r_{k}}(\partial_{s}+q_{k}z\partial_{z})\phi_{k}(\sigma+q_{k}z\partial_{z})\hat\Psi_{p}-q_{k}z(1-\mu_{k})\frac{r_{k}}{\partial_{z}r_{k}}\partial_{z}\phi_{k}\partial_{z}\hat\Psi_{p}\nonumber\\&-(1-\mu_{k})(\partial_{s}+q_{k}z\partial_{z})\phi_{k}\hat\Psi_{p}+(1-\mu_{k})q_{k}z\partial_{z}\phi_{k}\hat\Psi_{p}\nonumber\\&-(1-\mu_{k})\left(\frac{1}{2}r_{k}^{2}\left(\frac{(\partial_{s}+q_{k}z\partial_{z})\phi_{k}}{(\partial_{s}+q_{k}z\partial_{z})r_{k}}\right)^{2}+\frac{r_{k}(\partial_{s}+q_{k}z\partial_{z})\phi_{k}}{(\partial_{s}+q_{k}z\partial_{z})r_{k}}\mr{\phi}\right)(\sigma+q_{k}z\partial_{z})\hat r_{p}\nonumber\\&+q_{k}z\left(\frac{1}{2}\left(\frac{r_{k}}{\partial_{z}r_{k}}\right)^{2}(\partial_{z}\phi_{k})^{2}+\frac{r_{k}}{\partial_{z}r_{k}}\partial_{z}\phi_{k}\mr{\phi}\right)(1-\mu_{k})\partial_{z}\hat{r}_{p}\nonumber\\&+\left(\frac{r_{k}\mu_{k}}{2}\frac{((\partial_{s}+q_{k}z\partial_{z})\phi_{k})^{2}}{(\partial_{s}+q_{k}z\partial_{z})r_{k}}+(1-\mu_{k})\mr{\phi}+(1-\mu_{k})\frac{kr_{k}(\partial_{s}+q_{k}z\partial_{z})\phi_{k}}{(\partial_{s}+q_{k}z\partial_{z})r_{k}}\right)\hat{r}_{p}\nonumber\\&-q_{k}z\left(\frac{1}{2}\frac{r_{k}}{\partial_{z}r_{k}}(\partial_{z}\phi_{k})^{2}\mu_{k}+(1-\mu_{k})\partial_{z}\phi_{k}\mr{\phi}\right)\hat{r}_{p}.\label{scattering equation for mp}
\end{align}
\begin{remark}
The above equations can be obtained informally by taking the Fourier--Laplace transform \eqref{eq: taking the fourier-laplace transform} of equations \eqref{eq:wave eq for rp under self-similar coordinates}-\eqref{eq:v transport eq for m under self-similar coordinates} and applying integration by parts with the assumption that $(r_{p},\Psi_{p},m_{p})|_{s\leq0} = 0$. However, since the initial data $(r_{p},\Psi_{p},m_{p})|_{s = 0}$ is indeed non-trivial, in reality, we shall commute the equations \eqref{eq:wave eq for rp under self-similar coordinates}-\eqref{eq:v transport eq for m under self-similar coordinates} with a cut-off function supported away from $s = 0$ first. This procedure will generate some inhomogeneous terms on the right-hand side of \eqref{eq:scattering equation for rp}-\eqref{scattering equation for mp}. The corresponding equations will be studied in Section \ref{leading order expansion}. In this section, we will only focus on the linear homogeneous equations \eqref{eq:scattering equation for rp}-\eqref{scattering equation for mp}.
\end{remark}
 Since the last equation \eqref{scattering equation for mp} for $\hat{m}_{p}$ can be viewed as an algebraic equation, we can substitute $\hat{m}_{p}$ into the first two equations \eqref{eq:scattering equation for rp}-\eqref{eq:scattering equation for psip}. Schematically, the above linearized system for $(\hat{r}_{p},\hat{\Psi}_{p},\hat{m}_{p})$ can be written as \begin{equation}
\mathcal{L}(\hat{r}_{p},\hat{\Psi}_{p}) = 0,
\end{equation}
where $\mathcal{L}$ is a linear differential operator.

In the following subsections, we will establish the scattering theory for the linear operator $\mathcal{L}(\hat{r}_{p},\hat{\Psi}_{p})$.
\subsection{Preliminary estimates for $\mathcal{L}$}
In this section, we will focus on the estimates for the coefficients of the linear operator $\mathcal{L}(\hat{r}_{p},\hat{\Psi}_{p})$.
\begin{proposition}
Let:$$
V(z): =-\left(\frac{r_{k}}{\partial_{s}r_{k}+q_{k}z\partial_{z}r_{k}}\left((\partial_{s}+q_{k}z\partial_{z})\phi_{k}\right)^{2}-q_{k}z\frac{r_{k}}{\partial_{z}r_{k}}(\partial_{z}\phi_{k})^{2}\right).$$Then $V(z)>0$ for all $z\in[-1,0]$. Moreover, we have the following estimate for the size of $V(z)$:\begin{equation}
0\leq V(z)<O(k^{2}).
\end{equation}
\end{proposition}
\begin{proof}
Straightforward computation using Proposition \ref{prop: estimate on the exact k self similar spacetime}.
\end{proof}
Note that for $\sigma\in Im{V}([-1,0])$, the coefficient of $\hat{m}_{p}$ will be degenerate for some $z$. To avoid this degeneracy, we restrict $\sigma$ in the region $\Re\sigma<-\eta$, for some $\eta\in(0,\frac{1}{2})$ which will be chosen later. We define \begin{equation*}
    \mathbb{I}_{[-1-\eta,-\eta]}: = \{\sigma\in\mathbb{C}| -1-\eta\leq\Re{\sigma}\leq -\eta\}.
\end{equation*}
Recall that $\lambda_{0}:= \frac{d\mr{r}}{dz}(0)$. Then we can reduce the equations \eqref{eq:scattering equation for rp}-\eqref{scattering equation for mp} in a convenient schematic form in the following proposition:

\begin{proposition}
\label{prop: schematic form of equations}
The equations \eqref{eq:scattering equation for rp}-\eqref{scattering equation for mp} can be written as:\begin{align}
q_{k}z\frac{d^{2}\hat{r}_{p}}{dz^{2}}+(1+\sigma)\frac{d\hat{r}_{p}}{dz}+d_{1}\hat{r}_{p}+d_{2}\hat{\Psi}_{p}=&N_1(\sigma,z)\frac{d\hat{r}_{p}}{dz}+N_{2}(\sigma,z)\frac{d\hat{\Psi}_{p}}{dz}+h_{1}(\sigma,z)\hat{r}_{p}+h_{2}(\sigma,z)\hat{\Psi}_{p},\label{schematic form of rp}\\
q_{k}z\frac{d^{2}\hat\Psi_{p}}{dz^{2}}+k\frac{d\hat{r}_{p}}{dz}+(q_{k}+\sigma)\frac{d\hat{\Psi}_{p}}{dz}+d_{3}\hat{\Psi}_{p}=&N_{3}(\sigma,z)\frac{d\hat{r}_{p}}{dz}+N_{4}(\sigma,z)\frac{d\hat{\Psi}_{p}}{dz}+h_{3}(\sigma,z)\hat{r}_{p}+h_{4}(\sigma,z)\hat{\Psi}_{p}.
\label{schematic form of psip}
\end{align}
where $d_{i}$ are constants \begin{align}
    &\label{expression of d1}d_{1} = -\left(k^{2}(2+k^{2})+k^{2}\sigma\right)\lambda_{0}\frac{\sigma+1}{\sigma-k^{2}},
    \\& d_{2} =-\frac{2k\lambda_{0}(1+k^{2})}{\sigma-k^{2}}(\sigma+1),\label{expression of d2}\\&
    d_{3} = k^{2}\lambda_{0},\label{expression of d3}
\end{align}
and $N_{i}(\sigma,z)$ are complex-valued functions depending on $\sigma,k$ and the background spacetime, satisfying the following estimate:\begin{align}
    \vert N_{i}(\sigma,z)\vert\lesssim_{\eta} \frac{k}{\epsilon}(-z)^{1-\epsilon}.\label{estimate for Ni}
\end{align}
Here $h_{i}(\sigma,z)$ are holomorphic functions of $\sigma$ and $z$ for $\sigma\in\mathbb{I}_{[-1-\eta,-\eta]}$, satisfying:\begin{equation}
    \label{eq: key estimate on Ni and hi}
    \vert h_{1}(\sigma,z)\vert\lesssim_{\eta} \frac{k\vert\sigma\vert}{\epsilon}(-z)^{1-\epsilon},\quad \vert h_{2}(\sigma,z)\vert\lesssim_{\eta} \frac{k}{\epsilon}(-z)^{1-\epsilon},\quad \vert h_{3}(\sigma,z)\vert\lesssim_{\eta} \frac{k\vert\sigma\vert}{\epsilon}(-z)^{1-\epsilon},\quad \vert h_{4}(\sigma,z)\vert\lesssim_{\eta} \frac{k}{\epsilon}(-z)^{1-\epsilon}.
\end{equation}
\end{proposition}
\begin{proof}
The proof follows from straightforward computation using Proposition \ref{prop: estimate on the exact k self similar spacetime} and the null structure of the Einstein-scalar field equations.
\end{proof}
Motivated by the schematic form of the equations \eqref{schematic form of rp}-\eqref{schematic form of psip}, in order to keep track of the $k$-smallness of the solutions, we have the following definition:
\begin{definition}
Define $\mathcal{L}_{p}(\hat{r}_{p},\hat{\Psi}_{p})$ to be the differential operator\begin{equation}
\begin{aligned}
\mathcal{L}_{p}(\hat{r}_{p},\hat{\Psi}_{p}): =q_{k}z\frac{d^{2}}{dz^{2}}\begin{bmatrix}
    \hat{r}_{p}\\\hat{\Psi}_{p}
\end{bmatrix}
+\begin{bmatrix}
    1+\sigma&0\\k&q_{k}+\sigma
\end{bmatrix}\frac{d}{dz}\begin{bmatrix}
    \hat{r}_{p}\\\hat{\Psi}_{p}
\end{bmatrix}+\begin{bmatrix}
    d_{1}&d_{2}\\0&d_{3}
\end{bmatrix}
\begin{bmatrix}
    \hat{r}_{p}\\\hat{\Psi}_{p}
\end{bmatrix}
\end{aligned}\label{defintion of mathcal Lp}
\end{equation}
\end{definition}
We have the following relation between $\mathcal{L}_{p}$ and $\mathcal{L}$:
\begin{proposition}
\label{prop: relation between Lp and L}
The equation $\mathcal{L}(\hat{r}_{p},\hat{\Psi}_{p}) = 0$ is equivalent to the following form\begin{equation}
\label{eq: perturb+rhs form of scattering equations}
\begin{aligned}
\mathcal{L}_{p}(\hat{r}_{p},\hat{\Psi}_{p}) =\begin{bmatrix}
N_{1}(\sigma,z)&N_{2}(\sigma,z)\\
N_{3}(\sigma,z)&N_{4}(\sigma,z)
\end{bmatrix}\frac{d}{dz}
\begin{bmatrix}
    \hat{r}_{p}\\\hat{\Psi}_{p}
\end{bmatrix}
+\begin{bmatrix}
    h_{1}&h_{2}\\h_{3}&h_{4}
\end{bmatrix}
\begin{bmatrix}
    \hat{r}_{p}\\\hat{\Psi}_{p}
\end{bmatrix}.
\end{aligned}
\end{equation}

\end{proposition}
It is straightforward to check that $z =0$ is a regular singularity of the linear differential equations $\mathcal{L} = 0$. However, due to the low regularity nature of the coefficients, one cannot directly apply the standard theory of regular singularities. We construct explicit solutions to $\mathcal{L} = 0$ in the next several sections.
\subsection{Model construction of the outgoing solutions and ingoing solutions for a system of ODEs}
To illustrate the main idea and key structures in constructing the outgoing and ingoing solutions and to avoid tedious computation and notations, in this section, we consider the following system of ODEs\begin{align}
&q_{k}z\frac{d^{2}f}{dz^{2}}+(1+\sigma)\frac{df}{dz}+\alpha_{1}f+\alpha_{2}g = 0\label{eq:toy model equation1}\\
&q_{k}z\frac{d^{2}g}{dz^{2}}+k\frac{df}{dz}+(q_{k}+\sigma)\frac{dg}{dz}+\alpha_{3}g = 0,\label{eq:toy model equation2}
\end{align}
where each $\alpha_{i}$ is allowed to depend on $k,\sigma$. Note that the equations \eqref{eq:toy model equation1}-\eqref{eq:toy model equation2} can be reduced to $\mathcal{L}_{p}(f,g) = 0$ for some specific values of $\alpha_{i}$.\begin{proposition}
    The equations \eqref{eq:toy model equation1}-\eqref{eq:toy model equation2} are reduced to $\mathcal{L}_{p}(f,g) = 0$ if \begin{equation}
    \begin{aligned}
    \alpha_{i} = d_{i},\quad i = 1,2,3.
    \end{aligned}\label{eq: special structure of alphai}
\end{equation}
\end{proposition}

Dropping the terms with $\alpha_{i}$ in \eqref{eq:toy model equation1}-\eqref{eq:toy model equation2}, the system admits an explicit solution basis:\begin{equation}
    \left\{\begin{bmatrix}
        1\\0
    \end{bmatrix},\ 
    \begin{bmatrix}
        0\\1
    \end{bmatrix},\ 
    (-z)^{-\frac{k^{2}+\sigma}{q_{k}}}\begin{bmatrix}
        1\\\frac{1}{k}
    \end{bmatrix},\ 
    (-z)^{-\frac{\sigma}{q_{k}}}\begin{bmatrix}
        0\\1
    \end{bmatrix}\right\}.\label{eq: solution basis}
\end{equation}
Motivated by this, we have the following definition:\begin{definition}
    We call a solution $(f,g)$ to \eqref{eq:toy model equation1}-\eqref{eq:toy model equation2} with the leading order term when $z\rightarrow 0$ spanned by $(1,0)$ and $(0,1)$ an outgoing solution; $(f,g)$ with leading order term when $z\rightarrow 0$ spanned by $$((-z)^{-\frac{k^{2}+\sigma}{q_{k}}},\frac{1}{k}(-z)^{-\frac{k^{2}+\sigma}{q_{k}}}),\quad (0,(-z)^{-\frac{\sigma}{q_{k}}})$$ an ingoing solution.
\end{definition}

\begin{remark}
    Since $k$ will be taken to be sufficiently small in this paper, one might ask whether in the limit $k\rightarrow 0$, the asymptotic behavior of the solutions to the system \eqref{eq:toy model equation1}-\eqref{eq:toy model equation2} is governed by the solutions of \begin{equation}
        z\frac{d^{2}f}{dz^{2}}+(1+\sigma)\frac{df}{dz} = 0,\quad z\frac{d^{2}g}{dz^{2}}+(1+\sigma)\frac{dg}{dz},\label{eq: naive toy model}
    \end{equation}
    which are obtained by taking $k = 0$ and $\alpha_{i} = 0$ in \eqref{eq:toy model equation1}-\eqref{eq:toy model equation2}. This is indeed the situation in \cite{singh2}, where only a single wave equation is considered. However, a crucial difference arises in the present setting: for the coupled system of two equations, the solution basis \eqref{eq: solution basis} fails to converge, as $k\rightarrow 0$, to the solution basis of \eqref{eq: naive toy model}.
\end{remark}
Let the solution $(f,g)$ to \eqref{eq:toy model equation1}-\eqref{eq:toy model equation2} take the following form:\begin{align}
f &=\sum_{n = 0}^{\infty}a_{n}(-z)^{n}+(-z)^{-\frac{\sigma}{q_{k}}}\sum_{n = 0}^{\infty}b_{n}(-z)^{n}+(-z)^{-\frac{k^{2}+\sigma}{q_{k}}}\sum_{n = 0}^{\infty}c_{n}(-z)^{n},\label{informal form of the solution 1}\\
g & = \sum_{n = 0}^{\infty}A_{n}(-z)^{n}+(-z)^{-\frac{\sigma}{q_{k}}}\sum_{n=0}^{\infty}B_{n}(-z)^{n}+(-z)^{-\frac{k^{2}+\sigma}{q_{k}}}\sum_{n = 0}^{\infty}C_{n}(-z)^{n}
.\label{informal form of the solutions 2}
\end{align}
Substituting the forms \eqref{informal form of the solution 1}-\eqref{informal form of the solutions 2} into the equations \eqref{eq:toy model equation1}-\eqref{eq:toy model equation2}, we can obtain the following algebraic equations for $(a_{n},b_{n},c_{n},A_{n},B_{n},C_{n})$ for $n\geq 0$:\begin{align}
    \left(q_{k}n+1+\sigma\right)(n+1)a_{n+1} = &\alpha_{1}a_{n}+\alpha_{2}A_{n},\label{eq:algebraic 1}\\
    (q_{k}(n+1)+\sigma)(n+1)A_{n+1}+k(n+1)a_{n+1} = &\alpha_{3}A_{n},\label{eq:An albegraic equation}\\b_{0} = 0,\quad \left(n+1-\frac{\sigma}{q_{k}}\right)\left(q_{k}n+1\right)b_{n+1} = &\alpha_{1}b_{n}+\alpha_{2}B_{n},\label{algebraic relation for bn}\\
    \left(n+1-\frac{\sigma}{q_{k}}\right)(q_{k}n+q_{k})B_{n+1}+k\left(n+1-\frac{\sigma}{q_{k}}\right)b_{n+1} =& \alpha_{3}B_{n},\label{algebraic relation for Bn}\\\left(n+1-\frac{k^{2}+\sigma}{q_{k}}\right)\left(q_{k}n+q_{k}\right)c_{n+1}=&\alpha_{1}c_{n}+\alpha_{2}C_{n}\label{algebraic relation for cn}\\C_{0} = \frac{1}{k}c_{0},\quad \left(n+1-\frac{k^{2}+\sigma}{q_{k}}\right)\left(q_{n}n+1-2k^{2}\right)C_{n+1}+k\left(n+1-\frac{k^{2}+\sigma}{q_{k}}\right)c_{n+1} = &\alpha_{3}C_{n}.\label{eq:algebraic 2}
\end{align}
Following the above expansion and algebraic relations, we can construct a solution basis for ingoing solutions to \eqref{eq:toy model equation1}-\eqref{eq:toy model equation2}.\begin{proposition}
\label{prop: construction of model ingoing solutions}
Define $\hat{d}$ to be\begin{equation}
    \hat{d} = \sum_{i = 1}^{8}\vert \alpha_{i}\vert.
\end{equation}
    For $\sigma\in\mathbb{I}_{[-1-\eta,-\eta]}$, there exist two linearly independent ingoing solutions \begin{equation*}
        \begin{bmatrix}
        f_{3,\sigma}\\g_{3,\sigma}
        \end{bmatrix}
        ,\begin{bmatrix}
            f_{4,\sigma}\\g_{4,\sigma}
        \end{bmatrix}
    \end{equation*} 
    to the equations \eqref{eq:toy model equation1}-\eqref{eq:toy model equation2}, having the following expansions \begin{align}
    \begin{bmatrix}
        f_{3,\sigma}\\g_{3,\sigma}
    \end{bmatrix}
    & = (-z)^{-\frac{k^{2}+\sigma}{q_{k}}}\begin{bmatrix}
        1\\\frac{1}{k}
    \end{bmatrix}
    +(-z)^{-\frac{k^{2}+\sigma}{q_{k}}+1}\left(\sum_{n = 1}^{\infty}\begin{bmatrix}
        \bar{c}_{n,\sigma}\\\bar{C}_{n,\sigma}
    \end{bmatrix}(-z)^{n-1}\right),\label{general ingoing solutions for r}\\
    \begin{bmatrix}
        f_{4,\sigma}\\g_{4,\sigma}
    \end{bmatrix}
    & = (-z)^{-\frac{\sigma}{q_{k}}}
    \begin{bmatrix}
        0\\1
    \end{bmatrix}
    +(-z)^{-\frac{\sigma}{q_{k}}+1}\left(
    \sum_{n = 1}^{\infty}\begin{bmatrix}
        \bar{b}_{n,\sigma}\\\bar{B}_{n,\sigma}
    \end{bmatrix}(-z)^{n-1}\right),\label{eq: general ingoing solutions for psi}
    \end{align}
    where the coefficients $\bar{c}_{n,\sigma},\bar{C}_{n,\sigma},\bar{b}_{n,\sigma},\bar{B}_{n,\sigma}$ are holomorphic functions in $\sigma$ for $\sigma\in\mathbb{I}_{[-1-\eta,-\eta]}$ and satisfy the estimates:\begin{align}
        \vert\bar{c}_{n,\sigma}\vert\leq C\frac{1}{kn!}\hat{d}^{n},\quad \vert \bar{C}_{n,\sigma}\vert\leq C\frac{1}{kn!}\hat{d}^{n},\quad \vert \bar{b}_{n,\sigma}\vert\leq C\frac{1}{n!}\hat{d}^{n},\quad \vert \bar{B}_{n,\sigma}\vert\leq C\frac{1}{n!}\hat{d}^{n}.\label{eq: estimate for the toy model ingoing solutions}
    \end{align}
\end{proposition}
\begin{proof}
Using the algebraic relations \eqref{algebraic relation for bn}-\eqref{eq:algebraic 2}, taking \begin{equation*}
    c_{0} = b_{0} = 1,\quad c_{1} = b_{1} = 0,
\end{equation*} 
and letting \begin{equation*}
    \bar{c}_{n,\sigma} = c_{n},\quad \bar{C}_{n,\sigma} = C_{n},\quad \bar{b}_{n,\sigma} = \bar{b}_{n},\quad \bar{B}_{n,\sigma} = B_{n},
\end{equation*}
the estimate \eqref{eq: estimate for the toy model ingoing solutions} follows inductively from the algebraic relations \eqref{algebraic relation for bn}-\eqref{eq:algebraic 2}. The convergence of the series is ensured by the $n!$ factor in the estimate \eqref{eq: estimate for the toy model ingoing solutions}. This concludes the proof.
\end{proof}
An immediate corollary of Proposition \ref{prop: construction of model ingoing solutions} is the construction of ingoing solutions to $\mathcal{L}_{p}(\hat{r}_{p},\hat{\Psi}_{p}) = 0$.

\begin{corollary}
\label{coro: construction of ingoing solution for Lp}
    For $\sigma\in\mathbb{I}_{[-1-\eta,-\eta]}$, there exist two linearly independent ingoing solutions \begin{equation*}
        \mathbf{r}_{-,\sigma} = \begin{bmatrix}
\mathbf{r}_{-,\sigma}^{(1)}\\[.5em]\mathbf{r}_{-,\sigma}^{(2)}
        \end{bmatrix},\ 
        \mathbf{\Psi}_{-,\sigma} = 
        \begin{bmatrix}
            \mathbf{\Psi}_{-,\sigma}^{(1)}\\[.5em]\mathbf{\Psi}_{-,\sigma}^{(2)}
        \end{bmatrix}
    \end{equation*}
    to the equation $\mathcal{L}_{p}(\hat{r}_{p},\hat{\Psi}_{p}) = 0$, having the following expansions:\begin{align}
       \mathbf{r}_{-,\sigma} &=  \begin{bmatrix}
            \mathbf{r}_{-,\sigma}^{(1)}\\[.5em]\mathbf{r}_{-,\sigma}^{(2)}
        \end{bmatrix} = (-z)^{-\frac{k^{2}+\sigma}{q_{k}}}\begin{bmatrix}
            1\\\frac{1}{k}
        \end{bmatrix}
        +(-z)^{-\frac{k^{2}+\sigma}{q_{k}}+1}\left(\sum_{n = 1}^{\infty}\begin{bmatrix}
            \bar{c}_{n,\sigma}\\\bar{C}_{n,\sigma}
        \end{bmatrix}(-z)^{n-1}\right),\\\mathbf{\Psi}_{-,\sigma}& = 
        \begin{bmatrix}
            \mathbf{\Psi}_{-,\sigma}^{(1)}\\[.5em]\mathbf{\Psi}_{-,\sigma}^{(2)}
        \end{bmatrix} = (-z)^{-\frac{\sigma}{q_{k}}}\begin{bmatrix}
            0\\1
        \end{bmatrix}
        +(-z)^{-\frac{\sigma}{q_{k}}+1}\left(\sum_{n=1}^{\infty}\begin{bmatrix}
            \bar{b}_{n,\sigma}\\\bar{B}_{n,\sigma}
        \end{bmatrix}(-z)^{n-1}\right),
    \end{align}
    where coefficients $\bar{c}_{n,\sigma},\bar{C}_{n,\sigma},\bar{b}_{n,\sigma},\bar{B}_{n,\sigma}$ are holomorphic functions in $\sigma$ for $\sigma\in\mathbb{I}_{[-1-\eta,-\eta]}$ and satisfy the estimates:\begin{equation}
        \vert\bar{c}_{n,\sigma}\vert\leq C\frac{k^{n-1}}{n!},\quad \vert \bar{C}_{n,\sigma}\vert\leq C\frac{k^{n-1}}{n!},\quad \vert\bar{b}_{n,\sigma}\vert\leq C\frac{k^{n}}{n!},\quad \vert\bar{B}_{n,\sigma}\vert\leq C\frac{k^{n}}{n!}.\label{estimate on the ingoing solution of Lp}
    \end{equation}
\end{corollary}

One can directly apply the same argument to construct the outgoing solutions to \eqref{eq:toy model equation1}-\eqref{eq:toy model equation2}. However, note that the coefficient before $a_{1}$ in the algebraic equation \eqref{eq:algebraic 1} will degenerate when $\sigma = -1$. Similarly, the coefficient before $A_{1}$ in \eqref{eq:An albegraic equation} will degenerate when $\sigma = -q_{k}$. Viewing $a_{n}$ and $A_{n}$ as meromorphic functions of $\sigma$, this degeneracy will lead to poles at $\sigma = -1$ and $\sigma = -q_{k}$ and will make the above construction fail at $\sigma= -1,-q_{k}$, since $a_{1}$ and $A_{1}$ are not well-defined at $\sigma = -1,-q_{k}$. To make things precise, we have the following proposition constructing the outgoing solutions for $\sigma\in\mathbb{I}_{[-1-\eta,-\eta]}\backslash\{-1,-q_{k}\}$:
\begin{proposition}
    For $\sigma\in\mathbb{I}_{[-1-\eta,-\eta]}\backslash\{-1,-q_{k}\}$, there exist two linearly independent outgoing solutions\begin{equation*}
        \begin{bmatrix}
            \bar{f}_{1,\sigma}\\\bar{g}_{1,\sigma}
        \end{bmatrix},\begin{bmatrix}
            \bar{f}_{2,\sigma}\\\bar{g}_{2,\sigma}
        \end{bmatrix}
    \end{equation*}
    to the equations \eqref{eq:toy model equation1}-\eqref{eq:toy model equation1}, having the following expansions\begin{align}
        \begin{bmatrix}
            \bar{f}_{1,\sigma}\\\bar{g}_{1,\sigma}
        \end{bmatrix}
        &=\begin{bmatrix}
            1\\0
        \end{bmatrix}
        +(-z)\left(
        \sum_{n = 1}^{\infty}\begin{bmatrix} 
    \grave{a}_{n,\sigma}\\\grave{A}_{n,\sigma}
        \end{bmatrix}(-z)^{n-1}\right),\\
        \begin{bmatrix}
            \bar{f}_{2,\sigma}\\\bar{g}_{2,\sigma}
        \end{bmatrix}
        & = \begin{bmatrix}
            0\\1
        \end{bmatrix}+(-z)\left(\sum_{n = 1}^{\infty}
        \begin{bmatrix}
            \acute{a}_{n,\sigma}\\\acute{A}_{n,\sigma}
        \end{bmatrix}(-z)^{n-1}\right),
    \end{align}
    where the coefficients $\grave{a}_{n,\sigma},\grave{A}_{n,\sigma},\acute{a}_{n,\sigma},\acute{A}_{n,\sigma}$ are meromorphic functions of $\sigma$ for $\sigma\in\mathbb{I}_{[-1-\eta,-\eta]}$ with poles at $\sigma = -1,-q_{k}$, satisfying the following estimate:\begin{equation}
        \vert \grave{a}_{n,\sigma}\vert\lesssim_{\sigma}\frac{\hat{d}^{n}}{n!},\quad \vert \grave{A}_{n,\sigma}\vert\lesssim_{\sigma}\frac{\hat{d}^{n}}{n!},\quad \vert\acute{a}_{n,\sigma}\vert\lesssim_{\sigma}\frac{\hat{d}^{n}}{n!},\quad \vert \acute{A}_{n,\sigma}\vert\lesssim_{\sigma}\frac{\hat{d}^{n}}{n!}.\label{estimate on the toy toy model outgoing solutions}
    \end{equation}
\end{proposition}
\begin{proof}
    The proof follows similarly to the proof of Proposition \ref{prop: construction of model ingoing solutions}.
\end{proof}
\begin{remark}
The above proposition will not play a role in our proof of the main result of this paper. We state the proposition to highlight the key role played by the special structure of the background $k$-self-similar spacetime. Without some additional structures, the construction of the outgoing solution will just fail at $\sigma = -1,\ -q_{k}$.
\end{remark}
\begin{remark}
    Note that in the above proposition, we did not track the $\sigma$-dependence of the estimate \eqref{estimate on the toy toy model outgoing solutions}. However, since $\grave{a}_{n,\sigma},\grave{A}_{n,\sigma},\acute{a}_{n,\sigma},\acute{A}_{n,\sigma}$ have poles at $\sigma = -1,-q_{k}$, the $\sigma$-dependence will become singular at those points.
\end{remark}
\subsection{Construction of outgoing solutions to $\mathcal{L}_{p} = 0$}
As mentioned in the previous section, for a general system of ODEs \eqref{eq:toy model equation1}-\eqref{eq:toy model equation2}, in the region $\mathbb{I}_{[-1-\eta,-\eta]}$, while the construction of ingoing solutions holds true in the whole region $\mathbb{I}_{[-1-\eta,-\eta]}$, the construction of outgoing solutions is only valid for $\sigma$ away from $\sigma = -1,-q_{k}$, due to the degeneracy of some coefficients.

To avoid the presence of this degeneracy, one approach is to renormalize the values of $a_{0}$ and $A_{0}$ to be \begin{equation*}
    a_{0} = A_{0} = \left(1+\sigma\right)\left(1+\frac{\sigma}{q_{k}}\right).
\end{equation*}
However, this renormalization will cause a severe problem in determining the scattering resonance in Section \ref{sec: location of the scattering resonance}. Nonetheless, this issue can be resolved when considering the equations $\mathcal{L}_{p}(f,g) = 0$, i.e., taking $\alpha_{i}$ to be as in \eqref{eq: special structure of alphai}, due to a special cancellation from the coefficient of the Taylor expansion of background $k$-self similar spacetime at the singular horizon $z = 0$. This cancellation allows us to construct outgoing solutions at $\sigma = -1$. Hence, the only renormalization arises from the degeneracy of the coefficients at $\sigma = -q_{k}$. More specifically, we have: \begin{proposition}
\label{prop: model argument for square coefficient}
For $\sigma\in\mathbb{I}_{[-1-\eta,-\eta]}\backslash\{-q_{k}\}$, there exist two linearly independent outgoing solutions
\begin{equation*}
    \mathbf{r}_{+,\sigma} = \begin{bmatrix}
        \mathbf{r}_{+,\sigma}^{(1)}\\[.5em]\mathbf{r}_{+,\sigma}^{(2)}
    \end{bmatrix},\ \mathbf{\Psi}_{+,\sigma} = \begin{bmatrix}
        \mathbf{\Psi}_{+,\sigma}^{(1)}\\[.5em]\mathbf{\Psi}_{+,\sigma}^{(2)}
    \end{bmatrix}
\end{equation*}
to the equation $\mathcal{L}_{p}(f,g) = 0$, having the following expansions:\begin{align}
\mathbf{r}_{+,\sigma} = \begin{bmatrix}
    \mathbf{r}_{+,\sigma}^{(1)}\\[.5em]\mathbf{r}_{+,\sigma}^{(2)}
\end{bmatrix}
& = \begin{bmatrix}
    1\\\frac{kq_{k}}{\sigma}
\end{bmatrix}
+(-z)\left(\sum_{n = 1}^{\infty}\begin{bmatrix}
    \bar{a}_{n,\sigma}\\\bar{A}_{n,\sigma}
\end{bmatrix}(-z)^{n-1}\right),\\
\mathbf{\Psi}_{+,\sigma} = \begin{bmatrix}
    \mathbf{\Psi}_{+,\sigma}^{(1)}\\[.5em]\mathbf{\Psi}_{+,\sigma}^{(2)}
\end{bmatrix}
&  = \begin{bmatrix}
    0\\1+\frac{\sigma}{q_{k}}
\end{bmatrix}
+(-z)\left(\sum_{n = 1}^{\infty}\begin{bmatrix}
    \widetilde{a}_{n,\sigma}\\\widetilde{A}_{n,\sigma}
\end{bmatrix}(-z)^{n-1}\right),
\end{align}
where $\bar{a}_{n,\sigma},\bar{A}_{n,\sigma},\widetilde{a}_{n,\sigma},\widetilde{A}_{n,\sigma}$ are holomorphic functions in $\sigma$ for $\sigma\in\mathbb{I}_{[-1-\eta,-\eta]}$ and satisfy the estimates:
\begin{equation}
    \vert\bar{a}_{n,\sigma}\vert\leq C\frac{k^{n}}{n!},\quad \vert \bar{A}_{n,\sigma}\vert\leq C\frac{k^{n}}{n!},\quad \vert \widetilde{a}_{n,\sigma}\vert\leq C\frac{k^{n}}{n!},\quad \vert\widetilde{A}_{n,\sigma}\vert\leq C\frac{k^{n}}{n!}.\label{eq: estimate on the expansion of the outgoing solutions}
\end{equation}
\end{proposition}
\begin{proof}
The proof will follow similarly to the proof of Proposition \ref{prop: construction of model ingoing solutions}. To construct $\mathbf{r}_{+,\sigma}$, we take $(a_{0},A_{0}) = \left(1,-\frac{k(1+2k^{2})}{2(2+k^{2})}\right)$ and $(\bar{a}_{n},\bar{A}_{n}) = (a_{2n},A_{2n})$, where $(a_{2n},A_{2n})$ is determined inductively by the algebraic relations \eqref{eq:algebraic 1}-\eqref{eq:An albegraic equation}. To construct $\mathbf{\Psi}_{+,\sigma}$, we take $(a_{0},A_{0}) = (0,1)$ and $(\widetilde{a}_{n},\widetilde{A}_{n}) = (a_{2n},A_{2n})$, where $(a_{2n},A_{2n})$ again is determined inductively by the algebraic relations \eqref{eq:algebraic 1}-\eqref{eq:An albegraic equation}.

Since in the region $\sigma\in\mathbb{I}_{[-1-\eta,-\eta]}$, the coefficients \begin{equation*}
(q_{k}n+1+\sigma)(n+1),\quad (q_{k}(n+1)+\sigma)(n+1)
\end{equation*}
of $a_{n+1}$ and $A_{n+1}$ in \eqref{eq:algebraic 1}-\eqref{eq:An albegraic equation} can only be zero if $n = 0$, we write down the algebraic equations for $a_{1}$ and $A_{1}$ explicitly:\begin{align*}
    (1+\sigma)a_{1} &= d_{1}a_{0}+d_{2}A_{0},\\
    (q_{k}+\sigma)A_{1}+ka_{1}& = d_{3}A_{0}.
\end{align*}
Substituting the values of $\alpha_{i}$ \eqref{eq: special structure of alphai} in $\mathcal{L}_{p}$, we have\begin{align}
&(1+\sigma)a_{1} = -\left(k^{2}\sigma+k^{2}(2+k^{2})\right)\lambda_{0}\frac{\sigma+1}{\sigma-k^{2}}a_{0}-\frac{2k\lambda_{0}(1+k^{2})}{\sigma-k^{2}}(\sigma+1)A_{0},\label{eq: key cancellation for algebraic equation of a2}\\&
(q_{k}+\sigma)A_{1}+ka_{1} = k^{2}\lambda_{0}A_{0}.
\end{align}
By \eqref{eq: key cancellation for algebraic equation of a2}, the dangerous coefficient $1+\sigma$ will be canceled. Hence, we have\begin{align}
a_{1} &= -\left(k^{2}\sigma+k^{2}(2+k^{2})\right)\frac{\lambda_{0}}{\sigma-k^{2}}a_{0}-\frac{2k\lambda_{0}(1+k^{2})}{\sigma-k^{2}}A_{0},\\
A_{1} &= \frac{\lambda_{0}(k^{2}A_{0}+k^{3}a_{0})}{(\sigma-k^{2})(q_{k}+\sigma)}\left(\sigma+(2+k^{2})\right).
\end{align}
Hence, if $(a_{0},A_{0}) = \left(1
,\frac{kq_{k}}{\sigma}\right)$, we have \begin{equation*}
    \bar{a}_{1,\sigma} = a_{1} =-k^{2}\lambda_{0}\frac{\sigma+2+k^{2}}{\sigma-k^{2}}-\frac{2k^{2}\lambda_{0}(1+k^{2})q_{k}}{(\sigma-k^{2})\sigma},\quad\bar{A}_{1,\sigma} = A_{1} = \frac{\lambda_{0}k^{3}}{(\sigma-k^{2})\sigma}(\sigma+k^{2}+2),
\end{equation*}
which shows that $\bar{a}_{1,\sigma}$ and $\bar{A}_{1,\sigma}$ are holomorphic functions for $\sigma\in\mathbb{I}_{[-1-\eta,-\eta]}$. Similarly, if $(a_{0},A_{0}) = \left(0,1+\frac{\sigma}{q_{k}}\right)$, we have\begin{equation*}
    \widetilde{a}_{1,\sigma} = a_{1} =-\frac{2k\lambda_{0}(1+k^{2})}{\sigma-k^{2}}\left(1+\frac{\sigma}{q_{k}}\right),\quad \widetilde{A}_{1,\sigma} = A_{1} =\frac{k^{2}\lambda_{0}(\sigma+k^{2}+2)}{q_{k}(\sigma-k^{2})},
\end{equation*}
which shows that $\widetilde{a}_{1,\sigma}$ and $\widetilde{A}_{1,\sigma}$ are holomorphic functions for $\sigma\in\mathbb{I}_{[-1-\eta,-\eta]}$. The holomorphicity of $\bar{a}_{n,\sigma},\bar{A}_{n,\sigma},\widetilde{a}_{n,\sigma},\widetilde{A}_{n,\sigma}$ and the estimate \eqref{eq: estimate on the expansion of the outgoing solutions} follow inductively. This concludes the proof.
\end{proof}

\subsection{Construction of ingoing solutions to the linearized Einstein-scalar field system}
In this section, we derive the ingoing solutions to the equation $\mathcal{L}(f,g) = 0$. We consider the following equation \begin{equation}
\label{eq: general equation in the construction of ingoing solution}
\widetilde{L}_{p}(f,g) = E_{1,\sigma}(z)\frac{d}{dz}\begin{bmatrix}
    f\\g
\end{bmatrix}+E_{2,\sigma}(z)\begin{bmatrix}
    f\\g
\end{bmatrix},
\end{equation}
where $E_{1,\sigma}$ and $E_{2,\sigma}$ are $2\times2$ matrix-valued functions of $z$ of order $O_{k}(z)$ and $O(k)$ respectively, and $\widetilde{\mathcal{L}}_{p}$ takes the form of \begin{equation}
    \widetilde{\mathcal{L}}_{p}(f,g) = q_{k}z\frac{d^{2}}{dz^{2}}\begin{bmatrix}
        f\\g
    \end{bmatrix}+\begin{bmatrix}
        1+\sigma&0\\
        k&q_{k}+\sigma
    \end{bmatrix}\frac{d}{dz}\begin{bmatrix}
        f\\g
    \end{bmatrix}.
\end{equation}
\begin{remark}
The equation \eqref{eq: general equation in the construction of ingoing solution} can be reduced to $\mathcal{L} = 0$ by moving all the zeroth-order terms on the left-hand side of \eqref{schematic form of rp}-\eqref{schematic form of psip} to the right-hand side of \eqref{schematic form of rp}-\eqref{schematic form of psip} and defining $E_{1,\sigma}$ and $E_{2,\sigma}$ accordingly. However, since the ingoing solutions are insensitive to the precise structure of the equations, we provide a general framework here.
\end{remark}
\begin{remark}
    Note that $E_{1}$ and $E_{2}$ are not necessarily analytic. Hence, one can not directly apply the standard theory of constructing outgoing and ingoing solutions for a system of ODEs with regular singularities. In fact, since the coefficients on the right-hand side of \eqref{schematic form of rp}-\eqref{schematic form of psip} are from the $k$-self-similar naked singularity spacetime, $E_{1}$ and $E_{2}$ should be thought of as Hölder-continuous matrix-valued functions.
\end{remark}
We adopt the following convention: \begin{convention}
Let $N$ be a $n\times m$ matrix. We denote by $N_{ij}$ the element of $N$ in the $i$-th row and $j$-th column. For $m = 1$, we abbreviate the notation by writing $N_{i}$ the element of $N$ in the $i$-th row. For $m = n$, we denote the co-factor of $N_{ij}$ as $[N]_{ij}$.
\end{convention}
 We have the following proposition constructing ingoing solutions:
\begin{proposition}
\label{prop: general construction of outgoing and ingoing solutions}
For $\sigma\in\mathbb{I}_{[-1-\eta,-\eta]}$, there exist two linearly independent solutions \begin{equation*}
    \mathbf{f}_{in,\sigma} = \begin{bmatrix}
       \mathbf{f}_{in,\sigma}^{(1)}\\[.5em]\mathbf{f}_{in,\sigma}^{(2)} 
    \end{bmatrix},\mathbf{g}_{in,\sigma} = \begin{bmatrix}
        \mathbf{g}_{in,\sigma}^{(1)}\\[.5em]\mathbf{g}_{in,\sigma}^{(2)}
    \end{bmatrix}
\end{equation*}
to \eqref{eq: general equation in the construction of ingoing solution}, having the following expansions:\begin{align}
    \mathbf{f}_{in,\sigma} &= \begin{bmatrix} \mathbf{f}_{in,\sigma}^{(1)}\\[.5em]\mathbf{f}_{in,\sigma}^{(2)}
    \end{bmatrix}
     = (-z)^{-\frac{k^{2}+\sigma}{q_{k}}}\begin{bmatrix}
         1\\\frac{1}{k}
     \end{bmatrix}+O_{k,\sigma}\left((-z)^{-\frac{k^{2}+\sigma}{q_{k}}+1}\right),\label{expansion of real ingoing solution 1}\\[.5em]
     \mathbf{g}_{in,\sigma}& = \begin{bmatrix}
         \mathbf{g}_{in,\sigma}^{(1)}\\\mathbf{g}_{in,\sigma}^{(2)}
     \end{bmatrix}
      = (-z)^{-\frac{\sigma}{q_{k}}}\begin{bmatrix}
          0\\1
      \end{bmatrix}+O_{k,\sigma}((-z)^{-\frac{\sigma}{q_{k}}+1}).\label{expansion of real ingoing solution 2}
\end{align}
\end{proposition}
\begin{proof}
We can write down an explicit solution basis for $\widetilde{\mathcal{L}}_{p}(f,g) = 0$:\begin{equation*}
    \left\{\begin{bmatrix}
        1\\0
    \end{bmatrix},\ \begin{bmatrix}
        0\\1
    \end{bmatrix},\ (-z)^{-\frac{k^{2}+\sigma}{q_{k}}}\begin{bmatrix}
        1\\\frac{1}{k}
    \end{bmatrix},\ (-z)^{-\frac{\sigma}{q_{k}}}\begin{bmatrix}
        0\\1
    \end{bmatrix}\right\} = :\left\{\begin{bmatrix}
        \widetilde{f}_{1,\sigma}\\\widetilde{g}_{1,\sigma}
    \end{bmatrix},\ \begin{bmatrix}
        \widetilde{f}_{2,\sigma}\\\widetilde{g}_{2,\sigma}
    \end{bmatrix},\ \begin{bmatrix}
        \widetilde{f}_{3,\sigma}\\\widetilde{g}_{3,\sigma}
    \end{bmatrix},\ \begin{bmatrix}
        \widetilde{f}_{4,\sigma}\\\widetilde{g}_{4,\sigma}
    \end{bmatrix}\right\}.
\end{equation*}
\paragraph{Deriving the Volterra integral equation}
Assume the solution to \eqref{eq: general equation in the construction of ingoing solution} takes the form of \begin{equation}
\begin{bmatrix}
f\\g
\end{bmatrix}
 = A_{\sigma}(z)\begin{bmatrix}
 1\\0
 \end{bmatrix}
 +B_{\sigma}(z)\begin{bmatrix}
 0\\1
 \end{bmatrix}
 +C_{\sigma}(z)\begin{bmatrix}
 (-z)^{-\frac{k^{2}+\sigma}{q_{k}}}\\\frac{1}{k}(-z)^{-\frac{k^{2}+\sigma}{q_{k}}}
 \end{bmatrix}
 +D_{\sigma}(z)\begin{bmatrix}
 0\\(-z)^{-\frac{\sigma}{q_{k}}}
 \end{bmatrix}.
\end{equation}
Then one can compute the derivatives of $f$ and $g$:\begin{equation}
\begin{aligned}
\frac{d}{dz}\begin{bmatrix}
f\\g
\end{bmatrix}
 = &A_{\sigma}(z)\begin{bmatrix}
 \widetilde{f}_{1,\sigma}^{\prime}\\\widetilde{g}_{1,\sigma}^{\prime}
\end{bmatrix}
+B_{\sigma}(z)\begin{bmatrix}
\widetilde{f}^{\prime}_{2,\sigma}\\\widetilde{g}_{2,\sigma}^{\prime}
\end{bmatrix}
+C_{\sigma}(z)\begin{bmatrix}
\widetilde{f}_{3,\sigma}^{\prime}\\\widetilde{g}_{3,\sigma}^{\prime}
\end{bmatrix}
+D_{\sigma}(z)\begin{bmatrix}
\widetilde{f}_{4,\sigma}^{\prime}\\\widetilde{g}_{4,\sigma}^{\prime}
\end{bmatrix}
\\&+A_{\sigma}^{\prime}(z)\begin{bmatrix}
\widetilde{f}_{1,\sigma}\\\widetilde{g}_{1,\sigma}
\end{bmatrix}
+B_{\sigma}^{\prime}(z)\begin{bmatrix}
\widetilde{f}_{2,\sigma}\\\widetilde{g}_{2,\sigma}
\end{bmatrix}
+C_{\sigma}^{\prime}(z)\begin{bmatrix}
\widetilde{f}_{3,\sigma}\\\widetilde{g}_{3,\sigma}
\end{bmatrix}
+D_{\sigma}^{\prime}(z)\begin{bmatrix}
\widetilde{f}_{4,\sigma}\\\widetilde{g}_{4,\sigma}
\end{bmatrix}.
\end{aligned}
\label{eq: change of variable first order equation}
\end{equation}
We assume that \begin{equation}
A_{\sigma}^{\prime}(z)\begin{bmatrix}
\widetilde{f}_{1,\sigma}\\\widetilde{g}_{1,\sigma}
\end{bmatrix}
+B_{\sigma}^{\prime}(z)\begin{bmatrix}
\widetilde{f}_{2,\sigma}\\\widetilde{g}_{2,\sigma}
\end{bmatrix}
+C_{\sigma}^{\prime}(z)\begin{bmatrix}
\widetilde{f}_{3,\sigma}\\\widetilde{g}_{3,\sigma}
\end{bmatrix}
+D_{\sigma}^{\prime}(z)\begin{bmatrix}
\widetilde{f}_{4,\sigma}\\\widetilde{g}_{4,\sigma}
\end{bmatrix}
  =0.\label{eq: first order assumption}
\end{equation} 
Taking the second derivatives of $f$ and $g$, we have\begin{equation}
\begin{aligned}
\frac{d^{2}}{dz^{2}}\begin{bmatrix}
f\\g
\end{bmatrix}
 =& A_{\sigma}(z)\frac{d^{2}}{dz^{2}}\begin{bmatrix}
 \widetilde{f}_{1,\sigma}\\\widetilde{g}_{1,\sigma}
 \end{bmatrix}
 +B_{\sigma}(z)\frac{d^{2}}{dz^{2}}\begin{bmatrix}
 \widetilde{f}_{2,\sigma}\\\widetilde{g}_{2,\sigma}
 \end{bmatrix}
 +C_{\sigma}(z)\frac{d^{2}}{dz^{2}}\begin{bmatrix}
 \widetilde{f}_{3,\sigma}\\\widetilde{g}_{3,\sigma}
 \end{bmatrix}
 +D_{\sigma}(z)\frac{d^{2}}{dz^{2}}\begin{bmatrix}
\widetilde{f}_{4,\sigma}\\\widetilde{g}_{4,\sigma}
 \end{bmatrix}
 \\&+A_{\sigma}^{\prime}(z)\frac{d}{dz}\begin{bmatrix}
 \widetilde{f}_{1,\sigma}\\\widetilde{g}_{1,\sigma}
 \end{bmatrix}+B_{\sigma}^{\prime}(z)\frac{d}{dz}\begin{bmatrix}
 \widetilde{f}_{2,\sigma}\\\widetilde{g}_{2,\sigma}
 \end{bmatrix}
 +C_{\sigma}^{\prime}(z)\frac{d}{dz}\begin{bmatrix}
 \widetilde{f}_{3,\sigma}\\\widetilde{g}_{3,\sigma}
 \end{bmatrix}
 +D_{\sigma}^{\prime}(z)\frac{d}{dz}\begin{bmatrix}
\widetilde{f}_{4,\sigma}\\\widetilde{g}_{4,\sigma}
 \end{bmatrix}.
\end{aligned}
\label{eq: change of variable second order equation}
\end{equation}
Combining \eqref{eq: first order assumption}, \eqref{eq: change of variable first order equation}, \eqref{eq: change of variable second order equation} and using the equation \eqref{eq: general equation in the construction of ingoing solution}, we can deduce that {\footnotesize\begin{equation}
\begin{bmatrix}
\widetilde{f}_{1,\sigma}&\widetilde{f}_{2,\sigma}&\widetilde{f}_{3,\sigma}&\widetilde{f}_{4,\sigma}\\
\widetilde{g}_{1,\sigma}&\widetilde{g}_{2,\sigma}&\widetilde{g}_{3,\sigma}&\widetilde{g}_{4,\sigma}\\
\widetilde{f}_{1,\sigma}^{\prime}&\widetilde{f}_{2,\sigma}^{\prime}&\widetilde{f}_{3,\sigma}^{\prime}&\widetilde{f}_{4,\sigma}^{\prime}\\\widetilde{g}_{1,\sigma}^{\prime}&\widetilde{g}_{2,\sigma}^{\prime}&\widetilde{g}_{3,\sigma}^{\prime}&\widetilde{g}_{4,\sigma}^{\prime}
\end{bmatrix}
\frac{d}{dz}\begin{bmatrix}
A_{\sigma}(z)\\B_{\sigma}(z)\\C_{\sigma}(z)\\D_{\sigma}(z)
\end{bmatrix}
 = \begin{bmatrix}
 0&0\\p_{k}(-z)^{-1}E_{2,\sigma}(z)&p_{k}(-z)^{-1}E_{1,\sigma}(z)
 \end{bmatrix}
 \begin{bmatrix}
 \widetilde{f}_{1,\sigma}&\widetilde{f}_{2,\sigma}&\widetilde{f}_{3,\sigma}&\widetilde{f}_{4,\sigma}\\
 \widetilde{g}_{1,\sigma}&\widetilde{g}_{2,\sigma}&\widetilde{g}_{3,\sigma}&\widetilde{g}_{4,\sigma}\\
 \widetilde{f}_{1,\sigma}^{\prime}&\widetilde{f}_{2,\sigma}^{\prime}&\widetilde{f}_{3,\sigma}^{\prime}&\widetilde{f}_{4,\sigma}^{\prime}\\\widetilde{g}_{1,\sigma}^{\prime}&\widetilde{g}_{2,\sigma}^{\prime}&\widetilde{g}_{3,\sigma}^{\prime}&\widetilde{g}_{4,\sigma}^{\prime}
 \end{bmatrix}
 \begin{bmatrix}
 A_{\sigma}(y)\\B_{\sigma}(y)\\C_{\sigma}(y)\\D_{\sigma}(y)
 \end{bmatrix}.
 \label{eq: differential form of the volterra equation}
\end{equation}
}

Denote \begin{equation}
F_{\sigma}: = \begin{bmatrix}
\widetilde{f}_{1,\sigma}&\widetilde{f}_{2,\sigma}&\widetilde{f}_{3,\sigma}&\widetilde{f}_{4,\sigma}\\\widetilde{g}_{1,\sigma}&\widetilde{g}_{2,\sigma}&\widetilde{g}_{3,\sigma}&\widetilde{g}_{4,\sigma}\\\widetilde{f}_{1,\sigma}^{\prime}&\widetilde{f}_{2,\sigma}^{\prime}&\widetilde{f}_{3,\sigma}^{\prime}&\widetilde{f}_{4,\sigma}^{\prime}\\
\widetilde{g}_{1,\sigma}^{\prime}&\widetilde{g}_{2,\sigma}^{\prime}&\widetilde{g}_{3,\sigma}^{\prime}&\widetilde{g}_{4,\sigma}^{\prime}
\end{bmatrix},\quad Q_{\sigma}:=\begin{bmatrix}
0&0\\p_{k}(-z)^{-1}E_{2,\sigma}(z)&p_{k}(-z)^{-1}E_{1,\sigma}(z)
\end{bmatrix}
,\quad \widetilde{Q}_{\sigma} = F_{\sigma}^{-1}Q_{\sigma}F_{\sigma}.
\label{eq: definition of Fsigma and Qsigma}
\end{equation}
Then we can write the equation \eqref{eq: differential form of the volterra equation} in the form of the Volterra integral equation: \begin{equation}
\begin{bmatrix}
A_{\sigma}(z)\\B_{\sigma}(z)\\C_{\sigma}(z)\\D_{\sigma}(z)
\end{bmatrix}
 = \begin{bmatrix}
 A_{\sigma}(0)\\B_{\sigma}(0)\\C_{\sigma}(0)\\D_{\sigma}(0)
 \end{bmatrix}
 +\int_{0}^{z}F_{\sigma}^{-1}Q_{\sigma}
 F_{\sigma}\begin{bmatrix}
 A(y)\\B(y)\\C(y)\\D(y)
 \end{bmatrix}
 dy.\label{Voltera equation}
\end{equation}
Substituting $(A_{\sigma},B_{\sigma},C_{\sigma},D_{\sigma})$ on the right-hand side of \eqref{Voltera equation} using \eqref{Voltera equation} again, we have the following sum:\begin{align}
\begin{bmatrix}
A(z)\\B(z)\\C(z)\\D(z)
\end{bmatrix}
 =& \begin{bmatrix}
 A(0)\\B(0)\\C(0)\\D(0)
 \end{bmatrix}
 +\sum_{n = 0}^{\infty}\int_{0}^{z}\int_{0}^{y_{1}}\cdots\int_{0}^{y_{n}}\widetilde{Q}(y_{1})\cdots\widetilde{Q}(y_{n+1})\begin{bmatrix}
 A(0)\\B(0)\\C(0)\\D(0)
 \end{bmatrix}dy_{1}\cdots dy_{n+1}\\
=&:\begin{bmatrix}
A(0)\\B(0)\\C(0)\\D(0)
\end{bmatrix}
+\sum_{n = 0}^{\infty}\begin{bmatrix}
Q_{n,\sigma}^{(1)}(z)\\Q_{n,\sigma}^{(2)}(z)\\Q_{n,\sigma}^{(3)}(z)\\Q_{n,\sigma}^{(4)}(z)
\end{bmatrix}
.
\end{align}
\paragraph{Estimates on $F_{\sigma},Q_{\sigma}$, and $\widetilde{Q}_{\sigma}$}
To estimate $$\begin{bmatrix}
Q^{(1)}_{n,\sigma}\\Q^{(2)}_{n,\sigma}\\Q^{(3)}_{n,\sigma}\\Q^{(4)}_{n,\sigma}
\end{bmatrix},$$ we first establish the estimate for $F_{\sigma}$ and $F_{\sigma}^{-1}$. The matrix $F_{\sigma}$ takes the form of \begin{equation}
F_{\sigma}=
\begin{bmatrix}
 1&0&(-z)^{-\frac{k^{2}+\sigma}{q_{k}}}&0\\[1em]
 0&1&\frac{1}{k}(-z)^{-\frac{k^{2}+\sigma}{q_{k}}}&(-z)^{-\frac{\sigma}{q_{k}}}\\[1em]
 0&0&\frac{k^{2}+\sigma}{q_{k}}(-z)^{-\frac{k^{2}+\sigma}{q_{k}}-1}&0\\[1em]
 0&0&\frac{1}{k}\frac{k^{2}+\sigma}{q_{k}}(-z)^{-\frac{k^{2}+\sigma}{q_{k}}-1}&\frac{\sigma}{q_{k}}(-z)^{-\frac{\sigma}{q_{k}}-1}
 \end{bmatrix}.
\end{equation}

We can compute the inverse of $F_{\sigma}$:\begin{equation}
F_{\sigma}^{-1} = \begin{bmatrix}
1&0&\frac{q_{k}}{\sigma+k^{2}}z&0\\[1em]
0&1&\frac{k^{2}q_{k}}{k\sigma(\sigma+k^{2})}(-z)&\frac{q_{k}}{\sigma}z\\[1em]
0&0&\frac{q_{k}}{\sigma+k^{2}}(-z)^{\frac{k^{2}+\sigma}{q_{k}}+1}&0\\[1em]
0&0&-\frac{q_{k}}{k\sigma}(-z)^{\frac{\sigma}{q_{k}}+1}&\frac{q_{k}}{\sigma}(-z)^{\frac{\sigma}{q_{k}}+1}
\end{bmatrix}.
\end{equation}

Then we have for $\sigma\in\mathbb{I}_{[-1-\eta,-\eta]}$, $F_{\sigma}$ is invertible for $z\neq0$. For $F_{\sigma}^{-1}Q_{\sigma}F_{\sigma}$, we have the estimate \begin{equation*}
    \left\vert F_{\sigma}^{-1}Q_{\sigma}F_{\sigma}\right\vert = \begin{bmatrix}
        O(1)&O(1)&O((-z)^{-\frac{k^{2}+\sigma}{q_{k}}})&O((-z)^{-\frac{\sigma}{q_{k}}})\\
        O(1)&O(1)&O((-z)^{-\frac{k^{2}+\sigma}{q_{k}}})&O((-z)^{-\frac{\sigma}{q_{k}}})\\O((-z)^{\frac{k^{2}+\sigma}{q_{k}}})&O((-z)^{\frac{k^{2}+\sigma}{q_{k}}})&O(1)&O((-z)^{\frac{k^{2}}{q_{k}}})\\O((-z)^{\frac{\sigma}{q_{k}}})&O((-z)^{\frac{\sigma}{q_{k}}})&O(1)&O(1)
    \end{bmatrix}.
\end{equation*}

For $A(0) = B(0) = 0$, we have \begin{equation*}
    \widetilde{Q}_{\sigma}\begin{bmatrix}
        A(0)\\B(0)\\C(0)\\D(0)
    \end{bmatrix}=F_{\sigma}^{-1}Q_{\sigma}F_{\sigma}\begin{bmatrix}
        A(0)\\B(0)\\C(0)\\D(0)
    \end{bmatrix}
     = \begin{bmatrix}
         C(0)O_{k,\sigma}((-z)^{-\frac{k^{2}+\sigma}{q_{k}}})+D(0)(-z)^{-\frac{\sigma}{q_{k}}}\\
         C(0)(-z)^{-\frac{k^{2}+\sigma}{q_{k}}}+D(0)(-z)^{-\frac{\sigma}{q_{k}}}\\C(0)\\C(0)(-z)^{-\frac{k^{2}}{q_{k}}}+D(0)
     \end{bmatrix}.
\end{equation*}
Hence, we have \begin{equation*}
    \begin{bmatrix}
        \left\vert Q_{1,\sigma}^{(1)}\right\vert\\\left\vert Q_{1,\sigma}^{(2)}\right\vert\\\left\vert Q_{1,\sigma}^{(3)}\right\vert\\\left\vert Q_{1,\sigma}^{(4)}\right\vert
    \end{bmatrix}
    \lesssim_{k,\sigma}\begin{bmatrix}
       |C(0)|(-z)^{1-\frac{k^{2}+\sigma}{q_{k}}}+|D(0)|(-z)^{1-\frac{\sigma}{q_{k}}}\\
       |C(0)|(-z)^{1-\frac{k^{2}+\sigma}{q_{k}}}+|D(0)|(-z)^{1-\frac{\sigma}{q_{k}}}\\|C(0)|(-z)\\|C(0)|(-z)^{1-\frac{k^{2}}{q_{k}}}+|D(0)|(-z)
    \end{bmatrix}.
\end{equation*}

\paragraph{Convergence of the infinite sum}
Arguing inductively, we have \begin{equation*}
    \left\vert\begin{bmatrix}
        Q_{n,\sigma}^{(1)}\\Q_{n,\sigma}^{(2)}\\Q_{n,\sigma}^{(3)}\\Q_{n,\sigma}^{(4)}
    \end{bmatrix}\right\vert\lesssim_{k,\sigma}\frac{1}{n!}\begin{bmatrix}
       |C(0)|(-z)^{n-\frac{k^{2}+\sigma}{q_{k}}}+|D(0)|(-z)^{n-\frac{\sigma}{q_{k}}}\\
       |C(0)|(-z)^{n-\frac{k^{2}+\sigma}{q_{k}}}+|D(0)|(-z)^{n-\frac{\sigma}{q_{k}}}\\|C(0)|(-z)^{n}\\|C(0)|(-z)^{n-\frac{k^{2}}{q_{k}}}+|D(0)|(-z)^{n}
    \end{bmatrix}.
\end{equation*}
Taking $\mathbf{f}_{in,\sigma}$ to be the solution with $$(A(0),B(0),C(0),D(0)) = (0,0,1,0);$$ $\mathbf{g}_{in,\sigma}$ to be the solution with $$(A(0),B(0),C(0),D(0)) = \left(0,0,0,1\right),$$we can conclude the proof.
\end{proof}
\begin{remark}
Note that we did not track carefully the $k$-dependence of the expansions of the ingoing solutions constructed in Proposition \ref{prop: general construction of outgoing and ingoing solutions}. In fact, the terms $O_{k,\sigma}$ in \eqref{expansion of real ingoing solution 1}-\eqref{expansion of real ingoing solution 2} will be singular when $k\rightarrow0$.
\end{remark}

We can immediately prove the following corollary:
\begin{corollary}
\label{coro: construction of the true ingoing solution}
    For $\sigma\in\mathbb{I}_{[-1-\eta,-\eta]}$, there exist two linearly independent solutions\begin{equation*}
        \hat{\mathbf{r}}_{in,\sigma} = \begin{bmatrix}
            \hat{\mathbf{r}}_{in,\sigma}^{(1)}\\\hat{\mathbf{r}}_{in,\sigma}^{(2)}
        \end{bmatrix},\quad \hat{\mathbf{\Psi}}_{in,\sigma} = \begin{bmatrix}
            \hat{\mathbf{\Psi}}_{in,\sigma}^{(1)}\\\hat{\mathbf{\Psi}}_{in,\sigma}^{(2)}
        \end{bmatrix}
    \end{equation*}
    to the equation $\mathcal{L}(\hat{r}_{p},\hat{\Psi}_{p}) = 0$, having the following expansions:\begin{align}
        \hat{\mathbf{r}}_{in,\sigma} &= (-z)^{-\frac{k^{2}+\sigma}{q_{k}}}\begin{bmatrix}
            1\\\frac{1}{k}
        \end{bmatrix}
        +O_{k,\sigma}((-z)^{-\frac{k^{2}+\sigma}{q_{k}}+1}),\\
        \hat{\mathbf{\Psi}}_{in,\sigma}&=(-z)^{-\frac{\sigma}{q_{k}}}\begin{bmatrix}
            0\\1
        \end{bmatrix}
        +O_{k,\sigma}((-z)^{-\frac{\sigma}{q_{k}}+1}).
    \end{align}
\end{corollary}
\begin{remark}
In fact, one can obtain much more refined information on the ingoing solutions using the approach in Section \ref{sec: construction of outgoing solutions}. However, in the later analysis of the scattering resonance, it suffices to know the asymptotic behavior of the ingoing solutions. Therefore, we only provide crude estimates on the ingoing solutions in this section.
\end{remark}
\subsection{Construction of outgoing solutions for the linearized system and its expansions}
\label{sec: construction of outgoing solutions}
Now we are ready to construct outgoing solutions to $\mathcal{L}(\hat{r}_{p},\hat{\Psi}_{p}) = 0$. Since we can write the equation $\mathcal{L} = 0$ in the form of \eqref{eq: general equation in the construction of ingoing solution} , the proof in  Proposition \ref{prop: general construction of outgoing and ingoing solutions} might also be used to construct outgoing solutions. However, the proof of Proposition \ref{prop: general construction of outgoing and ingoing solutions} has two limitations. First, the argument in Proposition \ref{prop: general construction of outgoing and ingoing solutions} relies on the fact that $A(0) = B(0) = 0$, while for the outgoing solutions, due to the presence of non-trivial values of $A(0)$ and $B(0)$, the same approach will fail to give a convergent Volterra sequence. Second, we need to track the $k$-smallness and $\sigma$-dependence of the expansion of the outgoing solutions while the proof of Proposition \ref{prop: general construction of outgoing and ingoing solutions} fails to give a detailed information on the $(k,\sigma)$-dependence.

To conclude, the construction of the outgoing solutions has three challenges \begin{itemize}
    \item The existence of the outgoing solutions in the whole region $\mathbb{I}_{[-1-\eta,-\eta]}$.
    \item The $k$-smallness in the expansion of the outgoing solutions.
    \item The $\sigma$-boundedness in the expansion of the outgoing solutions.
\end{itemize}

In this section, we exploit the extra structures $\frac{d\hat{r}_{p}}{dz}$ and $\frac{d\hat\Psi_{p}}{dz}$ appearing on the right-hand side of the equation \eqref{eq: perturb+rhs form of scattering equations} to construct the outgoing solutions and track the leading order behaviors and $k$-smallness of the sub-leading terms. Moreover, we will achieve the uniform $\sigma$-boundedness of the coefficients in the expansions.

Note that the equation $\mathcal{L}_{p}(\hat{r}_{p},\hat{\Psi}_{p})$ satisfies the form of \eqref{eq:toy model equation1}-\eqref{eq:toy model equation2} with $d_{i} = \alpha_{i}(\sigma)$ in Proposition \ref{prop: schematic form of equations}. Using Corollary \ref{coro: construction of ingoing solution for Lp} and Proposition \ref{prop: model argument for square coefficient}, we can construct a solution basis of the equation $\mathcal{L}_{p}(\hat{r}_{p},\hat{\Psi}_{p}) = 0$ in the region $\mathbb{I}_{[-1-\eta,-\eta]}$:\begin{equation}
\left\{\mathbf{r}_{+,\sigma},\mathbf{\Psi}_{+,\sigma},\mathbf{r}_{-,\sigma},\mathbf{\Psi}_{-,\sigma}
\right\},
\end{equation}
where each solution in the basis has the following expansion:\begin{align}
\mathbf{r}_{+,\sigma} &= \begin{bmatrix}
\mathbf{r}_{+,\sigma}^{(1)}\\[.5em]\mathbf{r}_{+,\sigma}^{(2)}
\end{bmatrix}
 = \begin{bmatrix}
 1+\sum_{n = 1}^{\infty}\bar{a}_{n,\sigma}(-z)^{n}\\[.5em]\frac{kq_{k}}{\sigma}+
 \sum_{n = 1}^{\infty}\bar{A}_{n,\sigma}(-z)^{n}
 \end{bmatrix},\label{eq: expansion of r-outgoing solution}\\
\mathbf{r}_{-,\sigma} &= \begin{bmatrix}
\mathbf{r}_{-,\sigma}^{(1)}\\[.5em]\mathbf{r}_{-,\sigma}^{(2)}
\end{bmatrix}
 = (-z)^{-\frac{\sigma+k^{2}}{q_{k}}}\begin{bmatrix}
 1+\sum_{n = 1}^{\infty}\bar{c}_{n,\sigma}(-z)^{n}\\[.5em]
 \frac{1}{k}+\sum_{n = 1}^{\infty}\bar{C}_{n,\sigma}(-z)^{n}
 \end{bmatrix},
 \\
 \mathbf{\Psi}_{+,\sigma}& = \begin{bmatrix}
 \mathbf{\Psi}_{+,\sigma}^{(1)}\\[.5em]\mathbf{\Psi}_{+,\sigma}^{(2)}
 \end{bmatrix}
  = \begin{bmatrix}
  \sum_{n = 1}^{\infty}\widetilde{a}_{n,\sigma}(-z)^{n}\\[.5em]
  1+\frac{\sigma}{q_{k}}+\sum_{n = 1}^{\infty}\widetilde{A}_{n,\sigma}(-z)^{n}
  \end{bmatrix},\label{eq: expansion of psi-outgoing solution}\\
  \mathbf{\Psi}_{-,\sigma}& = \begin{bmatrix}
  \mathbf{\Psi}_{-,\sigma}^{(1)}\\[.5em]\mathbf{\Psi}_{-,\sigma}^{(2)}
  \end{bmatrix}
   = (-z)^{-\frac{\sigma}{q_{k}}}\begin{bmatrix}
   \sum_{n = 1}^{\infty}\bar{b}_{n,\sigma}(-z)^{n}\\[.5em]
   1+\sum_{n = 1}^{\infty}\bar{B}_{n,\sigma}(-z)^{n}
   \end{bmatrix}.
\end{align}
Hence, the fundamental matrix of this solution basis takes the form of \begin{equation}
    Y_{\sigma} = \begin{bmatrix}
        \mathbf{r}_{+,\sigma}&\mathbf{\Psi}_{+,\sigma}&\mathbf{r}_{-,\sigma}&\mathbf{\Psi}_{-,\sigma}\\
        \mathbf{r}_{+,\sigma}^{\prime}&\mathbf{\Psi}_{+,\sigma}^{\prime}&\mathbf{r}_{-,\sigma}^{\prime}&\mathbf{\Psi}_{-,\sigma}^{\prime}
    \end{bmatrix}.
\end{equation}
\subsubsection{Analysis on the fundamental matrix $Y_{\sigma}$}
To facilitate our later analysis, we first analyze the fundamental matrix $Y_{\sigma}$. Recall that in the construction of $\mathbf{\Psi}_{+,\sigma}$, we need to renormalize the value of $(a_{0}, A_{0})$ to be $(0,1+\frac{\sigma}{q_{k}})$ to avoid the poles for the coefficients in the expansion. However, this renormalization will result in the degeneracy of the coefficient of the leading order term of $\mathbf{\Psi}_{+,\sigma}$ when $\sigma = -q_{k}$. The following lemma illustrates the consequence of this degeneracy.
\begin{lemma}
    The construction of $\mathbf{r}_{+,\sigma}$ and $\mathbf{\Psi}_{+,\sigma}$ can be extended holomorphically to the whole region $\sigma\in\mathbb{I}_{[-1-\eta,-\eta]}$. When $\sigma = -q_{k}$, $\mathbf{\Psi}_{+,\sigma}$ is linearly dependent on $\mathbf{\Psi}_{-,\sigma}$.
\end{lemma}
\begin{proof}
The holomorphic extension follows from the fact that in the expansions of $\mathbf{r}_{+,\sigma}$ and $\mathbf{\Psi}_{+,\sigma}$, the coefficients $\bar{a}_{n,\sigma},\bar{A}_{n,\sigma},\widetilde{a}_{n,\sigma},\widetilde{A}_{n,\sigma}$ are holomorphic functions for $\sigma\in\mathbb{I}_{[-1-\eta,-\eta]}$. For $\sigma = -q_{k}$, we can compare the algebraic relations \eqref{eq:algebraic 1} and \eqref{eq:An albegraic equation} with \eqref{algebraic relation for bn} and \eqref{algebraic relation for Bn}:\begin{align*}
    &(q_{k}n+k^{2})(n+1)a_{n+1} = \alpha_{1}a_{n}+\alpha_{2}A_{n},\\&
    q_{k}n(n+1)A_{n+1}+k(n+1)a_{n+1} = \alpha_{3}A_{n},\\&
    (n+2)(q_{k}n+1)b_{n+1} = \alpha_{1}b_{n}+\alpha_{2}B_{n},\\&
    (n+2)(q_{k}n+q_{k})B_{n+1}+k(n+2)b_{n+1} = \alpha_{3}B_{n}
\end{align*}
Note that the equations for $a_{n+1}$ and $A_{n+1}$ take the same forms as that for $b_{n}$ and $B_{n}$ when $\sigma = -q_{k}$, respectively. Hence, we can conclude that when $\sigma = -q_{k}$, $\mathbf{\Psi}_{+,\sigma}$ is linearly dependent to $\mathbf{\Psi}_{-,\sigma}$.
\end{proof}
Next, we analyze the determinant of $Y_{\sigma}$. We can decompose $Y_{\sigma}$ as:\begin{equation}
\begin{aligned}
    Y_{\sigma} &= J_{\sigma}+\sum_{n = 1}^{\infty}J_{n,\sigma}(z)(-z)^{n}\\
    &=:\begin{bmatrix}
1&0&(-z)^{-\frac{k^{2}+\sigma}{q_{k}}}&0\\
\frac{kq_{k}}{\sigma}&1+\frac{\sigma}{q_{k}}&\frac{1}{k}(-z)^{-\frac{k^{2}+\sigma}{q_{k}}}&(-z)^{-\frac{\sigma}{q_{k}}}\\0&0&\frac{k^{2}+\sigma}{q_{k}}(-z)^{-\frac{k^{2}+\sigma}{q_{k}}-1}&0\\0&0&\frac{k^{2}+\sigma}{kq_{k}}(-z)^{-\frac{k^{2}+\sigma}{q_{k}}-1}&\frac{\sigma}{q_{k}}(-z)
^{-\frac{\sigma}{q_{k}}-1}\end{bmatrix}+\sum_{n = 1}^{\infty}J_{n,\sigma}(-z)^{n},
    \end{aligned}
\end{equation}
where $J_{n,\sigma}$ takes the form of {\footnotesize\begin{equation*}
    \begin{bmatrix}
        \bar{a}_{n,\sigma}&\widetilde{a}_{n,\sigma}&\bar{c}_{n,\sigma}(-z)^{-\frac{\sigma+k^{2}}{q_{k}}}&\bar{b}_{n,\sigma}(-z)^{-\frac{\sigma}{q_{k}}}\\[1em]
        \bar{A}_{n,\sigma}&\widetilde{A}_{n,\sigma}&\bar{C}_{n,\sigma}(-z)^{-\frac{\sigma+k^{2}}{q_{k}}}&\bar{B}_{n,\sigma}(-z)^{-\frac{\sigma}{q_{k}}}\\[1em]
        (-n)\bar{a}_{n,\sigma}(-z)^{-1}&(-n)\widetilde{a}_{n,\sigma}(-z)^{-1}&\left(\frac{\sigma+k^{2}}{q_{k}}-n\right)\bar{c}_{n,\sigma}(-z)^{-1-\frac{\sigma+k^{2}}{q_{k}}}&\left(\frac{\sigma}{q_{k}}-n\right)\bar{b}_{n,\sigma}(-z)^{-1-\frac{\sigma}{q_{k}}}\\[1em]
        (-n)\bar{A}_{n,\sigma}(-z)^{-1}&(-n)\widetilde{A}_{n,\sigma}(-z)^{-1}&\left(\frac{\sigma+k^{2}}{q_{k}}-n\right)\bar{C}_{n,\sigma}(-z)^{-1-\frac{\sigma+k^{2}}{q_{k}}}&\left(\frac{\sigma}{q_{k}}-n\right)\bar{B}_{n,\sigma}(-z)^{-1-\frac{\sigma}{q_{k}}}
    \end{bmatrix}.
\end{equation*}}
The role of $J_{\sigma}$ is to determine the leading order behavior of $Y_{\sigma}$. We have the following lemma computing the determinant of $Y_{\sigma}$ explicitly:
\begin{lemma}
\label{lemma: determinant of Y}
    For $\sigma\in\mathbb{I}_{[-1-\eta,-\eta]}$, the determinant of $Y_{\sigma}$ is: \begin{equation}
       \det{Y}_{\sigma} = \left(1+\frac{\sigma}{q_{k}}\right)\frac{k^{2}+\sigma}{q_{k}}\frac{\sigma}{q_{k}}(-z)^{-\frac{k^{2}+2\sigma}{q_{k}}-2}.
    \end{equation}
\end{lemma}
In particular, this means the fundamental matrix $Y_{\sigma}$ is not invertible if and only if $\sigma = -q_{k}$.
\begin{proof}
We can rewrite the equation $\mathcal{L}_{p}(f,g) = 0$ as \begin{equation*}
    \frac{d}{dz}\begin{bmatrix}
        f\\g\\[.5em]\frac{df}{dz}\\[.5em]\frac{dg}{dz}
    \end{bmatrix} = L_{p}(z)\begin{bmatrix}
        f\\g\\[.5em]\frac{df}{dz}\\[.5em]\frac{dg}{dz}
    \end{bmatrix}.
\end{equation*}
Then by the Jacobi formula, we have\begin{equation*}
    \frac{d}{dz}\det{Y}_{\sigma}(z) = \det{Y}_{\sigma}(z)\text{Tr}L_{p}.
\end{equation*}
For $\text{Tr}L_{p}$, we can compute:\begin{equation*}
    \text{Tr}L_{p} = \frac{1+q_{k}+2\sigma}{q_{k}(-z)} = \frac{2}{(-z)}+\frac{k^{2}+2\sigma}{q_{k}(-z)}.
\end{equation*}
Directly integrating the Jacobi formula, we have\begin{equation*}
    \det{Y}_{\sigma}(z) = C(-z)^{-2-\frac{k^{2}+2\sigma}{q_{k}}},
\end{equation*}
where $C$ is a constant. To compute the constant $C$, viewing each element of $Y_{\sigma}$ as an infinite sum, we have that the leading order of $\det{Y}_{\sigma}$ is $\det{J}_{\sigma}$. Then we have\begin{equation*}
    \det{J}_{\sigma} = \left(1+\frac{\sigma}{q_{k}}\right)\frac{k^{2}+\sigma}{q_{k}}\frac{\sigma}{q_{k}}(-z)^{-\frac{k^{2}+2\sigma}{q_{k}}-2}.
\end{equation*}
Then we have\begin{equation*}
    C = \lim_{z\rightarrow0}(-z)^{2+\frac{k^{2}+2\sigma}{q_{k}}}\det{Y}_{\sigma}(z) = \lim_{z\rightarrow0}(-z)^{2+\frac{k^{2}+2\sigma}{q_{k}}}\det{J}_{\sigma}(z) = \left(1+\frac{\sigma}{q_{k}}\right)\frac{k^{2}+\sigma}{q_{k}}\frac{\sigma}{q_{k}}.
\end{equation*}
Hence, we have \begin{equation*}
    \det{Y}_{\sigma} = \left(1+\frac{\sigma}{q_{k}}\right)\frac{k^{2}+\sigma}{q_{k}}\frac{\sigma}{q_{k}}(-z)^{-\frac{k^{2}+2\sigma}{q_{k}}-2}.
\end{equation*}
\end{proof}
Since the fundamental matrix $Y_{\sigma}$ is not invertible when $\sigma = -q_{k}$, viewing each element in $Y_{\sigma}^{-1}$ as a function of $\sigma$ for $\sigma\in\mathbb{I}_{[-1-\eta,-\eta]}$, we conduct a careful analysis on the $\sigma$-dependence of each element in $Y_{\sigma}^{-1}$ in the following lemma.
\begin{lemma}
\label{lemma: carefully study of the inverse}
    For $\sigma = -q_{k}$, we have \begin{align*}
        [Y_{-q_{k}}]_{j1} = [Y_{-q_{k}}]_{j3} &= 0,\\(Y_{-q_{k}})_{l2}(z)[Y_{-q_{k}}]_{j2}(z^{\prime})+(Y_{-q_{k}})_{l4}(z)[Y_{-q_{k}}]_{j4}(z^{\prime}) &=0,\quad  j,l = 1,2,3,4. 
    \end{align*}
Moreover, the functions $$(Y_{\sigma}^{-1})_{ij},\quad i = 1,3,\quad j = 1,2,3,4$$are holomorphic functions of $\sigma$ in the region $\mathbb{I}_{[-1-\eta,-\eta]}$, and the functions $$(Y_{\sigma}^{-1})_{2j},\ (Y_{\sigma}^{-1})_{4j},\quad j = 1,2,3,4$$
are meromorphic functions of $\sigma$ in the region $\mathbb{I}_{[-1-\eta,-\eta]}$ with a pole at $\sigma = -q_{k}$. Finally, we have the following expansion of $Y_{\sigma}^{-1}$: \begin{equation}
    Y_{\sigma}^{-1} = J_{\sigma}^{-1}+\sum_{n = 1}^{\infty}k^{n-1}(Y_{\sigma}^{-1})_{n}(-z)^{n},
\end{equation}
where $J_{\sigma}^{-1}$ is the inverse matrix of $J_{\sigma}$, taking the form of \begin{equation}
    J_{\sigma}^{-1} = \begin{bmatrix}
        1&0&\frac{q_{k}}{\sigma+k^{2}}z&0\\[1em]-\frac{kq_{k}}{\sigma\left(1+\frac{\sigma}{q_{k}}\right)}&\frac{1}{1+\frac{\sigma}{q_{k}}}&\frac{kq_{k}^{2}+kq_{k}}{\sigma\left(\sigma+k^{2}\right)\left(1+\frac{\sigma}{q_{k}}\right)}(-z)&\frac{q_{k}}{\sigma}\frac{1}{1+\frac{\sigma}{q_{k}}}z\\[1em]
        0&0&\frac{q_{k}}{\sigma+k^{2}}(-z)^{\frac{k^{2}+\sigma}{q_{k}}+1}&0\\[1em]
        0&0&-\frac{q_{k}}{k\sigma}(-z)^{\frac{\sigma}{q_{k}}+1}&\frac{q_{k}}{\sigma}(-z)^{\frac{\sigma}{q_{k}}+1}
    \end{bmatrix}
\end{equation}
and $(Y_{\sigma}^{-1})_{n}$ is a matrix, taking the form of {\footnotesize\begin{equation}
    (Y_{\sigma}^{-1})_{n} = \begin{bmatrix}
        \bar{Y}_{11,n}(\sigma,k)&\bar{Y}_{12,n}(\sigma,k)&\bar{Y}_{13,n}(\sigma,k)(-z)&\bar{Y}_{14,n}(\sigma,k)(-z)\\[1em]
        \bar{Y}_{21,n}(\sigma,k)&\bar{Y}_{22,n}(\sigma,k)&\bar{Y}_{23,n}(\sigma,k)(-z)&\bar{Y}_{24,n}(\sigma,k)(-z)\\[1em]
        \bar{Y}_{31,n}(\sigma,k)(-z)^{\frac{\sigma+k^{2}}{q_{k}}}&\bar{Y}_{32,n}(\sigma,k)(-z)^{\frac{\sigma+k^{2}}{q_{k}}}&\bar{Y}_{33,n}(\sigma,k)(-z)^{\frac{\sigma+k^{2}}{q_{k}}+1}&\bar{Y}_{34,n}(\sigma,k)(-z)^{\frac{\sigma+k^{2}}{q_{k}}+1}\\[1em]
        \bar{Y}_{41,n}(\sigma,k)(-z)^{\frac{\sigma}{q_{k}}}&\bar{Y}_{42,n}(\sigma,k)(-z)^{\frac{\sigma}{q_{k}}}&\bar{Y}_{43,n}(\sigma,k)(-z)^{\frac{\sigma}{q_{k}}+1}&\bar{Y}_{44,n}(\sigma,k)(-z)^{\frac{\sigma}{q_{k}}+1}
    \end{bmatrix},
\end{equation}}
where $\bar{Y}_{1i,n},\bar{Y}_{3i,n}$ are holomorphic functions of $\sigma$ in the region $\mathbb{I}_{[-1-\eta,-\eta]}$ and are uniformly bounded independent of $\sigma$ and $k$, $\bar{Y}_{2i,n},\bar{Y}_{4i,n}$ are meromorphic functions of $\sigma$ in the region $\mathbb{I}_{[-1-\eta,-\eta]}$ with a single pole at $\sigma = -q_{k}$, and $(q_{k}+\sigma)\bar{Y}_{2i,n},(q_{k}+\sigma)\bar{Y}_{4i,n}$ are uniformly bounded independent of $\sigma$ and $k$.
\end{lemma}
\begin{proof}
Since $\mathbf{\Psi}_{+,-q_{k}}$ is linearly dependent to $\mathbf{\Psi}_{-,-q_{k}}$, immediately we have $$[Y_{-q_{k}}]_{j1} = [Y_{-q_{k}}]_{j3} = 0.$$

Assume $(Y_{-q_{k}})_{i2} = \kappa (Y_{-q_{k}})_{i4}$, then the co-factors\begin{equation*}
    -\kappa[Y_{-q_{k}}]_{i2}(z) = [Y_{-q_{k}}]_{i4}(z).
\end{equation*}
Hence, we have \begin{align*}
    &(Y_{-q_{k}})_{l2}(z)[Y_{-q_{k}}]_{j2}(z^{\prime})+(Y_{-q_{k}})_{l4}(z)[Y_{-q_{k}}]_{j4}(z^{\prime}) \\=&\kappa (Y_{-q_{k}})_{l4}(z)[Y_{-q_{k}}]_{j2}(z^{\prime})-\kappa(Y_{-q_{k}})_{l4}(z)[Y_{-q_{k}}]_{j2}(z^{\prime})\\=&0.
\end{align*}
In view of the cofactor formulation of $Y_{\sigma}^{-1}$ \begin{equation*}
    (Y_{\sigma}^{-1})_{ij} = \frac{[Y_{\sigma}]_{ji}}{\det{Y}_{\sigma}},
\end{equation*}
since $[Y_{\sigma}]_{j1}$ and $[Y_{\sigma}]_{j3}$ are degenerate at $\sigma = -q_{k}$, and $\det{Y}_{\sigma}$ has a single zero at $\sigma = -q_{k}$, we have \begin{equation*}
    (Y_{-q_{k}}^{-1})_{1i} = \frac{\partial_{\sigma}[Y_{\sigma}]_{i1}}{\partial_{\sigma}\det{Y}_{\sigma}}\Bigg|_{\sigma=-q_{k}}<\infty,\quad (Y_{-q_{k}}^{-1})_{3i} = \frac{\partial_{\sigma}[Y_{\sigma}]_{i3}}{\partial_{\sigma}\det{Y}_{\sigma}}\Biggl|_{\sigma = -q_{k}}<\infty.
\end{equation*}
Hence, the functions $(Y_{\sigma}^{-1})_{1i},(Y_{\sigma}^{-1})_{3i}$ can be extended to holomorphic functions in $\mathbb{I}_{[-1-\eta,-\eta]}$. By the co-factor formulation of the inverse again, we have $J_{\sigma}^{-1}$ is the leading order behavior of $Y_{\sigma}^{-1}$. By the expansion property of $\{\mathbf{r}_{+,\sigma},\mathbf{\Psi}_{+,\sigma},\mathbf{r}_{-,\sigma},\mathbf{\Psi}_{-,\sigma}\}$, we can show the expansion of $Y_{\sigma}^{-1}$. This concludes the proof. 
\end{proof}
Although many of the elements in $Y_{\sigma}^{-1}$ are only meromorphic functions of $\sigma$, noticing that $Y_{\sigma}(z)Y_{\sigma}^{-1}(z) = I$, we can still prove the holomorphy of the elements in $Y_{\sigma}(z)Y_{\sigma}^{-1}(z^{\prime})$. 
\begin{lemma}
\label{lemma: estimate on YYinverse}
Let $-1\leq z\leq z^{\prime}< 0$. Each element in $Y_{\sigma}(z)Y_{\sigma}^{-1}(z^{\prime})$ can be extended to a holomorphic function of $\sigma$ in the region $\mathbb{I}_{[-1-\eta,-\eta]}$. The element $\left(Y_{\sigma}(z)Y_{\sigma}^{-1}(z^{\prime})\right)_{ij}$ can be written as \begin{align}
    \left(Y_{\sigma}(z)Y_{\sigma}^{-1}(z^{\prime})\right)_{ij} = &T_{ij}(\sigma,k,z,z^{\prime})+\widetilde{T}_{ij}(\sigma,k,z,z^{\prime})(-z)^{-\frac{\sigma}{q_{k}}}(-z^{\prime})^{\frac{\sigma}{q_{k}}+1},\quad  i = 1,2,\label{estimate for Yyinverse}\\\left(Y_{\sigma}(z)Y_{\sigma}^{-1}(z^{\prime})\right)_{ij} = &T_{ij}(\sigma,k,z,z^{\prime})+\widetilde{T}_{ij}(\sigma,k,z,z^{\prime})(-z)^{-\frac{\sigma}{q_{k}}-1}(-z^{\prime})^{\frac{\sigma}{q_{k}}+1},\quad  i = 3,4,\label{estimate on yyinverse 2}.
\end{align}
where $T_{ij}$ and $\widetilde{T}_{ij}$ are holomorphic functions of $\sigma$ in the region $\mathbb{I}_{[-1-\eta,-\eta]}$ and smooth functions of $z$ and $z^{\prime}$ for fixed $\sigma$, with the estimate: \begin{align}
    &\left\vert\partial_{z^{\prime}}^{p}T_{ij}\right\vert\lesssim \left(1+\vert\log(-z^{\prime})\vert\right),\quad \left\vert\partial_{z^{\prime}}^{p}\widetilde{T}_{ij}\right\vert\lesssim \frac{1}{\vert\sigma\vert}\left(1+\vert\log(-z^{\prime})\vert\right) ,\quad i = 1,2,\ j = 1,2,\\&
   \left\vert\partial_{z^{\prime}}^{p}T_{ij}\right\vert\lesssim \frac{1}{\vert\sigma\vert}\left(1+\vert\log(-z^{\prime})\vert\right),\quad \left\vert\partial_{z^{\prime}}^{p}\widetilde{T}_{ij}\right\vert\lesssim \frac{1}{\vert\sigma\vert}\left(1+\vert\log(-z^{\prime})\vert\right) ,\quad i = 1,2,\ j = 3,4, \\&\left\vert\partial_{z^{\prime}}^{p}T_{ij}\right\vert\lesssim \left(1+\vert\log(-z^{\prime})\vert\right),\quad \left\vert\partial_{z^{\prime}}^{p}\widetilde{T}_{ij}\right\vert\lesssim \left(1+\vert\log(-z^{\prime})\vert\right) ,\quad i = 3,4,\ j = 1,2,\\&\left\vert\partial_{z^{\prime}}^{p}T_{ij}\right\vert\lesssim \frac{1}{\vert\sigma\vert}\left(1+\vert\log(-z^{\prime})\vert\right),\quad \left\vert\partial_{z^{\prime}}^{p}\widetilde{T}_{ij}\right\vert\lesssim \left(1+\vert\log(-z^{\prime})\vert\right) ,\quad i = 3,4,\ j = 3,4.
\end{align}
\end{lemma}
\begin{proof}
Using the co-factor formulation of the inverse, we have \begin{align*}
&\left(Y_{\sigma}(z)Y_{\sigma}^{-1}(z^{\prime})\right)_{ij} \\=&
\sum_{l = 1}^{4}\left(Y_{
\sigma
}\right)_{il}(z)\left(Y_{\sigma}^{-1}\right)_{lj}(z^{\prime})\\=& \sum_{l = 1}^{4}\frac{\left(Y_{\sigma}\right)_{il}(z)[Y_{\sigma}]_{jl}(z^{\prime})}{\det{Y}_{\sigma}(z^{\prime})}\\=&
\frac{(Y_{\sigma})_{i1}(z)[Y_{\sigma}]_{j1}(z^{\prime})+(Y_{\sigma})_{i3}(z)[Y_{\sigma}]_{j3}(z^{\prime})+(Y_{\sigma})_{i2}(z)[Y_{\sigma}]_{j2}(z^{\prime})+(Y_{\sigma})_{i4}(z)[Y_{\sigma}]_{j4}(z^{\prime})}{\det{Y}_{\sigma}(z^{\prime})}\\=&(Y_{\sigma})_{i1}(z)(Y_{\sigma}^{-1})_{1j}(z^{\prime})+(Y_{\sigma})_{i3}(z)(Y_{\sigma}^{-1})_{3j}(z^{\prime})+\frac{(Y_{\sigma})_{i2}(z)[Y_{\sigma}]_{j2}(z^{\prime})+(Y_{\sigma})_{i4}(z)[Y_{\sigma}]_{j4}(z^{\prime})}{\det Y_{\sigma}(z^{\prime})}.
\end{align*}
By Lemma \ref{lemma: carefully study of the inverse}, we have\begin{align*}
    [Y_{-q_{k}}]_{j1}  = [Y_{-q_{k}}]_{j3} &= 0,\\
    (Y_{-q_{k}})_{i2}(z)[Y_{-q_{k}}]_{j2}(z^{\prime})+(Y_{-q_{k}})_{i4}(z)[Y_{-q_{k}}]_{j4}(z^{\prime})&= 0.
\end{align*}
Hence, we have $\left(Y_{\sigma}(z)Y_{\sigma}^{-1}(z^{\prime})\right)_{ij}$ is well-defined at $\sigma=-q_{k}$ for any $i,j = 1,2,3,4$. Then, we have that $(Y_{\sigma}(z)Y_{\sigma}^{-1}(z^{\prime}))_{ij}$ can be extended to a holomorphic function of $\sigma$ in the region $\sigma\in\mathbb{I}_{[-1-\eta,-\eta]}$.

To estimate $\left(Y_{\sigma}(z)Y_{\sigma}^{-1}(z^{\prime})\right)_{ij}$, we first estimate $(Y_{\sigma})_{i1}(Y_{\sigma}^{-1})_{1j}$ and $(Y_{\sigma})_{i3}(Y_{\sigma}^{-1})_{3j}$. Since $[Y_{\sigma}]_{j1}$ and $[Y_{\sigma}]_{j3}$ are degenerate at $\sigma = -q_{k}$, using the expansions of $\{\mathbf{r}_{+,\sigma},\mathbf{\Psi}_{+,\sigma},\mathbf{r}_{-,\sigma},\mathbf{\Psi}_{-,\sigma}\}$, we have the following expansions of $\left(Y_{\sigma}^{-1}\right)_{1j}$ and $\left(Y_{\sigma}^{-1}\right)_{3j}$:\begin{align*}
\left(Y_{\sigma}^{-1}\right)_{1j}(z) &= \left(J_{\sigma}^{-1}\right)_{1j}(z)+\sum_{n = 1}^{\infty}k^{n-1}\bar{Y}_{1j,n}(\sigma,k)(-z)^{n},\quad j = 1,2,\\
\left(Y_{\sigma}^{-1}\right)_{1j}(z) &=(J_{\sigma}^{-1})_{1j}(z)+\sum_{n = 1}^{\infty}k^{n-1}\bar{Y}_{1j,n}(\sigma,k)(-z)^{n+1},\quad j = 3,4, \\
\left(Y_{\sigma}^{-1}\right)_{3j}(z)& = \left(J_{\sigma}^{-1}\right)_{3j}(z)+\sum_{n = 1}^{\infty}k^{n-1}\bar{Y}_{3j,n}(\sigma,k)(-z)^{\frac{\sigma+k^{2}}{q_{k}}+n},\quad j = 1,2,\\
(Y_{\sigma}^{-1})_{3j}(z)& = (J_{\sigma}^{-1})_{3j}(z)+\sum_{n = 1}^{\infty}k^{n-1}\bar{Y}_{3j,n}(\sigma,k)(-z)^{\frac{\sigma+k^{2}}{q_{k}}+1+n},\quad j = 3,4.
\end{align*}
Then, for $\left(Y_{\sigma}\right)_{i1}(z)\left(Y_{\sigma}^{-1}\right)_{1j}(z^{\prime})$, we have 
\begin{align*}
    (Y_{\sigma})_{i1}(z)\left(Y_{\sigma}^{-1}\right)_{1j}(z^{\prime}) =&T_{1}^{(ij)}(\sigma,k,z,z^{\prime}),\quad i = 1,2,3,4,\quad j = 1,2,\\
    (Y_{\sigma})_{i1}(z)\left(Y_{\sigma}^{-1}\right)_{1j}(z^{\prime}) = &T_{1}^{(ij)}(\sigma,k,z,z^{\prime})(-z^{\prime}),\quad i = 1,2,3,4,\quad j = 3,4.
\end{align*}
For $\left(Y_{\sigma}\right)_{i3}(z)\left(Y_{\sigma}^{-1}\right)_{3j}(z^{\prime})$, we have\begin{equation}
\begin{aligned}
    &\left(Y_{\sigma}\right)_{i3}(z)\left(Y_{\sigma}^{-1}\right)_{3j}(z^{\prime}) = T_{3}^{(ij)}(\sigma,k,z,z^{\prime})(-z)^{-\frac{k^{2}+\sigma}{q_{k}}}(-z^{\prime})^{\frac{k^{2}+\sigma}{q_{k}}+1},\quad i = 1,2,\quad j = 1,2,3,4,\\&\left(Y_{\sigma}\right)_{i3}(z)\left(Y_{\sigma}^{-1}\right)_{3j}(z^{\prime}) = T_{3}^{(ij)}(\sigma,k,z,z^{\prime})(-z)^{-\frac{k^{2}+\sigma}{q_{k}}-1}(-z^{\prime})^{\frac{k^{2}+\sigma}{q_{k}}+1},\quad i = 3,4,\quad j = 1,2,3,4.
\end{aligned}
\label{def of T2}
\end{equation}
Tracking the expansions, we have that $T_{l}^{(ij)}(\sigma,k,z,z^{\prime})$ above are holomorphic functions of $\sigma$ in the region $\mathbb{I}_{[-1-\eta,-\eta]}$ and smooth functions in $z$ and $z^{\prime}$, with the following estimates for $p\geq 0$
\begin{align*}
&\left\vert\partial_{z^{\prime}}^{p}T_{1}^{(i1)}(\sigma,k,z,z^{\prime})\right\vert\leq C,\quad i = 1,3,\\&
\left\vert\partial_{z^{\prime}}^{p}T_{1}^{(ij)}(\sigma,k,z,z^{\prime})\right\vert\leq \frac{C}{\vert\sigma\vert},\quad i = 1,2,3,4,\ j = 2,3,4 \text{ or } i = 2,4,\ j = 1,2,3,4
\\&
\left\vert\partial_{z^{\prime}}^{p}T_{3}^{(ij)}(\sigma,k,z,z^{\prime})\right\vert\leq \frac{C}{k},\quad i = 3,4,\quad j = 1,2,3,4,\\&
\left\vert\partial_{z^{\prime}}^{p}T_{3}^{(ij)}(\sigma,k,z,z^{\prime})\right\vert\leq\frac{C}{k\vert\sigma\vert},\quad i = 1,2,\ j =1,2,3,4.
\end{align*}
Next, we estimate $(Y_{\sigma})_{i2}(Y_{\sigma}^{-1})_{2j}+(Y_{\sigma})_{i4}(Y_{\sigma}^{-1})_{4j}$. Similar to the above argument, we can write down the expansions of $(Y_{\sigma})_{i2}(z)[Y_{\sigma}]_{j2}(z^{\prime})$:
\begin{align*}
    \left(Y_{\sigma}\right)_{i2}(z)\left[Y_{\sigma}\right]_{j2}(z^{\prime}) =\left(Y_{\sigma}\right)_{i2}(z)\left(Y_{\sigma}^{-1}\right)_{2j}(z^{\prime})\det{Y}_{\sigma}(z^{\prime})=:\sum_{n,m = 0}^{\infty}\hat{\gamma}_{nm}^{(ij)}(\sigma,k)(-z)^{n}(-z^{\prime})^{-\frac{k^{2}+2\sigma}{q_{k}}-2+m},
\end{align*}
 and the expansions of $(Y_{\sigma})_{i4}(z)\left[Y_{\sigma}\right]_{j4}(z^{\prime})$:\begin{align*}
     &\left(Y_{\sigma}\right)_{i4}(z)\left[Y_{\sigma}\right]_{j4}(z^{\prime}) = \sum_{n,m = 0}^{\infty}\check{\gamma}_{nm}^{(ij)}(\sigma,k)(-z)^{-\frac{\sigma}{q_{k}}+n}(-z^{\prime})^{-\frac{k^{2}+\sigma}{q_{k}}+m-1},\quad i = 1,2,\\&\left(Y_{\sigma}\right)_{i4}(z)\left[Y_{\sigma}\right]_{j4}(z^{\prime})= \sum_{n,m = 0}^{\infty}\check{\gamma}_{nm}^{(ij)}(\sigma,k)(-z)^{-\frac{\sigma}{q_{k}}+n-1}(-z^{\prime})^{-\frac{k^{2}+\sigma}{q_{k}}+m-1},\quad i = 3,4,
 \end{align*}
where $\hat{\gamma}_{nm}^{(ij)}(\sigma,k)$ and $\check{\gamma}_{nm}^{(ij)}(\sigma,k)$ are holomorphic functions of $\sigma$ in the region $\mathbb{I}_{[-1-\eta,-\eta]}$ with the following estimates for $p = 0,1$:\begin{equation}
\label{eq: estimate on the singular components}
\begin{aligned}
&\left\vert\partial_{\sigma}^{p}\hat{\gamma}_{nm}^{(ij)}\right\vert\lesssim \vert\sigma\vert^{2-p}k^{n+m},\quad (i,j)\neq (2,2),\quad
\left\vert\partial_{\sigma}\hat{\gamma}_{nm}^{(22)}\right\vert\lesssim \vert\sigma\vert^{3-p}k^{n+m},\\&
\hat{\gamma}_{00}^{(ij)} = 0,\quad i = 1,3,4,\quad \check{\gamma}_{00}^{(23)} = -\frac{q_{k}}{k\sigma}\left(1+\frac{\sigma}{q_{k}}\right)\frac{k^{2}+\sigma}{q_{k}}\frac{\sigma}{q_{k}},\quad \check{\gamma}_{00}^{(43)} = -\frac{1}{k}\left(1+\frac{\sigma}{q_{k}}\right)\frac{k^{2}+\sigma}{q_{k}}\frac{\sigma}{q_{k}},\\&
\left\vert\partial_{\sigma}^{p}\check{\gamma}_{nm}^{(23)}\right\vert\lesssim \vert\sigma\vert^{2-p}k^{n+m-1},\qquad
\left\vert\partial_{\sigma}^{p}\check{\gamma}_{nm}^{(43)}\right\vert\lesssim \vert\sigma\vert^{3-p}k^{n+m-1},\\&\check{\gamma}_{00}^{(ij)} = 0,\quad (i,j)\neq(2,3)\ \text{and}\ (i,j)\neq (4,3),\\&\vert\partial_{\sigma}^{p}\check{\gamma}_{nm}^{(ij)}\vert\lesssim \vert\sigma\vert^{2-p}k^{n+m-1},\quad i = 1,2,\quad \vert\partial_{\sigma}^{p}\check{\gamma}_{nm}^{(ij)}\vert\lesssim \vert\sigma\vert^{3-p}k^{n+m-1},\quad i = 3,4.
\end{aligned}
\end{equation}
Since $(Y_{\sigma})_{i2}(z)[Y_{\sigma}]_{j2}(z^{\prime})+(Y_{\sigma})_{i4}(z)[Y_{\sigma}]_{j4}(z^{\prime}) = 0$ when $\sigma = -q_{k}$, we have \begin{align*}
    &\hat{\gamma}_{0m}^{(ij)}(-q_{k},k) = 0,\quad \hat{\gamma}_{(n+1)m}^{(ij)}(-q_{k},k)+\check{\gamma}_{nm}^{(ij)}(-q_{k},k) = 0,\quad i = 1,2,\\&\hat{\gamma}_{nm}^{(ij)}(-q_{k},k)+\check{\gamma}_{nm}^{(ij)}(-q_{k},k) = 0,\quad i = 3,4.
\end{align*}
Hence, for $i = 3,4$, we have \begin{align*}
&(Y_{\sigma})_{i2}(Y_{\sigma}^{-1})_{2j}+(Y_{\sigma})_{i4}(Y_{\sigma}^{-1})_{4j}\\[1em]=
&\frac{(Y_{\sigma})_{i2}(z)[Y_{\sigma}]_{j2}(z^{\prime})+(Y_{\sigma})_{i4}(z)[Y_{\sigma}]_{j4}(z^{\prime})}{\det{Y}_{\sigma}(z^{\prime})}\\=&
\frac{\sum_{n,m = 0}^{\infty}\hat{\gamma}_{nm}^{(ij)}(\sigma,k)(-z)^{n}(-z^{\prime})^{m}+\check{\gamma}_{nm}^{(ij)}(\sigma,k)(-z)^{-\frac{\sigma}{q_{k}}+n-1}(-z^{\prime})^{\frac{\sigma}{q_{k}}+1+m}}{\left(1+\frac{\sigma}{q_{k}}\right)\left(\frac{k^{2}+\sigma}{q_{k}}\right)\left(\frac{\sigma}{q_{k}}\right)}\\=&:
\frac{1}{1+\frac{\sigma}{q_{k}}}\left(\widetilde{T}_{2}^{(ij)}(\sigma,k,z,z^{\prime})+\widetilde{T}_{4}^{(ij)}(\sigma,k,z,z^{\prime})(-z)^{-\frac{\sigma}{q_{k}}-1}(-z^{\prime})^{\frac{\sigma}{q_{k}}+1}\right)
\end{align*}
with \begin{align*}
&\widetilde{T}_{2}^{(ij)} =\frac{1}{\frac{k^{2}+\sigma}{q_{k}}\frac{\sigma}{q_{k}}} \sum_{n,m = 0}\hat{\gamma}_{nm}^{(ij)}(-z)^{n}(-z^{\prime})^{m},\\& \widetilde{T}_{4}^{(ij)} = \frac{1}{\frac{k^{2}+\sigma}{q_{k}}\frac{\sigma}{q_{k}}}\sum_{n,m = 0}^{\infty}\check{\gamma}_{nm}^{(ij)}(-z)^{n}(-z^{\prime})^{m},\\&\widetilde{T}_{2}^{(ij)}(-q_{k},k,z,z^{\prime})+\widetilde{T}_{4}^{(ij)}(-q_{k},k,z,z^{\prime}) = 0,\\&
\left\vert\partial_{z^{\prime}}^{p^{\prime}}\partial_{\sigma}^{p}\widetilde{T}_{2}^{(ij)}\right\vert\leq C\vert\sigma\vert^{-p},\quad
\left\vert\partial_{z^{\prime}}^{p^{\prime}}\partial_{\sigma}^{p}\widetilde{T}_{4}^{(ij)}\right\vert\leq C\vert\sigma\vert^{1-p}k^{-1},\quad p = 0,1,\quad p^{\prime}\geq0.
\end{align*}
When $\sigma\rightarrow-q_{k}$, we have \begin{align*}
    &\lim_{\sigma\rightarrow-q_{k}}\frac{\widetilde{T}_{2}^{(ij)}+\widetilde{T}_{4}^{(ij)}}{1+\frac{\sigma}{q_{k}}} = \frac{1}{\frac{k^{2}+\sigma}{q_{k}}\frac{\sigma}{q_{k}}}\sum_{n,m = 0}\left({\partial_{\sigma}\hat{\gamma}_{nm}^{(ij)}+\partial_{\sigma}\check{\gamma}_{nm}^{(ij)}}\right)(-z)^{n}(-z^{\prime})^{m}<\infty,\\&
    \lim_{\sigma\rightarrow-q_{k}}\frac{\widetilde{T}_{4}^{(ij)}\left((-z)^{-\frac{\sigma}{q_{k}}-1}(-z^{\prime})^{\frac{\sigma}{q_{k}}+1}-1\right)}{1+\frac{\sigma}{q_{k}}}  = \widetilde{T}_{4}^{(ij)}(-q_{k},k,z,z^{\prime})(-z)^{-\frac{\sigma}{q_{k}}-1}(-z^{\prime})^{\frac{\sigma}{q_{k}}+1}\log\left(\frac{z^{\prime}}{z}\right).
\end{align*}
Similarly, for $i = 1,2$, we have \begin{align*}
    &(Y_{\sigma})_{i2}(Y_{\sigma}^{-1})_{2j}+(Y_{\sigma})_{i4}(Y_{\sigma}^{-1})_{4j}\\=&
    \frac{\sum_{n,m = 0}^{\infty}\hat{\gamma}_{nm}^{(ij)}(-z)^{n}(-z^{\prime})^{m}+\check{\gamma}_{nm}^{(ij)}(-z)^{-\frac{\sigma}{q_{n}}+n}(-z^{\prime})^{\frac{\sigma}{q_{k}}+1+m}}{\left(1+\frac{\sigma}{q_{k}}\right)\frac{k^{2}+\sigma}{q_{k}}\frac{\sigma}{q_{k}}}\\=&:\frac{1}{1+\frac{\sigma}{q_{k}}}\left(\widetilde{T}_{2}^{(ij)}(\sigma,k,z,z^{\prime})+\widetilde{T}_{4}^{(ij)}(\sigma,k,z,z^{\prime})(-z)^{-\frac{\sigma}{q_{k}}}(-z^{\prime})^{\frac{\sigma}{q_{k}}+1}\right)
\end{align*}
with\begin{align*}
    &\widetilde{T}_{2}^{(ij)}(-q_{k},k,z,z^{\prime})+\widetilde{T}_{4}^{(ij)}(-q_{k},k,z,z^{\prime}) = 0,\\&
    \left\vert\partial_{z^{\prime}}^{p^{\prime}}\partial_{\sigma}^{p}\widetilde{T}_{2}^{(ij)}\right\vert\lesssim \vert\sigma\vert^{1-p},\quad \left\vert\partial_{z^{\prime}}^{p^{\prime}}\partial_{\sigma}^{p}\widetilde{T}_{4}^{(ij)}\right\vert\lesssim \vert\sigma\vert^{2-p}k^{-1},\quad p = 0,1,\quad p^{\prime}\geq 0.
\end{align*}
Hence, we have\begin{align*}
\left(Y_{\sigma}(z)Y_{\sigma}^{-1}(z^{\prime})\right)_{ij} = &\left(T_{3}^{(ij)}(\sigma,k,z,z^{\prime})(-z)^{-\frac{k^{2}}{q_{k}}}(-z^{\prime})^{\frac{k^{2}}{q_{k}}}+\frac{\widetilde{T}_{4}^{(ij)}(\sigma,k,z,z^{\prime})}{1+\frac{\sigma}{q_{k}}}\right)(-z)^{-\frac{\sigma}{q_{k}}}(-z^{\prime})^{\frac{\sigma}{q_{k}}+1}\\&+T_{1}^{(ij)}(\sigma,k,z,z^{\prime})+\frac{\widetilde{T}_{2}^{(ij)}(\sigma,k,z,z^{\prime})}{1+\frac{\sigma}{q_{k}}},\quad  i = 1,2.\\
    \left(Y_{\sigma}(z)Y_{\sigma}^{-1}(z^{\prime})\right)_{ij} = &\left(T_{3}^{(ij)}(\sigma,k,z,z^{\prime})(-z)^{-\frac{k^{2}}{q_{k}}}(-z^{\prime})^{\frac{k^{2}}{q_{k}}}+\frac{\widetilde{T}_{4}^{(ij)}(\sigma,k,z,z^{\prime})}{1+\frac{\sigma}{q_{k}}}\right)(-z)^{-\frac{\sigma}{q_{k}}-1}(-z^{\prime})^{\frac{\sigma}{q_{k}}+1}\\&+T_{1}^{(ij)}(\sigma,k,z,z^{\prime})+\frac{\widetilde{T}_{2}^{(ij)}(\sigma,k,z,z^{\prime})}{1+\frac{\sigma}{q_{k}}},\quad  i = 3,4.
\end{align*}
Since $T_{3}^{(ij)}$ and $\widetilde{T}_{4}^{(ij)}$ can only be bounded by $\frac{1}{k}$, to improve the bound, we carefully track the presence of $\frac{1}{k}$-factor in the above expansions. Noticing that only when $(i,j) = (2,3),\ (4,3)$, the $\frac{1}{k}$-factor will appear. We can expand $T_{3}^{(ij)}$ and $\widetilde{T}_{4}^{(ij)}$ as 
\begin{align*}
    T_{3}^{(ij)}(\sigma,k,z,z^{\prime}) &= l_{0}^{(ij)}(\sigma,k)+\sum_{n,m = 0,(n,m)\neq(0,0)}^{\infty}l_{nm}^{(ij)}(\sigma,k)(-z)^{n}(-z^{\prime})^{m},\\
    \frac{\widetilde{T}_{4}^{(ij)}(\sigma,k,z,z^{\prime})}{1+\frac{\sigma}{q_{k}}}&=\frac{\check{\gamma}_{00}^{(ij)}(\sigma,k)}{\left(1+\frac{\sigma}{q_{k}}\right)\frac{\sigma+k^{2}}{q_{k}}\frac{\sigma}{q_{k}}}+\frac{1}{\left(1+\frac{\sigma}{q_{k}}\right)\frac{\sigma+k^{2}}{q_{k}}\frac{\sigma}{q_{k}}}\sum_{n,m = 0,(n,m)\neq(0,0)}^{\infty}\check{\gamma}_{nm}^{(ij)}(\sigma,k)(-z)^{n}(-z^{\prime})^{m}.
\end{align*}
By the estimate on $T_{2}^{(ij)}$ \eqref{def of T2} and the estimate on $\check{\gamma}_{nm}^{(ij)}$ in the above decomposition, only the terms $$l_{0}^{(ij)}(\sigma,k),\quad 
\frac{\check{\gamma}_{00}^{(ij)}(\sigma,k)}{\left(1+\frac{\sigma}{q_{k}}\right)\frac{k^{2}+\sigma}{q_{k}}\frac{\sigma}{q_{k}}}$$ are bounded by $\frac{1}{k}$, while all the rest terms are uniformly bounded independently of $k$. Notice that the above two quantities do not have a pole at $\sigma = -q_{k}$ due to \eqref{eq: estimate on the singular components} for $(i,j) = (2,3)$ or $(4,3)$. Since when $z = z^{\prime}$, the matrix $Y_{\sigma}(z)Y_{\sigma}^{-1}(z^{\prime})$ is reduced to the identity matrix, we have that \begin{equation*}
    l_{0}^{(ij)}(\sigma,k)+\frac{\check{\gamma}_{00}^{(ij)}(\sigma,k)}{\left(1+\frac{\sigma}{q_{k}}\right)\frac{k^{2}+\sigma}{q_{k}}\frac{\sigma}{q_{k}}}
\end{equation*}
is bounded independently of $k$. Then we have\begin{align*}
&l_{0}^{(ij)}(\sigma,k)(-z)^{-\frac{k^{2}}{q_{k}}}(-z^{\prime})^{\frac{k^{2}}{q_{k}}}+\frac{\check{\gamma}_{00}^{(ij)}(\sigma,k)}{\left(1+\frac{\sigma}{q_{k}}\right)\frac{k^{2}+\sigma}{q_{k}}\frac{\sigma}{q_{k}}}\\ =& l_{0}^{(ij)}(\sigma,k)+\frac{\check{\gamma}_{00}^{(ij)}(\sigma,k)}{\left(1+\frac{\sigma}{q_{k}}\right)\frac{k^{2}+\sigma}{q_{k}}\frac{\sigma}{q_{k}}}+l_{0}^{(ij)}(\sigma,k)\left(e^{\frac{k^{2}}{q_{k}}\log\left(\frac{z^{\prime}}{z}\right)}-1\right).
\end{align*}
Since $-1\leq z\leq z^{\prime}\leq 0$, we have \begin{align*}
    \left\vert l_{0}^{(ij)}(\sigma,k)(-z)^{-\frac{k^{2}}{q_{k}}}(-z^{\prime})^{\frac{k^{2}}{q_{k}}}+\frac{\check{\gamma}_{00}^{(ij)}(\sigma,k)}{\left(1+\frac{\sigma}{q_{k}}\right)\frac{k^{2}+\sigma}{q_{k}}\frac{\sigma}{q_{k}}}\right\vert\leq C+Ck(\log(-z)-\log(-z^{\prime}))\leq C(1-\log(-z^{\prime})).
\end{align*}
Putting everything together, we can conclude the proof of this proposition.
\end{proof}

\subsubsection{Deriving the Volterra equation}
\label{sec: deriving the volterra equation}
To construct outgoing solutions to $\mathcal{L}(\hat{r}_{p},\hat{\Psi}_{p}) = 0$, we use the special structures $\frac{d\hat{r}_{p}}{dz}$ and $\frac{d\hat{\Psi}_{p}}{dz}$ to redo the argument in Proposition \ref{prop: general construction of outgoing and ingoing solutions}. In this section, we derive the corresponding Volterra equation.

Let\begin{align*}
\begin{bmatrix}
\hat{r}_{p}\\\hat{\Psi}_{p}
\end{bmatrix}
 =& A_{\sigma}(z)\mathbf{r}_{+,\sigma}+B_{\sigma}(z)\mathbf{\Psi}_{+,\sigma}+C_{\sigma}(z)\mathbf{r}_{-,\sigma}+D_{\sigma}(z)\mathbf{\Psi}_{-,\sigma}\\=&
 \begin{bmatrix}
     \mathbf{r}_{+,\sigma}&\mathbf{\Psi}_{+,\sigma}&\mathbf{r}_{-,\sigma}&\mathbf{\Psi}_{-,\sigma}
 \end{bmatrix}\begin{bmatrix}
     A_{\sigma}\\B_{\sigma}\\C_{\sigma}\\D_{\sigma}
 \end{bmatrix}(z).
\end{align*}
Similarly as in Proposition \ref{prop: general construction of outgoing and ingoing solutions}, we have\begin{equation}
\label{eq: voltella for outgoing}
\begin{aligned}
&\begin{bmatrix}
 \mathbf{r}_{+,\sigma}&\mathbf{\Psi}_{+,\sigma}&\mathbf{r}_{-,\sigma}&\mathbf{\Psi}_{-,\sigma}\\[.5em]
\mathbf{r}_{+,\sigma}^{\prime}&\mathbf{\Psi}_{+,\sigma}^{\prime}&\mathbf{r}_{-,\sigma}^{\prime}&\mathbf{\Psi}_{-,\sigma}^{\prime}
 \end{bmatrix}\frac{d}{dz}\begin{bmatrix}
A_{\sigma}(z)\\B_{\sigma}(z)\\C_{\sigma}(z)\\D_{\sigma}(z)
\end{bmatrix}
 \\=& \frac{1}{q_{k}z}\begin{bmatrix}
0&0&0&0\\
0&0&0&0\\
h_{1}&h_{2}&N_{1}&N_{2}\\h_{3}&h_{4}&N_{3}&N_{4}
 \end{bmatrix}
 \begin{bmatrix}
     \mathbf{r}_{+,\sigma}&\mathbf{\Psi}_{+,\sigma}&\mathbf{r}_{-,\sigma}&\mathbf{\Psi}_{-,\sigma}\\[.5em]\mathbf{r}^{\prime}_{+,\sigma}&\mathbf{\Psi}_{+,\sigma}^{\prime}&\mathbf{r}_{-,\sigma}^{\prime}&\mathbf{\Psi}_{-,\sigma}^{\prime}
     \end{bmatrix}
 \begin{bmatrix}
 A_{\sigma}(z)\\B_{\sigma}(z)\\C_{\sigma}(z)\\D_{\sigma}(z)
 \end{bmatrix}
 \end{aligned}
\end{equation}
To simplify our notation, we denote
\begin{align*}
\mathbf{P}:=\frac{1}{q_{k}z}Y_{\sigma}^{-1}\begin{bmatrix}
    0&0&0&0\\
    0&0&0&0\\
    h_{1}&h_{2}&N_{1}&N_{2}\\h_{3}&h_{4}&N_{3}&N_{4}
\end{bmatrix}
Y_{\sigma}
\end{align*}
Let \begin{equation*}
    \lim_{z\rightarrow0}\begin{bmatrix}
        A_{\sigma}\\B_{\sigma}\\C_{\sigma}\\D_{\sigma}
    \end{bmatrix}(z)  = \begin{bmatrix}
        a_{\sigma}\\b_{\sigma}\\c_{\sigma}\\d_{\sigma}
    \end{bmatrix}.
\end{equation*}

Then integrating \eqref{eq: voltella for outgoing}, we have\begin{equation}
\begin{aligned}
\begin{bmatrix}
A_{\sigma}\\B_{\sigma}\\C_{\sigma}\\D_{\sigma}
\end{bmatrix}(z)
 = \begin{bmatrix}
a_{\sigma}\\b_{\sigma}\\c_{\sigma}\\d_{\sigma}
 \end{bmatrix}
 +\int_{0}^{z}\mathbf{P}
 \begin{bmatrix}
 A_{\sigma}\\B_{\sigma}\\C_{\sigma}\\D_{\sigma}
 \end{bmatrix}(z^{\prime})dz^{\prime}.
\end{aligned}
\label{volterra equation for studying the small expansion}
\end{equation}
Iterating the equation \eqref{volterra equation for studying the small expansion}, we have\begin{equation}
\begin{bmatrix}
A_{\sigma}\\B_{\sigma}\\C_{\sigma}\\D_{\sigma}
\end{bmatrix}(z) = \begin{bmatrix}
a_{\sigma}\\b_{\sigma}\\c_{\sigma}\\d_{\sigma}
\end{bmatrix}
+\sum_{n = 1}^{\infty}\begin{bmatrix}
A_{n,\sigma}\\B_{n,\sigma}\\C_{n,\sigma}\\D_{n,\sigma}
\end{bmatrix}(z)
\end{equation}
where $(A_{n,\sigma},B_{n,\sigma},C_{n,\sigma},D_{n,\sigma})$ takes the form of
\begin{equation}
\begin{bmatrix}
A_{n,\sigma}\\B_{n,\sigma}\\C_{n,\sigma}\\D_{n,\sigma}
\end{bmatrix}(z)
 = \int_{0}^{z}\cdots\int_{0}^{z_{n}}\mathbf{P}(z_{1})\cdots \mathbf{P}(z_{n+1})\begin{bmatrix}
 a_{\sigma}\\b_{\sigma}\\c_{\sigma}\\d_{\sigma}
 \end{bmatrix}dz_{1}\cdots dz_{n+1},
\end{equation}
where $z_{0}: = z$.
If the infinite sum converges, then we have\begin{align*}
    \begin{bmatrix}
    \hat{r}_{p}\\\hat{\Psi}_{p}
    \end{bmatrix}
     &= \begin{bmatrix}
         \mathbf{r}_{+,\sigma}&\mathbf{\Psi}_{+,\sigma}&\mathbf{r}_{-,\sigma}&\mathbf{\Psi}_{-,\sigma}
     \end{bmatrix}
     \begin{bmatrix}
         a_{\sigma}\\b_{\sigma}\\c_{\sigma}\\d_{\sigma}
     \end{bmatrix}
     +\sum_{n = 1}^{\infty}\begin{bmatrix}
         \mathbf{r}_{+,\sigma}&\mathbf{\Psi}_{+,\sigma}&\mathbf{r}_{-,\sigma}&\mathbf{\Psi}_{-,\sigma}
     \end{bmatrix}
     \begin{bmatrix}
     A_{n,\sigma}\\B_{n,\sigma}\\C_{n,\sigma}\\D_{n,\sigma}
     \end{bmatrix}\\&=:
     \sum_{n = 0}^{\infty}\begin{bmatrix}
         \hat{r}_{p,n}\\\hat{\Psi}_{p,n}
     \end{bmatrix},
\end{align*}
where $(\hat{r}_{p,n},\hat{\Psi}_{p,n})$ is defined to be \begin{equation*}
    \begin{bmatrix}
        \hat{r}_{p,n}\\\hat{\Psi}_{p,n}
    \end{bmatrix}
     = \begin{bmatrix}
         \mathbf{r}_{+,\sigma}&\mathbf{\Psi}_{+,\sigma}&\mathbf{r}_{-,\sigma}&\mathbf{\Psi}_{-,\sigma}
     \end{bmatrix}
     \begin{bmatrix}
         A_{n,\sigma}\\B_{n,\sigma}\\C_{n,\sigma}\\D_{n,\sigma}
     \end{bmatrix}.
\end{equation*}
Then, we have \begin{equation*}
    \begin{bmatrix}
        \hat{r}_{p,n}\\\hat{\Psi}_{p,n}\\\hat{r}_{p,n}^{\prime}\\\hat{\Psi}_{p,n}^{\prime}
    \end{bmatrix}
     = Y_{\sigma}(z)\begin{bmatrix}
         A_{n,\sigma}\\B_{n,\sigma}\\C_{n,\sigma}\\D_{n,\sigma}
     \end{bmatrix}(z).
\end{equation*}
\subsubsection{Estimates on the expansion}

First, we estimate \begin{align*}
    \begin{bmatrix}
        \hat{r}_{p,1}\\\hat{\Psi}_{p,1}\\\hat{r}_{p,1}^{\prime}\\\hat{\Psi}_{p,1}^{\prime}
    \end{bmatrix}(z) =& Y_{\sigma}(z)
    \int_{0}^{z}\mathbf{P}(z_{1})\begin{bmatrix}
        a_{\sigma}\\b_{\sigma}\\c_{\sigma}\\d_{\sigma}
    \end{bmatrix}
dz_{1} = \int_{0}^{z}Y_{\sigma}(z)Y_{\sigma}^{-1}(z_{1})\frac{1}{q_{k}z_{1}}\begin{bmatrix}
    0&0&0&0\\0&0&0&0\\h_{1}&h_{2}&N_{1}&N_{2}\\h_{3}&h_{4}&N_{3}&N_{4}
\end{bmatrix}
Y_{\sigma}(z_{1})\begin{bmatrix}
    a_{\sigma}\\b_{\sigma}\\c_{\sigma}\\d_{\sigma}
\end{bmatrix}.
\end{align*}
Recall that for $N_{i}$ and $h_{i}$, we have the estimates \begin{equation*}
    \vert N_{i}\vert\lesssim_{\eta}\frac{k}{\epsilon}(-z)^{1-\epsilon},\quad \vert h_{i}\vert\lesssim_{\eta} \frac{k\vert\sigma\vert}{\epsilon}(-z)^{1-\epsilon},\quad  i  = 1,3,\quad \vert h_{i}\vert\lesssim_{\eta} \frac{k}{\epsilon}(-z)^{1-\epsilon},\quad  i = 2,4.
\end{equation*}
In the construction of the outgoing solutions, we only consider the cases where $(c_{\sigma},d_{\sigma})=(0,0)$ and $(a_{\sigma},b_{\sigma}) = (1,0)$ or $(0,1)$. In both cases, we have \begin{equation*}
    Y_{\sigma}(z_{1})\begin{bmatrix}
        a_{\sigma}\\b_{\sigma}\\0\\0
    \end{bmatrix} = \begin{bmatrix}
        O(1)\\O(\sigma)\\O(1)\\O(1)
    \end{bmatrix}.
\end{equation*} 
Then, we have \begin{equation*}
   \frac{1}{q_{k}z_{1}} \begin{bmatrix}
        0&0&0&0\\0&0&0&0\\h_{1}&h_{2}&N_{1}&N_{2}\\h_{3}&h_{4}&N_{3}&N_{4}
    \end{bmatrix}(z_{1})Y_{\sigma}(z_{1})\begin{bmatrix}
        a_{\sigma}\\b_{\sigma}\\0\\0
    \end{bmatrix}
     = \begin{bmatrix}
         0\\0\\O(k\sigma(-z)^{-\epsilon})\\O(k\sigma(-z)^{-\epsilon})
     \end{bmatrix}.
\end{equation*}
 Then by Lemma \ref{lemma: estimate on YYinverse}, for $i = 1,2$, we have \begin{align*}
    &\left\vert \left(Y_{\sigma}(z)Y_{\sigma}^{-1}(z_{1})\frac{1}{q_{k}z_{1}}\begin{bmatrix}
        0&0&0&0\\0&0&0&0\\h_{1}&h_{2}&N_{1}&N_{2}\\h_{3}&h_{4}&N_{3}&N_{4}
    \end{bmatrix}(z_{1})Y_{\sigma}(z_{1})\begin{bmatrix}
        a_{\sigma}\\b_{\sigma}\\0\\0
    \end{bmatrix}\right)_{i}\right\vert\\\lesssim&\left(1+\left\vert\log(-z_{1})\right\vert\right)(-z_{1})^{-\epsilon}+\left(1+\vert\log(-z^{\prime})\vert\right)(-z)^{-\frac{\sigma}{q_{k}}}(-z_{1})^{\frac{\sigma}{q_{k}}+1-\epsilon}\\\lesssim&
    (-z_{1})^{-2\epsilon}+(-z)^{-\frac{\sigma}{q_{k}}}(-z_{1})^{\frac{\sigma}{q_{k}}+1-2\epsilon}.
\end{align*}
Similarly, for $i = 3,4$, we have \begin{align*}
    &\left\vert \left(Y_{\sigma}(z)Y_{\sigma}^{-1}(z_{1})\frac{1}{q_{k}z_{1}}\begin{bmatrix}
        0&0&0&0\\0&0&0&0\\h_{1}&h_{2}&N_{1}&N_{2}\\h_{3}&h_{4}&N_{3}&N_{4}
    \end{bmatrix}(z_{1})Y_{\sigma}(z_{1})\begin{bmatrix}
        a_{\sigma}\\b_{\sigma}\\0\\0
    \end{bmatrix}\right)_{i}\right\vert\\\lesssim&k\left(1+\left\vert\log(-z_{1})\right\vert\right)(-z_{1})^{-\epsilon}+k\vert\sigma\vert\left(1+\vert\log(-z^{\prime})\vert\right)(-z)^{-\frac{\sigma}{q_{k}}-1}(-z_{1})^{\frac{\sigma}{q_{k}}+1-\epsilon}\\\lesssim&
    k(-z_{1})^{-2\epsilon}+k\vert\sigma\vert(-z)^{-\frac{\sigma}{q_{k}}-1}(-z_{1})^{\frac{\sigma}{q_{k}}+1-2\epsilon}.
\end{align*}
Therefore, integrating the above expression, we have \begin{align*}
   \left\vert \begin{bmatrix}
        \hat{r}_{p,1}\\\hat{\Psi}_{p,1}\\\hat{r}_{p,1}^{\prime}\\\hat{\Psi}_{p,1}^{\prime}
\end{bmatrix}\right\vert\lesssim k(-z)^{1-2\epsilon}.
\end{align*}

To estimate $(\hat{r}_{p,n},\hat{\Psi}_{p,n})$ for $n>1$, since the presence of $\sigma$ in $\mathbf{P}(z)$, we need to carefully track the $\sigma$-dependence in the estimate. We first consider \begin{equation*}
    Y_{\sigma}(z)Y_{\sigma}^{-1}(z_{1})\frac{1}{q_{k}z_{1}}\begin{bmatrix}
        0&0&0&0\\0&0&0&0\\h_{1}&h_{2}&N_{1}&N_{2}\\h_{3}&h_{4}&N_{3}&N_{4}
    \end{bmatrix}\begin{bmatrix}
       p_{1}(z_{1},\sigma)\\ p_{2}(z_{1},\sigma)\\ p_{3}(z_{1},\sigma)\\ p_{4}(z_{1},\sigma)
    \end{bmatrix},
\end{equation*}
where $p_{i}(z,\sigma)$ are functions which are uniformly bounded with bounds independent of $\sigma$. Let $$p = \vert p_{1}\vert+\vert p_{2}\vert+\vert p_{3}\vert+\vert p_{4}\vert.$$ Then we have \begin{align*}
   \left\vert Y_{\sigma}(z)Y_{\sigma}^{-1}(z_{1})\frac{1}{q_{k}z_{1}}\begin{bmatrix}
        0&0&0&0\\0&0&0&0\\
        h_{1}&h_{2}&N_{1}&N_{2}\\h_{3}&h_{4}&N_{3}&N_{4}
    \end{bmatrix}
    \begin{bmatrix}
        p_{1}\\p_{2}\\p_{3}\\p_{4}
    \end{bmatrix}\right\vert = &\left\vert Y_{\sigma}(z)Y_{\sigma}^{-1}(z_{1})\begin{bmatrix}
        0\\0\\O(k\vert\sigma\vert p(-z_{1})^{-\epsilon})\\O(k\vert\sigma\vert p(-z_{1})^{-\epsilon})
    \end{bmatrix}\right\vert\\\lesssim&\begin{bmatrix}
        k p(-z_{1})^{-2\epsilon}+kp(-z)^{-\frac{\sigma}{q_{k}}}(-z_{1})^{\frac{\sigma}{q_{k}}+1-2\epsilon}\\kp(-z_{1})^{-2\epsilon}+kp(-z)^{-\frac{\sigma}{q_{k}}}(-z_{1})^{\frac{\sigma}{q_{k}}+1-2\epsilon}\\
        kp(-z_{1})^{-2\epsilon}+kp(-z)^{-\frac{\sigma}{q_{k}}-1}(-z_{1})^{\frac{\sigma}{q_{k}}+1-2\epsilon}\\
        kp(-z_{1})^{-2\epsilon}+kp(-z)^{-\frac{\sigma}{q_{k}}-1}(-z_{1})^{\frac{\sigma}{q_{k}}+1-2\epsilon}
    \end{bmatrix}.
\end{align*} 
Therefore, arguing inductively, for $-1\leq z\leq z_{1}\leq z_{2}\leq\cdots\leq z_{n}$, we have \begin{align*}
   \left\vert \int_{0}^{z_{1}}\cdots\int_{0}^{z_{n}}Y_{\sigma}(z)\mathbf{P}(z_{1})\mathbf{P}(z_{2})\cdots\mathbf{P}(z_{n})\begin{bmatrix}
        a_{\sigma}\\b_{\sigma}\\0\\0
    \end{bmatrix}dz_{1}\cdots dz_{n}\right\vert\lesssim \frac{1}{n!}k^{n}(-z)^{n(1-\epsilon)}.
\end{align*}
Putting everything together and taking $(a_{\sigma},b_{\sigma}) = (1,0)$ and $(0,1)$, respectively, we have the following proposition constructing outgoing solutions to $\mathcal{L}(\hat{r}_{p},\hat{\Psi}_{p}) = 0$:
\begin{proposition}
\label{prop: construction of the true outgoing solution}
    For $\sigma\in\mathbb{I}_{[-1-\eta,-\eta]}\backslash\{-q_{k}\}$, there exist two linearly independent solutions \begin{equation*}
        \hat{\mathbf{r}}_{out,\sigma} = \begin{bmatrix}
            \hat{\mathbf{r}}_{out,\sigma}^{(1)}\\\hat{\mathbf{r}}_{out,\sigma}^{(2)}
        \end{bmatrix},\quad \hat{\mathbf{\Psi}}_{out,\sigma} = \begin{bmatrix}
            \hat{\mathbf{\Psi}}_{out,\sigma}^{(1)}\\\hat{\mathbf{\Psi}}_{out,\sigma}^{(2)}
        \end{bmatrix}
    \end{equation*}
    to the equation $\mathcal{L}(\hat{r}_{p},\hat{\Psi}_{p}) = 0$, having the following expansions:\begin{align}
        \hat{\mathbf{r}}_{out,\sigma} &= \begin{bmatrix}
            \hat{\mathbf{r}}_{out,\sigma}^{(1)}\\\hat{\mathbf{r}}_{out,\sigma}^{(2)}
        \end{bmatrix}
         = \begin{bmatrix}
             1\\\frac{kq_{k}}{\sigma}
         \end{bmatrix}+\begin{bmatrix}
             Z_{1,\sigma}\\Z_{2,\sigma}
         \end{bmatrix}(z),\label{eq: expansion of the r outgoing solutions to the whole system}\\
        \hat{\mathbf{\Psi}}_{out,\sigma}&=\begin{bmatrix}
            \hat{\mathbf{\Psi}}^{(1)}_{out,\sigma}\\\hat{\mathbf{\Psi}}_{out,\sigma}^{(2)}
        \end{bmatrix} = \begin{bmatrix}
            0\\1+\frac{\sigma}{q_{k}}
        \end{bmatrix}+\begin{bmatrix}
            Z_{3,\sigma}\\Z_{4,\sigma}
        \end{bmatrix}(z),\label{eq: expansion of the psi outgoing solutions to the whole system}
    \end{align}
    where $Z_{i,\sigma}(z)$ are holomorphic functions of $\sigma$ in the region $\mathbb{I}_{[-1-\eta,-\eta]}$ with the estimate:\begin{equation}
        \left\vert Z_{i,\sigma}(z)\right\vert\leq \frac{C}{\epsilon}k(-z)^{1-2\epsilon}.\label{eq: estimate on Z}
    \end{equation}
    Moreover, we have the following estimates for $\hat{\mathbf{r}}_{out,\sigma}$ and $\hat{\mathbf{\Psi}}_{out,\sigma}$:\begin{align}
        &\left\vert\partial_{z}^{p}\hat{\mathbf{r}}_{out,\sigma}\right\vert\leq C\vert\sigma\vert^{p},\quad 0\leq p\leq 2,\label{eq: estimates on the r outgoing solutions}\\&
        \left\vert\partial_{z}^{p}\hat{\mathbf{\Psi}}_{out,\sigma}\right\vert\leq C\vert\sigma\vert^{p}\left(\left\vert1+\frac{\sigma}{q_{k}}\right\vert +k\vert z\vert^{1-2\epsilon}\right),\quad 0\leq p\leq2.\label{eq: estimates on the psi outgoing solutions}
    \end{align}
\end{proposition}
\begin{proof}
The proof follows from the fact that, by choosing $(a_{\sigma},b_{\sigma}) = (1,0), (0,1)$, we have that $(\hat{r}_{p,n},\hat{\Psi}_{p,n})$ constructed above is absolutely summable.
\end{proof}

\subsection{Construction of the Dirichlet solutions}
In this section, we construct the Dirichlet solutions to the equation \eqref{eq: general equation in the construction of ingoing solution}. First, we have the following definition on the Dirichlet solutions:\begin{definition}
For any $\sigma\in\mathbb{I}_{[-1-\eta,-\eta]}$, the Dirichlet solution $\mathbf{f}_{dir,\sigma}$ is the unique solution to \eqref{eq: general equation in the construction of ingoing solution} with the boundary condition\begin{equation*}
(f,g,,f^{\prime},g^{\prime})|_{z = -1} = (0,0,1,0);
\end{equation*}
the Dirichlet solution $\mathbf{g}_{dir,\sigma}$ is the unique solution to \eqref{eq: general equation in the construction of ingoing solution} with the boundary condition \begin{equation*}
(f,g,f^{\prime},g^{\prime})|_{z = -1}= (0,0,0,1).
\end{equation*}

\end{definition} 
The existence of the Dirichlet solutions in the whole region $z\in[-1,0]$ is given by the following proposition:
\begin{proposition}
\label{prop: general construction of the Dirichlet solution}
For $\sigma\in\mathbb{I}_{[-1-\eta,-\eta]}$, the Dirichlet solutions $(\mathbf{f}_{dir,\sigma},\mathbf{g}_{dir,\sigma})$ to \eqref{eq: general equation in the construction of ingoing solution} satisfies the following form\begin{align}
\mathbf{f}_{dir,\sigma} &= A_{f}(\sigma,z)\begin{bmatrix}
1\\0
\end{bmatrix}
+B_{f}(\sigma,z)\begin{bmatrix}
0\\1
\end{bmatrix}
+C_{f}(\sigma,z)\begin{bmatrix}
(-z)^{-\frac{k^{2}+\sigma}{q_{k}}}\\\frac{1}{k}(-z)^{-\frac{k^{2}+\sigma}{q_{k}}}
\end{bmatrix}
+D_{f}(\sigma,z)\begin{bmatrix}
0\\(-z)^{-\frac{\sigma}{q_{k}}}
\end{bmatrix},\\
\mathbf{g}_{dir,\sigma}& = A_{g}(\sigma,z)\begin{bmatrix}
1\\0
\end{bmatrix}
+B_{g}(\sigma,z)\begin{bmatrix}
0\\1
\end{bmatrix}
+C_{g}(\sigma,z)\begin{bmatrix}
(-z)^{-\frac{k^{2}+\sigma}{q_{k}}}\\\frac{1}{k}(-z)^{-\frac{k^{2}+\sigma}{q_{k}}}
\end{bmatrix}
+D_{g}(\sigma,z)\begin{bmatrix}
0\\(-z)^{-\frac{\sigma}{q_{k}}}
\end{bmatrix},
\end{align}
where $\left(A_{f}(\sigma,z),B_{f}(\sigma,z),C_{f}(\sigma,z),D_{f}(\sigma,z)\right)$ and $\left(A_{g}(\sigma,z),B_{g}(\sigma,z),C_{g}(\sigma,z),D_{g}(\sigma,z)\right)$ are holomorphic functions of $\sigma$ in the region $\mathbb{I}_{[-1-\eta,-\eta]}$ with the initial condition \begin{align*}
    &\left(A_{f}(\sigma,-1),B_{f}(\sigma,-1),C_{f}(\sigma,-1),D_{f}(\sigma,-1)\right) = \left(-\frac{q_{k}}{k^{2}+\sigma},\frac{kq_{k}}{\sigma(k^{2}+\sigma)},\frac{q_{k}}{k^{2}+\sigma},-\frac{q_{k}}{k\sigma}\right),\\[1em]&
    \left(A_{g}(\sigma,-1),B_{g}(\sigma,-1),C_{g}(\sigma,-1),D_{g}(\sigma,-1)\right) = \left(0,-\frac{q_{k}}{\sigma},0,\frac{q_{k}}{\sigma}\right),
\end{align*}
solving the following Volterra equation:
\begin{equation}
    \frac{d}{dz}\begin{bmatrix}
        A\\B\\C\\D
    \end{bmatrix}(z) = F_{\sigma}^{-1}Q_{\sigma}F_{\sigma}\begin{bmatrix}
        A\\B\\C\\D
    \end{bmatrix}(z),
\end{equation}
where the $4\times4$ matrices $F_{\sigma}$ and $Q_{\sigma}$ are given in \eqref{eq: definition of Fsigma and Qsigma}.
\end{proposition}
\begin{proof}
We adopt the same approach as in the proof of Proposition \ref{prop: general construction of outgoing and ingoing solutions}. Assuming the solution to \eqref{eq: general equation in the construction of ingoing solution} takes the form of \begin{equation*}
    \begin{bmatrix}
        f\\g
    \end{bmatrix}
     = A(\sigma,z)\begin{bmatrix}
         1\\0
     \end{bmatrix}+B(\sigma,z)\begin{bmatrix}
         0\\1
     \end{bmatrix}
     +C(\sigma,z)\begin{bmatrix}
         (-z)^{-\frac{k^{2}+\sigma}{q_{k}}}\\\frac{1}{k}(-z)^{-\frac{k^{2}+\sigma}{q_{k}}}
     \end{bmatrix}
     +D(\sigma,z)\begin{bmatrix}
         0\\(-z)^{-\frac{\sigma}{q_{k}}}
     \end{bmatrix}
\end{equation*}

We have \begin{equation*}
    \frac{d}{dz}\begin{bmatrix}
        A\\B\\C\\D
    \end{bmatrix}
     = F_{\sigma}^{-1}Q_{\sigma}F_{\sigma}\begin{bmatrix}
         A\\B\\C\\D
     \end{bmatrix}
\end{equation*}

From the Dirichlet boundary condition for $\mathbf{f}_{dir}$, we have\begin{equation}
\begin{aligned}
&A(\sigma,-1)+C(\sigma,-1) = 0,\quad B(\sigma,-1)+\frac{1}{k}C(\sigma,-1)+D(\sigma,-1) = 0,\\&
\frac{k^{2}+\sigma}{q_{k}}C(\sigma,-1) = 1,\quad \frac{1}{k}\frac{k^{2}+\sigma}{q_{k}}C(\sigma,-1)+\frac{\sigma}{q_{k}}D(\sigma,-1) = 0
\end{aligned}
\end{equation}
Solving the above equations, we have \begin{equation*}
    \left(A_{f}(\sigma,-1),B_{f}(\sigma,-1),C_{f}(\sigma,-1),D_{f}(\sigma,-1)\right) = \left(-\frac{q_{k}}{k^{2}+\sigma},\frac{kq_{k}}{\sigma(k^{2}+\sigma)},\frac{q_{k}}{k^{2}+\sigma},-\frac{q_{k}}{\sigma}\right)
\end{equation*}
Similarly, from the Dirichlet boundary condition for $\mathbf{g}_{dir,\sigma}$, we have \begin{equation*}
    \left(A_{g}(\sigma,-1),B_{g}(\sigma,-1),C_{g}(\sigma,-1),D_{g}(\sigma,-1)\right) = \left(0,-\frac{q_{k}}{\sigma},0,\frac{q_{k}}{\sigma}\right).
\end{equation*}

Then we have
\begin{equation}
\begin{aligned}
\begin{bmatrix}
A\\B\\C\\D
\end{bmatrix}(z)
=& e^{\int_{-1}^{z}F_{\sigma}^{-1}Q_{\sigma}F_{\sigma}}\begin{bmatrix}
    A\\B\\C\\D
\end{bmatrix}(-1)\\ =&\sum_{n = 1}^{\infty}\int_{-1}^{z}\int_{-1}^{z_{1}}\cdots\int_{-1}^{z_{n-1}}\left(F_{\sigma}^{-1}Q_{\sigma}F_{\sigma}\right)(z_{1})\cdots\left(F_{\sigma}^{-1}Q_{\sigma}F_{\sigma}\right)(z_{n})\begin{bmatrix}
    A\\B\\C\\D
\end{bmatrix}(-1)dz_{n}\cdots dz_{1}+\begin{bmatrix}
    A\\B\\C\\D
\end{bmatrix}(-1)\\=&:
\begin{bmatrix}
    A\\B\\C\\D
\end{bmatrix}(-1)+\sum_{n = 1}^{\infty}\begin{bmatrix}
    A_{n}\\B_{n}\\C_{n}\\D_{n}
\end{bmatrix}(\sigma,z).
\end{aligned}
\label{eq: volterra equation for dirichlet solution}
\end{equation}

This concludes the proof.

\end{proof}
Since the equations \eqref{schematic form of rp}-\eqref{schematic form of psip} satisfy the form of \eqref{eq: general equation in the construction of ingoing solution}, we can apply the above construction to \eqref{schematic form of rp}-\eqref{schematic form of psip} to obtain the following construction on the Dirichlet solutions to \eqref{schematic form of rp}-\eqref{schematic form of psip}.
\begin{corollary}
For $\sigma\in\mathbb{I}_{[-1-\eta,-\eta]}$, there exist two Dirichlet solutions $\hat{\mathbf{r}}_{dir,\sigma}$ and $\hat{\mathbf{\Psi}}_{dir,\sigma}$ to \eqref{schematic form of rp}-\eqref{schematic form of psip} with the boundary conditions \begin{equation*}
    \left(\hat{\mathbf{r}}^{(1)}_{dir,\sigma},\hat{\mathbf{r}}_{dir,\sigma}^{(2)},\hat{\mathbf{r}}^{\prime(1)}_{dir,\sigma},\hat{\mathbf{r}}_{dir,\sigma}^{\prime(2)}\right)|_{z = -1} = (0,0,1,0),\quad \left(\hat{\mathbf{\Psi}}^{(1)}_{dir,\sigma},\hat{\mathbf{\Psi}}_{dir,\sigma}^{(2)},\hat{\mathbf{\Psi}}^{\prime(1)}_{dir,\sigma},\hat{\mathbf{\Psi}}_{dir,\sigma}^{\prime(2)}\right)|_{z = -1} = (0,0,0,1).
\end{equation*}
    Moreover, we have the following estimates for $\hat{\mathbf{r}}_{dir,\sigma}$ and $\hat{\mathbf{\Psi}}_{dir,\sigma}$:\begin{align}
\left\vert\partial_{z}^{p}\hat{\mathbf{r}}_{dir,\sigma}\right\vert\lesssim_{k} \vert\sigma\vert^{p-1},\quad \left\vert\partial_{z}^{p}\hat{\mathbf{\Psi}}_{dir,\sigma}\right\vert\lesssim_{k}\vert\sigma\vert^{p-1},\quad p\geq 0.
\label{eq: estimates on the dirichlet solutions}
\end{align}
\end{corollary}
\begin{proof}
The existence of Dirichlet solutions follows from Proposition \ref{prop: general construction of the Dirichlet solution}. To estimate the Dirichlet solutions $\hat{\mathbf{r}}_{dir,\sigma}$ and $\hat{\mathbf{\Psi}}_{dir,\sigma}$, we first estimate $F_{\sigma}^{-1}Q_{\sigma}F_{\sigma}$. For \eqref{schematic form of rp}-\eqref{schematic form of psip}, we have that \begin{equation*}
    \left\vert E_{1,\sigma}\right\vert\lesssim \frac{1}{k}(-z),\quad \left\vert (E_{2,\sigma})_{ij}\right\vert\lesssim k\vert\sigma\vert,\quad i = 1,3,\quad \left\vert(E_{2,\sigma})_{ij}\right\vert\lesssim k,\quad i = 2,4.
\end{equation*}
Computing $F_{\sigma}^{-1}Q_{\sigma}F_{\sigma}$ and tracking the $\sigma$-dependence, we have \begin{equation*}
    \left\vert F_{\sigma}^{-1}Q_{\sigma}F_{\sigma}\right\vert = \begin{bmatrix}
        O_{k}(1)&O_{k}\left(\frac{1}{\sigma}\right)&O_{k}((-z)^{-\frac{k^{2}+\sigma}{q_{k}}})&O_{k}((-z)^{-\frac{\sigma}{q_{k}}})\\O_{k}(1)&O_{k}(1)&O_{k}((-z)^{-\frac{k^{2}+\sigma}{q_{k}}})&O_{k}((-z)^{-\frac{\sigma}{q_{k}}})\\
        O_{k}((-z)^{\frac{k^{2}+\sigma}{q_{k}}})&O_{k}\left(\frac{1}{\sigma}(-z)^{\frac{k^{2}+\sigma}{q_{k}}}\right)&O_{k}(1)&O_{k}\left(\frac{1}{\sigma}(-z)^{\frac{k^{2}}{q_{k}}}\right)\\O_{k}((-z)^{\frac{\sigma}{q_{k}}})&O_{k}(\frac{1}{\sigma}(-z)^{\frac{\sigma}{q_{k}}})&O_{k}(1)&O_{k}(1)
    \end{bmatrix}.
\end{equation*}
Hence, we have that \begin{equation*}
    \left\vert\begin{bmatrix}
        A_{1}\\B_{1}\\C_{1}\\D_{1}
    \end{bmatrix}\right\vert\lesssim\frac{1}{k\vert \sigma\vert}\begin{bmatrix}
        (-z)+1\\(-z)+1\\(-z)^{\frac{k^{2}+\sigma}{q_{k}}+1}+1\\(-z)^{\frac{\sigma}{q_{k}}+1}+1
    \end{bmatrix}.
\end{equation*}
The additional $\frac{1}{\vert\sigma\vert}$ comes from the choice of $(A(\sigma,-1),B(\sigma,-1),C(\sigma,-1),D(\sigma,-1))$ for $\hat{\mathbf{r}}_{dir,\sigma}$ and $\hat{\mathbf{\Psi}}_{dir,\sigma}$. Arguing inductively, we have \begin{equation*}
    \left\vert\begin{bmatrix}
        A_{n}\\B_{n}\\C_{n}\\D_{n}
    \end{bmatrix}\right\vert\lesssim_{k}\frac{1}{n!k\vert\sigma\vert}\begin{bmatrix}
        (-z)^{n}+1\\(-z)^{n}+1\\(-z)^{\frac{k^{2}+\sigma}{q_{k}}+n}+(-z)^{\frac{k^{2}+\sigma}{q_{k}}+1}+1\\
        (-z)^{\frac{\sigma}{q_{k}}+n}+(-z)^{\frac{k^{2}+\sigma}{q_{k}}+1}+1
    \end{bmatrix}.
\end{equation*}
Since \begin{equation*}
    \begin{bmatrix}
\hat{\mathbf{r}}_{dir,\sigma}\\\hat{\mathbf{r}}_{dir,\sigma}^{\prime}
\end{bmatrix}= F_{\sigma}\sum_{n = 1}^{\infty}\begin{bmatrix}
    A_{n}\\B_{n}\\C_{n}\\D_{n}
\end{bmatrix}
+F_{\sigma}\begin{bmatrix}
    A\\B\\C\\D
\end{bmatrix}(\sigma,-1),
\end{equation*}
we can derive
\begin{align*}
\left\vert\partial_{z}^{p}\hat{\mathbf{r}}_{dir,\sigma}\right\vert\lesssim_{k} \vert\sigma\vert^{p-1},\quad \left\vert\partial_{z}^{p}\hat{\mathbf{\Psi}}_{dir,\sigma}\right\vert\lesssim_{k}\vert\sigma\vert^{p-1}.
\end{align*}
This completes the proof.
\end{proof}
\subsection{Scattering resolvent and resonance}
Given the outgoing solutions $\hat{\mathbf{r}}_{out,\sigma}$, $\hat{\mathbf{\Psi}}_{out,\sigma}$ and the Dirichlet solutions ${\hat{\mathbf{r}}_{dir,\sigma},\hat{\mathbf{\Psi}}_{dir,\sigma}}$, we are particularly interested in the solutions to $\mathcal{L}(\hat{r}_{p},\hat{\Psi}_{p}) = F$ which satisfy the Dirichlet boundary conditions at the center $z = -1$ and asymptote to the outgoing behavior near the singular horizon $z= 0$. 

We assume the solution to the equation $\mathcal{L}(\hat{r}_{p},\hat{\Psi}_{p}) = F =(F_{1},F_{2})$ takes the form of \begin{equation}
\begin{bmatrix}
\hat{r}_{p}\\\hat{\Psi}_{p}
\end{bmatrix}
 = A(z)\hat{\mathbf{r}}_{dir,\sigma}+B(z)\hat{\mathbf{\Psi}}_{dir,\sigma}+C(z)\hat{\mathbf{r}}_{out,\sigma}+D(z)\hat{\mathbf{\Psi}}_{out,\sigma}.
\end{equation}
We can further assume that \begin{equation}
A^{\prime}(z)\hat{\mathbf{r}}_{dir,\sigma}(z)+B^{\prime}(z)\hat{\mathbf{\Psi}}_{dir,\sigma}(z)+C^{\prime}(z)\hat{\mathbf{r}}_{out,\sigma}(z)+D^{\prime}(z)\hat{\mathbf{\Psi}}_{out,\sigma}(z) = 0.
\end{equation}
Then we have \begin{align*}
\frac{d}{dz}\begin{bmatrix}
\hat{r}_{p}\\\hat{\Psi}_{p}
\end{bmatrix}
 =& A(z)\hat{\mathbf{r}}_{dir,\sigma}^{\prime}+B(z)\hat{\mathbf{\Psi}}_{dir,\sigma}^{\prime}+C(z)\hat{\mathbf{r}}_{out,\sigma}^{\prime}+D(z)\hat{\mathbf{\Psi}}_{out,\sigma}^{\prime},\\
 \frac{d^{2}}{dz^{2}}\begin{bmatrix}
 \hat{r}_{p}\\\hat{\Psi}_{p}
 \end{bmatrix}
 =& 
 A(z)\hat{\mathbf{r}}_{dir,\sigma}^{\prime\prime}+B(z)\hat{\mathbf{\Psi}}_{dir,\sigma}^{\prime\prime}+C(z)\hat{\mathbf{r}}_{out,\sigma}^{\prime\prime}+D(z)\hat{\mathbf{\Psi}}_{out,\sigma}^{\prime\prime}\\&+A^{\prime}(z)\hat{\mathbf{r}}_{dir,\sigma}^{\prime}+B^{\prime}(z)\hat{\mathbf{\Psi}}_{dir,\sigma}^{\prime}+C^{\prime}(z)\hat{\mathbf{r}}_{out,\sigma}^{\prime}+D^{\prime}(z)\hat{\mathbf{\Psi}}_{out,\sigma}^{\prime}.
\end{align*}
Substituting the above expression into $\mathcal{L}(\hat{r}_{p},\hat{\Psi}_{p}) = F = (F_{1},F_{2})$, we have\begin{align}
\label{eq: resolvent variation}
\begin{bmatrix}
\hat{\mathbf{r}}_{dir,\sigma}&\hat{\mathbf{\Psi}}_{dir,\sigma}&\hat{\mathbf{r}}_{out,\sigma}&\hat{\mathbf{\Psi}}_{out,\sigma}\\
\hat{\mathbf{r}}_{dir,\sigma}^{\prime}&\hat{\mathbf{\Psi}}_{dir,\sigma}^{\prime}&\hat{\mathbf{r}}_{out,\sigma}^{\prime}&\hat{\mathbf{\Psi}}_{out,\sigma}^{\prime}
\end{bmatrix}
\begin{bmatrix}
A^{\prime}(z)\\B^{\prime}(z)\\C^{\prime}(z)\\D^{\prime}(z)
\end{bmatrix}
 = \frac{1}{q_{k}z}\begin{bmatrix}
 0\\0\\F_{1}\\F_{2}
 \end{bmatrix}.
\end{align}
Let \begin{equation*}
M_{\sigma}: = \begin{bmatrix}
\hat{\mathbf{r}}_{dir,\sigma}&\hat{\mathbf{\Psi}}_{dir,\sigma}&\hat{\mathbf{r}}_{out,\sigma}&\hat{\mathbf{\Psi}}_{out,\sigma}\\\hat{\mathbf{r}}_{dir,\sigma}^{\prime}&\hat{\mathbf{\Psi}}_{dir,\sigma}^{\prime}&\hat{\mathbf{r}}_{out,\sigma}^{\prime}&\hat{\mathbf{\Psi}}_{out,\sigma}^{\prime}
\end{bmatrix}.
\end{equation*}
First, we study the invertibility of $M_{\sigma}$.
\begin{definition}
For fixed $\sigma\in\mathbb{I}_{[-1-\eta,-\eta]}$, we define the Wronskian of the Dirichlet solutions and outgoing solutions to be\begin{equation}
\mathcal{W}(\sigma,z) = \det\begin{bmatrix}
\hat{\mathbf{r}}_{dir,\sigma}&\hat{\mathbf{\Psi}}_{dir,\sigma}&\hat{\mathbf{r}}_{out,\sigma}&\hat{\mathbf{\Psi}}_{out,\sigma}\\
\hat{\mathbf{r}}_{dir,\sigma}^{\prime}&\hat{\mathbf{\Psi}}_{dir,\sigma}^{\prime}&\hat{\mathbf{r}}_{out,\sigma}^{\prime}&\hat{\mathbf{\Psi}}_{out,\sigma}^{\prime}
\end{bmatrix}.
\end{equation}
\end{definition}
By the Jacobi formula of the determinant of the fundamental matrix to ODEs, we have\begin{proposition}
$M_{\sigma}$ is invertible if and only if $\mathcal{W}(\sigma,-1)\neq0$. Moreover, we have the following estimates for $\mathcal{W}(\sigma,z)$:\begin{equation}
\frac{1}{C}\mathcal{W}(\sigma,-1)(-z)^{-\frac{k^{2}+2\sigma}{q_{k}}-2}\leq \mathcal{W}(\sigma,z)\leq C \mathcal{W}(\sigma,-1)(-z)^{-\frac{k^{2}+2\sigma}{q_{k}}-2}.\label{estimate for the wronskian}
\end{equation}
\end{proposition}
\begin{proof}
We can write the equations \eqref{eq:scattering equation for rp}-\eqref{scattering equation for mp} as:\begin{equation}
\frac{d}{dx}\begin{bmatrix}
\hat{r}_{p}\\\hat{\Psi}_{p}\\\hat{r}_{p}^{\prime}\\\hat{\Psi}_{p}^{\prime}
\end{bmatrix}
 = L(x)\begin{bmatrix}
 \hat{r}_{p}\\\hat{\Psi}_{p}\\\hat{r}_{p}^{\prime}\\\hat{\Psi}_{p}^{\prime}.
 \end{bmatrix}
\end{equation}
Then by the Jacobi formula, we have\begin{equation}
\frac{d}{dz}\mathcal{W}(\sigma,z) = \mathcal{W}(\sigma,z)\text{Tr}(L)(z).\label{Jacobi formula}
\end{equation}
Integrating \eqref{Jacobi formula}, we have\begin{equation}
\mathcal{W}(\sigma,z) = \mathcal{W}(\sigma,-1)e^{\int_{-1}^{z}\text{Tr}(L)(y)dy}=:\mathcal{W}(\sigma,0)W(\sigma,z),
\end{equation}
where $W(\sigma,-1)$ is a non-vanishing function. Hence, $M_{\sigma}$ is invertible if and only if $\mathcal{W}(\sigma,0)\neq 0$. From \eqref{eq:scattering equation for rp}-\eqref{scattering equation for mp}, we can compute $\text{Tr}L$:\begin{align*}
-\text{Tr}L(z) =&\frac{2(q_{k}+\sigma)}{q_{k}z}+\frac{\mu_{k}}{1-\mu_{k}}\frac{1}{q_{k}z}-\frac{\mu_{k}}{1-\mu_{k}}\frac{2\partial_{z}r_{k}}{r_{k}}\\&+\frac{(\partial_{s}+q_{k}z\partial_{z})r_{k}\partial_{z}r_{k}}{1-\mu_{k}}\frac{1}{\sigma-V_{k}}\left[\left(\frac{(\partial_{s}+q_{k}z\partial_{z})\phi_{k}}{(\partial_{s}+q_{k}z\partial_{z})r_{k}}\right)^{2}-\left(\frac{\partial_{z}\phi_{k}}{\partial_{z}r_{k}}\right)^{2}\right]\\=&
\frac{2(q_{k}+\sigma)+k^{2}}{q_{k}z}+\left(\frac{\mu_{k}}{1-\mu_{k}}-k^{2}\right)\frac{1}{q_{k}z}-\frac{\mu_{k}}{1-\mu_{k}}\frac{2\partial_{z}r_{k}}{r_{k}}\\&+\frac{(\partial_{s}+q_{k}z\partial_{z})r_{k}\partial_{z}r_{k}}{1-\mu_{k}}\frac{1}{\sigma-V_{k}}\left[\left(\frac{(\partial_{s}+q_{k}z\partial_{z})\phi_{k}}{(\partial_{s}+q_{k}z\partial_{z})r_{k}}\right)^{2}-\left(\frac{\partial_{z}\phi_{k}}{\partial_{z}r_{k}}\right)^{2}\right]
\end{align*}
By the estimates on the background solution $(r_{k},\phi_{k},\mu_{k})$, we have\begin{equation*}
\frac{1}{C}\leq (-z)^{\frac{2(q_{k}+\sigma)+k^{2}}{q_{k}}}e^{\int_{-1}^{z}\text{Tr}L(y)dy}\leq C.
\end{equation*}
Hence, we have the estimate \eqref{estimate for the wronskian}. This concludes the proof.
\end{proof}
Next, we derive the scattering resolvent operator, which can be viewed as an inverse of the linear operator $\mathcal{L}(\hat{r}_{p},\hat{\Psi}_{p})$, satisfying the Dirichlet boundary condition and outgoing boundary condition. By \eqref{eq: resolvent variation}, we have \begin{equation*}
    \frac{d}{dz}\begin{bmatrix}
        A\\B\\C\\D
    \end{bmatrix}(z) =\frac{1}{q_{k}z} M_{\sigma}^{-1}(z)\begin{bmatrix}
        0\\0\\F_{1}\\F_{2}
    \end{bmatrix}.
\end{equation*}
Let $G_{\sigma}(z)$ be a $4\times4$ matrix with $G_{\sigma}^{\prime}(z) = \frac{1}{q_{k}z}M_{\sigma}^{-1}(z)$. Then we have \begin{align*}
    &\frac{d}{dz}\begin{bmatrix}
        A\\B
    \end{bmatrix}(z) = \begin{bmatrix}
        \left(G_{\sigma}^{\prime}\right)_{13}&\left(G_{\sigma}^{\prime}\right)_{14}\\\left(G_{\sigma}^{\prime}\right)_{23}&\left(G_{\sigma}^{\prime}\right)_{24}
    \end{bmatrix}\begin{bmatrix}
        F_{1}\\F_{2}
    \end{bmatrix} = \frac{d}{dz}\left(\begin{bmatrix}
        \left(G_{\sigma}\right)_{13}&\left(G_{\sigma}\right)_{14}\\\left(G_{\sigma}\right)_{23}&\left(G_{\sigma}\right)_{24}
    \end{bmatrix}\begin{bmatrix}
        F_{1}\\F_{2}
    \end{bmatrix}\right)-\begin{bmatrix}
        \left(G_{\sigma}\right)_{13}&\left(G_{\sigma}\right)_{14}\\\left(G_{\sigma}\right)_{23}&\left(G_{\sigma}\right)_{24}
    \end{bmatrix}\begin{bmatrix}
        F_{1}^{\prime}\\F_{2}^{\prime}
    \end{bmatrix},\\[1em]&
    \frac{d}{dz}\begin{bmatrix}
        C\\D
    \end{bmatrix}(z) = \begin{bmatrix}
        \left(G_{\sigma}^{\prime}\right)_{33}&\left(G_{\sigma}^{\prime}\right)_{34}\\
        \left(G_{\sigma}^{\prime}\right)_{43}&\left(G_{\sigma}^{\prime}\right)_{44}
    \end{bmatrix}
    \begin{bmatrix}
        F_{1}\\F_{2}
    \end{bmatrix}.
    \end{align*}
Using the Dirichlet boundary condition, we have \begin{equation*}
    \begin{bmatrix}
        C\\D
    \end{bmatrix}(-1) = 0.
\end{equation*}
Then, we have \begin{align*}
   \begin{bmatrix}
       (M_{\sigma})_{13}&(M_{\sigma})_{14}\\
       (M_{\sigma})_{23}&(M_{\sigma})_{24}\\
   \end{bmatrix} \begin{bmatrix}
        C\\D
    \end{bmatrix}(z) =  \begin{bmatrix}
       (M_{\sigma})_{13}&(M_{\sigma})_{14}\\
       (M_{\sigma})_{23}&(M_{\sigma})_{24}\\
   \end{bmatrix}(z)\int_{-1}^{z}\begin{bmatrix}
        \left(G_{\sigma}^{\prime}\right)_{33}&\left(G_{\sigma}^{\prime}\right)_{34}\\
        \left(G_{\sigma}^{\prime}\right)_{43}&\left(G_{\sigma}^{\prime}\right)_{44}
    \end{bmatrix}
    \begin{bmatrix}
        F_{1}\\F_{2}
    \end{bmatrix}(z^{\prime})dz^{\prime}.
\end{align*}
By the definition of $G_{\sigma}$, we have \begin{align*}
    \begin{bmatrix}
        (M_{\sigma})_{11}&(M_{\sigma})_{12}\\(M_{\sigma})_{21}&(M_{\sigma})_{22}
    \end{bmatrix}\begin{bmatrix}
        \left(G_{\sigma}^{\prime}\right)_{13}&\left(G_{\sigma}^{\prime}\right)_{14}\\\left(G_{\sigma}^{\prime}\right)_{23}&\left(G_{\sigma}^{\prime}\right)_{24}
    \end{bmatrix}+\begin{bmatrix}
        (M_{\sigma})_{13}&(M_{\sigma})_{14}\\(M_{\sigma})_{23}&(M_{\sigma})_{24}
    \end{bmatrix}\begin{bmatrix}
        \left(G_{\sigma}^{\prime}\right)_{33}&\left(G_{\sigma}^{\prime}\right)_{34}\\\left(G_{\sigma}^{\prime}\right)_{43}&\left(G_{\sigma}^{\prime}\right)_{44}
    \end{bmatrix} = 0.
\end{align*}
We have \begin{align*}
    \begin{bmatrix}
        \left(G_{\sigma}^{\prime}\right)_{33}&\left(G_{\sigma}^{\prime}\right)_{34}\\\left(G_{\sigma}^{\prime}\right)_{43}&\left(G_{\sigma}^{\prime}\right)_{44}
    \end{bmatrix} \begin{bmatrix}
        F_{1}\\F_{2}
    \end{bmatrix}=& -\begin{bmatrix}
        (M_{\sigma})_{13}&(M_{\sigma})_{14}\\(M_{\sigma})_{23}&(M_{\sigma})_{24}
    \end{bmatrix}^{-1}\begin{bmatrix}
        (M_{\sigma})_{11}&(M_{\sigma})_{12}\\(M_{\sigma})_{21}&(M_{\sigma})_{22}
    \end{bmatrix}\begin{bmatrix}
        \left(G_{\sigma}^{\prime}\right)_{13}&\left(G_{\sigma}^{\prime}\right)_{14}\\\left(G_{\sigma}^{\prime}\right)_{23}&\left(G_{\sigma}^{\prime}\right)_{24}
    \end{bmatrix}\begin{bmatrix}
        F_{1}\\F_{2}
    \end{bmatrix}\\=&
    -\left(\begin{bmatrix}
        (M_{\sigma})_{13}&(M_{\sigma})_{14}\\(M_{\sigma})_{23}&(M_{\sigma})_{24}
    \end{bmatrix}^{-1}\begin{bmatrix}
        (M_{\sigma})_{11}&(M_{\sigma})_{12}\\(M_{\sigma})_{21}&(M_{\sigma})_{22}
    \end{bmatrix}\begin{bmatrix}
        \left(G_{\sigma}\right)_{13}&\left(G_{\sigma}\right)_{14}\\\left(G_{\sigma}\right)_{23}&\left(G_{\sigma}\right)_{24}
    \end{bmatrix}\begin{bmatrix}
        F_{1}\\F_{2}
    \end{bmatrix}\right)^{\prime}\\&+\left(\begin{bmatrix}
        (M_{\sigma})_{13}&(M_{\sigma})_{14}\\(M_{\sigma})_{23}&(M_{\sigma})_{24}
    \end{bmatrix}^{-1}\begin{bmatrix}
        (M_{\sigma})_{11}&(M_{\sigma})_{12}\\(M_{\sigma})_{21}&(M_{\sigma})_{22}
    \end{bmatrix}\right)^{\prime}\begin{bmatrix}
        \left(G_{\sigma}\right)_{13}&\left(G_{\sigma}\right)_{14}\\\left(G_{\sigma}\right)_{23}&\left(G_{\sigma}\right)_{24}
    \end{bmatrix}\begin{bmatrix}
        F_{1}\\F_{2}
    \end{bmatrix}\\&+\left(\begin{bmatrix}
        (M_{\sigma})_{13}&(M_{\sigma})_{14}\\(M_{\sigma})_{23}&(M_{\sigma})_{24}
    \end{bmatrix}^{-1}\begin{bmatrix}
        (M_{\sigma})_{11}&(M_{\sigma})_{12}\\(M_{\sigma})_{21}&(M_{\sigma})_{22}
    \end{bmatrix}\begin{bmatrix}
        \left(G_{\sigma}\right)_{13}&\left(G_{\sigma}\right)_{14}\\\left(G_{\sigma}\right)_{23}&\left(G_{\sigma}\right)_{24}
    \end{bmatrix}\right)\begin{bmatrix}
        F_{1}^{\prime}\\F_{2}^{\prime}
    \end{bmatrix}.
\end{align*}
Therefore, we have \begin{align*}
    &\begin{bmatrix}
        (M_{\sigma})_{13}& (M_{\sigma})_{14}\\ (M_{\sigma})_{23}& (M_{\sigma})_{24}
    \end{bmatrix}
    \begin{bmatrix}
        C\\D
    \end{bmatrix}(z) \\=& -\begin{bmatrix}
        (M_{\sigma})_{11}&(M_{\sigma})_{12}\\(M_{\sigma})_{21}&(M_{\sigma})_{22}
    \end{bmatrix}\begin{bmatrix}
        \left(G_{\sigma}\right)_{13}&\left(G_{\sigma}\right)_{14}\\\left(G_{\sigma}\right)_{23}&\left(G_{\sigma}\right)_{24}
    \end{bmatrix}
    \begin{bmatrix}
        F_{1}\\F_{2}
    \end{bmatrix}\\&+\begin{bmatrix}
        (M_{\sigma})_{13}& (M_{\sigma})_{14}\\ (M_{\sigma})_{23}& (M_{\sigma})_{24}
    \end{bmatrix}\int_{-1}^{z}\left(\begin{bmatrix}
        (M_{\sigma})_{13}&(M_{\sigma})_{14}\\(M_{\sigma})_{23}&(M_{\sigma})_{24}
    \end{bmatrix}^{-1}\begin{bmatrix}
        (M_{\sigma})_{11}&(M_{\sigma})_{12}\\(M_{\sigma})_{21}&(M_{\sigma})_{22}
    \end{bmatrix}\right)^{\prime}\begin{bmatrix}
        \left(G_{\sigma}\right)_{13}&\left(G_{\sigma}\right)_{14}\\\left(G_{\sigma}\right)_{23}&\left(G_{\sigma}\right)_{24}
    \end{bmatrix}\begin{bmatrix}
        F_{1}\\F_{2}
    \end{bmatrix}(z^{\prime})dz^{\prime}\\&+\begin{bmatrix}
        (M_{\sigma})_{13}& (M_{\sigma})_{14}\\ (M_{\sigma})_{23}& (M_{\sigma})_{24}
    \end{bmatrix}\int_{-1}^{z}\left(\begin{bmatrix}
        (M_{\sigma})_{13}&(M_{\sigma})_{14}\\(M_{\sigma})_{23}&(M_{\sigma})_{24}
    \end{bmatrix}^{-1}\begin{bmatrix}
        (M_{\sigma})_{11}&(M_{\sigma})_{12}\\(M_{\sigma})_{21}&(M_{\sigma})_{22}
    \end{bmatrix}\begin{bmatrix}
        \left(G_{\sigma}\right)_{13}&\left(G_{\sigma}\right)_{14}\\\left(G_{\sigma}\right)_{23}&\left(G_{\sigma}\right)_{24}
    \end{bmatrix}\right)\begin{bmatrix}
        F_{1}^{\prime}\\F_{2}^{\prime}
    \end{bmatrix}(z^{\prime})dz^{\prime}.
\end{align*}
For $(A,B)$, we have \begin{align*}
    \begin{bmatrix}
        (M_{\sigma})_{13}&(M_{\sigma})_{14}\\(M_{\sigma})_{23}&(M_{\sigma})_{24}
    \end{bmatrix}\begin{bmatrix}
        A\\B
    \end{bmatrix}(z) =&  \begin{bmatrix}
        (M_{\sigma})_{11}&(M_{\sigma})_{12}\\(M_{\sigma})_{21}&(M_{\sigma})_{22}
    \end{bmatrix}\begin{bmatrix}
        (G_{\sigma})_{13}&(G_{\sigma})_{14}\\(G_{\sigma})_{23}&(G_{\sigma})_{24}
    \end{bmatrix}\begin{bmatrix}
        F_{1}\\F_{2}
    \end{bmatrix}(z)\\&-\begin{bmatrix}
        (M_{\sigma})_{11}&(M_{\sigma})_{12}\\(M_{\sigma})_{21}&(M_{\sigma})_{22}
    \end{bmatrix}(z)\int_{0}^{z}\begin{bmatrix}
        (G_{\sigma})_{13}&(G_{\sigma})_{14}\\(G_{\sigma})_{23}&(G_{\sigma})_{24}
    \end{bmatrix}\begin{bmatrix}
        F_{1}^{\prime}\\F_{2}^{\prime}
    \end{bmatrix}(z^{\prime})dz^{\prime}.
\end{align*}
Therefore, solving $\mathcal{L}(\hat{r}_{p},\hat{\Psi}_{p}) = (F_{1},F_{2})$ with Dirichlet and outgoing boundary conditions, we have \begin{align*}
    \begin{bmatrix}
        \hat{r}_{p}\\\hat{\Psi}_{p}
    \end{bmatrix}(z)
     =& \begin{bmatrix}
         \hat{\mathbf{r}}_{out,\sigma}&\hat{\mathbf{\Psi}}_{out,\sigma}
     \end{bmatrix}(z)\int_{-1}^{z}\left(\begin{bmatrix}
         \hat{\mathbf{r}}_{out,\sigma}&\hat{\mathbf{\Psi}}_{out,\sigma}
     \end{bmatrix}^{-1}\begin{bmatrix}
         \hat{\mathbf{r}}_{dir,\sigma}&\hat{\mathbf{\Psi}}_{dir,\sigma}
     \end{bmatrix}\right)^{\prime}\begin{bmatrix}
         (G_{\sigma})_{13}&(G_{\sigma})_{14}\\(G_{\sigma})_{23}&(G_{\sigma})_{24}
     \end{bmatrix}
     \begin{bmatrix}
         F_{1}\\F_{2}
     \end{bmatrix}dz^{\prime}\\&+\begin{bmatrix}
         \hat{\mathbf{r}}_{out,\sigma}&\hat{\mathbf{\Psi}}_{out,\sigma}
     \end{bmatrix}(z)\int_{-1}^{z}\begin{bmatrix}
         \hat{\mathbf{r}}_{out,\sigma}&\hat{\mathbf{\Psi}}_{out,\sigma}
     \end{bmatrix}^{-1}\begin{bmatrix}
         \hat{\mathbf{r}}_{dir,\sigma}&\hat{\mathbf{\Psi}}_{dir,\sigma}
     \end{bmatrix}\begin{bmatrix}
         (G_{\sigma})_{13}&(G_{\sigma})_{14}\\(G_{\sigma})_{23}&(G_{\sigma})_{24}
     \end{bmatrix}
     \begin{bmatrix}
         F_{1}^{\prime}\\F_{2}^{\prime}
     \end{bmatrix}dz^{\prime}\\&+\begin{bmatrix}
         \hat{\mathbf{r}}_{dir,\sigma}&\hat{\mathbf{\Psi}}_{dir,\sigma}
     \end{bmatrix}(z)\int_{z}^{0}\begin{bmatrix}
         (G_{\sigma})_{13}&(G_{\sigma})_{14}\\(G_{\sigma})_{23}&(G_{\sigma})_{24}
     \end{bmatrix}\begin{bmatrix}
         F_{1}^{\prime}\\F_{2}^{\prime}
     \end{bmatrix}dz^{\prime}.
\end{align*}
This motivates the following definition.
\begin{definition}
\label{def: scattering resolvent}
    Let $p_{k}<\beta<\frac{3}{2}$ and $\beta+\frac{1}{4}<\widetilde{\beta}<2$ be given parameters, and define $\beta_{*} = \min\{\beta q_{k},1+\eta\}$. For fixed $\sigma\in \mathbb{I}_{[-\beta_{*},-\eta]}$, assuming that $\mathcal{W}(\sigma,0)\neq0$, we can define the scattering resolvent $R(\sigma):\partial_{z} \mathcal{C}^{\widetilde{\beta}}_{N}\times \partial_{z}\mathcal{C}_{N}^{\beta}\rightarrow L_{loc}^{2}\times L_{loc}^{2}$ as \begin{equation}
    \begin{aligned}
    R(\sigma)(F_{1},F_{2})(z) = & \begin{bmatrix}
         \hat{\mathbf{r}}_{out,\sigma}&\hat{\mathbf{\Psi}}_{out,\sigma}
     \end{bmatrix}(z)\int_{-1}^{z}\left(\begin{bmatrix}
         \hat{\mathbf{r}}_{out,\sigma}&\hat{\mathbf{\Psi}}_{out,\sigma}
     \end{bmatrix}^{-1}\begin{bmatrix}
         \hat{\mathbf{r}}_{dir,\sigma}&\hat{\mathbf{\Psi}}_{dir,\sigma}
     \end{bmatrix}\right)^{\prime}\begin{bmatrix}
         (G_{\sigma})_{13}&(G_{\sigma})_{14}\\(G_{\sigma})_{23}&(G_{\sigma})_{24}
     \end{bmatrix}
     \begin{bmatrix}
         F_{1}\\F_{2}
     \end{bmatrix}dz^{\prime}\\&+\begin{bmatrix}
         \hat{\mathbf{r}}_{out,\sigma}&\hat{\mathbf{\Psi}}_{out,\sigma}
     \end{bmatrix}(z)\int_{-1}^{z}\begin{bmatrix}
         \hat{\mathbf{r}}_{out,\sigma}&\hat{\mathbf{\Psi}}_{out,\sigma}
     \end{bmatrix}^{-1}\begin{bmatrix}
         \hat{\mathbf{r}}_{dir,\sigma}&\hat{\mathbf{\Psi}}_{dir,\sigma}
     \end{bmatrix}\begin{bmatrix}
         (G_{\sigma})_{13}&(G_{\sigma})_{14}\\(G_{\sigma})_{23}&(G_{\sigma})_{24}
     \end{bmatrix}
     \begin{bmatrix}
         F_{1}^{\prime}\\F_{2}^{\prime}
     \end{bmatrix}dz^{\prime}\\&+\begin{bmatrix}
         \hat{\mathbf{r}}_{dir,\sigma}&\hat{\mathbf{\Psi}}_{dir,\sigma}
     \end{bmatrix}(z)\int_{z}^{0}\begin{bmatrix}
         (G_{\sigma})_{13}&(G_{\sigma})_{14}\\(G_{\sigma})_{23}&(G_{\sigma})_{24}
     \end{bmatrix}\begin{bmatrix}
         F_{1}^{\prime}\\F_{2}^{\prime}
     \end{bmatrix}dz^{\prime}\\=&:\int_{-1}^{0}R_{\sigma}(z,z^{\prime})\begin{bmatrix}
         F_{1}\\F_{2}
     \end{bmatrix}+\widetilde{R}_{\sigma}(z,z^{\prime})\begin{bmatrix}
         F_{1}^{\prime}\\F_{2}^{\prime}
     \end{bmatrix}.
\end{aligned}
    \end{equation}
    where the kernel $R_{\sigma}(z,z^{\prime})$ and $\widetilde{R}_{\sigma}(z,z^{\prime})$ are $2\times2$ matrices taking the form of \begin{equation}
        \begin{aligned}
            &R_{\sigma}(z,z^{\prime}) = 0,\quad \widetilde{R}_{\sigma}(z,z^{\prime}) = \begin{bmatrix}
                \hat{\mathbf{r}}_{dir,\sigma}&\hat{\mathbf{\Psi}}_{dir,\sigma}
            \end{bmatrix}(z)\begin{bmatrix}
                (G_{\sigma})_{13}&(G_{\sigma})_{14}\\(G_{\sigma})_{23}&(G_{\sigma})_{24}
            \end{bmatrix}(z^{\prime}),\quad z<z^{\prime},\\&
            R_{\sigma}(z,z^{\prime}) = \begin{bmatrix}
                \hat{\mathbf{r}}_{out,\sigma}&\hat{\mathbf{\Psi}}_{out,\sigma}
            \end{bmatrix}(z)\left(\left(\begin{bmatrix}
                \hat{\mathbf{r}}_{out,\sigma}&\hat{\mathbf{\Psi}}_{out,\sigma}
            \end{bmatrix}^{-1}\begin{bmatrix}
                \hat{\mathbf{r}}_{dir,\sigma}&\hat{\mathbf{\Psi}}_{out,\sigma}
            \end{bmatrix}\right)^{\prime}\begin{bmatrix}
                (G_{\sigma})_{13}&(G_{\sigma})_{14}\\(G_{\sigma})_{23}&(G_{\sigma})_{24}
            \end{bmatrix}\right)(z^{\prime}),\quad z^{\prime}<z,\\&
            \widetilde{R}_{\sigma}(z,z^{\prime}) = \begin{bmatrix}
                \hat{\mathbf{r}}_{out,\sigma}&\hat{\mathbf{\Psi}}_{out,\sigma}
            \end{bmatrix}(z)\left(\begin{bmatrix}
                \hat{\mathbf{r}}_{out,\sigma}&\hat{\mathbf{\Psi}}_{out,\sigma}
            \end{bmatrix}^{-1}\begin{bmatrix}
                \hat{\mathbf{r}}_{dir,\sigma}&\hat{\mathbf{\Psi}}_{dir,\sigma}
            \end{bmatrix}\begin{bmatrix}
                (G_{\sigma})_{13}&(G_{\sigma})_{14}\\(G_{\sigma})_{23}&(G_{\sigma})_{24}
            \end{bmatrix}\right)(z^{\prime}),\quad z^{\prime}<z.
        \end{aligned}
    \end{equation}

    Let $\rho(z)$ be a decreasing cut-off function which is equal to $1$ for $-1\leq z\leq \frac{1}{4}$ and identically zero for $-1\leq z\leq 0$. Define the cut-off resolvent $\rho(z)R(\sigma)(F_{1},F_{2}):\partial_{z} \mathcal{C}_{N}^{\widetilde{\beta}}\times\partial_{z}\mathcal{C}_{N}^{\beta}\rightarrow L^{2}\times L^{2}$ with the integral kernel $\rho(z)R(\sigma)(z,z^{\prime})$ and $\rho(z)\widetilde{R}(\sigma)(z,z^{\prime})$. 
\end{definition}
Note that by the construction of the outgoing solutions, the matrix \begin{equation*}
    \begin{bmatrix}
    \hat{\mathbf{r}}_{out,\sigma}&\hat{\mathbf{\Psi}}_{out,\sigma}
    \end{bmatrix}
\end{equation*}
is always invertible. Therefore, the scattering resolvent $R(\sigma)(F_{1},F_{2})$ can be viewed as a meromorphic function with poles at the zero set of $\mathcal{W}(\sigma,-1)$, i.e., points where $M_{\sigma}$ is not invertible. We have the following definition.

\begin{definition}
    For any $(F_{1},F_{2})\in\partial_{z}\mathcal{C}_{N}^{\widetilde{\beta}}\times\partial_{z}\mathcal{C}_{N}^{\beta}$, $\rho(z)R(\sigma)(F_{1},F_{2})$ is a meromorphic function for $\sigma\in\mathbb{I}_{[-\beta_{*},-\eta]}$ with poles at the zero set of $\mathcal{W}(\sigma,-1)$.

    We define any $\sigma_{*}\in\mathbb{I}_{[-\beta_{*},-\eta]}$ such that $\mathcal{W}(\sigma_{*},-1) = 0$ to be a scattering resonance. The order of vanishing is called the multiplicity of the resonance.
\end{definition}

\subsection{Location of the scattering resonance}
\label{sec: location of the scattering resonance}
In this section, we consider the equation $\mathcal{W}(\sigma,-1) = 0$, the solution of which corresponds to the scattering resonance. First, we show that for $\sigma\in \mathbb{I}_{[-1-\eta,-\eta]}\cap\{\vert\Im\sigma\vert\geq1\}$, $\mathcal{W}(\sigma,-1)$ can never be zero.
\begin{lemma}
For $k$ sufficiently small, the region $\mathbb{I}_{[-1-\eta,-\eta]}\cap\{\vert\Im\sigma\vert\geq1\}$ is free of scattering resonance.
\end{lemma}
\begin{proof}
By the definition of the Dirichlet boundary conditions, we have\begin{equation}
\mathcal{W}(\sigma,-1) = \det\begin{bmatrix}
\hat{\mathbf{r}}_{out,\sigma}&\hat{\mathbf{\Psi}}_{out,\sigma}
\end{bmatrix}
(-1).
\end{equation}
By \eqref{eq: expansion of the r outgoing solutions to the whole system}, \eqref{eq: expansion of the psi outgoing solutions to the whole system}, and \eqref{eq: estimate on Z}, we have the estimates \begin{align}
&\begin{aligned}
\hat{\mathbf{r}}_{out,\sigma} = \begin{bmatrix}
    \hat{\mathbf{r}}_{out,\sigma}^{(1)}\\\hat{\mathbf{r}}_{out,\sigma}^{(2)}
\end{bmatrix} = \begin{bmatrix}
    1\\\frac{kq_{k}}{\sigma}
\end{bmatrix}+O(k(-z)^{1-2\epsilon}),
 \end{aligned}\\[1em]
 &\begin{aligned}
 \hat{\mathbf{\Psi}}_{out,\sigma} = &\begin{bmatrix}
\hat{\mathbf{\Psi}}_{out,\sigma}^{(1)}\\\hat{\mathbf{\Psi}}_{out,\sigma}^{(2)}
 \end{bmatrix}
  = \begin{bmatrix}
      0\\1+\frac{\sigma}{q_{k}}
  \end{bmatrix}+O(k(-z)^{1-2\epsilon}),
 \end{aligned}
 \label{key estimate for determining the scattering resonance of L}
\end{align}
Taking $k$ sufficiently small, we have \begin{align*}
    \mathcal{W}(\sigma,-1) &= \det\begin{bmatrix}
        1&0\\\frac{kq_{k}}{\sigma}&\left(1+\frac{\sigma}{q_{k}}\right)
    \end{bmatrix}+O(k)\\&=\left(1+\frac{\sigma}{q_{k}}\right)+O(k)\neq0,\quad \sigma\in\mathbb{I}_{[-1-\eta,-\eta]}\cap\{\vert\Im\sigma\vert\geq 1\}.
\end{align*}
\end{proof}
Moreover, for $\sigma\in\mathbb{I}_{[-1-\eta,-\eta]}$, we can show that $\mathcal{W}(\sigma,-1)$ has exactly one zero point.
\begin{lemma}
\label{lemma: only one zero point}
For $k$ sufficiently small and $\vert\eta\vert\leq\frac{1}{10}$, in the region $\mathbb{I}_{[-1-\eta,-\eta]}\cap\{\vert\Im\sigma\vert\leq 1\}$, $\mathcal{W}(\sigma,-1)$ has only one solution $\sigma = \sigma_{*}$.
\end{lemma}
\begin{proof}
By the estimate \eqref{key estimate for determining the scattering resonance of L}, on the boundary of the region $\mathbb{I}_{[-1-\eta,-\eta]}\cap\{\vert\Im\sigma\vert\leq1\}$, we have\begin{equation}
\det\begin{bmatrix}
\hat{\mathbf{r}}_{out}&\hat{\mathbf{\Psi}}_{out}
\end{bmatrix}(-1)
 = \left(1+\frac{\sigma}{q_{k}}\right)+O(k).
\end{equation}
Since for $\sigma\in\partial\left( \mathbb{I}_{[-1-\eta,-\eta]}\cap\{\vert\Im\sigma\vert\leq1\}\right)$ and $\vert\eta\vert\leq\frac{1}{10}$, we have\begin{equation}
\left\vert1+\frac{\sigma}{q_{k}}\right\vert\geq \frac{1}{2}\eta.
\end{equation}
Hence, we have\begin{equation}
\left\vert \mathcal{W}(\sigma,-1)-\left(1+\frac{\sigma}{q_{k}}\right)\right\vert\leq O(k)\leq\frac{1}{2}\eta\leq\left\vert1+\frac{\sigma}{q_{k}}\right\vert,\quad \sigma\in\partial\left(\mathbb{I}_{[-1-\eta,-\eta]}\cap\{\vert\Im\sigma\vert\leq 1\}\right).
\end{equation}
It is obvious that $\mathcal{W}(\sigma,-1)$ is a holomorphic function of $\sigma$. Hence, according to Rouché Theorem, in the region $\mathbb{I}_{[-1-\eta,-\eta]}\cap\{\vert\Im\sigma\vert\leq1\}$, $\mathcal{W}(\sigma,-1)$ has the same number of zero point as $1+\frac{\sigma}{q_{k}}$. Since $1+\frac{\sigma}{q_{k}}$ has a single zero point at $\sigma = -q_{k}$, we can conclude the proof of this lemma.
\end{proof}
Now, we determine the exact location of the zero of $\mathcal{W}(\sigma,-1)$. We need two more preliminaries. 

First, by our construction of the outgoing and ingoing solutions to the equation $\mathcal{L}(\hat{r}_{p},\hat{\Psi}_{p}) = 0$, any solution $(\hat{r}_{p},\hat{\Psi}_{p})$ will be a linear combination of $\hat{\mathbf{r}}_{out,\sigma},\hat{\mathbf{\Psi}}_{out,\sigma},\hat{\mathbf{r}}_{in,\sigma},\hat{\mathbf{\Psi}}_{in,\sigma}$. In particular, near  $z = 0$, the asymptotic behavior of $(\hat{r}_{p},\hat{\Psi}_{p})$ will be given by the linear combination of the asymptotic behaviors of the outgoing and ingoing solutions. More precisely, the asymptotic behavior of $(\hat{r}_{p},\hat{\Psi}_{p})$ will be\begin{align*}
  \hat{r}_{p}&\approx A_{1} +B_{1}(-z)^{-\frac{\sigma}{q_{k}}}+C_{1}(-z)^{-\frac{k^{2}+\sigma}{q_{k}}},\\
  \hat{\Psi}_{p}&\approx A_{2}+B_{2}(-z)^{-\frac{\sigma}{q_{k}}}+C_{2}(-z)^{-\frac{k^{2}+\sigma}{q_{k}}}.
\end{align*}
Since for $\sigma\in\mathbb{I}_{[-1-\eta,-\eta]}$, if $A_{1}\neq0,A_{2}\neq 0$, the asymptotic behavior of the outgoing branch will be the leading order. In other words, it is not enough to tell whether a solution $(\hat{r}_{p},\hat{\Psi}_{p})$ to the equation $\mathcal{L}(\hat{r}_{p},\hat{\Psi}_{p}) = 0$ is outgoing or not by only looking at the leading order asymptotic behavior at $z = 0$.  In the analysis of the scattering resonance, to show $\mathcal{W}(\sigma,-1) = 0$ for some $\sigma$, our approach is to construct a special solution to $\mathcal{L}(\hat{r}_{p},\hat{\Psi}_{p}) = 0$ that satisfies the Dirichlet boundary condition at $z  = -1$ and the outgoing boundary condition at $z = 0$. In view of this purpose, we shall establish a sharp criterion for functions $(f,g)$ to satisfy the outgoing boundary condition. Therefore, we provide the following lemma, which connects the outgoing solution with its regularity on the outgoing cone $\{u = -1\}$ under self-similar coordinates.
\begin{lemma}[Regularity perspective on outgoing boundary conditions]
\label{lemma: regularity perspective}
Fix $\sigma\in\mathbb{I}_{[-1-\eta,-\eta]}\backslash\{-q_{k}\}$. Then $(\hat{r}_{p},\hat{\Psi}_{p})$ satisfies the outgoing boundary condition if and only if $(\hat{r}_{p},\hat{\Psi}_{p})\in C_{z}^{\alpha}$ up to $z = 0$, where $\alpha$ is the smallest integer larger than $-\frac{\sigma+k^{2}}{q_{k}}$.
\end{lemma}

\begin{proof}
By our construction of the outgoing and ingoing solutions to $\mathcal{L}(\hat{r}_{p},\hat{\Psi}_{p}) = 0$, we can expand the solution $(\hat{r}_{p},\hat{\Psi}_{p})$ as\begin{equation}
\begin{bmatrix}
\hat{r}_{p}\\\hat{\Psi}_{p}
\end{bmatrix}
= A\hat{\mathbf{r}}_{out,\sigma}+B\hat{\mathbf{\Psi}}_{out,\sigma}+C\hat{\mathbf{r}}_{in,\sigma}+D\hat{\mathbf{\Psi}}_{in,\sigma}.
\end{equation}
Using the expansion of the outgoing and ingoing solutions, for $\sigma\neq-q_{k}$, we have\begin{align*}
    &\hat{\mathbf{r}}_{out} = \begin{bmatrix}
    \mathcal{E}_{1}(z)\\\mathcal{O}_{1}(z)
    \end{bmatrix},
    \quad \hat{\mathbf{\Psi}}_{out} = \begin{bmatrix}
    \mathcal{E}_{2}(z)\\\mathcal{O}_{2}(z)
    \end{bmatrix},\\[2em]&
    \hat{\mathbf{r}}_{in} = (-z)^{-\frac{k^{2}+\sigma}{q_{k}}}\begin{bmatrix}
    \mathcal{E}_{3}(z)\\\mathcal{O}_{3}(z)
    \end{bmatrix},\quad 
    \hat{\mathbf{\Psi}}_{in} = (-z)^{-\frac{\sigma}{q_{k}}}\begin{bmatrix}
    \mathcal{E}_{4}(z)\\\mathcal{O}_{4}(z)
    \end{bmatrix},
\end{align*}
where $\mathcal{E}_{i}(z)$ and $\mathcal{O}_{i}(z)$ are $C^{2}$ functions up to $z = 0$. It is immediate to see that if $(\hat{r}_{p},\hat{\Psi}_{p})$ is outgoing, then $(\hat{r}_{p},\hat{\Psi}_{p})$ is $C^{2}$ up to $z = 0$. However, if $(\hat{r}_{p},\hat{\Psi}_{p})$ has the ingoing branch, then due to the presence of $(-z)^{-\frac{\sigma}{q_{k}}}$, the regularity of $(\widetilde{r}_{p},\widetilde{\Psi}_{p})$ at $z = 0$ will be strictly less than $C_{z}^{\alpha}$. Hence, we can conclude the proof of this lemma.
\end{proof}
Second, as we have mentioned above, to show that the solutions $\hat{\mathbf{r}}_{out,\sigma},\hat{\mathbf{\Psi}}_{out,\sigma},\hat{\mathbf{r}}_{dir,\sigma},\hat{\mathbf{\Psi}}_{dir,\sigma}$ are linearly dependent for some $\sigma\in\mathbb{I}_{[-1-\eta,-\eta]}$, we use some special solutions to the equation $\mathcal{L}(\hat{r}_{p},\hat{\Psi}_{p}) = 0$. The following lemma provides a family of special solutions we will need.
\begin{lemma}
For any $c\in\mathbb{R}$, $(r_{p},\Psi_{p},m_{p}) = (0,cr_{k},0)$ is a special solution to the linearized Einstein-scalar field equations \eqref{eq:linearized r equation}-\eqref{eq:linearized dvm equation}.
\end{lemma}
\begin{proof}
We observe that the Einstein-scalar field equations \eqref{eq:u-Ray equation}-\eqref{eq:wave equation for phi} are translation invariant. More precisely, since only the derivatives of $\phi$ appear in the equations, if $(r,\phi,m)$ is a solution, then for any $c\in\mathbb{R}$, $(r,\phi+c,m)$ is also a solution. Thus, in the $(r_{p},\phi_{p},m_{p})$-formulation of the linearized system, $(r_{p},\phi_{p},m_{p}) = (0,c,0)$ is a special solution. Since $\Psi_{p} = r_{k}\phi_{p}+\frac{\Psi_{k}}{r_{k}}r_{p}$, we have $(r_{p},\Psi_{p},m_{p}) = (0,cr_{k},0)$ is a special solution to the linearized equations \eqref{eq:linearized r equation}-\eqref{eq:linearized dvm equation}.
\end{proof}

Now, we can determine the precise location of the zero of $\mathcal{W}(\sigma,-1)$. We have the following proposition:
\begin{proposition}
\label{prop: only one zero for the wronski}
    Let $k$ be sufficiently small. In the region $\mathbb{I}_{[-1-\eta,-\eta]}$, $\mathcal{W}(\sigma,-1)$ has only one zero at $\sigma = -1$. 
\end{proposition}
\begin{proof}
Recall that $r_{k} = (-u)\mr{r}(z) =e^{-s}\mr{r}(z)$. Substituting $(r_{p},\Psi_{p},m_{p}) = (0,r_{k},0)$ in the equation \eqref{eq:wave eq for rp under self-similar coordinates}-\eqref{eq:v transport eq for m under self-similar coordinates}, we have\begin{equation}
q_{k}z\frac{d^{2}\mr{r}}{dz^{2}}+(q_{k}-1)\frac{d\mr{r}}{dz}-\frac{\mu_{k}}{1-\mu_{k}}\frac{(\partial_{s}r_{k}+q_{k}z\partial_{z}r_{k})\partial_{z}r_{k}}{r_{k}^{2}}\mr{r} = 0
\end{equation}
Therefore, $(\hat{r}_{p},\hat{\Psi}_{p},\hat{m}_{p}) = (0,\mr{r}(z),0)$ solves the equations \eqref{eq:scattering equation for rp}-\eqref{scattering equation for mp} with $\sigma = -1$. Then $(\hat{r}_{p},\hat{\Psi}_{p}) = (0,\mr{r}(z))$ is a special solution to the equation $\mathcal{L}(\hat{r}_{p},\hat{\Psi}_{p}) = 0$. Moreover, since $z = -1$ corresponds to the center of the spacetime, we have $(\hat{r}_{p},\hat{\Psi}_{p})(-1) = (0,\mr{r}(-1)) = (0,0)$, which means that $(0,\mr{r}(z))$ satisfies the Dirichlet boundary condition. Since $\mr{r}\in C^{2}_{z}$, we have that $(0,\mr{r})$ also satisfies the outgoing boundary condition by Lemma \ref{lemma: regularity perspective}. 

Then by our construction, $\begin{bmatrix}
0\\\mr{r}(z)
\end{bmatrix}$
is a multiple of $\hat{\mathbf{\Psi}}_{out,-1}$ and a multiple of $\hat{\mathbf{\Psi}}_{dir,-1}$. Therefore, $\hat{\mathbf{\Psi}}_{out,\sigma}$ and $\hat{\mathbf{\Psi}}_{dir,\sigma}$ is linearly dependent at $\sigma = -1$. Then we have $\mathcal{W}(-1,-1) = 0$. Since by Lemma \ref{lemma: only one zero point}, $\mathcal{W}(\sigma,-1)$ has only one zero in $\mathbb{I}_{[-1-\eta,-\eta]}$, we can conclude the proof of this Proposition.
\end{proof}
The above proof motivates the following definition:
\begin{definition}
For $\sigma = -1$, define two constants $a_{*}$ and $b_{*}$ such that $\hat{\mathbf{\Psi}}_{out,-1} = (0,a_{*}\mr{r})$ and $\hat{\mathbf{\Psi}}_{dir,-1} = (0,b_{*}\mr{r})$.  
\end{definition}

\subsection{Expansion of scattering resolvent}
In this section, we establish the expansion of the scattering resolvent $R(\sigma)(F_{1},F_{2})$ in terms of $\sigma$ for $(F_{1},F_{2})\in\partial_{z}\mathcal{C}_{N}^{\widetilde{\beta}}\times\partial_{z}\mathcal{C}_{N}^{\beta}$ with $p_{k}<\beta<\frac{3}{2}$ and $\beta+\frac{1}{4}<\widetilde{\beta}<2$.

We have the following proposition.
\begin{proposition}
\label{prop: high regularity resolvent estimate}
Let $k$ be sufficiently small, $p_{k}<\beta<\frac{3}{2}$, and $\beta+\frac{1}{4}<\widetilde{\beta}<2$. For fixed $(F_{1},F_{2})\in\mathcal{C}_{N}^{\widetilde{\beta}}\times \mathcal{C}_{N}^{\beta}$, the cut-off resolvent $\rho(z)R(\sigma)(\partial_{z}F_{1},\partial_{z}F_{2})(z)$ is smooth in $(\sigma,z)$, meromorphic in $\sigma$ for fixed $z$. We have the following expansion of $\rho(z)R(\sigma)(\partial_{z}F_{1},\partial_{z}F_{2})$:\begin{equation}
\rho(z)R(\sigma)(\partial_{z}F_{1},\partial_{z}F_{2})(z) = \frac{\mathcal{R}(\sigma,z)}{\sigma+1},
\end{equation}
where $\mathcal{R}$ is smooth in $(\sigma,z)$, holomorphic in $\sigma$ for fixed $z$, and bounded uniformly in terms of the $\mathcal{C}_{N}^{\widetilde{\beta}}\times\mathcal{C}_{N}^{\beta}$-norm of $(F_{1},F_{2})$. 

Moreover, $\mathcal{R}(-1,z)$ can be computed explicitly. There exist a constant $c_{*}$ and two functions $(\widetilde{G}_{-1})_{23}$ and $(\widetilde{G}_{-1})_{24}\neq0$ such that\begin{align}
\bigl(\mathcal{R}_{1}(-1,z)\bigr)_{1} &= 0,\label{eq: expansion of R}\\
\bigl(\mathcal{R}_{1}(-1,z)\bigr)_{2}& = c_{*}\rho(z)\mr{r}(z)\int_{-1}^{0}(\widetilde{G}_{-1})_{23}\partial_{z}^{2}F_{1}+(\widetilde{G}_{-1})_{24}\partial_{z}^{2}F_{2}dz^{\prime}.\label{eq: expansion of R2}
\end{align}
Finally, we have the following estimate for $\rho(z)R(\sigma)(\partial_{z}F_{1},\partial_{z}F_{2})(z)$:\begin{equation}
\left\vert\frac{d^{j}}{dz^{j}}\left(\rho(z)R(\sigma)(\partial_{z}F_{1},\partial_{z}F_{2})\right)\right\vert\lesssim_{k}\vert\sigma\vert^{-1+j}\left(1+\frac{1}{\vert\sigma+1\vert}\right)\frac{1}{\beta_{*}+\Re\sigma}\left(\Vert F_{1}\Vert_{\mathcal{C}_{N}^{\widetilde{\beta}}}+\Vert F_{2}\Vert_{\mathcal{C}_{N}^{\beta}}\right),\label{eq: resolvent estimate}
\end{equation}
where $\beta_{*}: = \beta q_{k}$.
\end{proposition}
\begin{proof}
By Proposition \ref{prop: only one zero for the wronski}, we have the following expansion of $\mathcal{W}(\sigma,z)$:\begin{equation}
\mathcal{W}(\sigma,z) = (\sigma+1)\mathcal{W}_{r}(\sigma)W(\sigma,z),
\end{equation}
where $\mathcal{W}_{r}(\sigma)$ is an analytic function in $\sigma$ and non-vanishing for $\sigma\in\mathbb{I}_{[-1-\eta,-\eta]}$. Then $G_{\sigma}$ can be expressed as \begin{equation*}
    \begin{bmatrix}
        \left(G_{\sigma}^{\prime}\right)_{13}&\left(G_{\sigma}^{\prime}\right)_{14}\\\left(G_{\sigma}^{\prime}\right)_{23}&\left(G_{\sigma}^{\prime}\right)_{24}
    \end{bmatrix}
     = \frac{1}{q_{k}z}\frac{1}{(\sigma+1)\mathcal{W}_{r}(\sigma)W(\sigma,z)}
     \begin{bmatrix}
         [M_{\sigma}]_{31}&[M_{\sigma}]_{41}\\
         [M_{\sigma}]_{32}&[M_{\sigma}]_{42}
     \end{bmatrix}=:\frac{1}{(\sigma+1)\mathcal{W}_{r}(\sigma)}\begin{bmatrix}
         \left(\widetilde{G}_{\sigma}^{\prime}\right)_{13}&\left(\widetilde{G}_{\sigma}^{\prime}\right)_{14}\\[1em]\left(\widetilde{G}_{\sigma}^{\prime}\right)_{23}&\left(\widetilde{G}_{\sigma}^{\prime}\right)_{24}
     \end{bmatrix}.
\end{equation*}
Hence, we can renormalize the resolvent kernel $R_{\sigma}(z,z^{\prime})$ and $\widetilde{R}_{\sigma}(z,z^{\prime})$ as \begin{equation*}
    R_{\sigma}^{(nor)}(z,z^{\prime}) = (\sigma+1)\mathcal{W}_{r}(\sigma)R_{\sigma}(z,z^{\prime}),\quad \widetilde{R}_{\sigma}^{(nor)}(z,z^{\prime}) = (\sigma+1)\mathcal{W}_{r}(\sigma)\widetilde{R}_{\sigma}(z,z^{\prime}).
\end{equation*}
Then we have \begin{equation}
\begin{aligned}
&\lim_{\sigma\rightarrow-1}(\sigma+1)\rho(z)R(\sigma)(\partial_{z}F_{1},\partial_{z}F_{2})(z)\\ = &\frac{1}{\mathcal{W}_{r}(-1)}\lim_{\sigma\rightarrow-1}\int_{-1}^{0}\rho(z)R_{\sigma}^{(nor)}(z,z^{\prime})\begin{bmatrix}
    \partial_{z}F_{1}\\
    \partial_{z}F_{2}
\end{bmatrix}(z^{\prime})+\rho(z)\widetilde{R}_{\sigma}^{(nor)}(z,z^{\prime})\begin{bmatrix}
    \partial_{z}^{2}F_{1}\\
    \partial_{z}^{2}F_{2}
\end{bmatrix}dz^{\prime}\\=&\frac{1}{\mathcal{W}_{r}(-1)}\left(\lim_{\sigma\rightarrow-1}\int_{-1}^{z} \rho(z)R_{\sigma}^{(nor)}(z,z^{\prime})\begin{bmatrix}
    \partial_{z}F_{1}\\\partial_{z}F_{2}
\end{bmatrix}(z^{\prime})+\rho(z)\widetilde{R}_{\sigma}^{(nor)}(z,z^{\prime})\begin{bmatrix}
    \partial_{z}^{2}F_{1}\\\partial_{z}^{2}F_{2}
\end{bmatrix}dz^{\prime}\right)\\&+\frac{1}{\mathcal{W}_{r}(-1)}\left(\lim_{\sigma\rightarrow-1}\int_{z}^{0}\rho(z)\widetilde{R}_{\sigma}^{(nor)}(z,z^{\prime})\begin{bmatrix}
\partial_{z}^{2}F_{1}\\\partial_{z}^{2}F_{2}
\end{bmatrix}(z^{\prime})dz^{\prime}\right).
\end{aligned}
\label{eq: limiting equation of sigma}
\end{equation}
By the dominated convergence theorem, on the compact region $[-1,z]$, we have\begin{align*}
    &\lim_{\sigma\rightarrow-1}\int_{-1}^{z}\rho(z)R_{\sigma}^{(nor)}(z,z^{\prime})\begin{bmatrix}
    \partial_{z}F_{1}\\\partial_{z}F_{2}
    \end{bmatrix}(z^{\prime})+\rho(z)\widetilde{R}_{\sigma}^{(nor)}(z,z^{\prime})\begin{bmatrix}
        \partial_{z}^{2}F_{1}\\ \partial_{z}^{2}F_{2}
    \end{bmatrix}dz^{\prime} \\=& \int_{-1}^{z}\rho(z)R_{-1}^{(nor)}(z,z^{\prime})\begin{bmatrix}
    \partial_{z}F_{1}\\\partial_{z}F_{2}
    \end{bmatrix}(z^{\prime})+\rho(z)\widetilde{R}_{-1}^{(nor)}(z,z^{\prime})\begin{bmatrix}
        \partial_{z}^{2}F_{1}\\\partial_{z}^{2}F_{2}
    \end{bmatrix}dz^{\prime}.
\end{align*}
For the last term on the right-hand side of \eqref{eq: limiting equation of sigma}, since the resolvent kernel might be singular when $z^{\prime}\rightarrow 0$ for $\sigma\in\mathbb{I}_{[-1-\eta,-\eta]}$, we first analyze the behavior of $\widetilde{R}_{\sigma}^{(nor)}$. Assume that \begin{align*}
   &\hat{\mathbf{r}}_{dir,\sigma} = \bar{\zeta}_{1}(\sigma)\hat{\mathbf{r}}_{out,\sigma}+\bar{\zeta}_{2}(\sigma)\hat{\mathbf{\Psi}}_{out,\sigma}+\bar{\zeta}_{3}(\sigma)\hat{\mathbf{r}}_{in,\sigma}+\bar{\zeta}_{4}(\sigma)\hat{\mathbf{\Psi}}_{in,\sigma},\\&
\hat{\mathbf{\Psi}}_{dir,\sigma} = \widetilde{\zeta}_{1}(\sigma)\hat{\mathbf{r}}_{out,\sigma}+\widetilde{\zeta}_{2}(\sigma)\hat{\mathbf{\Psi}}_{out,\sigma}+\widetilde{\zeta}_{3}(\sigma)\hat{\mathbf{r}}_{in,\sigma}+\widetilde{\zeta}_{4}(\sigma)\hat{\mathbf{\Psi}}_{in,\sigma},
\end{align*} 
where $\zeta_{i}$ and $\widetilde{\zeta}_{i}$ are uniformly bounded in the compact region of $\sigma$. Moreover, due to the estimates on $\hat{\mathbf{r}}_{dir,\sigma},\hat{\mathbf{\Psi}}_{dir,\sigma}$ \eqref{eq: estimates on the dirichlet solutions} and the estimates on $\hat{\mathbf{r}}_{out,\sigma},\hat{\mathbf{\Psi}}_{out,\sigma}$ \eqref{eq: estimates on the r outgoing solutions}-\eqref{eq: estimates on the psi outgoing solutions}, we have the following estimates on $\bar{\zeta}_{i}$ and $\widetilde{\zeta}_{i}$: \begin{equation*}
    \left\vert\bar{\zeta}_{2}\right\vert\lesssim \frac{1}{\vert\sigma\vert^{2}},\quad \left\vert\widetilde{\zeta}_{2}\right\vert\lesssim \frac{1}{\vert\sigma\vert^{2}}\quad,\left\vert\bar{\zeta}_{i}\right\vert\lesssim \frac{1}{\vert\sigma\vert},\quad \left\vert\widetilde{\zeta}_{i}\right\vert\lesssim\frac{1}{\vert\sigma\vert}, i = 1,3,4.
\end{equation*}
Then for $[M_{\sigma}]_{31}$, we have
\begin{align*}
    [M_{\sigma}]_{31} (z^{\prime})=& \det\begin{bmatrix}
        \hat{\mathbf{\Psi}}_{dir,\sigma}^{(1)}&\hat{\mathbf{r}}_{out,\sigma}^{(1)}&\hat{\mathbf{\Psi}}_{out,\sigma}^{(1)}\\
    \hat{\mathbf{\Psi}}_{dir,\sigma}^{(2)}&\hat{\mathbf{r}}_{out,\sigma}^{(2)}&\hat{\mathbf{\Psi}}_{out,\sigma}^{(2)}\\
    \hat{\mathbf{\Psi}}_{dir,\sigma}^{\prime(2)}&\hat{\mathbf{r}}_{out,\sigma}^{\prime(2)}&\hat{\mathbf{\Psi}}_{out,\sigma}^{\prime(2)}
    \end{bmatrix}(z^{\prime})\\=&
    \widetilde{\zeta}_{3}(\sigma)\det\begin{bmatrix}
\hat{\mathbf{r}}_{in,\sigma}^{(1)}&\hat{\mathbf{r}}_{out,\sigma}^{(1)}&\hat{\mathbf{\Psi}}_{out,\sigma}^{(1)}\\\hat{\mathbf{r}}_{in,\sigma}^{(2)}&\hat{\mathbf{r}}_{out,\sigma}^{(2)}&\hat{\mathbf{\Psi}}_{out,\sigma}^{(2)}\\\hat{\mathbf{r}}_{in,\sigma}^{\prime(2)}&\hat{\mathbf{r}}_{out,\sigma}^{\prime(2)}&\hat{\mathbf{\Psi}}_{out,\sigma}^{\prime(2)}
    \end{bmatrix}(z^{\prime})+\widetilde{\zeta}_{4}(\sigma)\det\begin{bmatrix}
        \hat{\mathbf{\Psi}}_{in,\sigma}^{(1)}&\hat{\mathbf{r}}^{(1)}_{out,\sigma}&\hat{\mathbf{\Psi}}_{out,\sigma}^{(1)}\\
        \hat{\mathbf{\Psi}}_{in,\sigma}^{(2)}&\hat{\mathbf{r}}_{out,\sigma}^{(2)}&\hat{\mathbf{\Psi}}_{out,\sigma}^{(2)}\\
        \hat{\mathbf{\Psi}}_{in,\sigma}^{\prime(2)}&\hat{\mathbf{r}}_{out,\sigma}^{\prime(2)}&\hat{\mathbf{\Psi}}_{out,\sigma}^{\prime(2)}
        \end{bmatrix}(z^{\prime}).
\end{align*}
Then, by Corollary \ref{coro: construction of the true ingoing solution} and Proposition \ref{prop: construction of the true outgoing solution}, we have \begin{equation*}
\left\vert[M_{\sigma}]_{31}\right\vert\lesssim_{k}\vert\sigma\vert (-z^{\prime})^{-\frac{\sigma+k^{2}}{q_{k}}-1}.
\end{equation*}
Similarly, we can show that 
\begin{equation}
\begin{aligned}
    \left\vert [M_{\sigma}]_{ij}\right\vert&\lesssim_{k}\vert\sigma\vert (-z^{\prime})^{-\frac{\sigma+k^{2}}{q_{k}}-1},\quad i = 3,4,\ j = 1,2,\\
    \left\vert[M_{\sigma}]_{i3}\right\vert&\lesssim_{k} (-z^{\prime})^{-\frac{\sigma+k^{2}}{q_{k}}-1},\ i =3,4,\quad \left\vert[M_{\sigma}]_{i4}\right\vert\lesssim_{k}\frac{1}{\vert\sigma\vert}(-z^{\prime})^{-\frac{\sigma+k^{2}}{q_{k}}-1},\ i = 3,4.
\end{aligned}
\label{eq: estimates on M}
\end{equation}
Hence, by the estimates \eqref{eq: estimates on the r outgoing solutions}, \eqref{eq: estimates on the psi outgoing solutions}, \eqref{eq: estimates on the dirichlet solutions}, and \eqref{estimate for the wronskian}, we have \begin{equation*}
    \left\vert\frac{1}{q_{k}zW(\sigma,z)}\begin{bmatrix}
        [M_{\sigma}]_{31}&[M_{\sigma}]_{41}\\ [M_{\sigma}]_{32}&[M_{\sigma}]_{42}
    \end{bmatrix}(z)\right\vert\lesssim 
    \vert\sigma\vert(-z)^{\frac{\sigma}{q_{k}}}.
\end{equation*}
Then, we can choose $G_{\sigma}$ in such a way that we have \begin{equation*}
    \left\vert\begin{bmatrix}
        (\widetilde{G}_{\sigma})_{13}&(\widetilde{G}_{\sigma})_{14}\\(\widetilde{G}_{\sigma})_{23}&(\widetilde{G}_{\sigma})_{24}
    \end{bmatrix}\right\vert\lesssim (-z)^{\frac{\sigma}{q_{k}}+1}.
\end{equation*}
Since $(F_{1},F_{2})\in\mathcal{C}_{N}^{\widetilde{\beta}}\times\mathcal{C}^{\beta}_{N}$, we have that \begin{equation*}
    \vert \partial_{z}^{2}F_{1}\vert+\vert \partial_{z}^{2}F_{2}\vert\lesssim (-z)^{\beta-2}.
\end{equation*}
Hence, we have that \begin{equation*}
    \begin{bmatrix}
        (\widetilde{G}_{\sigma})_{13}&(\widetilde{G}_{\sigma})_{14}\\(\widetilde{G}_{\sigma})_{23}&(\widetilde{G}_{\sigma})_{24}
    \end{bmatrix}\begin{bmatrix}
        \partial_{z}^{2}F_{1}\\\partial_{z}^{2}F_{2}
    \end{bmatrix}
    \lesssim (-z)^{\beta+\frac{\sigma}{q_{k}}-1}
\end{equation*}
is integrable for $\sigma\in\mathbb{I}_{[-\beta_{*},-\eta]}$. Then, by the dominated convergence theorem, we have \begin{align*}
    \lim_{\sigma\rightarrow-1}\int_{z}^{0}\rho(z)\widetilde{R}_{\sigma}^{(nor)}(z,z^{\prime})\begin{bmatrix}
    \partial_{z}^{2}F_{1}\\\partial_{z}^{2}F_{2}
    \end{bmatrix}dz^{\prime} = &\lim_{\sigma\rightarrow-1}\int_{z}^{0}\rho(z)\left(\begin{bmatrix}
        \hat{\mathbf{r}}_{dir,\sigma}&\hat{\mathbf{\Psi}}_{dir,\sigma}
    \end{bmatrix}\begin{bmatrix}
        (\widetilde{G}_{\sigma})_{13}&(\widetilde{G}_{\sigma})_{14}\\(\widetilde{G}_{\sigma})_{23}&(\widetilde{G}_{\sigma})_{24}
    \end{bmatrix}\begin{bmatrix}
        \partial_{z}^{2}F_{1}\\\partial_{z}^{2}F_{2}
    \end{bmatrix}\right)(z^{\prime})dz^{\prime}\\=&\int_{z}^{0}\rho(z)\left(\begin{bmatrix}
        \hat{\mathbf{r}}_{dir,-1}&\hat{\mathbf{\Psi}}_{dir,-1}
    \end{bmatrix}\begin{bmatrix}
        (\widetilde G_{-1})_{13}&(\widetilde{G}_{-1})_{14}\\(\widetilde{G}_{-1})_{23}&(\widetilde{G}_{-1})_{24}
    \end{bmatrix}\begin{bmatrix}
        \partial_{z}^{2}F_{1}\\\partial_{z}^{2}F_{2}
    \end{bmatrix}\right)(z^{\prime})dz^{\prime}.
\end{align*}
Then we have
\begin{align*}
&\lim_{\sigma\rightarrow-1}(\sigma+1)\rho(z)R(\sigma)(\partial_{z}F_{1},\partial_{z}F_{2})(z)\\=&
\frac{\rho(z)}{\mathcal{W}_{r}(-1)}\int_{-1}^{z}R_{-1}^{(nor)}(z,z^{\prime})\begin{bmatrix}
    \partial_{z}F_{1}\\\partial_{z}F_{2}
\end{bmatrix}(z^{\prime})+\widetilde{R}_{-1}^{(nor)}(z,z^{\prime})\begin{bmatrix}
    \partial_{z}^{2}F_{1}\\\partial_{z}^{2}F_{2}
\end{bmatrix}(z^{\prime})dz^{\prime}+\frac{\rho(z)}{\mathcal{W}_{r}(-1)}\int_{z}^{0}\widetilde{R}_{-1}^{(nor)}(z,z^{\prime})\begin{bmatrix}
    \partial_{z}^{2}F_{1}\\\partial_{z}^{2}F_{2}
\end{bmatrix}(z^{\prime})dz^{\prime}.
\end{align*}
Therefore, we can express the resolvent operator $\rho(z)R(\sigma)(\partial_{z}F_{1},\partial_{z}F_{2})$ as $\rho(z)R(\sigma)(\partial_{z}F_{1},\partial_{z}F_{2}) = \frac{\mathcal{R}(\sigma,z)}{(\sigma+1)}$.

Next, we compute the explicit structure of $\mathcal{R}(-1,z)$. When $\sigma = -1$, since $\hat{\mathbf{\Psi}}_{dir,-1}$ is linearly dependent to $\hat{\mathbf{\Psi}}_{out,-1}$, we have \begin{equation*}
    [M_{-1}]_{31} = [M_{-1}]_{41} = 0.
\end{equation*}
We can also compute $[M_{-1}]_{32}$ and $[M_{-1}]_{42}$ explicitly: \begin{align*}
    [M_{-1}]_{32} = \det\begin{bmatrix}
    \hat{\mathbf{r}}_{dir,-1}^{(1)}&\hat{\mathbf{r}}_{out,-1}^{(1)}&0\\
    \hat{\mathbf{r}}_{dir,-1}^{(2)}&\hat{\mathbf{r}}_{out,-1}^{(2)}&a_{*}\mr{r}(z)\\
    \hat{\mathbf{r}}^{\prime(2)}_{dir,-1}&\hat{\mathbf{r}}^{\prime(2)}_{out,-1}&a_{*}\mr{r}^{\prime}(z)
    \end{bmatrix},\quad 
    [M_{-1}]_{42} = \det\begin{bmatrix}
\hat{\mathbf{r}}_{dir,-1}^{(1)}&\hat{\mathbf{r}}_{out,-1}^{(1)}&0\\
\hat{\mathbf{r}}_{dir,-1}^{(2)}&\hat{\mathbf{r}}_{out,-1}^{(2)}&a_{*}\mr{r}(z)\\
\hat{\mathbf{r}}^{\prime(1)}_{dir,-1}&\hat{\mathbf{r}}_{out,-1}^{\prime(1)}&0
    \end{bmatrix}.
\end{align*}
Let $Z_{1} = [M_{-1}]_{32}$ and $Z_{2} = [M_{-1}]_{42}$. We show that $Z_{2}$ is a non-trivial function. For $[M_{-1}]_{42}$, we have \begin{equation*}
    [M_{-1}]_{42} = -a_{*}\mr{r}(z)\det\begin{bmatrix}
        \hat{\mathbf{r}}_{dir,-1}^{(1)}&\hat{\mathbf{r}}_{out,-1}^{(1)}\\\hat{\mathbf{r}}_{dir,-1}^{\prime(1)}&\hat{\mathbf{r}}_{out,-1}^{\prime(1)}
    \end{bmatrix}.
\end{equation*}
By \eqref{eq: expansion of the r outgoing solutions to the whole system}, we have $\hat{\mathbf{r}}_{out,\sigma}^{(1)}(-1)\neq0$. Then we have that $(\hat{\mathbf{r}}_{out,-1}^{(1)},\hat{\mathbf{r}}_{out,-1}^{\prime(1)})$ is linearly independent of $(\hat{\mathbf{r}}_{dir,-1}^{(1)},\hat{\mathbf r}_{dir,-1}^{\prime(1)})$ in a neighborhood of $z = -1$. Hence, the function $Z_{2}$ is non-trivial. Then we have
\begin{equation*}
\begin{bmatrix}
    (\widetilde{G}_{-1}^{\prime})_{13}&(\widetilde G_{-1}^{\prime})_{14}\\(\widetilde G_{-1}^{\prime})_{23}&(\widetilde G_{-1}^{\prime})_{24}
\end{bmatrix} = \frac{1}{q_{k}z}\frac{1}{W(-1,z)}\begin{bmatrix}
    0&0\\Z_{1}&Z_{2}
\end{bmatrix}.
\end{equation*}
Moreover, considering the asymptotic expansion of $\hat{\mathbf{r}}_{dir,-1}$ at $z = 0$, we have \begin{equation}
    \vert Z_{2}\vert\lesssim_{k} 1. \label{estimate on Z2}
\end{equation}
Similarly, we can derive the following bound for $Z_{1}$: \begin{equation}
    \label{eq: estimate on Z1}
    \left\vert Z_{1}\right\vert\lesssim_{k}1.
\end{equation}
Therefore, for $z\leq z^{\prime}<0$, we have \begin{equation*}
    \widetilde{R}^{(nor)}_{-1}(z,z^{\prime}) = \begin{bmatrix}
        0&0\\b_{*}\mr{r}(z)(\widetilde{G}_{-1})_{23}&b_{*}\mr{r}(z)(\widetilde{G}_{-1})_{24}
    \end{bmatrix}.
\end{equation*}
Similarly, when $-1\leq z^{\prime}\leq z$, for $R_{-1}$ we can deduce \begin{align*}
    &R_{-1}^{(nor)}(z,z^{\prime})\\= &\begin{bmatrix}
        \hat{\mathbf{r}}_{out,-1}^{(1)}&0\\
        \hat{\mathbf{r}}_{out,-1}^{(2)}&a_{*}\mr{r}
    \end{bmatrix}(z)\left(\frac{1}{a_{*}\hat{\mathbf{r}}^{(1)}_{out,-1}\mr{r}}\begin{bmatrix}
        a_{*}\mr{r}&0\\-\hat{\mathbf{r}}_{out,-1}^{(2)}&\hat{\mathbf{r}}_{out,-1}^{(2)}
    \end{bmatrix}\begin{bmatrix}
        \hat{\mathbf{r}}_{dir,-1}^{(1)}&0\\\hat{\mathbf{r}}_{dir,-1}^{(2)}&b_{*}\mr{r}\end{bmatrix}(z^{\prime})\right)^{\prime}\begin{bmatrix}
            0&0\\(\widetilde{G}_{-1})_{23}&(\widetilde{G}_{-1})_{24}(z^{\prime})
        \end{bmatrix}\\=&
        0.
\end{align*}
For $\widetilde{R}_{-1}^{(nor)}(z,z^{\prime})$ when $-1\leq z^{\prime}\leq z$, we have \begin{align*}
    \widetilde{R}_{-1}(z,z^{\prime}) = \begin{bmatrix}
        0&0\\b_{*}\mr{r}(\widetilde{G}_{-1})_{23}&b_{*}\mr{r}(\widetilde{G}_{-1})_{24}
    \end{bmatrix}.
\end{align*}
Letting $c_{*} = \frac{b_{*}}{\mathcal{W}_{r}(-1)}$ concludes the proof of \eqref{eq: expansion of R}-\eqref{eq: expansion of R2}.

Lastly, we establish the resolvent estimate \eqref{eq: resolvent estimate}. Taking derivatives, we have \begin{align*}
    &\frac{d^{j}}{dz^{j}}\int_{-1}^{0}\rho(z)R_{\sigma}(z,z^{\prime})\begin{bmatrix}
        \partial_{z}F_{1}\\\partial_{z}F_{2}
    \end{bmatrix}(z^{\prime})+\rho(z)\widetilde{R}_{\sigma}(z,z^{\prime})\begin{bmatrix}
        \partial_{z}^{2}F_{1}\\\partial_{z}^{2}F_{2}
    \end{bmatrix}dz^{\prime}\\=&
    \int_{-1}^{z}\frac{d^{j}}{dz^{j}}\left(\rho(z)R_{\sigma}(z,z^{\prime})\begin{bmatrix}
       \partial_{z} F_{1}\\\partial_{z}F_{2}
    \end{bmatrix}(z^{\prime})\right)+\frac{d^{j}}{dz^{j}}\left(\rho(z)\widetilde{R}_{\sigma}(z,z^{\prime})\right)\begin{bmatrix}
        \partial_{z}^{2}F_{1}\\\partial_{z}^{2}F_{2}\end{bmatrix}dz^{\prime}\\&+\int_{z}^{0}\frac{d^{j}}{dz^{j}}\left(\rho(z)\widetilde{R}_{\sigma}(z,z^{\prime})\begin{bmatrix}
        \partial_{z}^{2}F_{1}\\\partial_{z}^{2}F_{2}
    \end{bmatrix}(z^{\prime})\right)dz^{\prime}.
\end{align*}
We have the following estimate on $G_{\sigma}$\begin{equation}
    \left\vert
    \begin{bmatrix}(G_{\sigma})_{13}&(G_{\sigma})_{14}\\(G_{\sigma})_{23}&(G_{\sigma})_{24}
    \end{bmatrix}\right\vert\lesssim_{k} \frac{1}{\vert\sigma+1\vert}(-z)^{\frac{\sigma}{q_{k}}+1}.
    \label{eq: estimate on Gsigma}
\end{equation}

For $-1\leq z^{\prime}\leq z$, we have \begin{align*}
    &\frac{d^{j}}{dz^{j}}\left(\rho(z)\widetilde R_{\sigma}(z,z^{\prime})\begin{bmatrix}
        \partial_{z}^{2}F_{1}\\\partial_{z}^{2}F_{2}
    \end{bmatrix}\right)\\ =& \frac{d^{j}}{dz^{j}}\left(\rho(z)\begin{bmatrix}
        \hat{\mathbf{r}}_{out,\sigma}&\hat{\mathbf{\Psi}}_{out,\sigma}
    \end{bmatrix}(z)\right)\begin{bmatrix}
        \hat{\mathbf{r}}_{out,\sigma}&\hat{\mathbf{\Psi}}_{out,\sigma}
    \end{bmatrix}^{-1}\begin{bmatrix}
        \hat{\mathbf{r}}_{dir,\sigma}&\hat{\mathbf{\Psi}}_{dir,\sigma}
    \end{bmatrix}\begin{bmatrix}
        (G_{\sigma})_{13}&(G_{\sigma})_{14}\\(G_{\sigma})_{23}&(G_{\sigma})_{24}
    \end{bmatrix}\begin{bmatrix}
        \partial_{z}^{2}F_{1}\\\partial_{z}^{2}F_{2}
    \end{bmatrix}(z^{\prime}).
\end{align*}
Hence, using \eqref{eq: estimates on M}, \eqref{eq: estimates on the r outgoing solutions}, \eqref{eq: estimates on the psi outgoing solutions}, and \eqref{eq: estimates on the dirichlet solutions}, for $-1\leq z^{\prime}\leq z$ and $0\leq j\leq2$, we have \begin{equation*}
    \left\vert\frac{d^{j}}{dz^{j}}\left(\rho(z)\widetilde{R}_{\sigma}(z,z^{\prime})\begin{bmatrix}
        \partial_{z}^{2}F_{1}\\\partial_{z}^{2}F_{2}
    \end{bmatrix}\right)\right\vert\lesssim_{k}\left(\left\Vert F_{1}\right\Vert_{\mathcal{C}_{N}^{\widetilde{\beta}}}+\left\Vert F_{2}\right\Vert_{\mathcal{C}_{N}^{\beta}}\right)\frac{\vert\sigma\vert^{j}}{\vert\sigma+1\vert}(-z^{\prime})^{\frac{\sigma}{q_{k}}+\beta-1},\quad 0\leq j\leq 2.
\end{equation*}
For $R_{\sigma}$ when $-1\leq z^{\prime}\leq z$, similarly we have \begin{align*}
    \left\vert\frac{d^{j}}{dz^{j}}\left(\rho(z)R_{\sigma}(z,z^{\prime})\begin{bmatrix}
        \partial_{z}F_{1}\\\partial_{z}F_{2}
    \end{bmatrix}\right)\right\vert\lesssim_{k} \left(\left\Vert F_{1}\right\Vert_{\mathcal{C}_{N}^{\widetilde{\beta}}}+\left\Vert F_{2}\right\Vert_{\mathcal{C}_{N}^{\beta}}\right)\frac{\vert\sigma\vert^{j}}{\vert\sigma+1\vert}(-z^{\prime})^{\frac{\sigma}{q_{k}}+\beta-1},\quad 0\leq j\leq 2.
\end{align*}
For $z\leq z^{\prime}$, we have \begin{align*}
   \frac{d^{j}}{dz^{j}}\left(\rho(z)\widetilde R_{\sigma}(z,z^{\prime})\begin{bmatrix}
       \partial_{z}^{2}F_{1}\\\partial_{z}^{2}F_{2}
   \end{bmatrix}\right)\lesssim_{k}\left(\left\Vert F_{1}\right\Vert_{\mathcal{C}_{N}^{\widetilde{\beta}}}+\left\Vert F_{2}\right\Vert_{\mathcal{C}_{N}^{\beta}}\right)\frac{\vert\sigma\vert^{j}}{\vert\sigma+1\vert}(-z^{\prime})^{\frac{\sigma}{q_{k}}+\beta-1},\quad 0\leq j\leq 2.
\end{align*}
Hence, for $p_{k}<\beta<\frac{3}{2}$ and $\beta+\frac{1}{4}<\widetilde{\beta}<2$, we have the resolvent estimate \eqref{eq: resolvent estimate}.
\end{proof}

\section{Leading expansion for the linearized system in the near-axis region}
\label{leading order expansion}
\subsection{Setup of the problem}
In this section, we establish decay results for the inhomogeneous linearized equations: \begin{equation}
    \mathcal{P}(r_{p},\Psi_{p},m_{p}) = \mathcal{I}(s,z),\quad s\geq 0,\ -1\leq z\leq0,\label{eq: inhomogeneous problem}
\end{equation}
with the initial and boundary conditions:\begin{equation}
\label{eq: trivial initial and boundary conditions}
\begin{aligned}
    &\left(r_{p},\Psi_{p},m_{p}\right)(s,-1) = (0,0,0),\quad 
    \left(r_{p},\Psi_{p},m_{p}\right)(0,z) = 0,
    \end{aligned}
\end{equation}
where $\mathcal{I}(s,z) = \left(\partial_{z}I_{1},\partial_{z}I_{2},I_{3},0\right)$ satisfies
\begin{equation}
\label{eq: assumption on the inhomogeneous terms}
\begin{aligned}
    &\left(I_{1},I_{2},I_{3}\right)(s,\cdot)\in  C_{s}^{3}\left(\mathbb{R}_{+}\right)\times C_{s}^{3}\left(\mathbb{R}_{+}\right)\times C_{s}^{3}\left(\mathbb{R}_{+}\right)\\&
    \textup{supp}\ (I_{1},I_{2},I_{3})(s,\cdot)\subset [1,\infty)\times[1,\infty)\times[1,\infty),\\&
   \sum_{j = 0}^{3} \sup_{s>0}e^{\frac{3}{2}s}\left(\Vert \partial_{s}^{j}I_{1}\Vert_{\mathcal{C}_{N}^{\widetilde{\beta}}}+\Vert \partial_{s}^{j}I_{2}\Vert_{\mathcal{C}_{N}^{\beta}}+\Vert \partial_{s}^{j}I_{3}\Vert_{\mathcal{C}_{N}^{p_{k}}}\right)\leq C,
\end{aligned}
\end{equation}
for some $p_{k}<\beta<\frac{3}{2}$, $\beta+\frac{1}{4}<\widetilde{\beta}<2$, and $C>0$. For the linearized operator $\mathcal{P}$, we further assume $\mathcal{P}^{(1)}$ denotes the wave equation for $r_{p}$; $\mathcal{P}^{(2)}$ denotes the wave equation for $\Psi_{p}$; $\mathcal{P}^{(3)}$ denotes the $\partial_{s}$-transport equation for $m_{p}$; and $\mathcal{P}^{(4)}$ denotes the $\partial_{z}$-transport equation for $m_{p}$.

Moreover, we assume that the solution $(r_{p},\Psi_{p},m_{p})$ to \eqref{eq: inhomogeneous problem} has the following decay rate under the self-similar coordinates:\begin{equation}
    \left\vert\xi\right\vert(s,z)\lesssim e^{-2\eta s},\quad \xi\in\{r_{p},\Psi_{p},m_{p}\}.\label{eq: decay rate assumption} 
\end{equation}
where $\eta$ is defined in Section \ref{sec: scattering theory for the linearized system}.
\begin{remark}
The motivations for studying the inhomogeneous problem \eqref{eq: inhomogeneous problem} are as follows. First, when studying the linearized Einstein-scalar field equations with non-trivial initial data, to facilitate the Fourier--Laplace transform, we will truncate the solution to be supported away from $\{s = 0\}$. This truncation will produce some inhomogeneous terms. One should think of the truncation function as only a function of $s$. Hence, the truncation procedure will only produce some inhomogeneous terms for $(\mathcal{P}^{(1)},\mathcal{P}^{(2)},\mathcal{P}^{(3)})$. Second, in the final study of the nonlinear stability, we will write the Einstein-scalar field equations as \begin{equation*}
    \text{Linearized part} = \textup{Nonlinearity}
\end{equation*}
and use the fixed point argument to prove some decay rates for the solutions. In the fixed point argument, we will treat the nonlinearity as some known functions, which motivates the study of the inhomogeneous problem \eqref{eq: inhomogeneous problem}.
\end{remark}

\subsection{Freedom of the inhomogeneous terms}
\label{sec: freedom of the inhomogeneous terms}
However, since the inhomogeneous $\mathcal{P} = \mathcal{I}$ originates from the nonlinear Einstein-scalar field equations, to have well-posedness, we lose some freedom to choose the inhomogeneous functions $\mathcal{I} = (\partial_{z}I_{1},\partial_{z}I_{2},I_{3},0)$. In this section, we discuss one formulation for the inhomogeneous terms.

Due to the $v$-transport equation \eqref{eq:linearized dvm equation} (equivalently \eqref{eq:v transport eq for m under self-similar coordinates}), integrating from the center in the $v$-direction, we can write $m$ as a function of $r$ and $\Psi$. Formally, we have that \begin{equation*}
    m = m[r,\Psi],\quad \mu = \mu[r,\Psi].
\end{equation*}
Hence, schematically we can write the Einstein-scalar field equations as \begin{align*}
    &\partial_{u}\partial_{v}r = \frac{\mu[r,\Psi]}{1-\mu[r,\Psi]}\frac{\partial_{u}r\partial_{v}r}{r},\\&
    \partial_{u}\partial_{v}\Psi+\frac{k}{(-u)}\partial_{v}r = \frac{\mu[r,\Psi]}{1-\mu[r,\Psi]}\frac{\partial_{u}r\partial_{v}r}{r^{2}}\Psi.
\end{align*}
Hence, we can only choose the inhomogeneous terms $I_{1}$ and $I_{2}$ freely and, depending on the exact formulation of the problem, $I_{3}$ will be determined accordingly once we fix $(I_{1},I_{2})$. For now, we use the word ``compatible'' to refer to any reasonable inhomogeneous terms.

\subsection{Leading order expansion of the inhomogeneous problem}
Since $r_{p}$, $\Psi_{p}$, and $m_{p}$ have trivial initial values, for each fixed $z$, we can view $r_{p}$, $\Psi_{p}$, and $m_{p}$ as real-valued functions defined on $s\in\mathbb{R}$ with $(r_{p},\Psi_{p},m_{p}) = (0,0,0)$ for $s\leq 0$. Taking the Fourier--Laplace transform of $(r_{p},\Psi_{p},m_{p})$ \begin{equation*}
    \hat{r}_{p}(\sigma,z): = \int_{0}^{\infty}e^{-\sigma s}r_{p}(s,z)ds,\quad \hat{\Psi}_{p}(\sigma,z) = \int_{0}^{\infty}e^{-\sigma s}\Psi_{p}(s,z)ds,\quad \hat{m}_{p}(\sigma,z) = \int_{0}^{\infty}e^{-\sigma s}m_{p}(s,z)ds,
\end{equation*}
we have that $(\hat{r}_{p},\hat{\Psi}_{p},\hat{m}_{p})$ is well-defined for $\Re\sigma>-2\eta$ by the assumption \eqref{eq: decay rate assumption}. Then the equation \eqref{eq: inhomogeneous problem} is reduced to \begin{equation*}
    \mathcal{L}(\hat{r}_{p},\hat{\Psi}_{p}) = F(\sigma,z),
\end{equation*}
where $F(\sigma,z)$ depends on $\hat{I}_{1}$, $\hat{I}_{2}$, $\hat{I}_{3}$, and the background spacetime. More specifically, we have \begin{proposition}
\label{prop: properties on the inhomogeneous equations}
Taking the Fourier--Laplace transform of $(r_{p},\Psi_{p},m_{p})$, the inhomogeneous term $F(\sigma,z)$ in the reduced equation $\mathcal{L}(\hat{r}_{p},\hat{\Psi}_{p}) = F(\sigma,z)$ takes the form of \begin{align}
F(\sigma,z)& = \begin{bmatrix}
    \partial_{z}F_{1}\\\partial_{z}F_{2}
\end{bmatrix},\\
    \partial_{z}F_{1}(\sigma,z) &= \partial_{z}\hat{I}_{1}(\sigma,z)+\frac{2}{(1-\mu_{k})^{2}}\frac{(\partial_{s}+q_{k}z\partial_{z})r_{k}\partial_{z}r_{k}}{r_{k}^{2}}\frac{\hat{I}_{3}(\sigma,z)}{\sigma+\left(\frac{r_{k}}{\partial_{s}r_{k}+q_{k}z\partial_{z}r_{k}}((\partial_{s}+q_{k}z\partial_{z})\phi_{k})^{2}-q_{k}z\frac{r_{k}}{\partial_{z}r_{k}}(\partial_{z}\phi_{k})^{2}\right)},\label{eq: expression of F1}\\
   \partial_{z} F_{2}(\sigma,z)& = \partial_{z}\hat{I}_{2}(\sigma,z)+\frac{2\Psi_{k}}{r_{k}^{3}}\frac{(\partial_{s}+q_{k}z\partial_{z})r_{k}\partial_{z}r_{k}}{(1-\mu_{k})^{2}}\frac{\hat{I}_{3}(\sigma,z)}{\sigma+\left(\frac{r_{k}}{\partial_{s}r_{k}+q_{k}z\partial_{z}r_{k}}((\partial_{s}+q_{k}z\partial_{z})\phi_{k})^{2}-q_{k}z\frac{r_{k}}{\partial_{z}r_{k}}(\partial_{z}\phi_{k})^{2}\right)}\label{eq: expression of F2},
\end{align}
where $\left(\hat{I}_{1},\hat{I}_{2},\hat{I}_{3}\right)$ satisfies \begin{equation*}
    \hat{I}_{1}(\sigma,\cdot)\in\mathcal{C}_{N}^{\widetilde{\beta}},\quad \hat{I}_{2}(\sigma,\cdot)\in\mathcal{C}_{N}^{\beta},\quad \hat{I}_{3}(\sigma,\cdot)\in\mathcal{C}_{N}^{p_{k}},
\end{equation*}
and is well-defined for $\Re\sigma>-\frac{3}{2}$, and $(F_{1},F_{2})$ satisfies\begin{equation*}
    F_{1}(\sigma,\cdot)\in\mathcal{C}_{N}^{\widetilde{\beta}},\quad F_{2}(\sigma,\cdot)\in\mathcal{C}_{N}^{\beta}.
\end{equation*}
\end{proposition}
\begin{proof}
The relations \eqref{eq: expression of F1}-\eqref{eq: expression of F2} follow from a straightforward algebraic computation. By \eqref{eq: assumption on the inhomogeneous terms}, for $\Re\sigma>-\frac{3}{2}$, we have \begin{equation*}
    \left\Vert\hat{I}_{1}(\sigma,\cdot)\right\Vert_{L_{z}^{\infty}} = \left\Vert\int_{0}^{\infty}e^{-\sigma s}I_{1}(s,z)ds\right\Vert_{L_{z}^{\infty}} = \left\Vert\int_{0}^{\infty}e^{-\left(\frac{3}{2}+\sigma\right)s}e^{\frac{3}{2}s}I_{1}(s,z)ds\right\Vert_{L_{z}^{\infty}}\lesssim \sup_{s\geq 0}\left\Vert e^{\frac{3}{2}s}I_{1}(s,z)\right\Vert_{L_{z}^{\infty}}\lesssim C.
\end{equation*}
Similarly, for higher-order $\partial_{z}$-derivatives of $\hat{I}_{1}$, we can derive \begin{equation*}
    \left\Vert \hat{I}_{1}(\sigma,\cdot)\right\Vert_{\mathcal{C}_{N}^{\widetilde{\beta}}}\lesssim \sup_{s\geq0}e^{\frac{3}{2}s}\left\Vert I_{1}(s,\cdot)\right\Vert_{\mathcal{C}_{N}^{\widetilde{\beta}}}\lesssim C.
\end{equation*}
Hence, we have that $\hat{I}_{1}(\sigma,z)$ is well-defined for $\Re\sigma>-\frac{3}{2}$ and \begin{equation*}
    \hat{I}_{1}(\sigma,\cdot)\in\mathcal{C}_{N}^{\widetilde{\beta}}.
\end{equation*}
Similarly, we can prove the conclusions for $\hat{I}_{2}(\sigma,z)$ and $\hat{I}_{3}(\sigma,z)$. This concludes the proof.
\end{proof}
To facilitate the later estimates, we prove the following lemma for $F_{1}(\sigma,z)$ and $F_{2}(\sigma,z)$.
\begin{lemma}
    Assume that $F_{1}(\sigma,z)$ and $F_{2}(\sigma,z)$ are functions given in Proposition \ref{prop: properties on the inhomogeneous equations}. Then we have \begin{align}
        &\sup_{\sigma\in\mathbb{I}_{(-\beta_{*},-\eta]}}\vert\sigma\vert^{3}\left\Vert F_{1}\right\Vert_{\mathcal{C}_{N}^{\widetilde{\beta}}}\leq \sup_{s\geq0}\sum_{j = 0}^{3}e^{\frac{3}{2}s}\left(\left\Vert \partial_{s}^{j}I_{1}\right\Vert_{\mathcal{C}_{N}^{\widetilde{\beta}}}+\left\Vert\partial_{s}^{j}I_{3}\right\Vert_{\mathcal{C}_{N}^{p_{k}}}\right),\label{eq: losing derivative on F1}\\&
        \sup_{\sigma\in\mathbb{I}_{(-\beta_{*},-\eta]}}\vert\sigma\vert^{3}\left\Vert F_{2}\right\Vert_{\mathcal{C}_{N}^{\beta}}\leq \sup_{s\geq 0}\sum_{j = 0}^{3}e^{\frac{3}{2}s}\left(\left\Vert\partial_{s}^{j}I_{2}\right\Vert_{\mathcal{C}_{N}^{\beta}}+\left\Vert\partial_{s}^{j}I_{3}\right\Vert_{\mathcal{C}_{N}^{p_{k}}}\right).\label{eq: losing derivative on F2}
    \end{align}
\end{lemma}
\begin{proof}
    We only prove \eqref{eq: losing derivative on F1}, and the proof of \eqref{eq: losing derivative on F2} will follow similarly. Using \eqref{eq: expression of F1} and the properties on the background spacetime,
    we have \begin{equation}
        \vert\sigma\vert^{3}\left\vert \partial_{z}F_{1}(\sigma,z)\right\vert\lesssim \vert\sigma\vert^{3}\left\vert\partial_{z}\hat{I}_{1}(\sigma,z)\right\vert+\vert\sigma\vert^{3}\vert \hat{I}_{3}(\sigma,z)\vert.\label{eq: triangle inequality}
    \end{equation}
    For the term with $\partial_{z}\hat{I}_{1}$ on the right-hand side of \eqref{eq: triangle inequality}, applying the integration by parts, we have \begin{equation}
    \begin{aligned}
    \vert\sigma\vert^{3}\vert\partial_{z}\hat{I}_{1}(\sigma,z)\vert\lesssim \vert\sigma\vert^{3}\left\vert\int_{0}^{\infty}e^{-\sigma s}\partial_{z}I_{1}(s,z)ds\right\vert\lesssim \left\vert\int_{0}^{\infty}e^{-\sigma s}\partial_{s}^{3}\partial_{z}I_{1}(s,z)ds\right\vert\lesssim \sup_{s\geq 0}e^{\frac{3}{2}s}\left\Vert \partial_{s}^{3}\partial_{z}I_{1}(s,\cdot)\right\Vert_{L_{z}^{\infty}}.
    \end{aligned}
    \label{eq: estimate on I1}
    \end{equation}
    Similarly, for the second term on the right-hand side of \eqref{eq: triangle inequality}, we have \begin{equation}
        \begin{aligned}
            \vert\sigma\vert^{3}\vert\hat{I}_{3}(\sigma,z)\vert\lesssim \sup_{s\geq0}e^{\frac{3}{2}s}\left\Vert\partial_{s}^{3}I_{3}(s,\cdot)\right\Vert_{L_{z}^{\infty}}.
        \end{aligned}
        \label{eq: estimate on I3}
    \end{equation}
    The higher-order estimates follow similarly. Hence, combining \eqref{eq: estimate on I1}, \eqref{eq: estimate on I3}, and the corresponding higher-order estimates, we can conclude the proof of \eqref{eq: losing derivative on F1}.
\end{proof}

Using the scattering theory established in Section \ref{sec: scattering theory for the linearized system}, we can get the following estimate for the solution to \eqref{eq: inhomogeneous problem} with trivial initial data.
\begin{proposition}[Near-axis decay for the regularity above the threshold]
\label{prop: near axis decay}
For $k$ sufficiently small, let $p_{k}<\beta<\frac{3}{2}$, $\beta+\frac{1}{4}<\widetilde{\beta}<2$, and $\eta\in(0,\frac{1}{10})$ be a fixed real number. Define $\beta_{*} = \min\{\beta q_{k},1+\eta\}$. Let $(r_{p},\Psi_{p},m_{p})$ be the solution to \eqref{eq: inhomogeneous problem} with the initial and boundary conditions given by \eqref{eq: trivial initial and boundary conditions} and the inhomogeneous term $\mathcal{I}$ satisfying \eqref{eq: assumption on the inhomogeneous terms}. Then there exist constants $\epsilon_{*}$ sufficiently small, $c_{\infty}$ depending on $\mathcal{I}$, and $\beta_{*}^{\prime}: = \beta_{*}-\epsilon_{*}$, such that for any $0\leq i+j\leq 2$, $-1\leq z\leq \frac{1}{4}$, and $N\geq 5$, we have \begin{align}
   &\sum_{0\leq i+j\leq 2}\left\vert \partial_{s}^{i}\partial_{z}^{j}r_{p}\right\vert\lesssim_{k,\epsilon_{*}} e^{-\beta_{*}^{\prime}s}\sum_{l = 0}^{3}\sup_{s\geq0}e^{\frac{3}{2}s}\left(\left\Vert\partial_{s}^{l}I_{1}\right\Vert_{\mathcal{C}_{N}^{\widetilde{\beta}}}+\left\Vert\partial_{s}^{l}I_{2}\right\Vert_{\mathcal{C}_{N}^{\beta}}+\left\Vert\partial_{s}^{l}I_{3}\right\Vert_{\mathcal{C}_{N}^{p_{k}}}\right),\label{eq: near-axis decay estimate for rp}\\&
   \sum_{0\leq i+j\leq 2}\left\vert \partial_{s}^{i}\partial_{z}^{j}\left(\Psi_{p}-c_{\infty}r_{k}\right)\right\vert\lesssim_{k,\epsilon_{*}} e^{-\beta_{*}^{\prime}s}\sum_{l = 0}^{3}\sup_{s\geq0}e^{\frac{3}{2}s}\left(\left\Vert\partial_{s}^{l}I_{1}\right\Vert_{\mathcal{C}_{N}^{\widetilde{\beta}}}+\left\Vert\partial_{s}^{l}I_{2}\right\Vert_{\mathcal{C}_{N}^{\beta}}+\left\Vert\partial_{s}^{l}I_{3}\right\Vert_{\mathcal{C}_{N}^{p_{k}}}\right),\label{eq: near-axis decay estimate for psip}\\&
   \sum_{0\leq i+j\leq 2}\left\vert\partial_{s}^{i}\partial_{z}^{j}m_{p}\right\vert\lesssim_{k,\epsilon_{*}}e^{-\beta_{*}^{\prime}s}\sum_{l = 0}^{3}\sup_{s\geq0}e^{\frac{3}{2}s}\left(\left\Vert\partial_{s}^{l}I_{1}\right\Vert_{\mathcal{C}_{N}^{\widetilde{\beta}}}+\left\Vert\partial_{s}^{l}I_{2}\right\Vert_{\mathcal{C}_{N}^{\beta}}+\left\Vert \partial_{s}^{l}I_{3}\right\Vert_{\mathcal{C}_{N}^{p_{k}}}\right).\label{eq: near-axis decay estimate for mp}
\end{align}
Moreover, we have the following estimate on $c_{\infty}$:
\begin{equation}
    \vert c_{\infty}\vert\lesssim_{k} \sup_{s\geq0}e^{\frac{3}{2}s}\left(\left\Vert I_{1}(s,\cdot)\right\Vert_{\mathcal{C}_{N}^{\widetilde{\beta}}}+\left\Vert I_{2}(s,\cdot)\right\Vert_{\mathcal{C}_{N}^{\beta}}+\left\Vert I_{3}(s,\cdot)\right\Vert_{\mathcal{C}_{N}^{p_{k}}}\right).
\end{equation}
\end{proposition}
\begin{proof}
Taking the Fourier--Laplace transform of $(r_{p},\Psi_{p},m_{p})$, by \eqref{eq: decay rate assumption}, we have $\left(\hat{r}_{p},\hat{\Psi}_{p},\hat{m}_{p}\right)$ is well-defined for $\Re\sigma>-2\eta$. Then the equation \eqref{eq: inhomogeneous problem} is reduced to $\mathcal{L}(\hat{r}_{p},\hat{\Psi}_{p}) = F(\sigma,z)$ with $F = (\partial_{z}F_{1},\partial_{z}F_{2})$ satisfying the properties in Proposition \ref{prop: properties on the inhomogeneous equations}. Since $r_{p}$, $\Psi_{p}$, and $m_{p}$ vanish on the axis $\Gamma$, we have \begin{equation*}
    \hat{r}_{p}(\sigma,-1) = \hat{\Psi}_{p}(\sigma,-1) = \hat{m}_{p}(\sigma,-1) = 0.
\end{equation*}
In other words, $(\hat{r}_{p},\hat{\Psi}_{p})$ satisfies the Dirichlet boundary condition. On the other hand, for $\Re\sigma = -\eta$, since $(r_{p},\Psi_{p})$ satisfies the decay estimate in the assumption \eqref{eq: decay rate assumption}, we have that\begin{align*}
    \hat{r}_{p} = \int_{0}^{\infty}e^{-\sigma s}r_{p}(s,z)ds\in C_{z}^{2}, \quad \hat{\Psi}_{p} = \int_{0}^{\infty}e^{-\sigma s}\Psi_{p}(s,z)ds\in C_{z}^{1}\text{ for } \Re\left(\frac{-\sigma}{q_{k}}\right) = \frac{\eta}{q_{k}}<1.
\end{align*}
Hence, by Lemma \ref{lemma: regularity perspective}, we have that $(\hat{r}_{p},\hat{\Psi}_{p})$ satisfies the outgoing boundary condition. Moreover, since when $\Re\left(\frac{-\sigma}{q_{k}}\right) = \frac{\eta}{q_{k}}$, the Dirichlet solutions will never be linearly dependent on the outgoing solutions, we have that the only solution $(\hat{r}_{p},\hat{\Psi}_{p})$ to $\mathcal{L} = F$ will be given by the resolvent operator $R(\sigma)(\partial_{z}F_{1},\partial_{z}F_{2})(z)$.

Followed by the argument above, since we have shown that $(\hat{r}_{p},\hat{\Psi}_{p})$ satisfies the Dirichlet boundary condition and outgoing boundary condition, we have that \begin{equation*}
    \rho(z)\begin{bmatrix}
        \hat{r}_{p}\\\hat{\Psi}_{p}
    \end{bmatrix}(z) = \rho(z)R(\sigma)(\partial_{z}F_{1},\partial_{z}F_{2})(z).
\end{equation*}
By the a priori decay estimate established in Theorem~\ref{thm:a priori energy estimate} and the standard Fourier--Laplace inversion theorem, we can recover $(r_{p},\Psi_{p})$ by  \begin{equation*}
    \rho(z)\begin{bmatrix}
        r_{p}\\\Psi_{p}
    \end{bmatrix} (s,z)= \int_{\Re\sigma = -\eta}e^{\sigma s}\rho(z)R(\sigma)(\partial_{z}F_{1},\partial_{z}F_{2})d\sigma.
\end{equation*}
Then for $\beta\in(p_{k},\frac{3}{2})$, we have $R(\sigma)(\partial_{z}F_{1},\partial_{z}F_{2})$ is well-defined for $\sigma\in\mathbb{I}_{[-\beta_{*},-\eta]}$. Taking $\epsilon_{*}$ small enough such that $\beta_{*}-2\epsilon_{*}>1$, we define the oriented contours \begin{equation}
    \begin{cases}
    \Gamma_{\epsilon} &= -1+\epsilon_{*}e^{-i\theta},\quad \theta\in[0,2\pi),\\
    \Gamma_{1,R}^{\pm}&=-\eta\pm iR\pm it,\quad t\in[0,\infty),\\
    \Gamma_{2,R}^{\pm}& = -\eta\pm iR\pm t,\quad t\in[0,-\eta+\beta_{*}-\epsilon_{*}],\\
    \Gamma_{3,R}& = (-\beta_{*}+\epsilon_{*})+t,\quad t\in[-R,R],
    \end{cases}
\end{equation}
and the path $R_{R}$ and $\bar{P}_{R}$  \begin{align}
    P_{R}: &= -\Gamma_{1,R}^{-}\cup\Gamma_{2,R}^{-}\cup\Gamma_{3,R}\cup-\Gamma_{2,R}^{+}\cup \Gamma_{1,R}^{+},\\
    \bar{P}_{R}:& = -\Gamma_{1,R}^{-}\cup\Gamma_{1,R}^{+}.
\end{align}
The Figure \ref{fig: contour} depicts this contour.
\begin{figure}[H]
\centering
\begin{tikzpicture}[x=1.4cm,y=1.4cm,>=stealth]

    \def\R{2.2}
    \def\eta{0.6}
    \def\B{2.5}      
    \def\eps{0.2}

    \draw[->] (-4.8,0) -- (4.8,0) node[right] {$\Re \sigma$};
    \draw[->] (0,-3.4) -- (0,0.8) node[above] {$\Im \sigma$};

    \draw (-\R,0.08) -- (-\R,-0.08);
    \draw (\R,0.08) -- (\R,-0.08);
    \node[above] at (-\R,0) {$-iR$};
    \node[above] at (\R,0) {$iR$};

    \draw (0,-1+0.08) -- (0,-1-0.08);
    \draw (0,-\eta+0.08) -- (0,-\eta-0.08);
    \draw (0,-\B+0.08) -- (0,-\B-0.08);

    \node[right] at (0,-2) {$-1$};

    \draw[thick,->] (-4.4,-\eta) -- (-\R,-\eta);
    \draw[thick,->] (-\R,-\eta) -- (-\R,-\B);
    \draw[thick,->] (-\R,-\B) -- (\R,-\B);
    \draw[thick,->] (\R,-\B) -- (\R,-\eta);
    \draw[thick,->] (\R,-\eta) -- (4.4,-\eta);

    \draw[thick,->]
      ({\eps*cos(35)},{-2+\eps*sin(35)})
      arc[start angle=35,end angle=-300,radius=\eps];

    \node[above] at (-3.9,-\eta) {$-\Gamma_{1,R}^{-}$};
    \node[left]  at (-\R,-1.9) {$\Gamma_{2,R}^{-}$};
    \node[below] at (0,-\B) {$\Gamma_{3,R}$};
    \node[right] at (\R,-1.9) {$-\Gamma_{2,R}^{+}$};
    \node[above] at (3.9,-\eta) {$\Gamma_{1,R}^{+}$};
    \node[right] at (0.55,-0.72) {$\Gamma_{\epsilon}$};

\end{tikzpicture}
\caption{Oriented contours}
\label{fig: contour}
\end{figure}
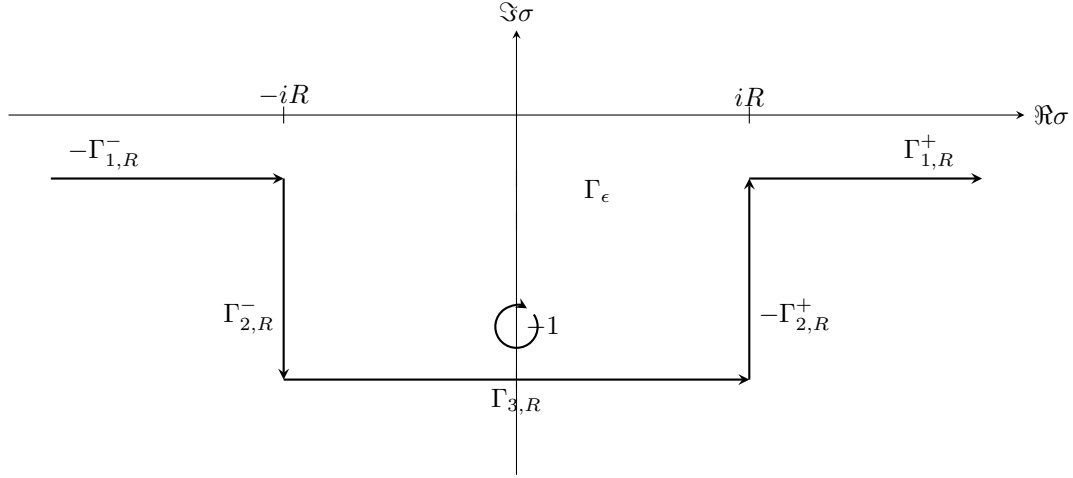

Then, by the residue theorem in complex analysis, since $\rho(z)R(\sigma)(\partial_{z}F_{1},\partial_{z}F_{2})$ is holomorphic in $\sigma$ for $\sigma\in\mathbb{I}_{(-\beta_{*},-\eta]}\backslash\{-1\}$, we have\begin{equation}
    \rho(z)\begin{bmatrix}
        r_{p}\\\Psi_{p}
    \end{bmatrix}(z) = \frac{1}{2\pi}\int_{P_{R}}e^{\sigma s}\rho(z)R(\sigma)(\partial_{z}F_{1},\partial_{z}F_{2})d\sigma+\frac{1}{2\pi}\int_{\Gamma_{\epsilon}}e^{\sigma s}\rho(z)R(\sigma)(\partial_{z}F_{1},\partial_{z}F_{2})d\sigma,\label{eq: inverse laplace transform to recover the solutions}
\end{equation} 
for any $R>0$ and any $\epsilon$ sufficiently small.

Again, by the residue theorem, we can compute \begin{align*}
\frac{1}{2\pi}\int_{\Gamma_{\epsilon}}e^{\sigma s}\rho(z)R(\sigma)(\partial_{z}F_{1},\partial_{z}F_{2})d\sigma &=\frac{1}{2\pi} \int_{\Gamma_{\epsilon}}e^{\sigma s}\frac{\mathcal{R}_{1}(\sigma,z)}{\sigma+1}d\sigma \\& =  e^{-s}\mathcal{R}_{1}(-1,z)\\[1em]& =  c_{\infty} \begin{bmatrix}
0\\\rho(z)e^{-s}\mr{r}(z)
\end{bmatrix} = c_{\infty}\begin{bmatrix}
    0\\\rho(z)r_{k}(s,z)
\end{bmatrix},
\end{align*}
where $c_{\infty}$ is the constant depending on $(F_{1},F_{2})$ \begin{equation}
    c_{\infty}(F_{1},F_{2}) = c_{*}\int_{-1}^{0}\left((\widetilde{G}_{-1})_{23}\partial_{z}^{2}F_{1}+(\widetilde{G}_{-1})_{24}\partial_{z}^{2}F_{2}\right)dz^{\prime}.
\end{equation}
To estimate $c_{\infty}$, by \eqref{eq: estimate on Gsigma}, we have \begin{equation}
    \begin{aligned}
       \left\vert c_{\infty}(F_{1},F_{2})\right\vert\lesssim \int_{-1}^{0}(-z)^{\frac{\sigma}{q_{k}}+\beta-1}\left(\left\Vert F_{1}\right\Vert_{\mathcal{C}_{N}^{\widetilde{\beta}}}+\left\Vert F_{2}\right\Vert_{\mathcal{C}_{N}^{\beta}}\right)dz\lesssim \frac{1}{\left\vert\beta-p_{k}\right\vert}\left(\left\Vert F_{1}\right\Vert_{\mathcal{C}_{N}^{\widetilde{\beta}}}+\left\Vert F_{2}\right\Vert_{\mathcal{C}_{N}^{\beta}}\right).
    \end{aligned}
\end{equation}

It remains to estimate the first integral term in \eqref{eq: inverse laplace transform to recover the solutions}. We further decompose this integration as \begin{equation*}
    \int_{P_{R}} = \int_{\bar{P}_{R}}+\int_{P_{R}\backslash\left(\bar{P}_{R}\cup \Gamma_{3,R}\right)}+\int_{\Gamma_{3,R}}.
\end{equation*}
For the integration on $\bar{P}_{R}$ and $0\leq i+j\leq 2$, by \eqref{eq: resolvent estimate}, we have \begin{equation}
\begin{aligned}
    &\left\vert\partial_{s}^{i}\partial_{z}^{j}\frac{1}{2\pi}\int_{\bar{P}_{R}}e^{\sigma s}\rho(z)R(\sigma)(\partial_{z}F_{1},\partial_{z}F_{2})d\sigma\right\vert\\\lesssim_{k}&\int_{\bar{P}_{R}}e^{(\Re\sigma) s}\vert\sigma\vert \left(\Vert F_{1}\Vert_{\mathcal{C}_{N}^{\widetilde{\beta}}}+\left\Vert F_{2}\right\Vert_{\mathcal{C}_{N}^{\beta}}\right)d\sigma\\\lesssim_{k}&e^{-\eta s}\int_{\bar{P}_{R}}\vert\sigma\vert\left(\Vert F_{1}\Vert_{\mathcal{C}_{N}^{\widetilde{\beta}}}+\left\Vert F_{2}\right\Vert_{\mathcal{C}_{N}^{\beta}}\right) d\sigma\\\lesssim_{k}&e^{-\eta s}\left(\sup_{s\geq0}\sum_{l = 0}^{3}e^{\frac{3}{2}s}\left(\left\Vert\partial_{s}^{l}I_{1}\right\Vert_{\mathcal{C}_{N}^{\widetilde{\beta}}}+\left\Vert\partial_{s}^{l}I_{2}\right\Vert_{\mathcal{C}_{N}^{\beta}}+\left\Vert \partial_{s}^{l}I_{3}\right\Vert_{\mathcal{C}_{N}^{p_{k}}}\right)\right)\int_{\bar{P}_{R}}\vert\sigma\vert^{-2}d\sigma\\\lesssim_{k}&e^{-\eta s}R^{-1}\sum_{l = 0}^{3}\sup_{s\geq 0}e^{\frac{3}{2}s}\left(\left\Vert\partial_{s}^{l}I_{1}\right\Vert_{\mathcal{C}_{N}^{\widetilde{\beta}}}+\left\Vert \partial_{s}^{l}I_{2}\right\Vert_{\mathcal{C}_{N}^{\beta}}+\left\Vert \partial_{s}^{l}I_{3}\right\Vert_{\mathcal{C}_{N}^{p_{k}}}\right),
    \end{aligned}
    \label{eq: estimate on barPr part}
\end{equation}
where the third inequality follows from \eqref{eq: losing derivative on F1}-\eqref{eq: losing derivative on F2}.

For the integration on $P_{R}\backslash\left(\bar{P}_{R}\cup\Gamma_{3,R}\right)$ and $0\leq i+j\leq2$, by \eqref{eq: resolvent estimate}, we have \begin{equation}
    \begin{aligned}
        &\left\vert\partial_{s}^{i}\partial_{z}^{j}\frac{1}{2\pi}\int_{P_{R}\backslash\left(\bar{P}_{R}\cup\Gamma_{3,R}\right)}e^{\sigma s}\rho(z)R(\sigma)(\partial_{z}F_{1},\partial_{z}F_{2})d\sigma\right\vert\\\lesssim&_{k}
    \int_{P_{R}\backslash\left(\bar{P}_{R}\cup\Gamma_{3,R}\right)}\vert\sigma\vert e^{\left(\Re\sigma \right)s}\left(1+\frac{1}{\vert\sigma+1\vert}\right)\frac{1}{\beta_{*}+\Re\sigma}\left(\left\Vert F_{1}\right\Vert_{\mathcal{C}_{N}^{\widetilde{\beta}}}+\left\Vert F_{2}\right\Vert_{\mathcal{C}_{N}^{\beta}}\right)d\sigma\\\lesssim&_{k,\epsilon_{*}}e^{-\eta s}R^{-1}\sum_{l = 0}^{3}\sup_{s\geq0}e^{\frac{3}{2}s}\left(\left\Vert\partial_{s}^{l}I_{1}\right\Vert_{\mathcal{C}_{N}^{\widetilde{\beta}}}+\left\Vert\partial_{s}^{l}I_{2}\right\Vert_{\mathcal{C}_{N}^{\beta}}+\left\Vert\partial_{s}^{l}I_{3}\right\Vert_{\mathcal{C}_{N}^{p_{k}}}\right).
    \end{aligned}
    \label{eq: estimate on the vertical part}
\end{equation}
For the integration on $\Gamma_{3,R}$ and $0\leq i+j\leq 2$, again by \eqref{eq: resolvent estimate}, we have \begin{equation}
    \begin{aligned}
    &\left\vert\partial_{s}^{i}\partial_{z}^{j}\frac{1}{2\pi}\int_{\Gamma_{3,R}}e^{\sigma s}\rho(z)R(\sigma)(\partial_{z}F_{1},\partial_{z}F_{2})d\sigma\right\vert\\\lesssim&_{k,\epsilon_{*}}\int_{\Gamma_{3,R}}\vert\sigma\vert e^{-(\beta_{*}-\epsilon_{*})s}\left(\left\Vert F_{1}\right\Vert_{\mathcal{C}_{N}^{\widetilde{\beta}}}+\left\Vert F_{2}\right\Vert_{\mathcal{C}_{N}^{\beta}}\right)d\sigma\\\lesssim&_{k,\epsilon_{*}}e^{-(\beta_{*}-\epsilon_{*})s}\sum_{l = 0}^{3}\sup_{s\geq0}e^{\frac{3}{2}s}\left(\left\Vert\partial_{s}^{l}I_{1}\right\Vert_{\mathcal{C}_{N}^{\widetilde{\beta}}}+\left\Vert\partial_{s}^{l}I_{2}\right\Vert_{\mathcal{C}_{N}^{\beta}}+\left\Vert\partial_{s}^{l}I_{3}\right\Vert_{\mathcal{C}_{N}^{p_{k}}}\right)\int_{\Gamma_{3,R}}\vert\sigma\vert^{-2}d\sigma\\\lesssim&_{k,\epsilon_{*}}
    e^{-(\beta_{*}-\epsilon_{*})s}\sum_{l= 0}^{3}\sup_{s\geq0}e^{\frac{3}{2}s}\left(\left\Vert\partial_{s}^{l}I_{1}\right\Vert_{\mathcal{C}_{N}^{\widetilde{\beta}}}+\left\Vert\partial_{s}^{l}I_{2}\right\Vert_{\mathcal{C}_{N}^{\beta}}+\left\Vert\partial_{s}^{l}I_{3}\right\Vert_{\mathcal{C}_{N}^{p_{k}}}\right).
    \end{aligned}
    \label{eq: estimate on Gamma3 part}
\end{equation}
Combining \eqref{eq: estimate on barPr part}, \eqref{eq: estimate on the vertical part}, and \eqref{eq: estimate on Gamma3 part}, and taking $R\rightarrow\infty$, for $0\leq i+j\leq 2$, we have \begin{equation}
    \begin{aligned}
&\left\vert\partial_{s}^{i}\partial_{z}^{j}\left(\rho(z)\begin{bmatrix}
        r_{p}\\\Psi_{p}
    \end{bmatrix}-c_{\infty}\begin{bmatrix}
        0\\\rho(z)r_{k}
    \end{bmatrix}\right)\right\vert\\\lesssim&_{k,\epsilon_{*}}e^{-(\beta_{*}-\epsilon_{*})s}\sum_{l = 0}^{3}\sup_{s\geq0}e^{\frac{3}{2}s}\left(\left\Vert\partial_{s}^{l}I_{1}\right\Vert_{\mathcal{C}_{N}^{\widetilde{\beta}}}+\left\Vert\partial_{s}^{l}I_{2}\right\Vert_{\mathcal{C}_{N}^{\beta}}+\left\Vert\partial_{s}^{l}I_{3}\right\Vert_{\mathcal{C}_{N}^{p_{k}}}\right).
    \end{aligned}
\end{equation}
It will also be useful to get the corresponding estimates for $\phi_{p}$. Recall that \begin{equation*}
    \Psi_{p} = r_{k}\phi_{p}+\mr{\phi}r_{p}.
\end{equation*}
Then for $0\leq i+j\leq1$, we have \begin{equation}
\begin{aligned}
\sum_{0\leq i+j\leq 1}\left\vert\partial_{s}^{i}\partial_{z}^{j}\left(\phi_{p}-c_{\infty}\right)\right\vert =& \left\vert\partial_{s}^{i}\partial_{z}^{j}\left(\frac{1}{r_{k}}\left(\Psi_{p}-c_{\infty}r_{k}\right)-\frac{\mr{\phi}}{r_{k}}r_{p}\right)\right\vert\\\leq&
\left\vert \partial_{s}^{i}\partial_{z}^{j}\left(\frac{1}{r_{k}}\left(\Psi_{p}-c_{\infty}r_{k}\right)\right)\right\vert+\left\vert \partial_{s}^{i}\partial_{z}^{j}\left(\frac{\mr{\phi}}{r_{k}}r_{p}\right)\right\vert\\\lesssim&_{k,\epsilon_{*}}e^{-(\beta_{*}-\epsilon_{*}-1)s}\sum_{l = 0}^{3}\sup_{s\geq 0}e^{\frac{3}{2}s}\left(\left\Vert\partial_{s}^{l}I_{1}\right\Vert_{\mathcal{C}_{N}^{\widetilde{\beta}}}+\left\Vert\partial_{s}^{l}I_{2}\right\Vert_{\mathcal{C}_{N}^{\beta}}+\left\Vert\partial_{s}^{l}I_{3}\right\Vert_{\mathcal{C}_{N}^{p_{k}}}\right).
\end{aligned}
\end{equation}
In particular, for $i+j =  1$, we have \begin{equation*}
    \sum_{i+j=1}\left\vert\partial_{s}^{i}\partial_{z}^{j}\phi_{p}\right\vert\lesssim_{k,\epsilon_{*}}e^{-(\beta_{*}-\epsilon_{*}-1)s}\sum_{l = 0}^{3}\sup_{s\geq0}e^{\frac{3}{2}s}\left(\left\Vert\partial_{s}^{l}I_{1}\right\Vert_{\mathcal{C}_{N}^{\widetilde{\beta}}}+\left\Vert\partial_{s}^{l}I_{2}\right\Vert_{\mathcal{C}_{N}^{\beta}}+\left\Vert\partial_{s}^{l}I_{3}\right\Vert_{\mathcal{C}_{N}^{p_{k}}}\right).
\end{equation*}

It remains to conclude the near-axis decay estimate \eqref{eq: near-axis decay estimate for mp} for $m_{p}$. To see this, we use the $\partial_{z}$-transport equation for $m_{p}$ and the boundary condition $m_{p}\bigl|_{\Gamma} = 0$. We can schematically write the $\partial_{z}m_{p}$ equation in terms of $\phi_{p}$ and $r_{p}$ as \begin{align*}
    \partial_{z}m_{p}+\frac{1}{2}\left(\frac{r_{k}}{\partial_{z}r_{k}}\left(\partial_{z}\phi_{k}\right)^{2}\right)m_{p} = -\frac{1}{2}\left(\frac{r}{\partial_{z}r}(\partial_{z}\phi)^{2}\right)_{p}m_{k}+\left(\frac{r^{2}}{\partial_{z}r}(\partial_{z}\phi)^{2}\right)_{p}.
\end{align*}
Straightforward algebraic computation gives \begin{align*}
    &0\leq \int_{-1}^{0}\frac{r_{k}}{\partial_{z}r_{k}}(\partial_{z}\phi_{k})^{2}\leq Ck^{2},\\&\left\vert\left(\frac{r}{\partial_{z}r}(\partial_{z}\phi)^{2}\right)_{p}m_{k}\right\vert\lesssim e^{-(\beta_{*}-\epsilon_{*})s}\sum_{l = 0}^{3}\sup_{s\geq 0}e^{\frac{3}{2}s}\left(\left\Vert \partial_{s}^{l}I_{1}\right\Vert_{\mathcal{C}_{N}^{\widetilde{\beta}}}+\left\Vert \partial_{s}^{l}I_{2}\right\Vert_{\mathcal{C}_{N}^{\beta}}+\left\Vert \partial_{s}^{l}I_{3}\right\Vert_{\mathcal{C}_{N}^{p_{k}}}\right),\\&\left\vert\left(\frac{r^{2}}{\partial_{z}r}(\partial_{z}\phi)^{2}\right)_{p}\right\vert\lesssim_{k,\epsilon_{*}}e^{-(\beta_{*}-\epsilon_{*})s}\sum_{l = 0}^{3}\sup_{s\geq0}e^{\frac{3}{2}s}\left(\left\Vert\partial_{s}^{l}I_{1}\right\Vert_{\mathcal{C}_{N}^{\widetilde{\beta}}}+\left\Vert\partial_{s}^{l}I_{2}\right\Vert_{\mathcal{C}_{N}^{\beta}}+\left\Vert\partial_{s}^{l}I_{3}\right\Vert_{\mathcal{C}_{N}^{p_{k}}}\right).
\end{align*}
Then, directly integrating the $\partial_{z}$-transport equation for $m_{p}$, for $-1\leq z\leq \frac{1}{4}$, we have \begin{equation}
   \sum_{0\leq i+j\leq 2} \left\vert\partial_{s}^{i}\partial_{z}^{j} m_{p}\right\vert(s,z)\lesssim_{k,\epsilon_{*}}e^{-(\beta_{*}-\epsilon_{*})s}\sum_{l = 0}^{3}\sup_{s\geq0}e^{\frac{3}{2}s}\left(\left\Vert\partial_{s}^{l}I_{1}\right\Vert_{\mathcal{C}_{N}^{\widetilde{\beta}}}+\left\Vert\partial_{s}^{l}I_{2}\right\Vert_{\mathcal{C}_{N}^{\beta}}+\left\Vert\partial_{s}^{l}I_{3}\right\Vert_{\mathcal{C}_{N}^{p_{k}}}\right).
\end{equation}
Letting $\beta_{*}^{\prime} = \beta_{*}-\epsilon_{*}$ concludes the proof of this proposition.
\end{proof}

\section{Decay results for the linearized system in the whole interior region}
\label{sec: decay in the interior region}
In this section, we prove Theorem \ref{thm: linearized result}.
\begin{proof}
\textbf{Step 1: Reduction to an inhomogeneous problem with zero initial data.} For the linearized Einstein-scalar field equations $\mathcal{P}(r_{p},\Psi_{p},m_{p}) = 0$ with the initial data $(r_{p},\Psi_{p})|_{\Sigma_{-1}^{(in)}} = (r_{p}^{0},\Psi_{p}^{0})\in\mathcal{C}_{N}^{\widetilde{\beta}}\times\mathcal{C}_{N}^{\beta}$ where $p_{k}<\beta<\frac{3}{2}$ and $\beta+\frac{1}{4}<\widetilde{\beta}<2$, let $\xi = \xi(s)$ be a cut-off function with \begin{equation*}
    \xi = 0,\ \text{for } 0\leq s\leq \frac{1}{2},\quad \xi = 1,\ \text{for } s\geq 1.
\end{equation*}
Commuting the equations $\mathcal{P}(r_{p},\Psi_{p},m_{p}) = 0$ with $\xi(s)$, we have \begin{equation}
    \begin{aligned}
     &\mathcal{P}\left(\xi(s)r_{p},\xi(s)\Psi_{p},\xi (s)m_{p}\right) = (\partial_{z}I_{1},\partial_{z}I_{2},I_{3},0), \\&
\left(\xi(s)r_{p},\xi(s)\Psi_{p},\xi(s)m_{p}\right)|_{s = 0}.
    \end{aligned}
\end{equation}
where $(I_{1},I_{2},I_{3})$ is determined by the initial data $(r_{p}^{0},\Psi_{p}^{0})$, the cut-off function $\xi$, and the background spacetime. More precisely, we have \begin{align*}
    \partial_{z}I_{1}=& \partial_{s}\xi\partial_{z}r_{p}-\frac{\mu_{k}}{1-\mu_{k}}\frac{\partial_{z}r_{k}}{r_{k}}(\partial_{s}\xi)r_{p},\\
    \partial_{z}I_{2} =& \partial_{s}\xi\partial_{z}\Psi_{p}-\frac{1}{r_{k}^{2}}\frac{\mu_{k}}{1-\mu_{k}}\Psi_{k}\partial_{z}r_{k}(\partial_{s}\xi) r_{p},\\I_{3} = &(\partial_{s}\xi)m_{p}-(1-\mu_{k})\frac{r_{k}(\partial_{s}+q_{k}z\partial_{z})\phi_{k}}{(\partial_{s}+q_{k}z\partial_{z})r_{k}}(\partial_{s}\xi)\Psi_{p}\\&+(1-\mu_{k})\left(\frac{1}{2}\left(\frac{r_{k}(\partial_{s}+q_{k}z\partial_{z})\phi_{k}}{(\partial_{s}+q_{k}z\partial_{z})r_{k}}\right)^{2}+\frac{r_{k}(\partial_{s}+q_{k}z\partial_{z})\phi_{k}}{(\partial_{s}+q_{k}z\partial_{z})r_{k}}(\phi_{k}-ks)\right)(\partial_{s}\xi)r_{p}.
\end{align*}

Moreover, by the property of $\xi$, we have that $I_{i}$ is identically zero for $s\geq 1$. Checking the regularity, we have that \begin{equation}
\sup_{s\geq 0}e^{\frac{3}{2}s}\sum_{j = 0}^{3}\left(\left\Vert\partial_{s}^{j}I_{1}\right\Vert_{\mathcal{C}_{N}^{\widetilde{\beta}}}+\left\Vert\partial_{s}^{j}I_{2}\right\Vert_{\mathcal{C}_{N}^{\beta}}+\left\Vert\partial_{s}^{j}I_{3}\right\Vert_{\mathcal{C}_{N}^{p_{k}}}\right)\lesssim\left\Vert r_{p}^{0}\right\Vert_{\mathcal{C}_{N}^{\widetilde{\beta}}}+\left\Vert\Psi_{p}^{0}\right\Vert_{\mathcal{C}_{N}^{\beta}}.
\end{equation}

\textbf{Step 2: Verify the decay rate assumption \eqref{eq: decay rate assumption}} Due to the energy estimate in Theorem~\ref{thm:a priori energy estimate}, we have that $(\xi(s)r_{p},\xi(s)\Psi_{p},\xi(s)m_{p})$ satisfies the assumption \eqref{eq: decay rate assumption} for any $\vert\eta\vert\leq\frac{1}{10}$.

\textbf{Step 3: Near-axis decay}
Now, we can apply Proposition \ref{prop: near axis decay} since the inhomogeneous terms satisfy the assumption \eqref{eq: assumption on the inhomogeneous terms} and the solution $(\xi(s)r_{p},\xi(s)\Psi_{p},\xi(s)m_{p})$ satisfies the decay rate \eqref{eq: decay rate assumption}. Hence, we can close the proof of Theorem \ref{thm: linearized result} in the near-axis region.

\textbf{Step 4: Decay in the whole region}
Finally, note that the decay rates in the near-axis region satisfy the assumption \eqref{eq: good assumption on the axis}. Hence, using Proposition \ref{prop: near axis decay} and Proposition \ref{prop: key to close the linearized result}, we can conclude the proof of Theorem \ref{thm: linearized result}.

\textbf{Step 5: Decay result for the threshold initial data} Now consider the linearized Einstein-scalar field equations $\mathcal{P}(r_{p},\Psi_{p},m_{p}) = 0$ with the initial data $(r_{p},\Psi_{p})|_{\Sigma_{-1}^{(in)}} = (r_{p}^{0},\Psi_{p}^{0})\in\mathcal{C}^{\widetilde{\beta}}_{N}\times \mathcal{C}_{N}^{p_{k},\delta}$, where $p_{k}+\frac{1}{4}<\widetilde{\beta}<2$ and $0<\delta\leq \delta_{0}$. By Lemma \ref{lemma: decomposition of the function spaces}, there exists a function $\bar{\Psi}_{p}^{0}\in\mathcal{C}_{N}^{p_{k}+\delta^{\prime}}$ with $0<\delta^{\prime}<\min\{\beta-\frac{1}{4}-p_{k},\delta\}$ such that $\Psi_{p}^{0} = cr_{k}\mr{\phi}+\bar{\Psi}_{p}^{0}$. Since the Einstein-scalar field equations are scaling invariant, we have that $(r_{k},\Psi_{k},m_{k})$ is a solution to the linearized equations $\mathcal{P} = 0$. Therefore, it is equivalent to consider $\mathcal{P}(r_{p}-cr_{k},\Psi_{p}-c\Psi_{k},m_{p}-cm_{k}) = 0$ with the initial data $(r_{p},\Psi_{p})|_{\Sigma^{(in)}_{-1}}= (r_{p}^{0}-cr_{k},\bar{\Psi}_{p}^{0})$. Using the argument in the previous steps again, we can conclude the proof of Theorem~\ref{thm: linearized result at the threshold}.

\end{proof}

\bibliographystyle{plain}
\bibliography{references}

\end{document}